\newif\iffull
\fulltrue

\iffull
\documentclass[11pt,letterpaper]{article}
\usepackage[notes=true,later=true]{dtrt}
\usepackage{fullpage}
\else
\documentclass[runningheads]{llncs}
\usepackage[notes=xxx,llncssubsub]{dtrt}
\fi

\iffull
\usepackage[backend=bibtex8,style=alphabetic,maxnames=8,maxalphanames=6,sorting=anyt]{biblatex}
\bibliography{references}
\else
\usepackage[
style=alphabetic,
firstinits,
backend=bibtex,
maxalphanames=4,
minalphanames=3,
maxbibnames=15,
]{biblatex}
\AtEveryBibitem{%
\ifentrytype{inproceedings}{%
\clearfield{year}%
\clearfield{pages}%
\clearfield{booktitle}%
\clearfield{note}%
}{%
}%
}
\AtEveryBibitem{%
\ifentrytype{misc}{%
\clearfield{year}%
}{%
}%
}
\AtEveryBibitem{%
\ifentrytype{article}{%
\clearfield{pages}%
\clearfield{volume}%
\clearfield{number}%
\clearfield{month}%
}{%
}%
}

\DeclareFieldFormat[misc]{title}{\mkbibquote{#1\MLInpSizdot}}
\bibliography{references}
\fi

\iffull
\usepackage{amsthm}
\else

\let\endproof\relax
\usepackage{amsthm}
\fi
\usepackage{amssymb,amsfonts,amsmath}
\usepackage{ifthen}
\usepackage{bm,bbm}
\usepackage{times}
\usepackage{microtype}
\usepackage{tikz}
\usetikzlibrary{matrix}
\usetikzlibrary{shapes,arrows}
\usepackage{graphicx}
\usepackage{verbatim}
\usepackage{array}
\usepackage{multirow}
\usepackage{latexsym}
\usepackage{paralist}
\usepackage{enumitem}
\setlist[itemize]{leftmargin=*}
\setlist[enumerate]{leftmargin=*}
\usepackage[capitalise]{cleveref}
\crefname{step}{Step}{Steps}
\usepackage{dsfont}
\usepackage{url}
\usepackage{braket}
\usepackage{mathrsfs}
\usepackage{color}
\usepackage{soul}
\usepackage{mdframed}
\usepackage{complexity}
\usepackage{mathtools}
\usepackage{bbm}
\usepackage{setspace}
\usepackage{tabu}
\usepackage{multirow}
\usepackage{adjustbox}
\usepackage{booktabs}
\iffull
\usepackage[margin=4mm,small,labelfont=bf]{caption}
\fi

\newtheorem{itheorem}{Theorem}
\iffull
\newtheorem{theorem}{Theorem}[section]
\newtheorem{corollary}[theorem]{Corollary}

\newtheorem{definition}[theorem]{Definition}
\newtheorem{lemma}[theorem]{Lemma}
\newtheorem{claim}[theorem]{Claim}

\crefname{claim}{claim}{claims}
\Crefname{claim}{Claim}{Claims}
\else
\spnewtheorem{numberedclaim}{Claim}{\itshape}{\rmfamily}
\crefname{numberedclaim}{claim}{claims}
\Crefname{numberedclaim}{Claim}{Claims}
\fi

\theoremstyle{definition} 
\iffull

\newtheorem{remark}[theorem]{Remark}
\fi

\newtheorem{construction}[theorem]{Construction}

\allowdisplaybreaks
\newcommand{\FormatAuthor}[3]{
\begin{tabular}{c}
#1 \\ {\small\texttt{#2}} \\ {\small #3}
\end{tabular}
}
\newcommand{\doclearpage}{%
\iffull
\clearpage
\fi
}
\newcommand{\DoQuote}[1]{``#1''}
\newcommand{\defemph}[1]{\textbf{\emph{#1}}}
\newcommand{\keywords}[1]{\bigskip\par\noindent{\footnotesize\textbf{Keywords\/}: #1}}

\tikzstyle{source} = [rectangle, draw, fill=gray!20,
text width=10em, text centered, minimum height=8em]

\tikzstyle{protocol} = [rectangle, draw, fill=white!20,
text width=10em, text centered, minimum height=8em]

\tikzstyle{transformation} = [rectangle, draw, fill=white!20,
text width=8em, text centered, rounded corners, minimum height=6em]

\tikzstyle{line} = [draw, -latex']

\DeclareSymbolFont{bbold}{U}{bbold}{m}{n}
\DeclareMathSymbol{\bbpi}{\mathord}{bbold}{"19}

\newcommand{\Bits}{\{0,1\}}
\newcommand{\N}{\mathbb{N}}
\newcommand{\abs}[1]{\left\lvert{#1}\right\rvert}
\newcommand{\HammingWeight}[1]{{\sf hm}\left(#1\right)}
\newcommand{\setcomplement}[1]{\overline{#1}}

\newcommand{\id}[1]{I_{#1}}
\newcommand{\matnorm}[1]{\| #1 \|}
\newcommand{\vecnorm}[1]{\| #1 \|}
\newcommand{\dotp}[2]{\langle #1,#2 \rangle}
\newcommand{\accessVectorAt}[2]{#1[#2]}

\newcommand{\prob}[1]{\mathrm{Pr}\left[ #1\right]}
\newcommand{\negl}[1]{\mathsf{negl}(#1)}

\newcommand{\UniformFrom}[1]{{\cal{U}}(#1)}

\newcommand{\MixedState}[1]{{{\bm #1}}}
\newcommand{\reg}[2]{\ifthenelse{\equal{#1}{0}}{{\scriptscriptstyle{\mathrm{#2}}}}{{\mathrm{#2}}}}
\newcommand{\unitary}[1]{#1}

\renewcommand{\braket}[2]{\langle#1|#2\rangle}
\newcommand{\ketbra}[2]{\ket{#1}\!\bra{#2}}

\newcommand{\MixedStateSampleOne}{\MixedState{\rho}}

\newcommand{\PureStateSampleOne}{\ket{\phi}}
\newcommand{\PureStateSampleOneI}[1]{\ket{\phi_{#1}}}
\newcommand{\PureStateSampleTwo}{\ket{\psi}}
\newcommand{\PureStateSampleTwoI}[1]{\ket{\psi_{#1}}}

\newcommand{\braketPureStateSampleOne}{\braket{\phi}{\phi}}

\newcommand{\QuantumQueryMassFunc}[1]{\mathsf{mass}(#1)}
\newcommand{\QuantumTotalQueryMassFunc}[1]{\mathsf{TotalMass}(#1)}
\newcommand{\QuantumTotalQueryMass}{w}

\newcommand{\Commutator}[2]{\left[#1, #2\right]}
\newcommand{\OperatorSampleOne}{A}
\newcommand{\OperatorSampleTwo}{B}
\newcommand{\OperatorSampleThree}{C}

\newcommand{\Security}{\lambda}
\newcommand{\NumberOfRepetition}{\eta}

\newcommand{\Malicious}[1]{\widetilde{#1}}

\newcommand{\Relation}{R}
\newcommand{\Language}{L}
\newcommand{\Instance}{\mathbbm{x}}
\newcommand{\InstanceSize}{\nu}
\newcommand{\Input}{x}

\newcommand{\GetLanguage}[1]{\Language(#1)}
\newcommand{\RelationDecider}{{\sf Decider}}
\newcommand{\QuantumWitness}{\MixedState{\rho}}
\newcommand{\Alphabet}{\Sigma}

\newcommand{\Oracle}[1]{\ifthenelse{\equal{#1}{0}}{\scriptscriptstyle{f}}{f}}

\newcommand{\RandomOracle}[1]{\ifthenelse{\equal{#1}{0}}{\scriptscriptstyle{{\sf RO}}}{{\sf RO}}}
\newcommand{\DatabaseRegister}[1]{\reg{#1}{D}}
\newcommand{\DatabaseRegisterAt}[2]{\reg{#1}{\accessVectorAt{D}{#2}}}

\newcommand{\RandomOracleOutputLength}{m}
\newcommand{\RandomOracleDomain}{X}
\newcommand{\QueryRegister}[1]{\reg{#1}{X}}
\newcommand{\AnswerRegister}[1]{\reg{#1}{Y}}
\newcommand{\OracleUnitary}{{O}}
\newcommand{\SizeOfDatabase}[1]{|#1|}

\newcommand{\compress}{{\unitary{F}}}
\newcommand{\todatabase}{\unitary{U_\mathsf{ToDB}}}

\newcommand{\ProjectNoCollision}{\Pi_{\scriptstyle{\sf NoCol}}}
\newcommand{\NoCollision}{\ProjectNoCollision}
\newcommand{\ProjectSizeDatabase}[1]{\Pi_{#1}}

\newcommand{\ProjectNoHatZero}{\Pi_{\sf NoHatZero}}
\newcommand{\ProjectWhereNonEmpty}[1]{\Delta_{#1}}

\newcommand{\SetOfDatabase}{S_D}
\newcommand{\SetOfNoCollisionDatabase}{S_{\sf NoCollision}}
\newcommand{\SetOfDatabaseSizeAtMost}[1]{S_{#1}}

\newcommand{\Image}[1]{{\sf Im}(#1)}

\newcommand{\HadamardGate}{\unitary{H}}
\newcommand{\CNOT}{{\unitary{\sf CNOT}}}
\newcommand{\TGate}{\unitary{T}}
\newcommand{\SWAP}{{\unitary{\mathsf{SWAP}}}}
\newcommand{\SWAPI}[2]{\SWAP\left[#1, #2\right]}
\newcommand{\PauliX}{{\unitary{\sf X}}}

\newcommand{\Circuit}{C}
\newcommand{\Algorithm}{{\sf A}}
\newcommand{\AlgorithmI}[1]{\Algorithm^{(#1)}}

\newcommand{\OracleParty}{\mathbf{O}}

\newcommand{\Counter}{N}
\newcommand{\QVCPathSize}{T}

\newcommand{\CM}{\mathsf{CM}}

\newcommand{\NoCollisionVariant}[1]{{#1}^*}
\newcommand{\NoCollisionAlgVariant}[1]{{#1}_{\scriptscriptstyle{\text{NoCollision}}}}

\newcommand{\SimulateWorld}{{\sf SimWorld}}
\newcommand{\ExtractWorld}{{\sf ExtWorld}}
\newcommand{\OfflineExtractWorld}{{\sf OffExtWorld}}

\newcommand{\HybridIJ}[2]{{\sf H}_{#1}^{(#2)}}
\newcommand{\HybridIJPrime}[2]{{\sf H'}_{#1}^{(#2)}}
\newcommand{\HybridI}[1]{{\sf H}_{#1}}

\newcommand{\NumberOfQueries}{t}
\newcommand{\NumberOfQueriesBound}{8\QVCMessageLength}
\newcommand{\MessageLength}{n}
\newcommand{\ExtractionError}{\xi_{\scriptscriptstyle{\mathsf{Ext}}}}
\newcommand{\ExtractionErrorI}[1]{\ExtractionError^{(#1)}}

\newcommand{\Commit}{\mathsf{Com}}
\renewcommand{\Check}{\mathsf{Check}}

\newcommand{\CommitCircuit}{C}

\newcommand{\Simulator}{U_{\scriptscriptstyle{\mathsf{Sim}}}}
\newcommand{\SimulationError}{\xi_{\scriptscriptstyle{\mathsf{Sim}}}}
\newcommand{\RO}{{\mathsf{RO}}}

\newcommand{\Extractor}{{\cal E}}
\newcommand{\ExtractOracle}{U_{\scriptscriptstyle{\mathsf{Ext}}}}
\newcommand{\Adversary}{{A}}
\newcommand{\AdversaryPrime}{{B}}
\newcommand{\Distinguisher}{{D}}
\newcommand{\AltCheck}{{\mathsf{AltCheck}}}
\newcommand{\ExtractMessage}{{\mathsf{ExtMsg}}}

\newcommand{\MessageRegister}[1]{\reg{#1}{M}}
\newcommand{\CommitmentRegister}[1]{\reg{#1}{C}}
\newcommand{\OpeningRegister}[1]{\reg{#1}{O}}
\newcommand{\AltOpeningRegister}[1]{\reg{#1}{Aux}}

\newcommand{\StandardBasisHashValueRegister}[1]{\reg{#1}{Z}_1}
\newcommand{\HadamardBasisHashValueRegister}[1]{\reg{#1}{Z}_2}
\newcommand{\StandardBasisWorkingRegister}[1]{\reg{#1}{W}_1}
\newcommand{\HadamardBasisWorkingRegister}[1]{\reg{#1}{W}_2}
\newcommand{\StandardBasisMessageRegister}[1]{\reg{#1}{M}_1}
\newcommand{\HadamardBasisMessageRegister}[1]{\reg{#1}{M}_2}
\newcommand{\StandardBasisSuccessRegister}[1]{\reg{#1}{B}_1}
\newcommand{\HadamardBasisSuccessRegister}[1]{\reg{#1}{B}_2}
\newcommand{\FlagRegister}[1]{\reg{#1}{B}}
\newcommand{\AncillasRegister}[1]{\reg{#1}{A}}
\newcommand{\EnvironmentRegister}[1]{\reg{#1}{E}}

\newcommand{\WorkingRegister}[1]{\reg{#1}{W}}
\newcommand{\TargetRegister}[1]{\reg{#1}{T}}
\newcommand{\StateRegister}[1]{\reg{#1}{S}}

\newcommand{\InternalStateRegister}{\reg{1}{I}}
\newcommand{\AdvInternalStateRegister}{\reg{1}{J}}
\newcommand{\OutputRegister}{\reg{1}{O}}

\newcommand{\ProverMessageRegister}[1]{\reg{1}{Mp}_{#1}}
\newcommand{\VerifierMessageRegister}[1]{\reg{1}{Mv}_{#1}}
\newcommand{\ProverPrivateRegister}[1]{\reg{1}{P}_{#1}}
\newcommand{\VerifierPrivateRegister}[1]{\reg{1}{V}_{#1}}
\newcommand{\QueryLocationRegister}{\reg{1}{Loc}}

\newcommand{\QuantumMessage}{\MixedState{\rho}}
\newcommand{\CommitAndOpeningState}{\MixedState{\sigma}}
\newcommand{\QuantumMessageAndAncillas}{\MixedState{\varrho}}

\newcommand{\CheckStateEPR}{\phi_{\scriptscriptstyle{\mathrm{EPR}}}}
\newcommand{\NullStateInExp}{\psi}

\newcommand{\ValidityBit}{b}
\newcommand{\ValidityBitI}[1]{\ValidityBit_{#1}}
\newcommand{\UnitaryInExp}{U}
\newcommand{\Another}[1]{#1'}
\newcommand{\Span}[1]{\mathrm{span}\{#1\}}

\newcommand{\IBCS}{\mathsf{IBCS}}

\newcommand{\Transformation}{\mathbb{T}}

\newcommand{\QVC}{\mathsf{QSVC}}
\newcommand{\QSTC}{\mathsf{QSTC}}

\newcommand{\NumberOfCommitmentsPhases}{p}

\newcommand{\QVCCommit}{\QVC.\Commit}

\newcommand{\Query}{\mathsf{Query}}
\newcommand{\QVCQuery}{\QVC.\Query}
\newcommand{\Open}{\mathsf{Open}}
\newcommand{\QVCOpen}{\QVC.\Open}
\newcommand{\Recover}{\mathsf{Recover}}
\newcommand{\QVCRecover}{\QVC.\Recover}
\newcommand{\Update}{\mathsf{Update}}
\newcommand{\QVCUpdate}{\QVC.\Update}
\newcommand{\QVCTuple}{(\Commit, \Open, \Query, \Update, \Recover)}
\newcommand{\QVCMessageRegister}{\reg{1}{M}}

\newcommand{\QVCCommitmentRegister}{\reg{1}{C}}
\newcommand{\QVCOpeningRegister}{\reg{1}{O}}
\newcommand{\QVCAuxiliaryRegister}[1]{\reg{#1}{T}}
\newcommand{\QVCAuxiliaryNotUsedRegister}{\reg{1}{R}}

\newcommand{\QVCMessageLength}{n}
\newcommand{\QVCMessageDepth}{d}
\newcommand{\QVCValidityBit}{b}
\newcommand{\QVCBlockSize}{s}

\newcommand{\indexForMTNode}{\ell}

\newcommand{\emptystring}{\varepsilon}
\newcommand{\QSTCPath}[1]{\mathsf{Path}(#1)}
\newcommand{\QSTCAuthPath}[1]{\mathsf{AuthPath}(#1)}
\newcommand{\sibling}[1]{\mathsf{Sib}(#1)}
\newcommand{\parent}[1]{\mathsf{Par}(#1)}

\newcommand{\QVCQuerySet}{{\mathcal{Q}}}
\newcommand{\QVCQuerySetRegister}[1]{\reg{#1}{Q}}

\newcommand{\AltCommit}{\mathsf{AltCommit}}

\newcommand{\QVCSimulateWorld}{{\sf QSVCSimWorld}}
\newcommand{\QVCExtractWorld}{{\sf QSVCExtWorld}}
\newcommand{\QVCOfflineExtractWorld}{{\sf QSVCOffExtWorld}}
\newcommand{\QVCSimWorldSC}{{\sf QSVCSimWorld1Open}}
\newcommand{\QVCOfflineExtractWorldSC}{{\sf QSVCOffExtWorld1Open}}

\newcommand{\invert}{{\mathsf{Inv}}}
\newcommand{\PurifiedInvY}{M}
\newcommand{\InvY}{\Sigma}
\newcommand{\target}{y}

\newcommand{\QueryUnitary}{\UnitaryInExp^{\sf qry}}

\newcommand{\ArgVerifier}{{\cal{V}}}
\newcommand{\ArgProver}{{\cal{P}}}
\newcommand{\ArgAdv}{\Malicious{\ArgProver}}
\newcommand{\Arg}{\mathsf{QARG}}

\newcommand{\ArgCompleteness}{c_{\scriptscriptstyle{\Arg}}}
\newcommand{\ArgSoundness}{s_{\scriptscriptstyle{\Arg}}}

\newcommand{\QIOP}{{\sf QIOP}}

\newcommand{\QIOPProver}{{\bf P}}
\newcommand{\QIOPVerifier}{{\bf V}}
\newcommand{\QIOPAdv}{\Malicious{\QIOPProver}}

\newcommand{\QIOPRoundComplexity}{{\sf k}}
\newcommand{\QIOPProofTotalSize}{{\sf l}}
\newcommand{\QIOPProofSize}[1]{{\sf l}_{#1}}
\newcommand{\QIOPProofMaxSize}{{\sf l}_{\max}}
\newcommand{\QIOPCompleteness}{c}
\newcommand{\QIOPSoundness}{s}
\newcommand{\pq}{\scriptscriptstyle{\mathsf{pq}}}
\newcommand{\Leaked}[3]{{\mathsf{Leaked}}(#1, #2, #3)}
\newcommand{\QIOPPublicQuerySoundness}{s_{\pq}}
\newcommand{\QIOPQueryComplexity}{{\mathsf{q}}}
\newcommand{\QIOPQueryDepth}{{\mathsf{d}}}
\newcommand{\QIOPQueryWidth}{{\mathsf{w}}}
\newcommand{\QIOPQueryWidthI}[1]{{\mathsf{w}}_{#1}}
\newcommand{\QIOPQueryComplexityI}[1]{{\mathsf{q}}_{#1}}
\newcommand{\QIOPQueryDepthI}[1]{{\mathsf{d}}_{#1}}
\newcommand{\QIOPVerifierMsgSizeI}[1]{{\mathsf{vc}}_{#1}}
\newcommand{\QIOPVerifierMsgSize}{{\mathsf{vc}}}

\newcommand{\QIOPReturnIdxSet}[1]{J_{#1}}
\newcommand{\QIOPNotReturnIdxSet}[1]{I_{#1}}

\newcommand{\IBCSRoundComplexity}{k_{\scriptscriptstyle{\IBCS}}}
\newcommand{\IBCSVerifierMsgSize}{{\mathsf{vc}}_{\scriptscriptstyle{\IBCS}}}
\newcommand{\IBCSProverMsgSize}{{\mathsf{pc}}_{\scriptscriptstyle{\IBCS}}}

\newcommand{\IBCSVerifierQueryComplexity}{{\mathsf{vq}}}
\newcommand{\IBCSVerifierQueryComplexityI}[1]{{\mathsf{vq}_{#1}}}

\newcommand{\ErrorBoundForSingleCommitOffline}{\xi_{\mathsf{1Open}}}

\begin{document}

\title{
Succinct Arguments for QMA \\ in the Quantum Random Oracle Model
}
\iffull
\author{
\begin{tabular}[h!]{ccc}
\FormatAuthor{Alessandro Chiesa}{alessandro.chiesa@epfl.ch}{EPFL}
\FormatAuthor{Zihan Hu}{zihan.hu@epfl.ch}{EPFL}
\end{tabular}
}
\date{\today}
\fi

\maketitle

\begin{abstract}

Succinct arguments are a fundamental cryptographic primitive for verifying computational claims with small communication. In the classical setting, succinct arguments for $\NP$ can be constructed from unstructured hardness alone (e.g., hash functions) by compiling probabilistically checkable proofs (PCPs) or interactive oracle proofs (IOPs) for $\NP$ via the commit-and-open paradigm. In contrast, known succinct arguments for $\QMA$ rely on ``structured'' cryptographic primitives, or on the quantum PCP conjecture.

We construct the first succinct argument for $\QMA$ in the quantum random oracle model (QROM) without relying on additional cryptographic assumptions or unproven conjectures. This yields succinct arguments for $\QMA$ from unstructured hardness alone, showing that ideal hash functions not only suffice for succinct arguments for $\NP$ but also for $\QMA$.

Underlying our result is an efficiency-preserving transformation that compiles quantum interactive oracle proofs (QIOPs), a recently introduced interactive generalization of quantum PCPs, into quantum arguments for the same language, via a natural quantum commit-and-open paradigm. Our transformation applies to every QIOP with public-query soundness, a notion that we formalize to capture a natural requirement of the commit-and-open paradigm and is satisfied by a known QIOP for $\QMA$. As a key ingredient in our transformation, we formalize and construct extractable vector commitments for quantum states with local openings in the QROM, which may be of independent interest.

\keywords{succinct arguments for $\QMA$; quantum random oracle model; quantum interactive oracle proofs}
\end{abstract}

\iffull
\thispagestyle{empty}
\clearpage
\setcounter{tocdepth}{2}
\begin{spacing}{1}
{\small\tableofcontents}
\end{spacing}
\thispagestyle{empty}
\clearpage
\setcounter{page}{1}
\fi

\doclearpage
\section{Introduction}
\label{sec:introduction}

Succinct arguments enable a verifier to check a computational claim using communication substantially smaller than the cost of checking the claim directly. A remarkable feature of the classical theory is that such arguments do not require public-key cryptography or other highly structured assumptions: suitable hash functions suffice. Indeed, the \emph{commit-and-open paradigm} compiles probabilistically checkable proofs and interactive oracle proofs for $\NP$ into succinct arguments using hash-based vector commitments \cite{Kil92,Mic00,BG08,BCS16,CDGS23}. This paradigm is also known to remain secure against quantum adversaries under unstructured post-quantum assumptions \cite{CMS19,CMSZ21,CDDGS25}.

These results naturally raise the question of whether succinct arguments based solely on unstructured hardness can be extended from $\NP$ to its quantum analogue, $\QMA$, which involves a quantum witness and can be verified in quantum polynomial time. In particular, \cite{BLM26} poses the following question:
\begin{center}
\emph{Are there succinct arguments for $\QMA$ based on hash functions, \\or any unstructured hardness?}
\end{center}

A line of work \cite{BKLMMVVY21,MNZ24,GKNV25,BLM26} constructs succinct arguments for $\QMA$, even with a classical verifier. However, these constructions rely on structured hardness assumptions (see \Cref{sec:related-work}).

A different approach, more closely aligned with our goal of succinct arguments for $\QMA$ from unstructured hardness, is to extend the classical commit-and-open paradigm to the fully quantum setting~\cite{CM24,GJMZ23}. In particular, \cite{GJMZ23} proposes a quantum analogue of Kilian's protocol \cite{Kil92} that compiles a quantum probabilistically checkable proof (QPCP) into a quantum interactive argument while preserving the efficiency of the QPCP. Crucially, \cite{GJMZ23} requires only unstructured hardness: any collapsing hash function, or more generally any succinct quantum state commitment (see \Cref{sec:related-work} for more details).

However, obtaining a succinct argument for $\QMA$ via their transformation requires a QPCP for $\QMA$ with polynomial length and small query complexity, whose existence is a weak form of the QPCP conjecture \cite{AALV09,AAV13} and remains a major open problem in quantum complexity theory.

In sum, succinct arguments for $\QMA$ have not matched their classical counterparts in terms of assumptions. Known constructions rely on structured cryptographic hardness, or on an unresolved complexity conjecture.

\subsection{Our results}
\label{sec:our-results}

We construct the first succinct argument for $\QMA$ in the quantum random oracle model (QROM) \cite{BonehDFLSZ11}, without relying on additional cryptographic assumptions or unsolved conjectures.

\begin{itheorem}
\label{ithm:succinct-argument-for-QMA}
In the QROM, there exists a succinct quantum interactive argument for $\QMA$ with communication complexity $\poly(\Security, \log \InstanceSize)$ for instance size $\InstanceSize$ and security parameter $\Security$.
\end{itheorem}

Here all parties, including malicious ones, have quantum query access to a uniformly sampled random oracle with output length $\Security$. Thus, our result establishes the feasibility of succinct arguments for $\QMA$ from an idealized form of unstructured hardness, without additional cryptographic assumptions; this answers the aforementioned open question of \cite{BLM26} in the oracle model. Establishing an analogous result in the plain model (from collapsing hash functions) is done in independent and concurrent work.\footnote{\cite{BartusekM26} constructs succinct arguments for $\QMA$ from (a non-black-box use of) collapsing hash functions in the standard model, via quantum-succinct blind delegation and the communication-compression compiler of \cite{BLM26}. Our approach is different: we give a quantum commit-and-open compiler for QIOPs and construct the extractable quantum-state vector commitments to instantiate it in the QROM. The two works independently resolve the feasibility question through different techniques and in different models.}

We obtain this result from a general transformation, which is our main technical contribution: a quantum analogue of the classical hash-based commit-and-open transformation. We compile any public-query quantum interactive oracle proof (QIOP) into a quantum interactive argument in the QROM while essentially preserving its efficiency. Instantiating the compiler with the recent QIOP for $\QMA$ in \cite{SV26} yields \Cref{ithm:succinct-argument-for-QMA}. Public-query soundness is inherent to this commit-and-open approach because the prover must learn the queried locations in order to open them; we formulate the appropriate quantum analogue further below.

\begin{itheorem}[\Cref{thm:quantum-IBCS-soundness}, informal]
There exists a transformation $\Transformation$ satisfying the following. Let $\QIOP$ be a QIOP for a relation $\Relation$, with round complexity $\QIOPRoundComplexity$, proof length $\QIOPProofTotalSize$, query complexity $\QIOPQueryComplexity$, and verifier-to-prover communication complexity $\QIOPVerifierMsgSize$. Then $\Arg \coloneq \Transformation[\QIOP]$ is a quantum interactive argument for $\Relation$ in the QROM with communication complexity $O(\Security \QIOPRoundComplexity + \Security \QIOPQueryComplexity \log \QIOPProofTotalSize+\QIOPVerifierMsgSize)$. If $\QIOP$ has public-query soundness error $\QIOPPublicQuerySoundness$, then $\Arg$ has soundness error $O(\QIOPPublicQuerySoundness +  \QIOPRoundComplexity \cdot \poly(\NumberOfQueries, \QIOPProofTotalSize,\QIOPQueryComplexity) \cdot {2^{-\Security}})$ against $\NumberOfQueries$-query quantum adversaries.
\end{itheorem}

Extending the commit-and-open paradigm to an appropriate notion of QIOP requires overcoming delicate definitional and technical challenges. In particular, we introduce a suitable notion of \emph{quantum-state vector commitment} (QSVC) with a strong extraction property, and prove that a natural quantum-state tree commitment (QSTC) in the QROM satisfies it. Below we elaborate on these challenges and results.

\parhead{Flavor of QIOP}
Quantum interactive oracle proofs (QIOPs) \cite{SV25} extend quantum probabilistically checkable proofs (QPCPs) to the interactive setting, just as classical IOPs extend PCPs \cite{BCS16,ReingoldRR16}. Several recent works consider different, sometimes incomparable, models of QIOPs \cite{SV25,SV26,CGMV26}. We formulate our compiler for a QIOP model that simultaneously supports the relevant features appearing across these constructions: superposition queries, the return of previously submitted quantum proof states, and adaptive queries made before those states are returned. This formulation also identifies expressive QIOP features supported by our QROM compiler, thereby providing a flexible target model for future QIOP constructions. See \Cref{sec:QIOP-detail} for details.
\begin{itemize}[noitemsep]
\item \textbf{Superposition queries.}
We equip the QIOP verifier with \emph{superposition query access} to the quantum proof oracles \cite{CGMV26}. Informally, the verifier can specify query sets in superposition in a quantum register and coherently swap the corresponding locations of the quantum proof oracle to designated registers held by the verifier. Superposition query access generalizes the \DoQuote{point access} used in the definition of QPCPs and some QIOP constructions~\cite{SV25,SV26}.
\item \textbf{Returning proof oracles.}
The verifier may return quantum proof oracles from earlier rounds to the prover after querying them, and the size of the returned registers is not charged to its query complexity or verifier-to-prover communication. This feature plays a key role in known QIOPs \cite{SV25, SV26, CGMV26}. In contrast to a classical prover, a quantum prover cannot generally retain a copy of a transmitted proof state. The verifier may therefore need to return a previously submitted proof oracle for use by the prover in later rounds.
\item \textbf{Adaptive queries and delayed returns.}
We allow the verifier to retain quantum proof oracles across multiple rounds and make \emph{adaptive} superposition queries to them in later rounds until they are returned to the prover. Delayed returns are used in the QIOP in \cite{SV26} (but not, e.g., the QIOP in \cite{CGMV26}).
\end{itemize}
Each feature creates an obstacle for commit-and-open compilation, addressed via the QSVC interface below.

\parhead{Public-query soundness}
The commit-and-open paradigm in the classical setting does not apply to every IOP: the argument prover learns the verifier's query locations to open the corresponding proof locations, so the IOP must satisfy \emph{public-query soundness}. That is, soundness must hold even when the verifier's queries are leaked to the malicious prover whenever they occur \cite{CDGS23}. We need an appropriate quantum formulation because the query locations may be in superposition and therefore cannot, in general, be revealed by measurement without disturbing the verifier's computation.

We introduce a definition of \emph{public-query soundness} for QIOPs, where soundness holds even when the malicious prover has access to the query location register immediately before and after each superposition query. That is, immediately before the query, the location registers are transferred to the malicious prover, who may act jointly on them and its private state before returning them. The same interaction occurs immediately after the query. Public-query soundness requires the original soundness guarantee to continue to hold even with this additional interface. This notion extends classical public-query soundness \cite{CDGS23}, and is satisfied by the efficient QIOP in~\cite{SV26} (see more in \Cref{sec:QIOP-we-use}). See \Cref{def:pq-qiop} for the formal definition.

\parhead{Quantum-state vector commitment schemes}
At a minimum, achieving succinctness in the compiled protocol requires a commitment scheme that produces a short commitment to a long quantum state and provides efficient openings of a few locations. In the classical setting, succinct commitment schemes with such local openings are known as vector commitments, whose standard commitment, opening, and verification procedures suffice for the commit-and-open paradigm \cite{BCS16,CDGS23,CDDGS25}.

In the quantum setting, this standard vector-commitment interface does not suffice when query locations may be in superposition, even when formulated for quantum messages.
\begin{itemize}[nolistsep]
  \item In a Merkle opening, the prover selects authentication-path registers according to query locations. If those locations are in superposition, the registers retained by the prover become entangled with the verifier's location register. Unless this which-path information is coherently erased, the location register is dephased, and even an honest execution of the compiled protocol may not reproduce the QIOP verifier's computation.
  \item A distinct problem arises when the QIOP verifier returns a previously submitted proof oracle. The compiled verifier holds only a succinct commitment to that state, so the commitment scheme must allow the prover to recover the updated underlying state after the verifier has queried it.
\end{itemize}
Prior work \cite{GJMZ23} defines quantum-state commitments for compiling QPCPs with classical (non-superposition) queries. They do not provide a standalone quantum-state vector-commitment interface, nor provide the additional procedures that we require for coherent openings and returned proof states.

In light of this, we provide a new definition of \emph{quantum-state vector commitment} (QSVC) in the QROM. Our definition incorporates key additional procedures:
\begin{inparaenum}[(a)]
  \item an $\Update$ procedure to allow the honest prover to coherently erase its record of the query locations;
  \item a $\Recover$ procedure to allow the honest prover to recover the underlying quantum message when the verifier returns a commitment, thereby supporting the return of proof oracles in a QIOP.
\end{inparaenum}
These procedures (see \Cref{subsec:def-QVC-syntax} for details) ensure \emph{completeness} of the resulting succinct argument for the aforementioned broad class of QIOPs. We also need a suitable \emph{security} notion, discussed next.

\parhead{Online quantum extraction of quantum messages}
We formulate a notion of \emph{online extractability} for QSVCs in the QROM. Informally, a QSVC scheme is extractable if there is an efficient extractor, with access to the compressed-oracle database, such that the following two ways of implementing superposition query access to the committed quantum message are indistinguishable. (Both implementations abort if any validity check fails.)
\begin{itemize}[nolistsep]
\item \emph{Real access.} The receiver verifies the supplied local openings and performs the desired superposition query through the commitment.
\item \emph{Extracted access.} Before seeing any openings, the extractor recovers a complete quantum state associated with the commitment. It later validates each opening, performs the same query directly on the extracted state, and reconstructs the updated commitment.
\end{itemize}
Crucially, extraction occurs online and before the sender chooses its openings. Formalizing the above intuition is subtle. The extractor must check the validity of the openings after the commitment has already been used to extract the message; see~\Cref{def:qvc_extractability} for the formal definition.

Our online extractability notion enables a straightline reduction from the soundness of the compiled quantum argument to the public-query soundness of the underlying QIOP. Given a malicious argument prover, the reduction constructs a malicious QIOP prover by using the QSVC extractor to recover the committed proof states. When the QIOP verifier returns a proof state after querying it, the reduction reconstructs a corresponding commitment and resumes the simulation. The swap-binding notion for quantum-state commitments in \cite{GJMZ23} does not provide the online guarantee required for such a straightline reduction.

\parhead{Extractable quantum-state tree commitment}
We construct an extractable QSVC in the QROM (unconditionally) where a commitment has size $\poly(\Security)$ and an opening proof has size $\poly(\Security, q, \log \QVCMessageLength)$, for message length $\QVCMessageLength$ and a superposition query where each branch is for at most $q$ locations. (This efficiency is comparable to that of the Merkle commitment scheme in the ROM.) The extraction error is as follows.

\begin{itheorem}[\Cref{thm:vc-extractability}, informal]
\label{ithm:vc-extractability}
There exists a QSVC in the QROM (with the above efficiency and) with extraction error $\ExtractionError = O(\NumberOfQueries^3 \cdot \poly(\QVCMessageLength, \Security) \cdot 2^{-\Security})$ against $\NumberOfQueries$-query quantum adversaries that ask a single opening. Here $\QVCMessageLength$ is the length of the message and $\Security$ is the security parameter (output length of the oracle).
\end{itheorem}

For simplicity, the theorem states the extraction error for a single opening. More generally, our formal definition of extraction supports multiple openings (see \Cref{def:qvc_extractability}) to handle QIOPs with adaptive queries across rounds, and our analysis directly establishes an extraction error for this adaptive query setting.

The underlying construction is a \emph{quantum-state tree commitment} ($\QSTC$). We combine a Merkle tree structure with a basic extractable succinct quantum-state commitment (see \Cref{sec:def-basic-extractable-commitment,sec:construction-basic-extractable-commitment} for details). This follows the structure in \cite{GJMZ23} with two differences:
\begin{inparaenum}[(i)]
  \item a random oracle replaces the cryptographic hash function;
  \item our construction additionally includes the extra procedures required by our QSVC definition.
\end{inparaenum}

Informally, at each tree node, the basic commitment coherently evaluates the random oracle on the computational-basis labels of the message register and, separately, on its Hadamard-basis labels. The two resulting image registers form the commitment, while the message register is retained as the opening. Using the compressed-oracle database, the extractor coherently recovers preimages associated with both image registers. Opening verification can then be related to an EPR-type consistency test between the two recovered descriptions. Excluding branches containing a random-oracle collision (and other small probability events), this test ensures that extracted access behaves like honest recovery of the committed state. We compose this basic commitment in a Merkle tree and prove that extraction remains secure under local, adaptive openings.

Our analysis of this construction addresses online extraction, coherent openings, and adaptive queries; all of these are not considered in \cite{GJMZ23}. See \Cref{construction:quantum-state-vector-commitment} for the full construction.

\subsection{Related work}
\label{sec:related-work}

\parhead{Hash-based succinct arguments for $\NP$}
The commit-and-open paradigm \cite{Kil92, Mic00, BCS16, CDGS23} originated with Kilian's protocol \cite{Kil92} and was later generalized from PCPs to IOPs \cite{BCS16}; it has become a standard approach to constructing hash-based succinct arguments for $\NP$. Using a hash-based vector commitment scheme (typically a Merkle commitment scheme), these transformations map a public-coin IOP into a succinct public-coin interactive argument. The prover replaces each long proof string with a succinct commitment and subsequently opens only the locations queried by the verifier. (To obtain a non-interactive argument, one further applies the Fiat--Shamir transformation.) Security of the hash-based commit-and-open paradigm is studied in two settings:
\begin{inparaenum}[(i)]
\item \emph{in the plain model}, classical security is established in \cite{CDGS23} based on collision-resistant hash functions and post-quantum security is established in \cite{CDDGS25} based on collapsing hash functions; and
\item \emph{in the random oracle model}, classical security is established in \cite{BCS16,ChiesaYogev2024} and post-quantum security is established in \cite{CMS19,CDHZ26}.
\end{inparaenum}

\parhead{Succinct arguments for $\QMA$}
Succinct arguments for $\QMA$ can be constructed under various assumptions. A line of work constructs such arguments with classical verifiers based on either Mahadev's measurement protocol \cite{BKLMMVVY21,Faisal24,GKNV25}\footnote{\cite{Faisal24} additionally assumes a variant of the QPCP conjecture and constructs a classical-verifier succinct argument for $\QMA$ with non-succinct setup in the QROM.} or the compiled non-local game paradigm \cite{MNZ24,BLM26}. Most recently, \cite{BLM26} shows the feasibility of succinct arguments for $\QMA$ from oblivious state preparation (OSP) and collapsing hash functions. Since OSP can be based on plain trapdoor claw-free functions, their result is the first such construction that does not inherently rely on the hardness of learning with errors. No construction of OSP from purely unstructured assumptions is currently known.

A different approach, more closely related to ours, applies the commit-and-open paradigm directly to quantum proofs. \cite{CM24} proposes a candidate transformation that compiles a QPCP into a succinct interactive quantum argument by committing to the quantum proof succinctly and later opening only the locations queried by the verifier, and \cite{GJMZ23} constructs another transformation following the same paradigm and proves the security of their transformation. Neither transformation relies on structured cryptographic assumptions; we discuss the assumptions underlying their commitment schemes below. However, obtaining a succinct argument for $\QMA$ through their approaches requires an efficient QPCP for $\QMA$, whose existence remains a conjecture.

\parhead{Commitments to quantum states}
Two prior works study quantum-state vector commitments by adapting the Merkle tree paradigm to quantum states. \cite{CM24} proposes a candidate construction called the quantum Merkle tree in the quantum Haar random oracle model, using a Haar random unitary to compress the quantum state at each level of the tree until we reach the root, without proving the security of the candidate construction against a malicious sender. \cite{GJMZ23} studies how to obtain a quantum-state vector commitment in the plain model by first constructing a succinct quantum-state commitment (without local opening) and composing them using a Merkle tree. Their construction can be based on collapsing hash functions and even on potentially weaker quantum cryptographic assumptions. Neither work formalizes quantum-state vector commitments as a standalone cryptographic primitive with a dedicated security definition.

A related notion is the classical commitment to quantum states in \cite{GKNV25}. Their scheme allows a quantum sender to produce a classical commitment to a quantum state and later provide classical openings corresponding to measurements of selected qubits in either the computational or Hadamard basis. Openings for measurement results in these two bases do not suffice for our application.

\parhead{Extractable commitments}
Online extractability requires that an extractor, given some trapdoor information, can recover the underlying message from the commitment without rewinding the sender. \cite{Pas03} shows that the basic hash-based commitment (\DoQuote{hash the message}) is extractable in the ROM. \cite{DFMS22} prove the same scheme is post-quantum extractable in the QROM; given a classical commitment $y$ from a quantum adversary, it is possible to recover the message $x$ given the quantum database obtained via the compressed oracle technique \cite{Zha19}. \cite{CDHZ26} further strengthens this notion by showing that such an extraction can be applied coherently to commitments $y$ appearing in superposition within the adversary's random oracle queries without being detected by the adversary. These works consider only commitments to classical messages and therefore do not directly provide extractability for commitments to quantum states.

\doclearpage
\section{Preliminaries}
\label{sec:prelim}

For every non-negative integer $n$, we use $[n]$ to denote the set $\{1, 2, \cdots, n\}$. In particular, $[0]$ denotes the empty set. We denote the empty string by $\emptystring$. Let $\Bits^{< \ell}, \Bits^{\leq \ell}$ be the sets of strings of length less than $\ell$ and no more than $\ell$, respectively.

For every game ${\sf G}$, we use $\prob{{\sf G}}$ to denote the probability that the output of game ${\sf G}$ is 1.

We use $\HammingWeight{z}$ to denote the Hamming weight of a vector $z$ (the number of non-zero elements in $z$).

\subsection{Quantum states, operators, and circuits}

We use the standard bra-ket notation. We abbreviate the tensor product $\ket{0}^{\otimes n}$ as $\ket{0^n}$, or simply $\ket{0}$ when the number of qubits is clear from context.

A \emph{register} is a named finite-dimensional complex Hilbert space. We use serif font, for example, $\reg{1}{A}$, to represent registers. For registers $\reg{1}{A}, \reg{1}{B}, \reg{1}{C}$, the concatenation $\reg{1}{A}\reg{1}{B}\reg{1}{C}$ is the tensor product of the associated Hilbert spaces. Sometimes we need to group several registers into one register, in which case, we use notations like $\reg{1}{A} \coloneq (\reg{1}{B}, \reg{1}{C})$. We also often divide a register into several registers, in which case, we use notations like $(\reg{1}{A}, \reg{1}{B}) \coloneq \reg{1}{C}$ when the size of each register is clear from the context. For a linear transformation $L$ and a quantum state $\MixedStateSampleOne$, we sometimes add a subscript $\reg{1}{R}$ to them to emphasize that $L_{\reg{0}{R}}$ is acting on the register $\reg{1}{R}$, and $\MixedStateSampleOne_{\reg{0}{R}}$ is a quantum state in the register $\reg{1}{R}$. We denote the identity transformation over a register $\reg{1}{R}$ as $\id{\reg{0}{R}}$.

For a vector $\PureStateSampleOne$, we write $\vecnorm{\PureStateSampleOne}$ to denote its $\ell_2$ norm $\vecnorm{\PureStateSampleOne} \coloneq \sqrt{\braketPureStateSampleOne}$. For a linear transformation $L$, the operator norm of $L$ is denoted by $\matnorm{L}\coloneq \max_{\PureStateSampleOne}\vecnorm{L\PureStateSampleOne}$ where the max is among all the vectors of norm 1. Then for two operators $\OperatorSampleOne$ and $\OperatorSampleTwo$, $\matnorm{\OperatorSampleOne + \OperatorSampleTwo} \leq \matnorm{\OperatorSampleOne} + \matnorm{\OperatorSampleTwo}$ and $\matnorm{\OperatorSampleOne\OperatorSampleTwo} \leq \matnorm{\OperatorSampleOne}\matnorm{\OperatorSampleTwo}$. If $\OperatorSampleOne$ and $\OperatorSampleTwo$ satisfy $\OperatorSampleOne^\dagger \OperatorSampleTwo = 0$ and $\OperatorSampleOne \OperatorSampleTwo^\dagger = 0$ (i.e. they have orthogonal images and orthogonal supports), then
\begin{equation}
\label{eqn:NormOfSumOfOrthogonalOperator}
\matnorm{\OperatorSampleOne + \OperatorSampleTwo} \leq \max \{\matnorm{\OperatorSampleOne}, \matnorm{\OperatorSampleTwo}\}\enspace.
\end{equation}
In particular, the operator norm of a controlled operator $\OperatorSampleOne = \sum \ketbra{x}{x}\otimes \OperatorSampleOne^x$ can be upper bounded by the maximum of the norms of the operators $\OperatorSampleOne^x$ as shown in \cite{DFMS22}:
\begin{equation}
\label{eqn:NormOfControlledOperator}
\matnorm{\OperatorSampleOne} \leq \max_x \matnorm{\OperatorSampleOne^x}\enspace.
\end{equation}

For two linear transformations $\OperatorSampleOne$ and $\OperatorSampleTwo$, their commutator is $\Commutator{\OperatorSampleOne}{\OperatorSampleTwo}\coloneq \OperatorSampleOne\OperatorSampleTwo - \OperatorSampleTwo\OperatorSampleOne$, whose norm measures how nearly $\OperatorSampleOne$ and $\OperatorSampleTwo$ commute. For three operators $\OperatorSampleOne$, $\OperatorSampleTwo$ and $\OperatorSampleThree$ such that $\vecnorm{\OperatorSampleOne}, \vecnorm{\OperatorSampleTwo}, \vecnorm{\OperatorSampleThree} \leq 1$, if both $\OperatorSampleOne$ and $\OperatorSampleTwo$ are almost commutative with $\OperatorSampleThree$, $\OperatorSampleOne\OperatorSampleTwo$ is almost commutative with $\OperatorSampleThree$, formally, as shown in \cite{CMS19}:
\begin{equation}
\label{eqn:NormOfOpertorProduct}
\matnorm{\Commutator{\OperatorSampleOne\OperatorSampleTwo}{\OperatorSampleThree}}\leq \matnorm{\Commutator{\OperatorSampleOne}{\OperatorSampleThree}} + \matnorm{\Commutator{\OperatorSampleTwo}{\OperatorSampleThree}}\enspace.
\end{equation}

We fix the universal gate set $\{\HadamardGate, \CNOT, \TGate\}$ \cite{NC10}. In this work, we consider \emph{unitary quantum circuits} consisting of unitary gates from this gate set. The size of a unitary quantum circuit $\Circuit$ is the number of gates in $\Circuit$. For a unitary quantum circuit $\Circuit$, the inverse of $\Circuit$, denoted as $\Circuit^\dagger$, has size linear in the size of $\Circuit$.

A quantum algorithm can, without loss of generality, be written in the form of introducing ancilla qubits initialized as $\ket{0}$, applying a unitary quantum circuit, and making measurements, by the deferred measurement principle. Thus we sometimes specify a quantum algorithm with a unitary quantum circuit together with the ancilla register and the output register.

For two registers $\reg{1}{A}$ and $\reg{1}{B}$ of the same size, we use $\SWAP_{\reg{0}{A}\reg{0}{B}}$ to denote the unitary that maps $\PureStateSampleOne_{\reg{0}{A}}\PureStateSampleTwo_{\reg{0}{B}}$ to $\PureStateSampleTwo_{\reg{0}{A}}\PureStateSampleOne_{\reg{0}{B}}$ for each $\PureStateSampleOne$ and $\PureStateSampleTwo$. $\PauliX$ is the Pauli matrix which flips a qubit,
\begin{equation*}
\PauliX \coloneq \left(\begin{matrix}
0 & 1\\
1 & 0
\end{matrix}
\right)\enspace.
\end{equation*}

\subsection{Oracle circuits and oracle algorithms}

For a function $\Oracle{1}\colon\Bits^* \to \Bits^\RandomOracleOutputLength$, an \emph{oracle-aided unitary quantum circuit} $\Circuit^{\Oracle{0}}$ is a unitary quantum circuit with the additional query gate $U_{\Oracle{0}}\colon \ket{x}\ket{y} \to \ket{x}\ket{y \oplus \Oracle{1}(x)}$. Notice that the inverse of $U_{\Oracle{0}}$ is $U_{\Oracle{0}}$. For an oracle-aided quantum circuit $\Circuit^{\Oracle{0}}$, the inverse of $\Circuit^{\Oracle{0}}$ can also be implemented with oracle access to $\Oracle{1}$.

More generally, we consider stateful oracles, which can be written as $U(\StateRegister{1})$ for a unitary $U$ and a state register $\StateRegister{1}$. An oracle-aided unitary quantum circuit $\Circuit^{\scriptscriptstyle{{U(\StateRegister{1})}}}$ is a unitary quantum circuit where all the gates do not act on the state register $\StateRegister{1}$, with the additional query gate $U_{\QueryRegister{0}\AnswerRegister{0}\StateRegister{0}}$ acting on the registers $\QueryRegister{1}\AnswerRegister{1}\StateRegister{1}$ where $\QueryRegister{1}$ is the query register and $\AnswerRegister{1}$ is the answer register.

We use the sans-serif font with a stateful oracle in the superscript, e.g., $\Algorithm^{{\scriptscriptstyle{U(\StateRegister{0})}}}$, to denote an oracle-aided quantum algorithm. Similarly, an oracle-aided quantum algorithm can be written in the form of introducing some ancilla qubits in $\ket{0}$, applying an oracle-aided unitary quantum circuit, and making measurements, by the deferred measurement principle.

For an oracle-aided quantum algorithm $\Algorithm^{{\scriptscriptstyle{U(\StateRegister{0})}}}$ which has input register $\reg{1}{I}$ and output register $\reg{1}{O}$, we use $\reg{1}{O} \gets \Algorithm^{{\scriptscriptstyle{U(\StateRegister{0})}}}(\reg{1}{I})$ to denote the following process:
\begin{enumerate}[noitemsep]
\item $\Algorithm$ is given an input on register $\reg{1}{I}$, which might entangle with the state register $\StateRegister{1}$ of the oracle $U$.
\item After interacting with the oracle $U$, $\Algorithm$ outputs a state on register $\reg{1}{O}$, which might also be entangled with the state register $\StateRegister{1}$.
\end{enumerate}

An algorithm $\Algorithm$ with oracle access to multiple oracles $(U_i(\StateRegister{1}_i))_{i \in [t]}$ may query them in superposition, the registers $(\StateRegister{1}_i)_{i \in [t]}$ may not be distinct. Specifically, $\Algorithm$ has an oracle-selection register $\reg{1}{N}$, and a query is implemented by the controlled gate
\begin{equation*}
\sum_{i \in [t]}
\ketbra{i}{i}_{\reg{1}{N}}
\otimes
(U_i)_{\QueryRegister{1}\AnswerRegister{1}\StateRegister{1}_i}
\enspace.
\end{equation*}

We use \emph{query probability mass} to quantify how heavily an oracle is queried.

\begin{definition}
Let $\Algorithm$ be an algorithm with access to multiple oracles $(U_i(\StateRegister{1}_i))_{i \in [t]}$. We define the \defemph{query probability mass} of $\Algorithm$ on the $i$-th unitary $U_i$ at its $j$-th query as
\begin{equation*}
\QuantumQueryMassFunc{\Algorithm, i, j} \coloneq \vecnorm{\ketbra{i}{i}_{\reg{1}{N}}\ket{\psi_j}}^{2}\enspace,
\end{equation*}
where $\ket{\psi_j}$ is the joint quantum state of the algorithm $\Algorithm$ and the state registers of the unitaries $(U_i)_{i \in [t]}$ just before the algorithm makes its $j$-th oracle query.
\end{definition}

A quantum query can simultaneously access each oracle, but the sum of query probability mass across different oracles cannot exceed 1.

\begin{remark}
\label{remark:weight-properties}
Since $\sum_{i \in [t]} \ketbra{i}{i}_{\reg{1}{N}} = \id{\reg{1}{N}}$, for every algorithm $\Algorithm$ with access to oracles $(U_i(\StateRegister{1}_i))_{i \in [t]}$, and integer $j$,
\begin{equation*}
\sum_{i \in [t]} \QuantumQueryMassFunc{\Algorithm, i, j} = 1\enspace.
\end{equation*}
\end{remark}

We define the total (query) probability mass as the sum of the query probability masses over all oracle queries.

\begin{definition}
For an algorithm $\Algorithm$ that makes $q$ queries to oracles $(U_i(\StateRegister{1}_i))_{i \in [t]}$, we define the \defemph{total query probability mass} of $\Algorithm$ on the $i$-th unitary $U_i$ as follows:
\begin{equation*}
\QuantumTotalQueryMassFunc{\Algorithm, i} \coloneq \sum_{j = 1}^q \QuantumQueryMassFunc{\Algorithm, i, j}\enspace.
\end{equation*}
\end{definition}

In the paper, we slightly overload the notation to use $\QuantumTotalQueryMassFunc{\Algorithm, U_i(\StateRegister{1}_i)}$ to also denote the total query (probability) mass of $\Algorithm$ to $U_i$.

\subsection{Compressed oracles}
\label{prelim:compressed_oracle}

In this work, we consider quantum random oracle models (QROMs) where every party has quantum access to an oracle $\Oracle{1}$, sampled uniformly at random from the set of all the functions from $\Bits^*$ to $\Bits^\RandomOracleOutputLength$. We use $\RandomOracle{1}$ to emphasize that it is a random oracle and write the process of sampling a random oracle $\RandomOracle{1}$ as $\RandomOracle{1} \gets \UniformFrom{\RandomOracleOutputLength}$.

For $y \in \Bits^\RandomOracleOutputLength$, we define $\ket{\hat{y}} \coloneq 2^{-\RandomOracleOutputLength/2}\sum_{y' \in \Bits^\RandomOracleOutputLength}(-1)^{\dotp{y}{y'}}\ket{y'}$ where $\dotp{y}{y'}$ is the inner product of $y$ and $y'$. To simplify notation, we often use $\ket{\hat{0}^{\RandomOracleOutputLength}}$ to denote the state $\ket{\hat{0}}^{\otimes \RandomOracleOutputLength} = 2^{-\RandomOracleOutputLength/2}\left(\ket{0} + \ket{1}\right)^{\otimes \RandomOracleOutputLength} = \ket{\widehat{0^\RandomOracleOutputLength}}$.

Since an efficient oracle-aided quantum algorithm with running time $T$ cannot query $\RandomOracle{1}$ on input length greater than $T$, we often set the domain $\RandomOracleDomain$ to be a set of bit strings of bounded length instead of $\Bits^*$ in order to represent it in a finite-dimensional complex Hilbert space.

Applying the query gate $U_{\RandomOracle{0}}\colon \ket{x}\ket{y} \to \ket{x}\ket{y \oplus \RandomOracle{1}(x)}$ is equivalent to applying the unitary
\begin{equation*}
U\colon \ket{x}\ket{y}\ket{\Oracle{1}} \to \ket{x}\ket{y \oplus \Oracle{1}(x)}\ket{\Oracle{1}}
\end{equation*}
to a state whose last register is initialized as the truth table of $\RandomOracle{1}$ on $\RandomOracleDomain$. Since $U$ commutes with computational basis measurement on the last register, we can initialize a database register $\DatabaseRegister{1} \coloneq \bigotimes_{x \in \RandomOracleDomain}\DatabaseRegisterAt{1}{x}$ to be a uniform superposition of the truth tables of all the possible functions from $\RandomOracleDomain$ to $\Bits^\RandomOracleOutputLength$ and use $U$ as the query gate, instead of sampling $\RandomOracle{1}$ and using $U_{\RandomOracle{0}}$.

More interestingly, Zhandry \cite{Zha19} introduced a way to compress the database register $\DatabaseRegister{1}$, taking advantage of the fact that most of the registers $\DatabaseRegisterAt{1}{x}$ are not queried and thus remain in $\ket{\hat{0}^{\RandomOracleOutputLength}}$ during the process. Formally, the state of each $\DatabaseRegisterAt{1}{x}$ lies in the span of $\{\ket{y}\}_{y \in \Bits^\RandomOracleOutputLength} \cup \{\ket{\bot}\}$ where $\bot$ is a special symbol. We can apply
\begin{equation*}
\compress_{\DatabaseRegister{0}} \coloneq \bigotimes_{x \in \RandomOracleDomain} \compress_{\DatabaseRegisterAt{0}{x}}
\end{equation*}
with
\begin{equation*}
\compress_{\DatabaseRegisterAt{0}{x}} \coloneq \ketbra{\hat{0}^{\RandomOracleOutputLength}}{\bot}_{\DatabaseRegisterAt{0}{x}} + \ketbra{\bot}{\hat{0}^{\RandomOracleOutputLength}}_{\DatabaseRegisterAt{0}{x}} + \sum_{y \in \Bits^\RandomOracleOutputLength / \{0^\RandomOracleOutputLength\}}\ketbra{\hat{y}}{\hat{y}}_{\DatabaseRegisterAt{0}{x}}
\end{equation*}
to the state $\bigotimes_{x \in \RandomOracleDomain}\left(\frac{1}{\sqrt{2^\RandomOracleOutputLength}}\sum_{y \in \Bits^\RandomOracleOutputLength}\ket{y}\right)_{\DatabaseRegisterAt{0}{x}} = \bigotimes_{x \in \RandomOracleDomain}\ket{\hat{0}^{\RandomOracleOutputLength}}_{\DatabaseRegisterAt{0}{x}}$ to get an initial state $\bigotimes_{x \in \RandomOracleDomain}\ket{\bot}_{\DatabaseRegisterAt{0}{x}}$ before any query is made.

Upon receiving a query with query register $\QueryRegister{0}$ and answer register $\AnswerRegister{0}$, we can apply the oracle unitary
\begin{equation*}
\OracleUnitary_{\QueryRegister{0}\AnswerRegister{0}\DatabaseRegister{0}} \coloneq \sum_{x \in \RandomOracleDomain}\ketbra{x}{x}_{\QueryRegister{0}}\otimes \compress_{\DatabaseRegisterAt{0}{x}}\CNOT_{\DatabaseRegisterAt{0}{x}\AnswerRegister{0}}\compress_{\DatabaseRegisterAt{0}{x}}
\end{equation*}
where $\CNOT_{\DatabaseRegisterAt{0}{x}\AnswerRegister{0}}\ket{y_x}_{\DatabaseRegisterAt{0}{x}}\ket{y}_{\AnswerRegister{0}} = \ket{y_x}_{\DatabaseRegisterAt{0}{x}}\ket{y \oplus y_x}_{\AnswerRegister{0}}$ for $y, y_x \in \Bits^\RandomOracleOutputLength$ and $\CNOT_{\DatabaseRegisterAt{0}{x}\AnswerRegister{0}}$ acts as an identity on $\ket{\bot}_{\DatabaseRegisterAt{0}{x}}\ket{y}_{\AnswerRegister{0}}$ for $y \in \Bits^\RandomOracleOutputLength$.

For every set $S \subseteq \RandomOracleDomain$, the register $\accessVectorAt{\DatabaseRegister{1}}{S} \coloneq \bigotimes_{x \in S}\DatabaseRegisterAt{1}{x}$ records the images of the set $S$. We view a database $D$ as an array of length $|\RandomOracleDomain|$ with alphabet $\Bits^\RandomOracleOutputLength \cup \{\bot\}$. We use $\accessVectorAt{D}{x}$ to denote the image of $x$, as stored in the database $D$, and we use $\accessVectorAt{D}{S}$ to denote the array of images of $x \in S$, as stored in the database $D$. The size of $D$ is denoted as $\SizeOfDatabase{D}$, which is the number of non-$\bot$ elements in the array $D$.

Let $\ProjectNoCollision$ be the projector that projects to all the databases without any collision. Formally, we denote the set of all possible databases as $\SetOfDatabase \coloneq \left(\Bits^\RandomOracleOutputLength \cup \{\bot\}\right)^\RandomOracleDomain$ and the set of all databases without collisions as $\SetOfNoCollisionDatabase \coloneq \{D \in \SetOfDatabase : \forall x, x' \in \RandomOracleDomain \text{ s.t. } x \neq x', \accessVectorAt{D}{x} \neq \bot, \text{ and } \accessVectorAt{D}{x'} \neq \bot, \text{ we have that} \accessVectorAt{D}{x} \neq \accessVectorAt{D}{x'}\}$, and we have that
\begin{equation*}
\ProjectNoCollision \coloneq \sum_{D \in \SetOfNoCollisionDatabase}\ketbra{D}{D}
\enspace.
\end{equation*}
$\ProjectSizeDatabase{t}$ is the projector that projects to all the databases with size at most $t$. Formally,
\begin{equation*}
\ProjectSizeDatabase{t} \coloneq \sum_{D \in \SetOfDatabaseSizeAtMost{t}}\ketbra{D}{D}\enspace,
\end{equation*}
where
\begin{equation*}
\SetOfDatabaseSizeAtMost{t} \coloneq \{D \in \SetOfDatabase: \SizeOfDatabase{D} \leq t\}\enspace.
\end{equation*}

We will use the following lemma from \cite{CMS19}.

\begin{lemma}[Lemmas 5.10 and 6.11 in \cite{CMS19}]
\label{lem:NoCollisionToCollisionOneQuery}
For every random oracle output length $\RandomOracleOutputLength$ and bound $t$ for the size of the database,
\[\matnorm{(\id{\DatabaseRegister{1}} - \ProjectNoCollision)\left(\ProjectSizeDatabase{t}\OracleUnitary\ProjectSizeDatabase{t}\right)\ProjectNoCollision} \leq \sqrt{6t} \cdot 2^{-\RandomOracleOutputLength/2}\enspace.\]
\end{lemma}

Readers may wonder how to implement the above compressed oracle efficiently. As in \cite{Zha19}, we can make the compressed oracle efficient by mapping the extremely long table into a short database. There exists a unitary $\todatabase$ that takes a compressed table and returns a compressed database. Namely, given input $\bigotimes_{x \in \RandomOracleDomain}\ket{z_x}_{\DatabaseRegisterAt{0}{x}}$ where $z_x \in \Bits^\RandomOracleOutputLength \cup \{\bot\}$, $\todatabase$ returns $\ket{x_1, z_{x_1}, x_2, z_{x_2}, \ldots, x_k, z_{x_k}}$ where $x_1,\ldots,x_k$ enumerate all $x\in\RandomOracleDomain$ such that $z_x\neq\bot$, in increasing order. All the computations can also be done with this short database efficiently.

We show that restricting a procedure to run on databases without collisions won't change the final result too much.

\begin{lemma}
\label{lemma:collision-free}
Let $\Algorithm$ be a $\NumberOfQueries$-query algorithm with query access to the oracles $\OracleUnitary(\DatabaseRegister{1})$ and $U(\DatabaseRegister{1})$ for a unitary $U$ such that $\Commutator{\ProjectNoCollision}{U} = 0$ and $\Commutator{\ProjectSizeDatabase{q}}{U} = 0$ for every $q$. Let $\reg{1}{O}$ denote the output register of the algorithm $\Algorithm$. Denote $\Algorithm$'s query mass to $\OracleUnitary(\DatabaseRegister{1})$ as $\QuantumTotalQueryMass \coloneq \QuantumTotalQueryMassFunc{\Algorithm, \OracleUnitary(\DatabaseRegister{1})}$. Let $\Distinguisher$ be a computationally unbounded distinguisher.

Let $\NoCollisionAlgVariant{\Algorithm}$ be a variant of $\Algorithm$ that runs as $\Algorithm$ except that before and after each query, the measurement $\{\ProjectNoCollision, \id{\DatabaseRegister{1}} - \ProjectNoCollision\}$ is performed over $\DatabaseRegister{1}$ until $\Algorithm$ outputs register $\reg{1}{O}$. If at least one of the measurement outcomes is $\id{\DatabaseRegister{1}} - \ProjectNoCollision$, then $\NoCollisionAlgVariant{\Algorithm}$ outputs 0 and $\reg{1}{O}$; otherwise, $\NoCollisionAlgVariant{\Algorithm}$ outputs 1 and $\reg{1}{O}$.

Then for every random oracle output length $\RandomOracleOutputLength$,
\begin{align*}
&\abs{\sqrt{\prob{
b = 1
\;\middle\vert\;
\begin{array}{l}
\DatabaseRegister{1} \gets \ket{\bot}\\
\reg{1}{O} \gets \Algorithm^{\OracleUnitary(\DatabaseRegister{1}), U(\DatabaseRegister{1})}\\
b \gets \Distinguisher(\reg{1}{O}, \DatabaseRegister{1})
\end{array}}}
-
\sqrt{\prob{
b \land b' = 1
\;\middle\vert\;
\begin{array}{l}
\DatabaseRegister{1} \gets \ket{\bot}\\
(b', \reg{1}{O}) \gets \NoCollisionAlgVariant{\Algorithm}^{\OracleUnitary(\DatabaseRegister{1}), U(\DatabaseRegister{1})}\\
b \gets \Distinguisher(\reg{1}{O}, \DatabaseRegister{1})
\end{array}}}
}^2\\
&\leq 24 \cdot \NumberOfQueries^2 \QuantumTotalQueryMass \cdot 2^{-\RandomOracleOutputLength}\enspace.
\end{align*}

In general, the same bound holds for $\Another{\Algorithm}$, which runs as $\Algorithm$ except that the measurement $\{\ProjectNoCollision, \id{\DatabaseRegister{1}} - \ProjectNoCollision\}$ is performed over $\DatabaseRegister{1}$ before the $i$-th query for $i$ in a prescribed set $I \subseteq [\NumberOfQueries + 1]$:
\begin{align*}
&\abs{\sqrt{\prob{
b = 1
\;\middle\vert\;
\begin{array}{l}
\DatabaseRegister{1} \gets \ket{\bot}\\
\reg{1}{O} \gets \Algorithm^{\OracleUnitary(\DatabaseRegister{1}), U(\DatabaseRegister{1})}\\
b \gets \Distinguisher(\reg{1}{O}, \DatabaseRegister{1})
\end{array}}}
-
\sqrt{\prob{
b \land b' = 1
\;\middle\vert\;
\begin{array}{l}
\DatabaseRegister{1} \gets \ket{\bot}\\
(b', \reg{1}{O}) \gets {\Another{\Algorithm}}^{\OracleUnitary(\DatabaseRegister{1}), U(\DatabaseRegister{1})}\\
b \gets \Distinguisher(\reg{1}{O}, \DatabaseRegister{1})
\end{array}}}
}^2\\
&\leq 24 \cdot \NumberOfQueries^2 \QuantumTotalQueryMass \cdot 2^{-\RandomOracleOutputLength}\enspace.
\end{align*}
\end{lemma}

\begin{proof}
We prove the first part of the lemma. The second part of the lemma can be proved with the same idea.

By the deferred measurement principle, we can assume $\Algorithm$ always performs a unitary $U_{i}$ to prepare the $i$-th query for $i \in [\NumberOfQueries]$ and performs the unitary $U_{\NumberOfQueries + 1}$ to prepare the output register. Without loss of generality, we can assume the distinguisher $\Distinguisher$ performs a projective measurement $\{M_0, M_1\}$.

Define $\ket{\psi_i} \coloneq U_{i + 1}\OracleUnitary' \cdots \OracleUnitary' U_1\ket{\bar{0}}$ where the unitary
\[\OracleUnitary' \coloneq \ketbra{0}{0}_{\reg{1}{N}} \otimes \OracleUnitary + \ketbra{1}{1}_{\reg{1}{N}} \otimes U\]
is the controlled unitary that implements the superposition queries to the two oracles $\OracleUnitary$ and $U$. Then the state $\ket{\psi_{\NumberOfQueries}} = U_{\NumberOfQueries + 1}\OracleUnitary' \cdots \OracleUnitary' U_1\ket{\bar{0}}$ is the state on registers $(\reg{1}{O}, \DatabaseRegister{1})$ immediately after $\Algorithm$ is finished.

Define $\ket{\phi_i} \coloneq U_{i + 1}\NoCollisionAlgVariant{\OracleUnitary}' \cdots \NoCollisionAlgVariant{\OracleUnitary}' U_1\ket{\bar{0}}$, where
\[\NoCollisionAlgVariant{\OracleUnitary}' \coloneq \ProjectNoCollision\OracleUnitary'\ProjectNoCollision\enspace.\]
Then the subnormalized state $\ket{\phi_{\NumberOfQueries}} = U_{\NumberOfQueries + 1}\NoCollisionAlgVariant{\OracleUnitary}' \cdots \NoCollisionAlgVariant{\OracleUnitary}' U_1\ket{\bar{0}}$ is the state on registers $(\reg{1}{O}, \DatabaseRegister{1})$ immediately after $\Algorithm$ is finished conditioned on all of the measurement outcomes being $\ProjectNoCollision$.

Then we have that
\begin{align*}
&\abs{\sqrt{\prob{
b = 1
\;\middle\vert\;
\begin{array}{l}
\DatabaseRegister{1} \gets \ket{\bot}\\
\reg{1}{O} \gets \Algorithm^{\OracleUnitary(\DatabaseRegister{1}), U(\DatabaseRegister{1})}\\
b \gets \Distinguisher(\reg{1}{O}, \DatabaseRegister{1})
\end{array}}}
-
\sqrt{\prob{
b \land b' = 1
\;\middle\vert\;
\begin{array}{l}
\DatabaseRegister{1} \gets \ket{\bot}\\
(b', \reg{1}{O}) \gets \NoCollisionAlgVariant{\Algorithm}^{\OracleUnitary(\DatabaseRegister{1}), U(\DatabaseRegister{1})}\\
b \gets \Distinguisher(\reg{1}{O}, \DatabaseRegister{1})
\end{array}}}
}^2\\
&\leq \abs{\vecnorm{M_1 \ket{\psi_{\NumberOfQueries}}} - \vecnorm{M_1 \ket{\phi_{\NumberOfQueries}}}}^2\\
&\leq \vecnorm{M_1 (\ket{\psi_{\NumberOfQueries}} - \ket{\phi_{\NumberOfQueries}})}^2 \tag{By the triangle inequality}\\
&\leq \vecnorm{\ket{\psi_{\NumberOfQueries}} - \ket{\phi_{\NumberOfQueries}}}^2\enspace.
\end{align*}

Since one query to the oracle $\OracleUnitary$ increases the size of the database by at most 1, and one query to the oracle $U$ does not change the size of the database (as $\Commutator{\ProjectSizeDatabase{q}}{U} = 0$ for every $q$), the states satisfy that for every $i \in [\NumberOfQueries]$,
\[\ket{\psi_i} = U_{i + 1}\ProjectSizeDatabase{\NumberOfQueries}\OracleUnitary'\ProjectSizeDatabase{\NumberOfQueries} \cdots \ProjectSizeDatabase{\NumberOfQueries}\OracleUnitary'\ProjectSizeDatabase{\NumberOfQueries} U_1\ket{\bar{0}}\]
and
\[\ket{\phi_i} = U_{i + 1}\ProjectSizeDatabase{\NumberOfQueries}\NoCollisionAlgVariant{\OracleUnitary}'\ProjectSizeDatabase{\NumberOfQueries} \cdots \ProjectSizeDatabase{\NumberOfQueries}\NoCollisionAlgVariant{\OracleUnitary}'\ProjectSizeDatabase{\NumberOfQueries} U_1\ket{\bar{0}}\enspace.\]
As a result,
\begin{align*}
&\vecnorm{\ket{\psi_{\NumberOfQueries}} - \ket{\phi_{\NumberOfQueries}}}\\
&\leq \sum_{i = 0}^{\NumberOfQueries - 1} \vecnorm{U_{\NumberOfQueries + 1}\ProjectSizeDatabase{\NumberOfQueries}\NoCollisionAlgVariant{\OracleUnitary}' \cdots \NoCollisionAlgVariant{\OracleUnitary}'\ProjectSizeDatabase{\NumberOfQueries} U_{i + 2}\ProjectSizeDatabase{\NumberOfQueries}(\NoCollisionAlgVariant{\OracleUnitary}' - \OracleUnitary')\ProjectSizeDatabase{\NumberOfQueries} \ket{\psi_{i}}}\tag{By the triangle inequality}\\
&\leq \sum_{i = 0}^{\NumberOfQueries - 1} \vecnorm{U_{\NumberOfQueries + 1}\ProjectSizeDatabase{\NumberOfQueries}\NoCollisionAlgVariant{\OracleUnitary}' \cdots \NoCollisionAlgVariant{\OracleUnitary}'\ProjectSizeDatabase{\NumberOfQueries} U_{i + 2}(\ProjectNoCollision + \id{\DatabaseRegister{1}} - \ProjectNoCollision)\ProjectSizeDatabase{\NumberOfQueries}(\NoCollisionAlgVariant{\OracleUnitary}' - \OracleUnitary')\ProjectSizeDatabase{\NumberOfQueries} \ket{\psi_{i}}}\\
&\leq \sum_{i = 0}^{\NumberOfQueries - 1} \vecnorm{\ProjectNoCollision\ProjectSizeDatabase{\NumberOfQueries}(\NoCollisionAlgVariant{\OracleUnitary}' - \OracleUnitary')\ProjectSizeDatabase{\NumberOfQueries} \ket{\psi_{i}}} + \vecnorm{(\id{\DatabaseRegister{1}} - \ProjectNoCollision)\ProjectSizeDatabase{\NumberOfQueries}\OracleUnitary'\ProjectSizeDatabase{\NumberOfQueries} \ket{\psi_{\NumberOfQueries - 1}}}\\
&\leq \sum_{i = 0}^{\NumberOfQueries - 1} \vecnorm{\ProjectNoCollision\ProjectSizeDatabase{\NumberOfQueries}(\NoCollisionAlgVariant{\OracleUnitary}' - \OracleUnitary')\ProjectSizeDatabase{\NumberOfQueries}\ket{\psi_{i}}} + \sum_{i = 0}^{\NumberOfQueries - 1} \vecnorm{(\id{\DatabaseRegister{1}} - \ProjectNoCollision)\ProjectSizeDatabase{\NumberOfQueries}\OracleUnitary'\ProjectSizeDatabase{\NumberOfQueries} \ProjectNoCollision \ket{\psi_{i}}}\tag{By the triangle inequality}\\
&= \sum_{i = 0}^{\NumberOfQueries - 1} \vecnorm{\ProjectNoCollision\ProjectSizeDatabase{\NumberOfQueries}(\NoCollisionAlgVariant{\OracleUnitary}' - \OracleUnitary')\ProjectSizeDatabase{\NumberOfQueries} \ketbra{0}{0}_{\reg{1}{N}}\ket{\psi_{i}}} + \sum_{i = 0}^{\NumberOfQueries - 1} \vecnorm{(\id{\DatabaseRegister{1}} - \ProjectNoCollision)\ProjectSizeDatabase{\NumberOfQueries}\OracleUnitary'\ProjectSizeDatabase{\NumberOfQueries} \ProjectNoCollision \ketbra{0}{0}_{\reg{1}{N}}\ket{\psi_{i}}}\\
&\leq \sum_{i = 0}^{\NumberOfQueries - 1} \left(\matnorm{\ProjectNoCollision\ProjectSizeDatabase{\NumberOfQueries}(\NoCollisionAlgVariant{\OracleUnitary}' - \OracleUnitary')\ProjectSizeDatabase{\NumberOfQueries}} + \matnorm{(\id{\DatabaseRegister{1}} - \ProjectNoCollision)\ProjectSizeDatabase{\NumberOfQueries}\OracleUnitary'\ProjectSizeDatabase{\NumberOfQueries} \ProjectNoCollision}\right) \vecnorm{\ketbra{0}{0}_{\reg{1}{N}}\ket{\psi_{i}}}\\
&= 2\sum_{i \in [\NumberOfQueries]} \matnorm{(\id{\DatabaseRegister{1}} - \ProjectNoCollision)\left(\ProjectSizeDatabase{\NumberOfQueries}\OracleUnitary\ProjectSizeDatabase{\NumberOfQueries}\right)\ProjectNoCollision}\sqrt{\QuantumQueryMassFunc{\Algorithm, \OracleUnitary(\DatabaseRegister{1}), i}}\\
&\leq 2\sqrt{6\NumberOfQueries} \cdot 2^{-\RandomOracleOutputLength/2} \sqrt{\NumberOfQueries\sum_{i \in [\NumberOfQueries]} \QuantumQueryMassFunc{\Algorithm, \OracleUnitary(\DatabaseRegister{1}), i}} \tag{By \Cref{lem:NoCollisionToCollisionOneQuery} and Cauchy--Schwarz inequality}\\
&\leq 2\sqrt{6\NumberOfQueries} \cdot 2^{-\RandomOracleOutputLength/2} \sqrt{\NumberOfQueries \QuantumTotalQueryMass}\enspace,
\end{align*}
where in the sixth line, we use that $\Commutator{\ProjectNoCollision}{U} = 0$ and $\Commutator{\ProjectSizeDatabase{\NumberOfQueries}}{U} = 0$, and thus $\ProjectNoCollision\ProjectSizeDatabase{\NumberOfQueries}(\NoCollisionAlgVariant{\OracleUnitary}' - \OracleUnitary')\ProjectSizeDatabase{\NumberOfQueries} \ketbra{1}{1}_{\reg{1}{N}} = 0$ and $(\id{\DatabaseRegister{1}} - \ProjectNoCollision)\ProjectSizeDatabase{\NumberOfQueries}\OracleUnitary'\ProjectSizeDatabase{\NumberOfQueries} \ProjectNoCollision \ketbra{1}{1}_{\reg{1}{N}} = 0$.

The lemma follows from the above two inequalities.
\end{proof}

\subsection{Relations and languages}

We consider quantum arguments for relations, where the honest prover additionally gets a quantum witness as input. Thus we need to define a relation between the classical instance and the quantum witness state.

\begin{definition}[Relations and languages]
A relation $\Relation$ is a set of tuples of a classical string and a quantum state. Namely, $\Relation \subseteq \{(\Instance, \QuantumWitness): \Instance \in \Bits^*, \QuantumWitness \text{ is a quantum state}\}$.

A language for a relation $\Relation$ is defined as $\GetLanguage{\Relation} = \{\Instance: \exists \QuantumWitness, (\Instance, \QuantumWitness) \in \Relation\}$.

We say a relation has a decider $\RelationDecider$, if $\RelationDecider$ is a quantum algorithm such that for every $(\Instance, \QuantumWitness) \in \Relation$,
\[\prob{\RelationDecider(\Instance,\QuantumWitness) = 1} \geq \frac{2}{3}\enspace,\]
and for every $\Instance \notin \Language({\Relation})$ and quantum state $\QuantumWitness$,
\[\prob{\RelationDecider(\Instance,\QuantumWitness) = 1} \leq \frac{1}{3}\enspace.\]

A decider is efficient, if it takes a quantum state of $\poly(\InstanceSize)$ qubits, and runs in time $\poly(\InstanceSize)$ where $\InstanceSize$ is the length of the instance $\Instance$.
\end{definition}

Then $\QMA$ is the set of languages $\GetLanguage{\Relation}$ such that $\Relation$ has an efficient decider.

\subsection{Quantum interactive arguments}

A quantum interactive argument $(\ArgProver,\ArgVerifier)$ for a relation $\Relation$ consists of a quantum polynomial-time prover $\ArgProver$ and a quantum polynomial-time verifier $\ArgVerifier$ that exchange quantum messages. Both parties receive an instance $\Instance$ as input, while $\ArgProver$ receives a quantum state $\QuantumWitness$ as the witness. After the interaction, the verifier $\ArgVerifier$ outputs a bit, indicating whether $\ArgVerifier$ accepts. We write $\Pr[\langle \ArgProver(\Instance, \QuantumWitness),\ArgVerifier(\Instance)\rangle=1]$ for the probability that the verifier accepts after the interaction.

In this paper, we consider quantum interactive arguments in the quantum random oracle model where both prover $\ArgProver$ and $\ArgVerifier$ have quantum query access to a random oracle $\RandomOracle{1}$ sampled from $\UniformFrom{\RandomOracleOutputLength}$. We write $\Pr[\langle \ArgProver^{\RandomOracle{1}}(\Instance, \QuantumWitness),\ArgVerifier^{\RandomOracle{1}}(\Instance)\rangle=1]$ for the probability that the verifier accepts after the interaction when the prover and the verifier have access to $\RandomOracle{1}$.

\begin{definition}[Completeness]
A quantum interactive argument $(\ArgProver,\ArgVerifier)$ for a relation $\Relation$ in the quantum random oracle model has completeness $\ArgCompleteness$ if for every integer $\Security$, $\InstanceSize$, function $\RandomOracle{1}: \Bits^* \to \Bits^{\RandomOracleOutputLength}$ where the output length $\RandomOracleOutputLength$ of the random oracle is a function of the security parameter $\Security$ specified by the scheme, and $(\Instance, \QuantumWitness) \in \Relation$ such that $\abs{\Instance} \leq \InstanceSize$,
\[
\Pr[\langle \ArgProver^{\RandomOracle{1}}(\Instance, \QuantumWitness),\ArgVerifier^{\RandomOracle{1}}(\Instance)\rangle=1] \geq \ArgCompleteness(\InstanceSize, \Security)\enspace.
\]
\end{definition}

\begin{definition}[Soundness]
A quantum interactive argument $(\ArgProver,\ArgVerifier)$ for a relation $\Relation$ in the quantum random oracle model has soundness $\ArgSoundness$ if for every integer $\Security$, $\InstanceSize$, $\NumberOfQueries$, and $\NumberOfQueries$-query quantum adversary $\ArgAdv$,
\[\prob{
\begin{array}{l}
\abs{\Instance} \leq \InstanceSize \\
\land\, \Instance \notin \GetLanguage{\Relation}\\
\land\, b = 1
\end{array}
\;\middle\vert\;
\begin{array}{l}
\RandomOracle{1} \gets \UniformFrom{\RandomOracleOutputLength}\\
\Instance \gets \ArgAdv^{\RandomOracle{1}}\\
b \gets \langle \ArgAdv^{\RandomOracle{1}},\ArgVerifier^{\RandomOracle{1}}(\Instance)\rangle
\end{array}} \leq \ArgSoundness(\InstanceSize, \Security, \NumberOfQueries)\enspace,
\]
where the output length $\RandomOracleOutputLength$ of the random oracle is a function of the security parameter $\Security$ specified by the scheme.
\end{definition}

\doclearpage
\section{Quantum interactive oracle proofs}
\label{sec:QIOP-def}

A quantum interactive oracle proof (QIOP) system is a proof system that combines the QPCPs and the quantum interactive proof (QIP) systems. In this section, we define the syntax of the QIOPs and the notions of completeness and soundness for QIOP systems. Basically, QIOPs are QIPs such that the verifier does not necessarily read all qubits of the prover's message.

\subsection{Review: quantum interactive proofs}
\label{sec:QIP-def}

We review the quantum interactive proofs \cite{Watrous03} and introduce our notation conventions that are also useful to define QIOPs in \Cref{sec:QIOP-detail}.

In the standard notation, a quantum interactive proof system is an interactive protocol between two parties, the computationally unbounded prover $\QIOPProver$ and the quantum polynomial-time verifier $\QIOPVerifier$, where the prover $\QIOPProver$ has a private working register $\ProverPrivateRegister{}$ and the verifier $\QIOPVerifier$ has a private working register $\VerifierPrivateRegister{}$, and they exchange messages through a specified register $\MessageRegister{1}$ for $\QIOPRoundComplexity$ rounds.

In this work, we consider the doubly-efficient setting, where the honest prover $\QIOPProver$ can be implemented in quantum polynomial time given the witness $\QuantumWitness$ along with the instance $\Instance$ for the relation $\Relation$, and we specify different registers for messages and the internal states for different rounds. In more detail, we name the register that holds the $i$-th prover's message by $\ProverMessageRegister{i}$, and the register that holds the $i$-th prover's internal state by $\ProverPrivateRegister{i}$, and similarly, name the register that holds the $i$-th verifier's message by $\VerifierMessageRegister{i}$, and the register that holds the $i$-th verifier's internal state by $\VerifierPrivateRegister{i}$.

Before the interaction, both $\QIOPProver$ and $\QIOPVerifier$ are given an instance $\Instance$ for a relation $\Relation$. Furthermore, the honest prover $\QIOPProver$ is given a quantum witness $\QuantumWitness$ in a specified part of the private working register $\ProverPrivateRegister{1}$ in the first round.

The prover $\QIOPProver$ and the verifier $\QIOPVerifier$ interact as follows. In the $i$-th round for $i \in [\QIOPRoundComplexity]$, the prover $\QIOPProver$ applies a quantum polynomial-time algorithm $\QIOPProver_i$ on the instance $\Instance$, the verifier's message register $\VerifierMessageRegister{i - 1}$, and the private working register $\ProverPrivateRegister{i}$ to produce an output on the prover's message register $\ProverMessageRegister{i}$ and the private working register $\ProverPrivateRegister{i + 1}$. We write this procedure as $(\ProverMessageRegister{i}, \ProverPrivateRegister{i + 1}) \gets \QIOPProver_i(\Instance, \VerifierMessageRegister{i - 1}, \ProverPrivateRegister{i})$ where $\VerifierMessageRegister{0}$ is an empty register. In the $i$-th round for $i \in [\QIOPRoundComplexity - 1]$, the verifier $\QIOPVerifier$ applies a quantum polynomial-time algorithm $\QIOPVerifier_i$ on the instance $\Instance$, the prover's message register $\ProverMessageRegister{i}$, and the private working register $\VerifierPrivateRegister{i}$ to produce an output on the verifier's message register $\VerifierMessageRegister{i}$ and the private working register $\VerifierPrivateRegister{i + 1}$. We write this procedure as $(\VerifierMessageRegister{i}, \VerifierPrivateRegister{i + 1}) \gets \QIOPVerifier_i(\Instance, \ProverMessageRegister{i}, \VerifierPrivateRegister{i})$. In the final round, the verifier $\QIOPVerifier$ runs a quantum polynomial-time algorithm $\QIOPVerifier_{\QIOPRoundComplexity}$ on the instance $\Instance$, the prover's message register $\ProverMessageRegister{\QIOPRoundComplexity}$, and the private working register $\VerifierPrivateRegister{\QIOPRoundComplexity}$ to get a single-qubit output register $\reg{1}{O}$, and measures $\reg{1}{O}$ in the computational basis to get an outcome $b$, which indicates whether $\QIOPVerifier$ accepts or rejects. We write this procedure as $\ValidityBit \gets \QIOPVerifier_{\QIOPRoundComplexity}(\Instance, \ProverMessageRegister{\QIOPRoundComplexity}, \VerifierPrivateRegister{\QIOPRoundComplexity})$.

We overload the notation and write $\QIOPVerifier = (\QIOPVerifier_i)_{i \in [\QIOPRoundComplexity]}$ and $\QIOPProver = (\QIOPProver_i)_{i \in [\QIOPRoundComplexity]}$. Notice that here $\QIOPVerifier_i$ may touch all parts of $\ProverMessageRegister{i}$, similar to the classical interactive proofs, where the verifier may read the prover's message in full.

The protocol satisfies the usual completeness and soundness notions. The completeness states that given $(\Instance, \QuantumWitness) \in \Relation$, the honest prover $\QIOPProver$ can make the verifier accept with probability at least $\QIOPCompleteness(\abs{\Instance})$. The soundness states that no computationally unbounded cheating prover, who implements an arbitrary channel $\QIOPProver_i^*$ instead of the polynomial-time algorithms $\QIOPProver_i$ in the $i$-th round, can make the verifier $\QIOPVerifier$ accept on a no instance $\Instance \notin \GetLanguage{\Relation}$ with probability more than $\QIOPSoundness(\abs{\Instance})$.

\subsection{Our notion of QIOPs}
\label{sec:QIOP-detail}

We define QIOPs by restricting the access of the QIP verifier to the register $\ProverMessageRegister{i}$ to be ``local''. This definition is inspired by the QIPCP definition of \cite{SV26} and extends it by allowing the verifier to return the $i$-th round prover's message for any $i \in [\QIOPRoundComplexity]$.

Analogous to the classical case, where the proof oracle is a string divided into $\QIOPProofSize{}$ symbols over an alphabet $\Alphabet$, we need to first divide each prover's message register into subregisters before defining the local access to the prover's message registers.

\parhead{Dividing the prover's message registers into subregisters}
For round $i \in [\QIOPRoundComplexity]$, we decompose the $i$-th prover's message register $\ProverMessageRegister{i}$ to $\QIOPProofSize{i}$ registers $(\accessVectorAt{\ProverMessageRegister{i}}{j})_{j \in [\QIOPProofSize{i}]}$ for a prescribed proof length $\QIOPProofSize{i}$, where $\accessVectorAt{\ProverMessageRegister{i}}{j}$ is a quantum register with prescribed alphabet $\Alphabet$, whose Hilbert space can be written as $\Span{\ket{\sigma} : \sigma \in \Alphabet}$.

\parhead{The locality: the first attempt}
A natural way to define the locality of a QIOP verifier is to count the prover-message registers on which it acts.
More precisely, the algorithm $\QIOPVerifier_i$ that $\QIOPVerifier$ applies in the $i$-th round has locality $\QIOPQueryComplexity_i$ if $\QIOPVerifier_i$ only acts on registers $(\accessVectorAt{\ProverMessageRegister{i'}}{j})_{(i', j) \in \QVCQuerySet_i}$ for a query set $\QVCQuerySet_i$ of size at most $\QIOPQueryComplexity_i$, and the private working register $\VerifierPrivateRegister{i}$.

This definition, however, does not capture a coherent version of classical local query access.
For example, consider an operation that prepares a uniform superposition state over the location register, and then applies a $\CNOT$ gate over registers $(\accessVectorAt{\ProverMessageRegister{i}}{j}, \AnswerRegister{1})$ coherently, controlled on a location register containing $\ket{j}$. Under the above naive definition, this operation would be charged as $\QIOPProofSize{i}$ queries because it touches each of the $\QIOPProofSize{i}$ registers $(\accessVectorAt{\ProverMessageRegister{i}}{j})_{j \in [\QIOPProofSize{i}]}$, even though it is just a coherent version of reading a random location of the $i$-th prover's message and thus should be regarded as a local operation.

\parhead{Superposition access to quantum registers}
To capture such operations, we adopt the notion of \emph{superposition queries} introduced in \cite{CGMV26} and equip the verifier $\QIOPVerifier$ with quantum superposition access to the registers $(\accessVectorAt{\ProverMessageRegister{i}}{j})_{i \in [\QIOPRoundComplexity], j \in [\QIOPProofSize{i}]}$.
Generally, for a list of registers $(\reg{1}{B}_i)_{i \in L}$ with the same size, a quantum \emph{superposition query} of width $\QIOPQueryWidth$ to the oracle $(\reg{1}{B}_i)_{i \in L}$ consists of $\QIOPQueryWidth$ pairs of registers $(\QueryLocationRegister_{\iota}, \AnswerRegister{1}_{\iota})_{\iota \in [\QIOPQueryWidth]}$. The oracle applies the unitary
\[\QueryUnitary \coloneq \sum_{i \in L} \ketbra{i}{i}_{\QueryLocationRegister_{\iota}} \otimes \SWAPI{\reg{1}{B}_i}{\AnswerRegister{1}_{\iota}}\]
on the registers $(\QueryLocationRegister_{\iota}, \AnswerRegister{1}_{\iota}, (\reg{1}{B}_i)_{i \in L})$ for each $\iota \in [\QIOPQueryWidth]$ before sending the registers $(\QueryLocationRegister_{\iota}, \AnswerRegister{1}_{\iota})_{\iota \in [\QIOPQueryWidth]}$ back.
We say that an algorithm $\Algorithm$ has quantum \emph{superposition access} to a list of registers $(\reg{1}{B}_i)_{i \in L}$ with query depth $\QIOPQueryDepth$ and query width $\QIOPQueryWidth$ if its computation consists of unitaries $(\AlgorithmI{i})_{i \in \{0, 1, \ldots, \QIOPQueryDepth\}}$, none of which acts directly on $(\reg{1}{B}_i)_{i \in L}$, interleaved with $\QIOPQueryDepth$ quantum superposition queries with width at most $\QIOPQueryWidth$ to the oracle $(\reg{1}{B}_i)_{i \in L}$, where $\QIOPQueryDepth$ and $\QIOPQueryWidth$ are polynomially bounded. We denote an algorithm with such access by $\Algorithm^{(\reg{1}{B}_i)_{i \in L}}$. For simplicity, we only consider algorithms whose superposition queries all have the same width. We define the query complexity of $\Algorithm$ to be $\QIOPQueryComplexity \coloneq \QIOPQueryDepth \cdot \QIOPQueryWidth$.

\parhead{Returning the prover's message registers}
So far, the above QIOP model gives a reasonable definition of locality for the verifier, but it may be too restrictive for multi-round interactions. In the classical interactive oracle proofs, messages sent in later rounds may depend on messages from earlier rounds to help the verifier check claims about those earlier messages. A classical prover can support such interactions by retaining copies of its previous messages, whereas a quantum prover cannot in general do so because of the no-cloning theorem. We therefore allow previously submitted prover's message registers to be returned to the prover in later rounds, as in the QIOP constructions of \cite{SV26,CGMV26} and, more generally, in quantum interactive proofs. The remaining question is how such returns should be charged toward locality.

\parhead{The cost of returning registers}
Returning a prover's message register may involve many subregisters and can be implemented simply by having the verifier include the corresponding prover's message in the next message to the prover. However, the verifier does not ``read'' the contents of the register, but merely transfers the register back to the prover. Charging such an operation according to the size of the returned register would therefore not reflect the intended notion of locality. We thus adopt a more convenient way to formalize this model: we view all prover's message registers as being held by a trusted third party $\OracleParty$, who answers the verifier's superposition queries to these registers. When the prover needs a previously submitted register, the trusted third party removes that register from the verifier's query access and transfers it back to the prover. We therefore assign no additional cost to such a return operation, regardless of the size of the returned prover message.

\parhead{Putting the pieces together}
We give the formal definition of QIOPs. Let $\QIOP = (\QIOPProver, \QIOPVerifier)$ where $\QIOPProver = (\QIOPProver_i)_{i \in [\QIOPRoundComplexity]}$ is a tuple of polynomial-time quantum algorithms and $\QIOPVerifier = (\QIOPVerifier_i)_{i \in [\QIOPRoundComplexity]}$ is a tuple of polynomial-time quantum algorithms with superposition access to registers. We say that $\QIOP$ is a $\QIOPRoundComplexity$-round quantum interactive oracle proof for a relation $\Relation$ with completeness $\QIOPCompleteness$ and soundness $\QIOPSoundness$ if the following holds.

\begin{definition}[Completeness]
For every integer $\InstanceSize$ and instance-witness pair $(\Instance, \QuantumWitness) \in \Relation$ such that $\abs{\Instance} \leq \InstanceSize$,
\begin{align*}
\prob{
b = 1
\;\middle\vert\;
\begin{array}{l}
\ProverPrivateRegister{1} \gets \QuantumWitness \\
(\ProverMessageRegister{1}, \ProverPrivateRegister{2}) \gets \QIOPProver_1(\Instance, \ProverPrivateRegister{1})\\
(\accessVectorAt{\ProverMessageRegister{1}}{j})_{j \in [\QIOPProofSize{1}]} \coloneq \ProverMessageRegister{1}\\
\QIOPNotReturnIdxSet{1} \gets \{1\}, L \gets \{(1, j): j \in [\QIOPProofSize{1}]\}\\
\text{For } i = 2, \ldots, \QIOPRoundComplexity:\\
\quad (\QIOPReturnIdxSet{i}, \VerifierMessageRegister{i - 1}, \VerifierPrivateRegister{i}) \gets \QIOPVerifier_{i - 1}^{(\accessVectorAt{\ProverMessageRegister{i}}{j})_{(i, j) \in L}}(\Instance, \VerifierPrivateRegister{i - 1})\\
\quad (\ProverMessageRegister{i}, \ProverPrivateRegister{i + 1}) \gets \QIOPProver_i(\Instance, \VerifierMessageRegister{i - 1}, \ProverPrivateRegister{i}, (\ProverMessageRegister{j})_{j \in \QIOPReturnIdxSet{i}})\\
\quad (\accessVectorAt{\ProverMessageRegister{i}}{j})_{j \in [\QIOPProofSize{i}]} \coloneq \ProverMessageRegister{i}\\
\quad \QIOPNotReturnIdxSet{i} \gets \QIOPNotReturnIdxSet{i - 1} \cup \{i\} \setminus \QIOPReturnIdxSet{i}, L \gets \{(i', j'): i' \in \QIOPNotReturnIdxSet{i}, j' \in [\QIOPProofSize{i'}]\}\\
b \gets \QIOPVerifier_{\QIOPRoundComplexity}^{(\accessVectorAt{\ProverMessageRegister{i}}{j})_{(i, j) \in L}}(\Instance, \VerifierPrivateRegister{\QIOPRoundComplexity})
\end{array}}
\geq \QIOPCompleteness(\InstanceSize)\enspace.
\end{align*}

We abbreviate the interaction in the above experiment by \[b \gets \langle \QIOPProver(\Instance, \QuantumWitness), \QIOPVerifier(\Instance) \rangle\enspace.\]
Accordingly, the completeness condition can be written more compactly as
\[\prob{\langle \QIOPProver(\Instance, \QuantumWitness), \QIOPVerifier(\Instance) \rangle = 1} \geq \QIOPCompleteness(\InstanceSize)\enspace.\]
\end{definition}

\begin{definition}[Soundness]
For every integer $\InstanceSize$ and computationally unbounded quantum adversary $\Malicious{\QIOPProver}$,
\[\prob{
\begin{array}{l}
\abs{\Instance} \leq \InstanceSize \\
\land\, \Instance \notin \GetLanguage{\Relation}\\
\land\, b = 1
\end{array}
\;\middle\vert\;
\begin{array}{l}
\Instance \gets \Malicious{\QIOPProver}\\
(\ProverMessageRegister{1}, \ProverPrivateRegister{2}) \gets \Malicious{\QIOPProver}(\Instance, \ProverPrivateRegister{1})\\
(\accessVectorAt{\ProverMessageRegister{1}}{j})_{j \in [\QIOPProofSize{1}]} \coloneq \ProverMessageRegister{1}\\
\QIOPNotReturnIdxSet{1} \gets \{1\}, L \gets \{(1, j): j \in [\QIOPProofSize{1}]\}\\
\text{For } i = 2, \ldots, \QIOPRoundComplexity:\\
\quad (\QIOPReturnIdxSet{i}, \VerifierMessageRegister{i - 1}, \VerifierPrivateRegister{i}) \gets \QIOPVerifier_{i - 1}^{(\accessVectorAt{\ProverMessageRegister{i}}{j})_{(i, j) \in L}}(\Instance, \VerifierPrivateRegister{i - 1})\\
\quad (\ProverMessageRegister{i}, \ProverPrivateRegister{i + 1}) \gets \Malicious{\QIOPProver}(\Instance, \VerifierMessageRegister{i - 1}, \ProverPrivateRegister{i}, (\ProverMessageRegister{j})_{j \in \QIOPReturnIdxSet{i}})\\
\quad (\accessVectorAt{\ProverMessageRegister{i}}{j})_{j \in [\QIOPProofSize{i}]} \coloneq \ProverMessageRegister{i}\\
\quad \QIOPNotReturnIdxSet{i} \gets \QIOPNotReturnIdxSet{i - 1} \cup \{i\} \setminus \QIOPReturnIdxSet{i}, L \gets \{(i', j'): i' \in \QIOPNotReturnIdxSet{i}, j' \in [\QIOPProofSize{i'}]\}\\
b \gets \QIOPVerifier_{\QIOPRoundComplexity}^{(\accessVectorAt{\ProverMessageRegister{i}}{j})_{(i, j) \in L}}(\Instance, \VerifierPrivateRegister{\QIOPRoundComplexity})
\end{array}}
\leq \QIOPSoundness(\InstanceSize)\enspace.\]

We abbreviate the interaction in the above experiment by
\[
\Instance\gets\Malicious{\QIOPProver}, \;
b\gets \langle \Malicious{\QIOPProver}, \QIOPVerifier(\Instance) \rangle\enspace.
\]
Accordingly, the soundness condition can be written more compactly as
\[\prob{
\begin{array}{l}
\abs{\Instance} \leq \InstanceSize \\
\land\, \Instance \notin \GetLanguage{\Relation}\\
\land\, b = 1
\end{array}
\;\middle\vert\;
\begin{array}{l}
\Instance \gets \Malicious{\QIOPProver}\\
b \gets \langle \Malicious{\QIOPProver}, \QIOPVerifier(\Instance) \rangle
\end{array}}
\leq \QIOPSoundness(\InstanceSize)\enspace.\]
\end{definition}

We consider the following efficiency measures of a QIOP.
\begin{itemize}[noitemsep]
\item The \emph{round complexity} $\QIOPRoundComplexity$ is the number of rounds of interactions between the prover and the verifier.
\item The \emph{alphabet} $\Alphabet$ is the set indexing the computational basis of each prover-message subregister. In particular, each subregister has Hilbert space
\[
\Span{\ket{\sigma} : \sigma \in \Alphabet}.
\]
\item The \emph{query depth} is the total query depth of the verifier algorithms $(\QIOPVerifier_i)_{i \in [\QIOPRoundComplexity]}$. Namely, let $\QIOPQueryDepthI{i}$ denote the query depth of $\QIOPVerifier_i$, and we define the \emph{query depth} of the QIOP
\[\QIOPQueryDepth \coloneq \sum_{i \in [\QIOPRoundComplexity]}\QIOPQueryDepthI{i}\enspace.\]
\item The \emph{query width} is an upper bound for the query width of verifier algorithms $(\QIOPVerifier_i)_{i \in [\QIOPRoundComplexity]}$. Namely, let $\QIOPQueryWidthI{i}$ denote the query width of $\QIOPVerifier_i$, and we define the \emph{query width} of the QIOP
\[\QIOPQueryWidth \coloneq \max_{i \in [\QIOPRoundComplexity]}\QIOPQueryWidthI{i}\enspace.\]
\item The \emph{query complexity} is the total query complexity of $(\QIOPVerifier_i)_{i \in [\QIOPRoundComplexity]}$. Namely, let $\QIOPQueryComplexityI{i} \coloneq \QIOPQueryWidthI{i} \cdot \QIOPQueryDepthI{i}$ denote the query complexity of $\QIOPVerifier_i$, and we define the \emph{query complexity} of the QIOP
\[\QIOPQueryComplexity \coloneq \sum_{i \in [\QIOPRoundComplexity]}\QIOPQueryComplexityI{i}\enspace.\]
\item The \emph{total proof length} is the sum of the lengths of the prover messages over all rounds, while the \emph{maximum proof length} is the maximum length of a prover message in any round. Namely, \[\QIOPProofTotalSize \coloneq \sum_{i \in [\QIOPRoundComplexity]}\QIOPProofSize{i}, \; \QIOPProofMaxSize \coloneq \max_{i \in [\QIOPRoundComplexity]}\QIOPProofSize{i}\enspace.\]
\item The \emph{verifier-to-prover communication} is the total size of the verifier's message registers over all rounds. Namely, let $\QIOPVerifierMsgSizeI{i}$ denote the size, in qubits, of the verifier's message register $\VerifierMessageRegister{i}$ in the $i$-th round, and we define the \emph{verifier-to-prover communication} of the QIOP
\[\QIOPVerifierMsgSize \coloneq \sum_{i \in [\QIOPRoundComplexity]} \QIOPVerifierMsgSizeI{i}\enspace.\]
\end{itemize}

In particular, we do not count the returned prover messages toward the verifier-to-prover communication, since they are returned by the trusted third party $\OracleParty$, rather than by the verifier.

\subsection{Public-query QIOPs}
\label{subsec:pqQIOP}

As in the classical IBCS transformation \cite{CDGS23}, our compiler from quantum interactive oracle proofs to succinct quantum interactive arguments applies only to public-query QIOPs. This restriction arises because the query locations must be revealed to the argument prover in order to open the corresponding locations, and this additional information may be exploited by a malicious argument prover.

In the classical setting, the public-query property roughly requires soundness to hold even when the query locations are revealed to the prover at the time the queries are made \cite{CDGS23}. In the quantum setting, however, the query locations are stored in superposition in the registers $(\QueryLocationRegister_i)_{i \in [\QIOPQueryWidth]}$. We therefore adopt a natural quantum analogue of the classical notion. Namely, we say that a QIOP is \emph{public-query} if its soundness continues to hold even when the verifier's superposition queries are replaced by \emph{leaked superposition queries}.

Informally, a leaked superposition query with respect to an adversary $\AdversaryPrime$ is performed as follows. Upon receiving a query $(\QueryLocationRegister_i, \AnswerRegister{1}_i)_{i \in [\QIOPQueryWidth]}$, the trusted third party $\OracleParty$ sends the query location registers $(\QueryLocationRegister_i)_{i \in [\QIOPQueryWidth]}$ to $\AdversaryPrime$. The adversary $\AdversaryPrime$ may then apply an arbitrary quantum operation jointly to the registers $(\QueryLocationRegister_i)_{i \in [\QIOPQueryWidth]}$ and the private register before returning $(\QueryLocationRegister_i)_{i \in [\QIOPQueryWidth]}$ to $\OracleParty$, who answers the query by applying the same query unitary $\QueryUnitary$ as in the ordinary superposition query setting. After $\OracleParty$ performs the query, the query location registers $(\QueryLocationRegister_i)_{i \in [\QIOPQueryWidth]}$ are leaked to $\AdversaryPrime$ again, who may then apply another arbitrary quantum operation jointly to the registers $(\QueryLocationRegister_i)_{i \in [\QIOPQueryWidth]}$ and the private register before returning $(\QueryLocationRegister_i)_{i \in [\QIOPQueryWidth]}$. We formalize this interaction below.

\begin{definition}[Leaked superposition queries]
Let $(\reg{1}{M}_i)_{i \in L}$ be a list of registers of equal size, and let $\AdversaryPrime$ be an adversary with internal register $\AdvInternalStateRegister$.

Let $\Algorithm$ be an algorithm with leaked superposition access to
$(\reg{1}{M}_i)_{i \in L}$ of query depth $\QIOPQueryDepth$ and query width $\QIOPQueryWidth$.
Let $\Input$ denote the classical input of $\Algorithm$, let
$\InternalStateRegister^{(q)}$ denote its private register after the $q$-th invocation of its internal unitary, and let $\OutputRegister$ denote its output register.

The computation of $\Algorithm$ with superposition queries leaked to the adversary $\AdversaryPrime$ is given by the following experiment:
\[
\left[\;
\begin{aligned}
&((\QueryLocationRegister_{\iota}, \AnswerRegister{1}_{\iota})_{\iota \in [\QIOPQueryWidth]},
\InternalStateRegister^{(1)})
\gets
\AlgorithmI{0}(\Input, \InternalStateRegister^{(0)})
\\
&((\QueryLocationRegister_{\iota})_{\iota \in [\QIOPQueryWidth]},
\AdvInternalStateRegister)
\gets
\AdversaryPrime(
(\QueryLocationRegister_{\iota})_{\iota \in [\QIOPQueryWidth]},
\AdvInternalStateRegister)
\\
&\text{For } q = 1,\ldots,\QIOPQueryDepth-1:
\\
&\quad
((\QueryLocationRegister_{\iota}, \AnswerRegister{1}_{\iota})_{\iota \in [\QIOPQueryWidth]},
(\reg{1}{M}_i)_{i \in L})
\gets
\OracleParty(
(\QueryLocationRegister_{\iota}, \AnswerRegister{1}_{\iota})_{\iota \in [\QIOPQueryWidth]},
(\reg{1}{M}_i)_{i \in L})
\\&\quad((\QueryLocationRegister_{\iota})_{\iota \in [\QIOPQueryWidth]},
\AdvInternalStateRegister)
\gets
\AdversaryPrime(
(\QueryLocationRegister_{\iota})_{\iota \in [\QIOPQueryWidth]},
\AdvInternalStateRegister)
\\
&\quad
((\QueryLocationRegister_{\iota}, \AnswerRegister{1}_{\iota})_{\iota \in [\QIOPQueryWidth]},
\InternalStateRegister^{(q + 1)})
\gets
\AlgorithmI{q}(
\Input,
(\QueryLocationRegister_{\iota}, \AnswerRegister{1}_{\iota})_{\iota \in [\QIOPQueryWidth]},
\InternalStateRegister^{(q)})
\\
&\quad
((\QueryLocationRegister_{\iota})_{\iota \in [\QIOPQueryWidth]},
\AdvInternalStateRegister)
\gets
\AdversaryPrime(
(\QueryLocationRegister_{\iota})_{\iota \in [\QIOPQueryWidth]},
\AdvInternalStateRegister)
\\
&((\QueryLocationRegister_{\iota}, \AnswerRegister{1}_{\iota})_{\iota \in [\QIOPQueryWidth]},
(\reg{1}{M}_i)_{i \in L})
\gets
\OracleParty(
(\QueryLocationRegister_{\iota}, \AnswerRegister{1}_{\iota})_{\iota \in [\QIOPQueryWidth]},
(\reg{1}{M}_i)_{i \in L})
\\
&((\QueryLocationRegister_{\iota})_{\iota \in [\QIOPQueryWidth]},
\AdvInternalStateRegister)
\gets
\AdversaryPrime(
(\QueryLocationRegister_{\iota})_{\iota \in [\QIOPQueryWidth]},
\AdvInternalStateRegister)
\\
&(\OutputRegister, \InternalStateRegister^{(\QIOPQueryDepth + 1)})
\gets
\AlgorithmI{\QIOPQueryDepth}(
\Input,
(\QueryLocationRegister_{\iota}, \AnswerRegister{1}_{\iota})_{\iota \in [\QIOPQueryWidth]},
\InternalStateRegister^{(\QIOPQueryDepth)}),
\end{aligned}
\;\right]
\]
where $\OracleParty$ is the trusted party that answers the superposition queries by applying the unitary
\[\QueryUnitary \coloneq \sum_{i \in L} \ketbra{i}{i}_{\QueryLocationRegister_{\iota}} \otimes \SWAPI{\reg{1}{M}_i}{\AnswerRegister{1}_{\iota}}\]
on the registers $(\QueryLocationRegister_{\iota}, \AnswerRegister{1}_{\iota}, (\reg{1}{M}_i)_{i \in L})$ for each $\iota \in [\QIOPQueryWidth]$.

We abbreviate this computation as
\[
(\OutputRegister, \InternalStateRegister^{(\QIOPQueryDepth + 1)})
\gets
\Algorithm^{\Leaked{\OracleParty}{\AdversaryPrime}{(\reg{1}{M}_i)_{i \in L}}}
(\Input, \InternalStateRegister^{(0)}).
\]
For algorithms with multiple inputs or outputs, we use the analogous notation.
\end{definition}

We formalize the public-query soundness notion for QIOPs, where a trusted third party $\OracleParty$ is storing the prover's message registers and answering the superposition queries for the verifier, while leaking the query location registers to the adversary $\Malicious{\QIOPProver}$.

\begin{definition}[Public-query soundness]
\label{def:pq-qiop}
A $\QIOPRoundComplexity$-round quantum interactive oracle proof $\QIOP = (\QIOPProver, \QIOPVerifier)$ for a relation $\Relation$ is \emph{public-query} if the soundness property continues to hold even if the verifier's queries to the registers are leaked to the malicious prover $\Malicious{\QIOPProver}$.

Specifically, let $\OracleParty$ be the trusted party that implements the oracle access. $\QIOP$ has \emph{public-query soundness} $\QIOPPublicQuerySoundness$ if for every integer $\InstanceSize$ and computationally unbounded quantum adversary $\Malicious{\QIOPProver}$,
\[\prob{
\begin{array}{l}
\abs{\Instance} \leq \InstanceSize \\
\land\, \Instance \notin \GetLanguage{\Relation}\\
\land\, b = 1
\end{array}
\;\middle\vert\;
\begin{array}{l}
\Instance \gets \Malicious{\QIOPProver}\\
(\ProverMessageRegister{1}, \ProverPrivateRegister{2}) \gets \Malicious{\QIOPProver}(\Instance, \ProverPrivateRegister{1})\\
(\accessVectorAt{\ProverMessageRegister{1}}{j})_{j \in [\QIOPProofSize{1}]} \coloneq \ProverMessageRegister{1}\\
\QIOPNotReturnIdxSet{1} \gets \{1\}, L \gets \{(1, j): j \in [\QIOPProofSize{1}]\}\\
\text{For } i = 2, \ldots, \QIOPRoundComplexity:\\
\quad (\QIOPReturnIdxSet{i}, \VerifierMessageRegister{i - 1}, \VerifierPrivateRegister{i}) \gets \QIOPVerifier_{i - 1}^{\Leaked{\OracleParty}{\Malicious{\QIOPProver}}{(\accessVectorAt{\ProverMessageRegister{i}}{j})_{(i, j) \in L}}}(\Instance, \VerifierPrivateRegister{i - 1})\\
\quad (\ProverMessageRegister{i}, \ProverPrivateRegister{i + 1}) \gets \Malicious{\QIOPProver}(\Instance, \VerifierMessageRegister{i - 1}, \ProverPrivateRegister{i}, (\ProverMessageRegister{j})_{j \in \QIOPReturnIdxSet{i}})\\
\quad (\accessVectorAt{\ProverMessageRegister{i}}{j})_{j \in [\QIOPProofSize{i}]} \coloneq \ProverMessageRegister{i}\\
\quad \QIOPNotReturnIdxSet{i} \gets \QIOPNotReturnIdxSet{i - 1} \cup \{i\} \setminus \QIOPReturnIdxSet{i}, L \gets \{(i', j'): i' \in \QIOPNotReturnIdxSet{i}, j' \in [\QIOPProofSize{i'}]\}\\
b \gets \QIOPVerifier_{\QIOPRoundComplexity}^{\Leaked{\OracleParty}{\Malicious{\QIOPProver}}{(\accessVectorAt{\ProverMessageRegister{i}}{j})_{(i, j) \in L}}}(\Instance, \VerifierPrivateRegister{\QIOPRoundComplexity})
\end{array}}
\leq \QIOPPublicQuerySoundness(\InstanceSize)\enspace.\]

Furthermore, when the intermediate registers can be omitted, we further abbreviate the interaction in the above experiment by
\[
\Instance\gets\Malicious{\QIOPProver}, \;
b\gets \langle \Malicious{\QIOPProver}, \QIOPVerifier(\Instance) \rangle_{\pq}\enspace.
\]
Accordingly, the public-query soundness condition can be written more compactly as
\[\prob{
\begin{array}{l}
\abs{\Instance} \leq \InstanceSize \\
\land\, \Instance \notin \GetLanguage{\Relation}\\
\land\, b = 1
\end{array}
\;\middle\vert\;
\begin{array}{l}
\Instance \gets \Malicious{\QIOPProver}\\
b \gets \langle \Malicious{\QIOPProver}, \QIOPVerifier(\Instance) \rangle_{\pq}
\end{array}}
\leq \QIOPPublicQuerySoundness(\InstanceSize)\enspace.\]
\end{definition}

\parhead{Public-coin QIOPs}
A $\QIOP$ is \emph{public-coin} if each verifier's message register $\VerifierMessageRegister{i}$ consists of fresh classical random coins. Moreover, for each round $i \in [\QIOPRoundComplexity]$, both the returned index set $\QIOPReturnIdxSet{i}$ and the classical query set used by the verifier for point queries in that round are determined by publicly known deterministic functions of the instance $\Instance$ and the random coins contained in the preceding verifier's message registers $(\VerifierMessageRegister{j})_{j \in [i - 1]}$.

For public-coin QIOPs, we assume without loss of generality that, immediately after each leakage in round $i$, the verifier derives the query set from the instance $\Instance$ and the random coins in $(\VerifierMessageRegister{j})_{j \in [i - 1]}$, and the verifier rejects if the query set given by the prover differs from this derived set. These checks do not affect the ordinary execution of the public-coin QIOP and ensure that the malicious prover does not modify the query set. Moreover, the leaked query locations are determined by the instance $\Instance$ and the random coins in the preceding verifier's message registers $(\VerifierMessageRegister{j})_{j \in [i - 1]}$. Thus, they reveal no additional information to the malicious prover. Consequently, we obtain the following.

\begin{lemma}
\label{lemma:public-coin-public-query}
Every public-coin $\QIOP$ with soundness $\QIOPSoundness$ has public-query soundness
$\QIOPPublicQuerySoundness(\InstanceSize) = \QIOPSoundness(\InstanceSize)$.
\end{lemma}

\doclearpage
\section{Defining extractable quantum state commitments}
\label{sec:def-basic-extractable-commitment}

We describe the syntax of (non-interactive) quantum state commitments we will consider, and then provide a formal definition of extractability.

\subsection{A canonical quantum state commitment}

Following \cite{Yan22,GJMZ23}, we consider quantum state commitments in the following canonical form, where to commit to a quantum message, the sender always applies a unitary quantum circuit to the quantum message along with enough ancilla qubits initialized as $\ket{0}$, and the receiver always does the inverse of what the sender does and checks whether the ancilla qubits are $\ket{0}$.

\begin{definition}[Quantum state commitment scheme in the QROM]
\label{def:succinct_non_interactive_commitments}
Let $\RandomOracleOutputLength \in \N$ be the output length of the random oracle and $\MessageLength \in \N$ be the length of the quantum state which we want to commit to. Let $\RandomOracle{1}$ be sampled uniformly at random from all the functions with range $\Bits^{\RandomOracleOutputLength}$.

In the QROM, a \emph{quantum state commitment (QSC)} is a pair of oracle-aided quantum algorithms $(\Commit^{\RandomOracle{0}}, \Check^{\RandomOracle{0}})$ with oracle access to $\RandomOracle{1}$ where $\Commit^{\RandomOracle{0}}$ and $\Check^{\RandomOracle{0}}$ are specified by a $\poly(\RandomOracleOutputLength, \MessageLength)$-qubit ancilla register $\AncillasRegister{1}$ and an oracle-aided $\poly(\RandomOracleOutputLength, \MessageLength)$-size unitary quantum circuit $\CommitCircuit^{\RandomOracle{0}}_{\RandomOracleOutputLength, \MessageLength}$, as follows.

\begin{itemize}[noitemsep]
\item[]$\Commit^{\RandomOracle{0}}(\QuantumMessage_{\MessageRegister{0}})$:
\begin{enumerate}[nolistsep]
\item Initialize $\AncillasRegister{1}$ as all $\ket{0}$ states.
\item To commit to an $\MessageLength$-qubit quantum message $\QuantumMessage_{\MessageRegister{0}}$, apply $\CommitCircuit^{\RandomOracle{0}}_{\RandomOracleOutputLength, \MessageLength}$ to the registers $\MessageRegister{1}$ and $\AncillasRegister{1}$ to get a state $\CommitAndOpeningState_{\CommitmentRegister{0}\OpeningRegister{0}}$, which is divided into two registers $\CommitmentRegister{1}$ and $\OpeningRegister{1}$.
\item Output the state on $\CommitmentRegister{1}$ as the commitment, and the state on $\OpeningRegister{1}$ as the decommitment.
\end{enumerate}
\end{itemize}

\begin{itemize}[noitemsep]
\item[]$\Check^{\RandomOracle{0}}(\CommitAndOpeningState_{\CommitmentRegister{0}\OpeningRegister{0}})$:
\begin{enumerate}[nolistsep]
\item Apply the inverse of $\CommitCircuit^{\RandomOracle{0}}_{\RandomOracleOutputLength, \MessageLength}$ to the registers $\CommitmentRegister{1}$ and $\OpeningRegister{1}$ to get a state $\QuantumMessageAndAncillas_{\MessageRegister{0}\AncillasRegister{0}}$.
\item Make a computational measurement on the register $\AncillasRegister{1}$.
\item If the outcome is not all $0$, output $\ValidityBit = 0$ and a dummy quantum message $\ket{0^\MessageLength}$ on register $\MessageRegister{1}$.
\item Otherwise, output $\ValidityBit = 1$ and the state $\QuantumMessage_{\MessageRegister{0}}$ in register $\MessageRegister{1}$.
\end{enumerate}
\end{itemize}

A commitment is \emph{succinct} if the length of $\CommitmentRegister{1}$ is less than the length of the quantum message.

A commitment is \emph{$\MessageLength$-to-$t$} if the quantum message has $\MessageLength$ qubits while $\CommitmentRegister{1}$ has $t$ qubits.
\end{definition}

A canonical commitment scheme always satisfies the perfect correctness as defined below.

\begin{definition}[Perfect correctness]
\label{def:basic_commitment_perfect_correctness}
Let $\RandomOracleOutputLength\in\N$ be the output length of the random oracle. Let $\RandomOracle{1}$ be sampled uniformly at random from all the functions with range $\Bits^\RandomOracleOutputLength$. A quantum state commitment scheme $(\Commit^{\RandomOracle{0}}, \Check^{\RandomOracle{0}})$ has perfect correctness if for every function $\RandomOracle{1}$ with output length $\RandomOracleOutputLength$, unbounded quantum adversary $\Adversary$, and unbounded quantum distinguisher $\Distinguisher$, the following holds:
\[\prob{
1 \gets \Distinguisher^{\RandomOracle{0}}(\ValidityBit, \MessageRegister{1}, \EnvironmentRegister{1})
\;\middle\vert\;
\begin{array}{l}
(\MessageRegister{1}, \EnvironmentRegister{1}) \gets \Adversary^{\RandomOracle{0}}\\
(\CommitmentRegister{1}, \OpeningRegister{1}) \gets \Commit^{\RandomOracle{0}}(\MessageRegister{1})\\
(\ValidityBit, \MessageRegister{1}) \gets \Check^{\RandomOracle{0}}(\CommitmentRegister{1}, \OpeningRegister{1})
\end{array}}
= \prob{
1 \gets \Distinguisher^{\RandomOracle{0}}(1, \MessageRegister{1}, \EnvironmentRegister{1})
\;\middle\vert\;
(\MessageRegister{1}, \EnvironmentRegister{1}) \gets \Adversary^{\RandomOracle{0}}
}\enspace.
\]
\end{definition}

\begin{lemma}
\label{lem:basic_commitment_perfect_correctness}
Every canonical quantum state commitment scheme $(\Commit^{\RandomOracle{0}}, \Check^{\RandomOracle{0}})$ defined in \Cref{def:succinct_non_interactive_commitments} has perfect correctness (\Cref{def:basic_commitment_perfect_correctness}).
\end{lemma}

\begin{proof}
By definition, $\Check^{\RandomOracle{0}}$ does exactly the inverse of $\Commit^{\RandomOracle{0}}$. Since in $\Commit^{\RandomOracle{0}}$, $\AncillasRegister{1}$ is properly initialized, before the computational measurement in $\Check^{\RandomOracle{0}}$, qubits of $\AncillasRegister{1}$ are all in the $\ket{0}$ state, so $\Check^{\RandomOracle{0}}$ always outputs $\ValidityBit = 1$, and the original quantum message (while preserving the entanglement between $\MessageRegister{1}$ and $\EnvironmentRegister{1}$).
\end{proof}

\subsection{The extractability definition}

We first provide the syntax of an extractor in this case. In addition to the state in the commitment register $\CommitmentRegister{1}$, an extractor should have some side information of the oracle, provided by a quantum simulator. Ideally, the simulator maintains the side information in its internal state register $\StateRegister{1}$ without disturbing any adversary's execution.

\begin{definition}[Quantum simulator]
\label{def:quantum_simulator}
A stateful oracle $\Simulator$ with the state register $\StateRegister{1}$ is a $\SimulationError(\NumberOfQueries, \RandomOracleOutputLength)$-\emph{quantum simulator} for the random oracle if for any $\NumberOfQueries$-query quantum adversary $\Adversary$,
\[\abs{\prob{1 \leftarrow \Adversary^{\Simulator(\StateRegister{1})}(\ket{\bot}_{\StateRegister{0}})} - \prob{1 \leftarrow \Adversary^{\RandomOracle{0}}\middle\vert \RandomOracle{1} \gets \UniformFrom{\RandomOracleOutputLength}}} \leq \SimulationError(\NumberOfQueries, \RandomOracleOutputLength)\enspace,\]
where $\RandomOracleOutputLength$ is the output length of the random oracle.

A $\SimulationError$-quantum simulator is a \emph{perfect} quantum simulator for the random oracle if $\SimulationError = 0$.
\end{definition}

\begin{theorem}[\cite{Zha19}]
\label{thm:perfect-quantum-simulator}
$\OracleUnitary$ with the state register $\DatabaseRegister{1}$ is a perfect quantum simulator for the random oracle, where
\[\OracleUnitary_{\QueryRegister{0}\AnswerRegister{0}\DatabaseRegister{0}} = \sum_{x \in \RandomOracleDomain}\ketbra{x}{x}_{\QueryRegister{0}}\otimes \compress_{\DatabaseRegisterAt{0}{x}}\CNOT_{\DatabaseRegisterAt{0}{x}\AnswerRegister{0}}\compress_{\DatabaseRegisterAt{0}{x}}\] is the oracle unitary defined in \Cref{prelim:compressed_oracle}.
\end{theorem}

Now we are ready to define the syntax of an extractor.

\begin{definition}[Quantum message extractor]
\label{def:quantum_message_extractor}
A \emph{quantum message extractor} $\Extractor$ with query access to $\Simulator$ and $\ExtractOracle$ (both have the state register $\StateRegister{1}$) has the following syntax:
\begin{enumerate}[noitemsep]
\item $\Extractor.\ExtractMessage$ takes as input a commitment in register $\CommitmentRegister{1}$, makes queries to $\Simulator$ and $\ExtractOracle$, and outputs two registers $\MessageRegister{1}$ and $\AltOpeningRegister{1}$ (which will be used for $\AltCheck$).
\item $\Extractor.\AltCheck$ takes as input the opening in register $\OpeningRegister{1}$ and the auxiliary information in register $\AltOpeningRegister{1}$, and outputs $\ValidityBit = 0$ or $1$, indicating whether the sender gives a valid opening.
\end{enumerate}

We will write the above two procedures as $(\MessageRegister{1}, \AltOpeningRegister{1}) \gets \Extractor.\ExtractMessage^{\Simulator(\StateRegister{1}), \ExtractOracle(\StateRegister{1})}(\CommitmentRegister{1})$ and $\ValidityBit \gets \Extractor.\AltCheck(\OpeningRegister{1}, \AltOpeningRegister{1})$.
\end{definition}

\begin{definition}[Extractability]
\label{def:basic_commitment_strong_extractability}
A quantum state commitment scheme in the quantum random oracle model $(\Commit^{\RandomOracle{0}}, \Check^{\RandomOracle{0}})$ is \emph{$(\ExtractionErrorI{1}, \ExtractionErrorI{2}, \ExtractionErrorI{3})$-extractable} if there exist a $\ExtractionErrorI{1}$-quantum simulator $\Simulator$ with state register $\StateRegister{1}$ for the random oracle, a stateful oracle $\ExtractOracle$ with the same state register $\StateRegister{1}$, and a polynomial-time quantum message extractor $\Extractor$ with query access to $\Simulator$ and $\ExtractOracle$ as in \Cref{def:quantum_message_extractor} such that the following holds:

\begin{enumerate}
\item ($\ExtractOracle$ does not disturb the simulation.) For every integer $\RandomOracleOutputLength$, $\matnorm{\Commutator{\ExtractOracle}{\Simulator}}^2 \leq \ExtractionErrorI{2}(\RandomOracleOutputLength)$, where $\RandomOracleOutputLength$ is the output length of the random oracle.
\item Two parallel $\Simulator$ oracles commute, and two parallel $\ExtractOracle$ oracles commute.
\item ($\Extractor$ gives the only state that the adversary $\Adversary$ can open to.) For every integer $\RandomOracleOutputLength$, $\NumberOfQueries$, real numbers $\QuantumTotalQueryMass_1, \QuantumTotalQueryMass_2 \in [0, \NumberOfQueries]$ such that $\QuantumTotalQueryMass_1 + \QuantumTotalQueryMass_2 \leq \NumberOfQueries$, $\NumberOfQueries$-query quantum adversary $\Adversary$ with query access to $\Simulator$ and $\ExtractOracle$ such that $\QuantumTotalQueryMassFunc{\Adversary, \Simulator} \leq \QuantumTotalQueryMass_1$ and $\QuantumTotalQueryMassFunc{\Adversary, \ExtractOracle} \leq \QuantumTotalQueryMass_2$, and unbounded quantum distinguisher $\Distinguisher$,
\[\abs{\sqrt{\prob{\SimulateWorld(\Adversary, \Distinguisher)}} - \sqrt{\prob{\ExtractWorld(\Adversary, \Distinguisher)}}}^2 \leq \ExtractionErrorI{3}(\NumberOfQueries, \QuantumTotalQueryMass_1, \QuantumTotalQueryMass_2, \RandomOracleOutputLength)\enspace,\]
where $\RandomOracleOutputLength$ is the output length of the random oracle, and the games $\SimulateWorld(\Adversary, \Distinguisher)$ and $\ExtractWorld(\Adversary, \Distinguisher)$ are defined below:
\begin{itemize}
\item
\begin{itemize}[noitemsep]
\item[]$\SimulateWorld(\Adversary, \Distinguisher)$:
\begin{enumerate}[nolistsep]
\item The game initializes the state register: $\StateRegister{1} \gets \ket{\bot}$.
\item The adversary generates a commitment: $(\CommitmentRegister{1}, \EnvironmentRegister{1}) \gets \Adversary^{\Simulator(\StateRegister{1}), \ExtractOracle(\StateRegister{1})}$.
\item The adversary generates an opening: $(\OpeningRegister{1}, \EnvironmentRegister{1}) \gets \Adversary^{\Simulator(\StateRegister{1}), \ExtractOracle(\StateRegister{1})}(\EnvironmentRegister{1})$.
\item The game checks if the opening is valid: $(\ValidityBit, \MessageRegister{1}) \gets \Check^{\Simulator(\StateRegister{1})}(\CommitmentRegister{1}, \OpeningRegister{1})$.
\item The game generates the output: If $\ValidityBit = 0$, output 0; otherwise, compute $\ValidityBit' \gets \Distinguisher(\MessageRegister{1}, \EnvironmentRegister{1}, \StateRegister{1})$ and output $\ValidityBit'$.
\end{enumerate}
\end{itemize}
\item
\begin{itemize}[noitemsep]
\item[]$\ExtractWorld(\Adversary, \Distinguisher)$:
\begin{enumerate}[nolistsep]
\item The game initializes the state register: $\StateRegister{1} \gets \ket{\bot}$.
\item The adversary generates a commitment: $(\CommitmentRegister{1}, \EnvironmentRegister{1}) \gets \Adversary^{\Simulator(\StateRegister{1}), \ExtractOracle(\StateRegister{1})}$.
\item \textcolor{blue!70}{The game uses the extractor to extract the underlying message: $(\MessageRegister{1}, \AltOpeningRegister{1}) \gets \Extractor.\ExtractMessage^{\Simulator(\StateRegister{1}), \ExtractOracle(\StateRegister{1})}(\CommitmentRegister{1})$.}
\item The adversary generates an opening: $(\OpeningRegister{1}, \EnvironmentRegister{1}) \gets \Adversary^{\Simulator(\StateRegister{1}), \ExtractOracle(\StateRegister{1})}(\EnvironmentRegister{1})$.
\item \textcolor{blue!70}{The game uses the alternative check to check if the opening is valid: $\ValidityBit \gets \Extractor.\AltCheck(\OpeningRegister{1}, \AltOpeningRegister{1})$.}
\item The game generates the output: If $\ValidityBit = 0$, output 0; otherwise, compute $\ValidityBit' \gets \Distinguisher(\MessageRegister{1}, \EnvironmentRegister{1}, \StateRegister{1})$ and output $\ValidityBit'$.
\end{enumerate}
\end{itemize}
\end{itemize}
\end{enumerate}
\end{definition}

\doclearpage
\section{Construction of extractable commitments to quantum states}
\label{sec:construction-basic-extractable-commitment}

\subsection{An extractable basic succinct commitment scheme}

In this work, we will use the following basic succinct commitment in the quantum random oracle model. This is essentially an idealization of the commitment scheme in \cite{GJMZ23}, with a cryptographic hash function replaced by its idealization, a random oracle $\RandomOracle{1}$.

\begin{construction}
\label{construction:basic_commitment_circuit}
Let $\RandomOracle{1}$ be sampled uniformly from all the functions with range $\Bits^{\RandomOracleOutputLength}$. We construct the \emph{oracle-aided circuit $\CommitCircuit^{\RandomOracle{0}}_{\RandomOracleOutputLength, \MessageLength}$} as follows, where $\MessageLength$ is the length of the quantum message $\QuantumMessage_{\MessageRegister{0}}$.
\begin{itemize}[noitemsep]
\item[] $\CommitCircuit^{\RandomOracle{0}}_{\RandomOracleOutputLength, \MessageLength}(\QuantumMessage_{\MessageRegister{0}\AncillasRegister{0}})$:
\begin{enumerate}[nolistsep]
\item Rename the first half of the ancilla qubits as $\StandardBasisHashValueRegister{1}$, and the second half of the ancilla qubits as $\HadamardBasisHashValueRegister{1}$.
\item Compute the hash value in the standard basis: make a query $U_{\RandomOracle{0}}\colon \ket{x}\ket{y} \to \ket{x}\ket{y \oplus \RandomOracle{1}(x)}$ with query register $\MessageRegister{1}$ and answer register $\StandardBasisHashValueRegister{1}$.
\item Compute the hash value in the Hadamard basis.
\begin{enumerate}[nolistsep]
\item Apply the Hadamard gate on each qubit of the register $\MessageRegister{1}$.
\item Make a query $U_{\RandomOracle{0}}\colon \ket{x}\ket{y} \to \ket{x}\ket{y \oplus \RandomOracle{1}(x)}$ with query register $\MessageRegister{1}$ and answer register $\HadamardBasisHashValueRegister{1}$.
\item Apply the Hadamard gate on each qubit of the register $\MessageRegister{1}$.
\end{enumerate}
\item Set $\CommitmentRegister{1}\coloneq \StandardBasisHashValueRegister{1}\HadamardBasisHashValueRegister{1}$ and $\OpeningRegister{1} \coloneq \MessageRegister{1}$.
\end{enumerate}
\end{itemize}
\end{construction}

\begin{construction}[An extractable quantum state commitment]
\label{construction:basic_commitment}
Let $\RandomOracle{1}$ be sampled uniformly from all the functions with range $\Bits^{\RandomOracleOutputLength}$. $(\Commit^{\RandomOracle{0}}, \Check^{\RandomOracle{0}})$ is a canonical quantum state commitment as defined in \Cref{def:succinct_non_interactive_commitments}, where the ancilla register $\AncillasRegister{1}$ is a $2\RandomOracleOutputLength$-qubit register, and the oracle-aided circuit $\CommitCircuit_{\RandomOracleOutputLength, \MessageLength}^{\RandomOracle{0}}$ is instantiated as \Cref{construction:basic_commitment_circuit}.
\end{construction}

\parhead{Perfect correctness} This candidate construction in \Cref{construction:basic_commitment} is in the canonical form as \Cref{def:succinct_non_interactive_commitments}. By \Cref{lem:basic_commitment_perfect_correctness}, this candidate construction has perfect correctness.

\subsection{Extractor for the commitment scheme}

Let $(\Commit^{\RandomOracle{0}}, \Check^{\RandomOracle{0}})$ be the quantum state commitment scheme in \Cref{construction:basic_commitment}. We consider the following perfect quantum simulator for the random oracle\[\OracleUnitary_{\QueryRegister{0}\AnswerRegister{0}\DatabaseRegister{0}} = \sum_{x \in \RandomOracleDomain}\ketbra{x}{x}_{\QueryRegister{0}}\otimes \compress_{\DatabaseRegisterAt{0}{x}}\CNOT_{\DatabaseRegisterAt{0}{x}\AnswerRegister{0}}\compress_{\DatabaseRegisterAt{0}{x}}\enspace,\] which maintains some side information in register $\DatabaseRegister{1}$ during the simulation without disturbing any adversary's execution.

Let's construct an extractor for $(\Commit^{\RandomOracle{0}}, \Check^{\RandomOracle{0}})$ and the quantum simulator $\OracleUnitary(\DatabaseRegister{1})$. Our extractor will use the side information from the database register $\DatabaseRegister{1}$. Let's start with several setups.

Recall that $\RandomOracleOutputLength$ is the output length of the random oracle and $\RandomOracleDomain$ is the set of possible query positions made by the adversary. Without loss of generality, we can assume $\RandomOracleDomain = \Bits^\MessageLength$ where $\MessageLength$ is the length of the quantum message.

Let $\WorkingRegister{1} \coloneq \FlagRegister{1}\Another{\MessageRegister{1}}$ for a single-qubit register $\FlagRegister{1}$ and an $\MessageLength$-qubit register $\Another{\MessageRegister{1}}$. For each $\target \in \Bits^\RandomOracleOutputLength$, the purified measurement $\PurifiedInvY^{(\target)}_{\DatabaseRegister{0}\WorkingRegister{0}}$ writes the first $x \in \RandomOracleDomain$ (if any) such that the image of $x$ in $\DatabaseRegister{1}$ is $\target$ into the register $\Another{\MessageRegister{1}}$. Otherwise, it flips the flag register $\FlagRegister{1}$. Namely,
\begin{equation}
\label{equ:define_purified_measurement}
\PurifiedInvY^{(\target)}_{\DatabaseRegister{0}\WorkingRegister{0}} \coloneq \sum_{x \in \RandomOracleDomain} \InvY_x^{(\target)} \otimes \PauliX^x_{\Another{\MessageRegister{0}}} + \InvY_{\bot}^{(\target)} \otimes \PauliX_{\FlagRegister{0}}\enspace,
\end{equation}
where $\{\InvY_x^{(\target)}\}_{x \in \RandomOracleDomain \cup \{\bot\}}$ is the measurement on the database register $\DatabaseRegister{1}$ that outputs the first preimage $x \in \RandomOracleDomain$ (if any) of $\target$ and outputs $\bot$ otherwise; in other words, for every $x \in \RandomOracleDomain$,
\[
\InvY_x^{(\target)} \coloneq \bigotimes_{x' < x}\left(\id{\DatabaseRegisterAt{0}{x'}} - \ketbra{\target}{\target}_{\DatabaseRegisterAt{0}{x'}}\right)\otimes \ketbra{\target}{\target}_{\DatabaseRegisterAt{0}{x}}\enspace,
\]
and
\[\InvY_\bot^{(\target)} \coloneq \id{\DatabaseRegister{0}} - \sum_{x \in \RandomOracleDomain}\InvY_x^{(\target)}\enspace.\]

The following lemma states that the oracle unitary $\OracleUnitary_{\QueryRegister{0}\AnswerRegister{0}\DatabaseRegister{0}}$ almost commutes with the above purified measurement $\PurifiedInvY_{\DatabaseRegister{0}\WorkingRegister{0}}^{(y)}$.

\begin{lemma}[Theorem 3.1 and Lemma 3.4 in \cite{DFMS22}]
\label{lem:commutivity_of_purified_measurement}
For every $y \in \Bits^\RandomOracleOutputLength$ and every operator $V$ such that it acts only on $\DatabaseRegisterAt{1}{x}$ within the database register $\DatabaseRegister{1}$ and it does not act on $\WorkingRegister{1}$, the purified measurement $\PurifiedInvY_{\DatabaseRegister{0}\WorkingRegister{0}}^{(y)}$ defined in \Cref{equ:define_purified_measurement} almost commutes with $V$, as long as $\ketbra{y}{y}_{\DatabaseRegisterAt{0}{x}}$ almost commutes with $V$:
\[\matnorm{\Commutator{V}{\PurifiedInvY_{\DatabaseRegister{0}\WorkingRegister{0}}^{(y)}}} \leq 3\matnorm{\Commutator{V}{\ketbra{y}{y}_{\DatabaseRegisterAt{0}{x}}}} + \matnorm{\Commutator{V}{\InvY_\bot^{(\target)}}} \leq 4\matnorm{\Commutator{V}{\ketbra{y}{y}_{\DatabaseRegisterAt{0}{x}}}}\enspace.\]

Besides, the purified measurement $\PurifiedInvY_{\DatabaseRegister{1}\WorkingRegister{0}}^{(y)}$ almost commutes with the oracle unitary $\OracleUnitary_{\QueryRegister{0}\AnswerRegister{0}\DatabaseRegister{0}}$:
\[\matnorm{\Commutator{\OracleUnitary_{\QueryRegister{0}\AnswerRegister{0}\DatabaseRegister{0}}}{\PurifiedInvY_{\DatabaseRegister{0}\WorkingRegister{0}}^{(y)}}} \leq 8\sqrt{2} \cdot 2^{-\RandomOracleOutputLength/2}\enspace.\]
\end{lemma}

We can also use a value in a quantum register $\TargetRegister{1}$ as $\target$, and do $\PurifiedInvY^{(y)}_{\DatabaseRegister{0}\WorkingRegister{0}}$ coherently, as the unitary $\invert$:
\begin{align}\label{eqn:inv-def}
	\invert_{\TargetRegister{0}\DatabaseRegister{0}\WorkingRegister{0}} \coloneq \sum_{\target \in \Bits^\RandomOracleOutputLength} \ketbra{\target}{\target}_{\TargetRegister{0}} \otimes \PurifiedInvY^{(\target)}_{\DatabaseRegister{0}\WorkingRegister{0}}\enspace.
\end{align}

Now we are ready to present the quantum message extractor $\Extractor$.

\begin{construction}[The extractor $\Extractor$]
\label{construction:basic_commitment_extractor}
We consider the following $\Extractor$ for $(\Commit^{\RandomOracle{0}}, \Check^{\RandomOracle{0}})$ with query access to $\OracleUnitary(\DatabaseRegister{1})$ and $\invert(\DatabaseRegister{1})$.
\begin{itemize}[noitemsep]
\item[] $\Extractor$:
\begin{enumerate}[nolistsep]
\item $\Extractor.\ExtractMessage$: Upon receiving a commitment in register $\CommitmentRegister{1}$, $\Extractor.\ExtractMessage$ does the following.
\begin{enumerate}[nolistsep]
\item Parse $\CommitmentRegister{1}$ as $\StandardBasisHashValueRegister{1}$ and $\HadamardBasisHashValueRegister{1}$.
\item Initialize two $\MessageLength$-qubit registers $\StandardBasisMessageRegister{1}$ and $\HadamardBasisMessageRegister{1}$ in the zero state, and three single-qubit registers $\FlagRegister{1}$, $\StandardBasisSuccessRegister{1}$ and $\HadamardBasisSuccessRegister{1}$ in the zero state, and set working registers $\StandardBasisWorkingRegister{1} \coloneq \StandardBasisSuccessRegister{1}\StandardBasisMessageRegister{1}$, $\HadamardBasisWorkingRegister{1} \coloneq \HadamardBasisSuccessRegister{1}\HadamardBasisMessageRegister{1}$.
\item Call $\invert_{\HadamardBasisHashValueRegister{0}\DatabaseRegister{0}\HadamardBasisWorkingRegister{0}}$ and $\OracleUnitary_{\HadamardBasisMessageRegister{0}\HadamardBasisHashValueRegister{0}\DatabaseRegister{0}}$.
\item Call $\invert_{\StandardBasisHashValueRegister{0}\DatabaseRegister{0}\StandardBasisWorkingRegister{0}}$ and $\OracleUnitary_{\StandardBasisMessageRegister{0}\StandardBasisHashValueRegister{0}\DatabaseRegister{0}}$.
\item Apply the Hadamard gate on each qubit of the register $\HadamardBasisMessageRegister{1}$.
\item\label[step]{item:ExtMsgEndOfUnitary} Run $\CNOT$ on each qubit of the register $\StandardBasisMessageRegister{1}$ (as the control bit) and register $\HadamardBasisMessageRegister{1}$ (as the target bit).
\item Set $\MessageRegister{1} \coloneq \StandardBasisMessageRegister{1}$ and $\AltOpeningRegister{1} \coloneq \StandardBasisSuccessRegister{1}\HadamardBasisSuccessRegister{1}\HadamardBasisMessageRegister{1}\StandardBasisHashValueRegister{1}\HadamardBasisHashValueRegister{1}$. Output the state in registers $\MessageRegister{1}$ and $\AltOpeningRegister{1}$.
\end{enumerate}
\item $\Extractor.\AltCheck$: Upon receiving an opening in register $\OpeningRegister{1}$, $\Extractor.\AltCheck$ checks whether the state in registers $\OpeningRegister{1}$ and $\AltOpeningRegister{1}$ is \[\ket{\NullStateInExp_{\AltCheck}} \coloneq \frac{1}{\sqrt{|\RandomOracleDomain|}}\sum_{x \in \RandomOracleDomain}\ket{0, 0, 0^\RandomOracleOutputLength, 0^\RandomOracleOutputLength, x, x}_{\StandardBasisSuccessRegister{0}\HadamardBasisSuccessRegister{0}\StandardBasisHashValueRegister{0}\HadamardBasisHashValueRegister{0}\OpeningRegister{0}\HadamardBasisMessageRegister{0}}\]
and outputs $\ValidityBit = 1$ if so. Otherwise, output $\ValidityBit = 0$.
\end{enumerate}
\end{itemize}

\end{construction}

\subsection{Analysis of the commitment scheme}

\begin{theorem}[Extractability]
\label{thm:extractability}
\Cref{construction:basic_commitment} is $(\ExtractionErrorI{1}, \ExtractionErrorI{2}, \ExtractionErrorI{3})$-extractable with the quantum simulator $\OracleUnitary(\DatabaseRegister{1})$, the extractor oracle $\invert(\DatabaseRegister{1})$, and the extractor $\Extractor$ in \Cref{construction:basic_commitment_extractor} where
\begin{align*}
\ExtractionErrorI{1}&\coloneq 0\enspace,\\
\ExtractionErrorI{2}(\RandomOracleOutputLength) &\coloneq 2^{-\RandomOracleOutputLength + 7}\enspace,\\
\ExtractionErrorI{3}(\NumberOfQueries, \QuantumTotalQueryMass_1, \QuantumTotalQueryMass_2, \RandomOracleOutputLength) &\coloneq 2^{-\RandomOracleOutputLength + 14}(\NumberOfQueries + 2)^2(\QuantumTotalQueryMass_1 + \QuantumTotalQueryMass_2 + 3)\enspace.
\end{align*}
\end{theorem}

Define $\CNOT_{\QueryRegister{0}\Another{\QueryRegister{0}}} \ket{\bot}\ket{y} = \ket{y}\ket{\bot}$ (in case that at least one of $\QueryRegister{1}$ and $\Another{\QueryRegister{1}}$ is $\ket{\bot}$, do the swap unitary.)

We first present a lemma that is used to prove \Cref{thm:extractability}. The following lemma basically says that the extractor can extract offline the message from the commitment (after the adversary produces the opening) as long as there is no collision during the process of $\Extractor.\ExtractMessage$ and $\Check$.

To this end, we define a \emph{no-collision variant for an algorithm with access to $\OracleUnitary(\DatabaseRegister{1})$ and $\invert(\DatabaseRegister{1})$}.
For every algorithm $\Algorithm$ with oracle access to $\OracleUnitary(\DatabaseRegister{1})$ and $\invert(\DatabaseRegister{1})$, let $\NoCollisionAlgVariant{\Algorithm}$ denote the algorithm which has an additional query access that tells whether the database has a collision by making a projective measurement $\{\ProjectNoCollision, \id{\DatabaseRegister{0}} - \ProjectNoCollision\}$. $\NoCollisionAlgVariant{\Algorithm}$ does the same thing as $\Algorithm$ except that $\NoCollisionAlgVariant{\Algorithm}$ always makes a projective measurement $\{\ProjectNoCollision, \id{\DatabaseRegister{0}} - \ProjectNoCollision\}$ (by issuing a query) before and after applying each unitary, and outputs a bit $\NoCollisionAlgVariant{\ValidityBit}$, indicating whether all of the measurement results say there is no collision (i.e. the result is $\ProjectNoCollision$), in addition to the original outputs. We omit the additional query access and write $\NoCollisionAlgVariant{\Algorithm}^{\OracleUnitary(\DatabaseRegister{1}), \invert(\DatabaseRegister{1})}$ as the additional query access is clear from the context.

\begin{lemma}
\label{lemma:indistinguishablility-if-no-collisions}
Let $(\Commit^{\RandomOracle{0}}, \Check^{\RandomOracle{0}})$ be the quantum state commitment in \Cref{construction:basic_commitment}, and $\Extractor$ be the extractor in \Cref{construction:basic_commitment_extractor} with oracle access to $\OracleUnitary(\DatabaseRegister{1})$ and $\invert(\DatabaseRegister{1})$ that works on an internal database register $\DatabaseRegister{1}$.

For every integer $\RandomOracleOutputLength$, $\NumberOfQueries$, real numbers $\QuantumTotalQueryMass_1, \QuantumTotalQueryMass_2 \in [0, \NumberOfQueries]$ such that $\QuantumTotalQueryMass_1 + \QuantumTotalQueryMass_2 \leq \NumberOfQueries$, $\NumberOfQueries$-query quantum adversary $\Adversary$ such that $\QuantumTotalQueryMassFunc{\Adversary, \OracleUnitary} \leq \QuantumTotalQueryMass_1$ and $\QuantumTotalQueryMassFunc{\Adversary, \invert} \leq \QuantumTotalQueryMass_2$, and unbounded quantum distinguisher $\Distinguisher$,
\[\abs{\sqrt{\prob{\NoCollisionVariant{\SimulateWorld}(\Adversary, \Distinguisher)}} - \sqrt{\prob{\NoCollisionVariant{\OfflineExtractWorld}(\Adversary, \Distinguisher)}}}^2 \leq (12\NumberOfQueries + 16 + 32\NumberOfQueries\sqrt{\QuantumTotalQueryMass_2})^2\cdot 2^{-\RandomOracleOutputLength}\enspace,\]
where
\begin{itemize}
\item $\RandomOracleOutputLength$ is the output length of the random oracle;
\item The game $\NoCollisionVariant{\SimulateWorld}(\Adversary, \Distinguisher)$ is obtained by replacing the algorithm $\Check$ in the game $\SimulateWorld$ with the algorithm $\NoCollisionAlgVariant{\Check}$;
\item The game $\NoCollisionVariant{\OfflineExtractWorld}(\Adversary, \Distinguisher)$ is obtained by doing extraction after the adversary outputs the openings, and replacing the algorithm $\Extractor.\ExtractMessage$ by $\Extractor.\NoCollisionAlgVariant{\ExtractMessage}$ in the game $\ExtractWorld(\Adversary, \Distinguisher)$.
\end{itemize}
Specifically, the followings are the formal definitions of the games $\NoCollisionVariant{\SimulateWorld}(\Adversary, \Distinguisher)$ and $\NoCollisionVariant{\OfflineExtractWorld}(\Adversary, \Distinguisher)$:
\begin{itemize}
\item
\begin{itemize}[noitemsep]
\item[]$\NoCollisionVariant{\SimulateWorld}(\Adversary, \Distinguisher)$:
\begin{enumerate}[nolistsep]
\item The game initializes the database register: $\DatabaseRegister{1} \gets \ket{\bot}$.
\item The adversary generates a commitment: $(\CommitmentRegister{1}, \EnvironmentRegister{1}) \gets \Adversary^{\OracleUnitary(\DatabaseRegister{1}), \invert(\DatabaseRegister{1})}$.
\item The adversary generates an opening: $(\OpeningRegister{1}, \EnvironmentRegister{1}) \gets \Adversary^{\OracleUnitary(\DatabaseRegister{1}), \invert(\DatabaseRegister{1})}(\EnvironmentRegister{1})$.
\item \textcolor{blue!70}{The game checks if the opening is valid: $(\NoCollisionAlgVariant{\ValidityBit},\ValidityBit, \MessageRegister{1}) \gets \NoCollisionAlgVariant{\Check}^{\OracleUnitary(\DatabaseRegister{1})}(\CommitmentRegister{1}, \OpeningRegister{1})$.}
\item The game generates the output: If $\NoCollisionAlgVariant{\ValidityBit}\land \ValidityBit = 0$, output 0; otherwise, compute $\ValidityBit' \gets \Distinguisher(\MessageRegister{1}, \EnvironmentRegister{1}, \DatabaseRegister{1})$ and output $\ValidityBit'$.
\end{enumerate}
\end{itemize}
\item
\begin{itemize}[noitemsep]
\item[]$\NoCollisionVariant{\OfflineExtractWorld}(\Adversary, \Distinguisher)$:
\begin{enumerate}[nolistsep]
\item The game initializes the database register: $\DatabaseRegister{1} \gets \ket{\bot}$.
\item The adversary generates a commitment: $(\CommitmentRegister{1}, \EnvironmentRegister{1}) \gets \Adversary^{\OracleUnitary(\DatabaseRegister{1}), \invert(\DatabaseRegister{1})}$.
\item The adversary generates an opening: $(\OpeningRegister{1}, \EnvironmentRegister{1}) \gets \Adversary^{\OracleUnitary(\DatabaseRegister{1}), \invert(\DatabaseRegister{1})}(\EnvironmentRegister{1})$.
\item \textcolor{blue!70}{The game uses the extractor to extract the underlying message: $(\NoCollisionAlgVariant{\ValidityBit}, \MessageRegister{1}, \AltOpeningRegister{1}) \gets \Extractor.\NoCollisionAlgVariant{\ExtractMessage}^{\OracleUnitary(\DatabaseRegister{1}), \invert(\DatabaseRegister{1})}(\CommitmentRegister{1})$.}
\item {The game uses the alternative check to check if the opening is valid: $\ValidityBit \gets \Extractor.\AltCheck(\OpeningRegister{1}, \AltOpeningRegister{1})$.}
\item The game generates the output: If $\NoCollisionAlgVariant{\ValidityBit}\land \ValidityBit = 0$, output 0; otherwise, compute $\ValidityBit' \gets \Distinguisher(\MessageRegister{1}, \EnvironmentRegister{1}, \DatabaseRegister{1})$ and output $\ValidityBit'$.
\end{enumerate}
\end{itemize}
\end{itemize}
\end{lemma}

The proof of \Cref{lemma:indistinguishablility-if-no-collisions} is deferred to \Cref{subsec:proof-of-lemma}. We first show how \Cref{lemma:indistinguishablility-if-no-collisions} implies \Cref{thm:extractability}.
\begin{proof}[Proof of \Cref{thm:extractability}]
We use Zhandry's compressed oracle $\OracleUnitary$ as the unitary $\Simulator$ to do the simulation and use $\DatabaseRegister{1}$ as the state register $\StateRegister{1}$. By \Cref{thm:perfect-quantum-simulator}, $\OracleUnitary$ simulates the random oracle perfectly and thus $\ExtractionErrorI{1}= 0$.

We use the preimage-finding unitary $\invert$ as the unitary $\ExtractOracle$. By \Cref{lem:commutivity_of_purified_measurement} and \Cref{eqn:NormOfControlledOperator,eqn:inv-def},
\[\matnorm{\Commutator{\ExtractOracle}{\Simulator}}^2 \leq \max_{\target} \matnorm{\Commutator{\PurifiedInvY^{(\target)}_{\DatabaseRegister{0}\WorkingRegister{0}}}{\OracleUnitary_{\QueryRegister{0}\AnswerRegister{0}\DatabaseRegister{0}}}}^2 \leq (8\sqrt{2} \cdot 2^{-\RandomOracleOutputLength/2})^2 = 2^{-\RandomOracleOutputLength + 7}\enspace,\]
and thus $\ExtractionErrorI{2}(\RandomOracleOutputLength) = 2^{-\RandomOracleOutputLength + 7}$.

By the definition of $\OracleUnitary$ and $\invert$, two $\OracleUnitary$ oracles commute, and two $\invert$ oracles commute no matter which registers they act on.

To bound $\ExtractionErrorI{3}$, we first introduce a new game $\OfflineExtractWorld$ that does the extraction offline, where the difference between $\OfflineExtractWorld$ and ${\ExtractWorld}$ is highlighted in blue.
\begin{itemize}[noitemsep]
\item[]$\OfflineExtractWorld(\Adversary, \Distinguisher)$:
\begin{enumerate}[nolistsep]
\item The game initializes the database register: $\DatabaseRegister{1} \gets \ket{\bot}$.
\item The adversary generates a commitment: $(\CommitmentRegister{1}, \EnvironmentRegister{1}) \gets \Adversary^{\OracleUnitary(\DatabaseRegister{1}), \invert(\DatabaseRegister{1})}$.
\item \textcolor{blue!70}{The adversary generates an opening: $(\OpeningRegister{1}, \EnvironmentRegister{1}) \gets \Adversary^{\OracleUnitary(\DatabaseRegister{1}), \invert(\DatabaseRegister{1})}(\EnvironmentRegister{1})$.}
\item \textcolor{blue!70}{The game uses the extractor to extract the underlying message: $(\MessageRegister{1}, \AltOpeningRegister{1}) \gets \Extractor.\ExtractMessage^{\OracleUnitary(\DatabaseRegister{1}), \invert(\DatabaseRegister{1})}(\CommitmentRegister{1})$.}
\item {The game uses the alternative check to check if the opening is valid: $\ValidityBit \gets \Extractor.\AltCheck(\OpeningRegister{1}, \AltOpeningRegister{1})$.}
\item The game generates the output: If $\ValidityBit = 0$, output 0; otherwise, compute $\ValidityBit' \gets \Distinguisher(\MessageRegister{1}, \EnvironmentRegister{1}, \DatabaseRegister{1})$ and output $\ValidityBit'$.
\end{enumerate}
\end{itemize}

The only difference between $\OfflineExtractWorld$ and ${\ExtractWorld}$ is whether the adversary first generates the opening or the game first extracts the underlying message. As a result, for every $\NumberOfQueries$-query quantum adversary $\Adversary$ and unbounded quantum distinguisher $\Distinguisher$,
\begin{align}
\label{eqn:diff-ext-offext}
\abs{\sqrt{\prob{\ExtractWorld(\Adversary, \Distinguisher)}} - \sqrt{\prob{{\OfflineExtractWorld}(\Adversary, \Distinguisher)}}} \leq 4 \NumberOfQueries \cdot \matnorm{\Commutator{\ExtractOracle}{\Simulator}} \leq 32\sqrt{2}\cdot\NumberOfQueries \cdot 2^{-\RandomOracleOutputLength/2}\enspace.
\end{align}

It remains to bound $\abs{\sqrt{\prob{\SimulateWorld(\Adversary, \Distinguisher)}} - \sqrt{\prob{\OfflineExtractWorld(\Adversary, \Distinguisher)}}}$.

Then the only difference between $\SimulateWorld(\Adversary, \Distinguisher)$ and $\NoCollisionVariant{\SimulateWorld}(\Adversary, \Distinguisher)$ is whether we check if there are collisions before and after each $\OracleUnitary(\DatabaseRegister{1})$ query in the algorithm $\Check$. By \Cref{lemma:collision-free},
\begin{align}
\label{eqn:diff-simulate-collision}
&\abs{\sqrt{\prob{\SimulateWorld(\Adversary, \Distinguisher)}} - \sqrt{\prob{\NoCollisionVariant{\SimulateWorld}(\Adversary, \Distinguisher)}}}\nonumber\\
&\leq 2\sqrt{6(\NumberOfQueries + 2)^2(\QuantumTotalQueryMass_1 + 2)} \cdot 2^{-\RandomOracleOutputLength/2}\enspace,
\end{align}
since $\Check$ makes two queries to the random oracle.

By the same reasoning,
\begin{align}
\label{eqn:diff-offline-extract-collision}
&\abs{\sqrt{\prob{\OfflineExtractWorld(\Adversary, \Distinguisher)}} - \sqrt{\prob{\NoCollisionVariant{\OfflineExtractWorld}(\Adversary, \Distinguisher)}}} \leq 2\sqrt{6(\NumberOfQueries + 2)^2(\QuantumTotalQueryMass_1 + 2)} \cdot 2^{-\RandomOracleOutputLength/2}\enspace,
\end{align}
since $\Extractor.\ExtractMessage$ makes two queries to the random oracle.

Combining \Cref{eqn:diff-simulate-collision,eqn:diff-offline-extract-collision,lemma:indistinguishablility-if-no-collisions}, we can obtain
\begin{align*}
	&\abs{\sqrt{\prob{\SimulateWorld(\Adversary, \Distinguisher)}} - \sqrt{\prob{\OfflineExtractWorld(\Adversary, \Distinguisher)}}}\\
	&\leq 4\sqrt{6} \cdot (\NumberOfQueries + 2) \sqrt{\QuantumTotalQueryMass_1 + 2} \cdot 2^{-\RandomOracleOutputLength/2} + (12\NumberOfQueries + 16 + 32\NumberOfQueries\sqrt{w_2}) \cdot 2^{-\RandomOracleOutputLength/2} \enspace,
\end{align*}
which can be further combined with \Cref{eqn:diff-ext-offext} to get that
\begin{align*}
	&\abs{\sqrt{\prob{\SimulateWorld(\Adversary, \Distinguisher)}} - \sqrt{\prob{\ExtractWorld(\Adversary, \Distinguisher)}}}\\
	&\leq 4\sqrt{6} \cdot (\NumberOfQueries + 2) \sqrt{\QuantumTotalQueryMass_1 + 2} \cdot 2^{-\RandomOracleOutputLength/2} + (60\NumberOfQueries + 16 + 32\NumberOfQueries\sqrt{\QuantumTotalQueryMass_2}) \cdot 2^{-\RandomOracleOutputLength/2} \enspace,
\end{align*}
and therefore,
\begin{align*}
&\ExtractionErrorI{3}(\NumberOfQueries, \QuantumTotalQueryMass_1, \QuantumTotalQueryMass_2, \RandomOracleOutputLength)\\
&\leq \left(4\sqrt{6} \cdot (\NumberOfQueries + 2) \sqrt{\QuantumTotalQueryMass_1 + 2} \cdot 2^{-\RandomOracleOutputLength/2} + (60\NumberOfQueries + 16 + 32\NumberOfQueries\sqrt{\QuantumTotalQueryMass_2}) \cdot 2^{-\RandomOracleOutputLength/2}\right)^2\\
&\leq 2^{-\RandomOracleOutputLength + 2}\left(96 (\NumberOfQueries + 2)^2 (\QuantumTotalQueryMass_1 + 2) + 3600\NumberOfQueries^2 + 256 + 1024\NumberOfQueries^2\QuantumTotalQueryMass_2\right) \tag{By Cauchy--Schwarz inequality}\\
&\leq 2^{-\RandomOracleOutputLength + 14}(\NumberOfQueries + 2)^2(\QuantumTotalQueryMass_1 + \QuantumTotalQueryMass_2 + 3)\enspace.
\end{align*}
\end{proof}

\subsection{Proof of \Cref{lemma:indistinguishablility-if-no-collisions}}
\label{subsec:proof-of-lemma}

We need the following claims to prove \Cref{lemma:indistinguishablility-if-no-collisions}.

Let $\ProjectNoHatZero \coloneq \prod_{x} (\id{\DatabaseRegisterAt{0}{x}} - \ketbra{\hat{0}^\RandomOracleOutputLength}{\hat{0}^\RandomOracleOutputLength}_{\DatabaseRegisterAt{0}{x}})$ be the projector that checks whether each point of the database register is set to be non-$\hat{0}^\RandomOracleOutputLength$.
\iffull
\begin{claim}
\else
\begin{numberedclaim}
\fi
\label{claim:almost-commutativity-of-ProjectNoHatZero}
	For every query bound $\NumberOfQueries$,
	\[\matnorm{\ProjectSizeDatabase{\NumberOfQueries}\Commutator{\NoCollision}{\ProjectNoHatZero}\ProjectSizeDatabase{\NumberOfQueries}} \leq 2^{-(\RandomOracleOutputLength - 3)/2}\NumberOfQueries\enspace,\]
	and \[
	\matnorm{\ProjectSizeDatabase{\NumberOfQueries}\Commutator{\invert_{\TargetRegister{0}\DatabaseRegister{0}\WorkingRegister{0}}}{\ProjectNoHatZero}\ProjectSizeDatabase{\NumberOfQueries}} \leq 2^{-\RandomOracleOutputLength/2 + 4}\sqrt{\NumberOfQueries}\enspace.
	\]
\iffull
\end{claim}
\else
\end{numberedclaim}
\fi
We defer the proof to \Cref{subsec:Proof_of_almost-commutativity-of-ProjectNoHatZero}.

\iffull
\begin{claim}
\else
\begin{numberedclaim}
\fi
\label{claim:check_0=check_EPR}
\begin{align*}
\HadamardGate_{\MessageRegister{0}}^{\otimes \MessageLength}\CNOT_{\MessageRegister{0}\HadamardBasisMessageRegister{0}}\HadamardGate_{\MessageRegister{0}}^{\otimes \MessageLength}\CNOT_{\MessageRegister{0}\StandardBasisMessageRegister{0}}\ket{0, 0}_{\StandardBasisMessageRegister{0}\HadamardBasisMessageRegister{0}}= \HadamardGate_{\HadamardBasisMessageRegister{0}}^{\otimes \MessageLength}\CNOT_{\StandardBasisMessageRegister{0}\HadamardBasisMessageRegister{0}}\SWAP_{\StandardBasisMessageRegister{0}\MessageRegister{0}}\ket{\CheckStateEPR}_{\StandardBasisMessageRegister{0}\HadamardBasisMessageRegister{0}}\enspace,
\end{align*}
where $\ket{\CheckStateEPR} \coloneq \frac{1}{\sqrt{|\RandomOracleDomain|}}\sum_{z \in \RandomOracleDomain}\ket{z, z}$.
\iffull
\end{claim}
\else
\end{numberedclaim}
\fi

\begin{proof}
This equation can be shown by a direct computation. We omit the proof.
\end{proof}

\iffull
\begin{claim}
\else
\begin{numberedclaim}
\fi
\label{claim:diff=collision}
Let $\WorkingRegister{1} \coloneq \FlagRegister{1}\Another{\MessageRegister{1}}$. For every subnormalized state $\PureStateSampleOne$ on registers $\MessageRegister{1}$, $\TargetRegister{1}$, $\EnvironmentRegister{1}$ and $\DatabaseRegister{1}$,
\begin{align*}
&\vecnorm{\bra{0}_{\WorkingRegister{0}\TargetRegister{0}}\CNOT_{\MessageRegister{0}\Another{\MessageRegister{0}}}\OracleUnitary_{\Another{\MessageRegister{0}}\TargetRegister{0}\DatabaseRegister{0}}\invert_{\TargetRegister{0}\DatabaseRegister{0}\WorkingRegister{0}}\PureStateSampleOne_{\MessageRegister{0}\TargetRegister{0}\EnvironmentRegister{0}\DatabaseRegister{0}}\ket{0}_{\WorkingRegister{0}} - \bra{0}_{\TargetRegister{0}}\OracleUnitary_{\MessageRegister{0}\TargetRegister{0}\DatabaseRegister{0}}\PureStateSampleOne_{\MessageRegister{0}\TargetRegister{0}\EnvironmentRegister{0}\DatabaseRegister{0}}} \\
&\leq \frac{4}{\sqrt{2^{\RandomOracleOutputLength}}} + \vecnorm{(\id{\DatabaseRegister{0}} - \ProjectNoCollision)\PureStateSampleOne_{\MessageRegister{0}\TargetRegister{0}\EnvironmentRegister{0}\DatabaseRegister{0}}} + \vecnorm{(\id{\DatabaseRegister{0}} - \ProjectNoHatZero) \PureStateSampleOne_{\MessageRegister{0}\TargetRegister{0}\EnvironmentRegister{0}\DatabaseRegister{0}}}\enspace.
\end{align*}

Furthermore, for every subnormalized state $\PureStateSampleTwo$ on registers $\MessageRegister{1}$, $\EnvironmentRegister{1}$ and $\DatabaseRegister{1}$,
\begin{align*}
&\vecnorm{\ProjectNoCollision\invert_{\TargetRegister{0}\DatabaseRegister{0}\WorkingRegister{0}}\OracleUnitary_{\Another{\MessageRegister{0}}\TargetRegister{0}\DatabaseRegister{0}}\CNOT_{\MessageRegister{0}\Another{\MessageRegister{0}}}\PureStateSampleTwo_{\MessageRegister{0}\EnvironmentRegister{0}\DatabaseRegister{0}}\ket{0}_{\WorkingRegister{0}\TargetRegister{0}} - \ProjectNoCollision\OracleUnitary_{\MessageRegister{0}\TargetRegister{0}\DatabaseRegister{0}}\PureStateSampleTwo_{\MessageRegister{0}\EnvironmentRegister{0}\DatabaseRegister{0}}\ket{0}_{\WorkingRegister{0}\TargetRegister{0}}} \\
&\leq \frac{2}{\sqrt{2^{\RandomOracleOutputLength}}} + \sqrt{2}\vecnorm{(\id{\DatabaseRegister{0}} - \ProjectNoHatZero) \PureStateSampleTwo_{\MessageRegister{0}\EnvironmentRegister{0}\DatabaseRegister{0}}}\enspace.
\end{align*}
\iffull
\end{claim}
\else
\end{numberedclaim}
\fi

\begin{proof}
Let $\ket{\theta} \coloneq \frac{1}{\sqrt{2}}(\ket{\bot} - \ket{\hat{0}^\RandomOracleOutputLength})$ be a normalized quantum state. By the definition of $\compress$, we can get that for each $y \in \Bits^\RandomOracleOutputLength$, $\compress \ket{y} = \ket{y} + \frac{1}{\sqrt{2^{\RandomOracleOutputLength - 1}}}\ket{\theta}$.

We start with the first inequality. 
We write $\PureStateSampleOne_{\MessageRegister{0}\TargetRegister{0}\EnvironmentRegister{0}\DatabaseRegister{0}} \coloneq \sum_{x \in \RandomOracleDomain, y \in \Bits^{\RandomOracleOutputLength}, e, D}\alpha_{x, y, e, D}\ket{x, y, e, D}_{\MessageRegister{0}\TargetRegister{0}\EnvironmentRegister{0}\DatabaseRegister{0}}$.
Then by the definition of $\invert$, we can compute $\invert_{\TargetRegister{0}\DatabaseRegister{0}\WorkingRegister{0}}\PureStateSampleOne_{\MessageRegister{0}\TargetRegister{0}\EnvironmentRegister{0}\DatabaseRegister{0}}\ket{0}_{\WorkingRegister{0}}$ according to whether $y \in \Image{D}$. Specifically,
\begin{align*}
&\invert_{\TargetRegister{0}\DatabaseRegister{0}\WorkingRegister{0}}\PureStateSampleOne_{\MessageRegister{0}\TargetRegister{0}\EnvironmentRegister{0}\DatabaseRegister{0}}\ket{0}_{\WorkingRegister{0}}\\
&= \sum_{\substack{x, e, x', y, D\\ \text{ s.t. } \accessVectorAt{D}{x'} = y \text{ and } \forall x'' < x', \accessVectorAt{D}{x''} \neq y}}\alpha_{x, y, e, D}\ket{x, y, e, D, 0, x'}_{\MessageRegister{0}\TargetRegister{0}\EnvironmentRegister{0}\DatabaseRegister{0}\FlagRegister{0}\Another{\MessageRegister{0}}} + \sum_{x, e, D, y \notin \Image{D}}\alpha_{x, y, e, D}\ket{x, y, e, D, 1, 0}_{\MessageRegister{0}\TargetRegister{0}\EnvironmentRegister{0}\DatabaseRegister{0}\FlagRegister{0}\Another{\MessageRegister{0}}}\enspace,
\end{align*}
which implies
\begin{align*}
\bra{0}_{\FlagRegister{0}}\invert_{\TargetRegister{0}\DatabaseRegister{0}\WorkingRegister{0}}\PureStateSampleOne_{\MessageRegister{0}\TargetRegister{0}\EnvironmentRegister{0}\DatabaseRegister{0}}\ket{0}_{\WorkingRegister{0}} = \sum_{\substack{x, e, x', y, D\\ \text{ s.t. } \accessVectorAt{D}{x'} = y \text{ and } \forall x'' < x', \accessVectorAt{D}{x''} \neq y}}\alpha_{x, y, e, D}\ket{x, y, e, D, x'}_{\MessageRegister{0}\TargetRegister{0}\EnvironmentRegister{0}\DatabaseRegister{0}\Another{\MessageRegister{0}}}\enspace.
\end{align*}

Moreover,
\begin{align*}
&\bra{0}_{\TargetRegister{0}}\OracleUnitary_{\Another{\MessageRegister{0}}\TargetRegister{0}\DatabaseRegister{0}}\sum_{\substack{x, e, x', y, D\\ \text{ s.t. } \accessVectorAt{D}{x'} = y \text{ and } \forall x'' < x', \accessVectorAt{D}{x''} \neq y}}\alpha_{x, y, e, D}\ket{x, y, e, D, x'}_{\MessageRegister{0}\TargetRegister{0}\EnvironmentRegister{0}\DatabaseRegister{0}\Another{\MessageRegister{0}}}\\
&= \sum_{x', x, e}\ket{x', x, e}_{\Another{\MessageRegister{0}}\MessageRegister{0}\EnvironmentRegister{0}} \sum_{\substack{y, D \text{ s.t. } \accessVectorAt{D}{x'} = y\\ \text{and } \forall x'' < x', \accessVectorAt{D}{x''} \neq y}}\alpha_{x, y, e, D}\bra{0}_{\TargetRegister{0}}\compress_{\DatabaseRegisterAt{0}{x'}}\CNOT_{\DatabaseRegisterAt{0}{x'}\TargetRegister{0}}\compress_{\DatabaseRegisterAt{0}{x'}}\ket{y, D}_{\TargetRegister{0}\DatabaseRegister{0}}\\
&= \sum_{x', x, e}\ket{x', x, e}_{\Another{\MessageRegister{0}}\MessageRegister{0}\EnvironmentRegister{0}} \sum_{\substack{y, D \text{ s.t. } \accessVectorAt{D}{x'} = y\\ \text{and } \forall x'' < x', \accessVectorAt{D}{x''} \neq y}}\alpha_{x, y, e, D}\bra{0}_{\TargetRegister{0}}\compress_{\DatabaseRegisterAt{0}{x'}}\CNOT_{\DatabaseRegisterAt{0}{x'}\TargetRegister{0}}(\ket{y} + \frac{1}{\sqrt{2^{\RandomOracleOutputLength - 1}}}\ket{\theta})_{\DatabaseRegisterAt{0}{x'}}\ket{y, D - x'}_{\TargetRegister{0}\DatabaseRegisterAt{0}{\RandomOracleDomain/\{x'\}}}\\
&= \sum_{x', x, e}\ket{x', x, e}_{\Another{\MessageRegister{0}}\MessageRegister{0}\EnvironmentRegister{0}} \sum_{\substack{y, D \text{ s.t. } \accessVectorAt{D}{x'} = y\\ \text{and } \forall x'' < x', \accessVectorAt{D}{x''} \neq y}}\alpha_{x, y, e, D}\compress_{\DatabaseRegisterAt{0}{x'}}\ket{D}_{\DatabaseRegister{0}}\\
&+ \frac{1}{\sqrt{2^{\RandomOracleOutputLength - 1}}}\bra{0}_{\TargetRegister{0}}\sum_{\substack{x, e, x', y, D\\ \text{ s.t. } \accessVectorAt{D}{x'} = y \text{ and } \forall x'' < x', \accessVectorAt{D}{x''} \neq y}}\compress_{\DatabaseRegisterAt{0}{x'}}\CNOT_{\DatabaseRegisterAt{0}{x'}\TargetRegister{0}}\alpha_{x, y, e, D}\ket{x', x, e}_{\Another{\MessageRegister{0}}\MessageRegister{0}\EnvironmentRegister{0}}\ket{\theta}_{\DatabaseRegisterAt{0}{x'}}\ket{y, D - x'}_{\TargetRegister{0}\DatabaseRegisterAt{0}{\RandomOracleDomain/\{x'\}}}\enspace.
\end{align*}
Therefore,
\begin{align*}
&\bra{0}_{\Another{\MessageRegister{0}}\TargetRegister{0}}\CNOT_{\MessageRegister{0}\Another{\MessageRegister{0}}}\OracleUnitary_{\Another{\MessageRegister{0}}\TargetRegister{0}\DatabaseRegister{0}}\sum_{\substack{x, e, x', y, D\\ \text{ s.t. } \accessVectorAt{D}{x'} = y \text{ and } \forall x'' < x', \accessVectorAt{D}{x''} \neq y}}\alpha_{x, y, e, D}\ket{x, y, e, D, x'}_{\MessageRegister{0}\TargetRegister{0}\EnvironmentRegister{0}\DatabaseRegister{0}\Another{\MessageRegister{0}}}\\
&= \sum_{x, e}\ket{x, e}_{\MessageRegister{0}\EnvironmentRegister{0}} \sum_{\substack{y, D \text{ s.t. } \accessVectorAt{D}{x} = y\\ \text{and } \forall x'' < x, \accessVectorAt{D}{x''} \neq y}}\alpha_{x, y, e, D}\compress_{\DatabaseRegisterAt{0}{x}}\ket{D}_{\DatabaseRegister{0}}\\
&+ \frac{1}{\sqrt{2^{\RandomOracleOutputLength - 1}}}\bra{0}_{\TargetRegister{0}}\sum_{\substack{x, e, y, D\\ \text{ s.t. } \accessVectorAt{D}{x} = y \text{ and } \forall x'' < x, \accessVectorAt{D}{x''} \neq y}}\alpha_{x, y, e, D}\compress_{\DatabaseRegisterAt{0}{x}}\CNOT_{\DatabaseRegisterAt{0}{x}\TargetRegister{0}}\ket{x, e}_{\MessageRegister{0}\EnvironmentRegister{0}}\ket{\theta}_{\DatabaseRegisterAt{0}{x}}\ket{y, D - x}_{\TargetRegister{0}\DatabaseRegisterAt{0}{\RandomOracleDomain/\{x\}}}\enspace,
\end{align*}
which implies
\begin{align*}
&\|\bra{0}_{\WorkingRegister{0}\TargetRegister{0}}\CNOT_{\MessageRegister{0}\Another{\MessageRegister{0}}}\OracleUnitary_{\Another{\MessageRegister{0}}\TargetRegister{0}\DatabaseRegister{0}}\invert_{\TargetRegister{0}\DatabaseRegister{0}\WorkingRegister{0}}\PureStateSampleOne_{\MessageRegister{0}\TargetRegister{0}\EnvironmentRegister{0}\DatabaseRegister{0}}\ket{0}_{\WorkingRegister{0}} - \sum_{\substack{x, e, y, D \text{ s.t. } \accessVectorAt{D}{x} = y\\ \text{and } \forall x'' < x, \accessVectorAt{D}{x''} \neq y}}\alpha_{x, y, e, D}\compress_{\DatabaseRegisterAt{0}{x}}\ket{x, e, D}_{\MessageRegister{0}\EnvironmentRegister{0}\DatabaseRegister{0}}\|\\
&= \frac{1}{\sqrt{2^{\RandomOracleOutputLength - 1}}}\vecnorm{\bra{0}_{\TargetRegister{0}}\sum_{\substack{x, e, y, D\\ \text{ s.t. } \accessVectorAt{D}{x} = y \text{ and } \forall x'' < x, \accessVectorAt{D}{x''} \neq y}}\alpha_{x, y, e, D}\compress_{\DatabaseRegisterAt{0}{x}}\CNOT_{\DatabaseRegisterAt{0}{x}\TargetRegister{0}}\ket{x, e}_{\MessageRegister{0}\EnvironmentRegister{0}}\ket{\theta}_{\DatabaseRegisterAt{0}{x}}\ket{y, D - x}_{\TargetRegister{0}\DatabaseRegisterAt{0}{\RandomOracleDomain/\{x\}}}}\\
&\leq \frac{1}{\sqrt{2^{\RandomOracleOutputLength - 1}}}\vecnorm{\sum_{\substack{x, e, y, D\\ \text{ s.t. } \accessVectorAt{D}{x} = y \text{ and } \forall x'' < x, \accessVectorAt{D}{x''} \neq y}}\alpha_{x, y, e, D}\compress_{\DatabaseRegisterAt{0}{x}}\CNOT_{\DatabaseRegisterAt{0}{x}\TargetRegister{0}}\ket{x, e}_{\MessageRegister{0}\EnvironmentRegister{0}}\ket{\theta}_{\DatabaseRegisterAt{0}{x}}\ket{y, D - x}_{\TargetRegister{0}\DatabaseRegisterAt{0}{\RandomOracleDomain/\{x\}}}}\\
&= \frac{1}{\sqrt{2^{\RandomOracleOutputLength - 1}}}\vecnorm{\left(\sum_{x}\ketbra{x}{x}_{\MessageRegister{0}} \otimes \compress_{\DatabaseRegisterAt{0}{x}}\CNOT_{\DatabaseRegisterAt{0}{x}\TargetRegister{0}}\right)\sum_{\substack{x, e, y, D\\ \text{ s.t. } \accessVectorAt{D}{x} = y\\ \land \forall x'' < x, \accessVectorAt{D}{x''} \neq y}}\alpha_{x, y, e, D}\ket{x, e}_{\MessageRegister{0}\EnvironmentRegister{0}}\ket{\theta}_{\DatabaseRegisterAt{0}{x}}\ket{y, D - x}_{\TargetRegister{0}\DatabaseRegisterAt{0}{\RandomOracleDomain/\{x\}}}}\\
&= \frac{1}{\sqrt{2^{\RandomOracleOutputLength - 1}}}\vecnorm{\sum_{\substack{x, e, y, D\\ \text{ s.t. } \accessVectorAt{D}{x} = y\\ \land \forall x'' < x, \accessVectorAt{D}{x''} \neq y}}\alpha_{x, y, e, D}\ket{x, e, y, D - x}_{\MessageRegister{0}\EnvironmentRegister{0}\TargetRegister{0}\DatabaseRegisterAt{0}{\RandomOracleDomain/\{x\}}}}\\
&= \frac{1}{\sqrt{2^{\RandomOracleOutputLength - 1}}}\sqrt{\sum_{\substack{x, e, y, D\\ \text{ s.t. } \accessVectorAt{D}{x} = y\\ \land \forall x'' < x, \accessVectorAt{D}{x''} \neq y}}\abs{\alpha_{x, y, e, D}}^2}\\
&\leq \frac{1}{\sqrt{2^{\RandomOracleOutputLength - 1}}}\tag{By normalization, we have that $\sum_{x, y, e, D}\abs{\alpha_{x, y, e, D}}^2 \leq 1$ for the subnormalized state $\PureStateSampleOne$.}\enspace.
\end{align*}

On the other hand,
\begin{align*}
&\bra{0}_{\TargetRegister{0}}\OracleUnitary_{\MessageRegister{0}\TargetRegister{0}\DatabaseRegister{0}}\PureStateSampleOne_{\MessageRegister{0}\TargetRegister{0}\EnvironmentRegister{0}\DatabaseRegister{0}}\\
&=\sum_{x, e}\ket{x, e}_{\MessageRegister{0}\EnvironmentRegister{0}}\sum_{y, D}\alpha_{x, y, e, D}\bra{0}_{\TargetRegister{0}}\compress_{\DatabaseRegisterAt{0}{x}}\CNOT_{\DatabaseRegisterAt{0}{x}\TargetRegister{0}}\compress_{\DatabaseRegisterAt{0}{x}}\ket{y, D}_{\TargetRegister{0}\DatabaseRegister{0}}\\
&=\sum_{x, e}\ket{x, e}_{\MessageRegister{0}\EnvironmentRegister{0}}\sum_{y, D \text{ s.t. } D(x) = \bot}\alpha_{x, y, e, D}\bra{0}_{\TargetRegister{0}}\compress_{\DatabaseRegisterAt{0}{x}}\CNOT_{\DatabaseRegisterAt{0}{x}\TargetRegister{0}}\compress_{\DatabaseRegisterAt{0}{x}}\ket{y, D}_{\TargetRegister{0}\DatabaseRegister{0}}\\
&+\sum_{x, e}\ket{x, e}_{\MessageRegister{0}\EnvironmentRegister{0}}\sum_{y, D \text{ s.t. } D(x) = y}\alpha_{x, y, e, D}\bra{0}_{\TargetRegister{0}}\compress_{\DatabaseRegisterAt{0}{x}}\CNOT_{\DatabaseRegisterAt{0}{x}\TargetRegister{0}}\compress_{\DatabaseRegisterAt{0}{x}}\ket{y, D}_{\TargetRegister{0}\DatabaseRegister{0}}\\
&+\sum_{x, e}\ket{x, e}_{\MessageRegister{0}\EnvironmentRegister{0}}\sum_{y, D \text{ s.t. } D(x) \neq \bot \text{ and } D(x) \neq y}\alpha_{x, y, e, D}\bra{0}_{\TargetRegister{0}}\compress_{\DatabaseRegisterAt{0}{x}}\CNOT_{\DatabaseRegisterAt{0}{x}\TargetRegister{0}}\compress_{\DatabaseRegisterAt{0}{x}}\ket{y, D}_{\TargetRegister{0}\DatabaseRegister{0}}\\
&= \frac{1}{\sqrt{2^{\RandomOracleOutputLength}}}\sum_{x, e}\ket{x, e}_{\MessageRegister{0}\EnvironmentRegister{0}}\sum_{y, D \text{ s.t. } D(x) = \bot}\alpha_{x, y, e, D} \compress_{\DatabaseRegisterAt{0}{x}} \ket{y}_{\DatabaseRegisterAt{0}{x}}\ket{D}_{\DatabaseRegisterAt{0}{\RandomOracleDomain/\{x\}}}\\
&+ \sum_{x, e}\ket{x, e}_{\MessageRegister{0}\EnvironmentRegister{0}}\sum_{y, D \text{ s.t. } D(x) = y}\alpha_{x, y, e, D} \left(\compress_{\DatabaseRegisterAt{0}{x}}\ket{D}_{\DatabaseRegister{0}} + \frac{1}{\sqrt{2^{\RandomOracleOutputLength - 1}}}\bra{0}_{\TargetRegister{0}}\compress_{\DatabaseRegisterAt{0}{x}}\CNOT_{\DatabaseRegisterAt{0}{x}\TargetRegister{0}}\ket{y}_{\TargetRegister{0}}\ket{\theta}_{\DatabaseRegisterAt{0}{x}}\ket{D - x}_{\DatabaseRegisterAt{0}{\RandomOracleDomain/\{x\}}}\right)\\
&+ \frac{1}{\sqrt{2^{\RandomOracleOutputLength - 1}}}\sum_{x, e}\ket{x, e}_{\MessageRegister{0}\EnvironmentRegister{0}}\sum_{y, D \text{ s.t. } D(x) \neq \bot \text{ and } D(x) \neq y}\alpha_{x, y, e, D}\bra{0}_{\TargetRegister{0}}\compress_{\DatabaseRegisterAt{0}{x}}\CNOT_{\DatabaseRegisterAt{0}{x}\TargetRegister{0}}\ket{y}_{\TargetRegister{0}}\ket{\theta}_{\DatabaseRegisterAt{0}{x}}\ket{D - x}_{\DatabaseRegisterAt{0}{\RandomOracleDomain/\{x\}}}\enspace,
\end{align*}
which implies
\begin{align*}
&\vecnorm{\bra{0}_{\TargetRegister{0}}\OracleUnitary_{\MessageRegister{0}\TargetRegister{0}\DatabaseRegister{0}}\PureStateSampleOne_{\MessageRegister{0}\TargetRegister{0}\EnvironmentRegister{0}\DatabaseRegister{0}} - \sum_{\substack{x, e, y, D \text{ s.t. } \accessVectorAt{D}{x} = y}}\alpha_{x, y, e, D}\compress_{\DatabaseRegisterAt{0}{x}}\ket{x, e, D}_{\MessageRegister{0}\EnvironmentRegister{0}\DatabaseRegister{0}}}\\
&\leq \frac{1}{\sqrt{2^{\RandomOracleOutputLength}}}\vecnorm{\sum_{x, e}\ket{x, e}_{\MessageRegister{0}\EnvironmentRegister{0}}\sum_{y, D \text{ s.t. } D(x) = \bot}\alpha_{x, y, e, D} \compress_{\DatabaseRegisterAt{0}{x}} \ket{y}_{\DatabaseRegisterAt{0}{x}}\ket{D}_{\DatabaseRegisterAt{0}{\RandomOracleDomain/\{x\}}}}\\
& + \frac{1}{\sqrt{2^{\RandomOracleOutputLength - 1}}}\vecnorm{\sum_{x, e}\ket{x, e}_{\MessageRegister{0}\EnvironmentRegister{0}}\sum_{y, D \text{ s.t. } D(x) \neq \bot}\alpha_{x, y, e, D}\bra{0}_{\TargetRegister{0}}\compress_{\DatabaseRegisterAt{0}{x}}\CNOT_{\DatabaseRegisterAt{0}{x}\TargetRegister{0}}\ket{y}_{\TargetRegister{0}}\ket{\theta}_{\DatabaseRegisterAt{0}{x}}\ket{D - x}_{\DatabaseRegisterAt{0}{\RandomOracleDomain/\{x\}}}}\\
&\leq \frac{1}{\sqrt{2^{\RandomOracleOutputLength}}} + \frac{1}{\sqrt{2^{\RandomOracleOutputLength - 1}}}\vecnorm{\frac{1}{\sqrt{2^{\RandomOracleOutputLength + 1}}}\sum_{x, e, y, D \text{ s.t. } D(x) \neq \bot}\alpha_{x, y, e, D}\compress_{\DatabaseRegisterAt{0}{x}}\ket{x, e}_{\MessageRegister{0}\EnvironmentRegister{0}}\ket{y}_{\DatabaseRegisterAt{0}{x}}\ket{D - x}_{\DatabaseRegisterAt{0}{\RandomOracleDomain/\{x\}}}}\\
& + \frac{1}{\sqrt{2^{\RandomOracleOutputLength}}} \vecnorm{\sum_{x, e, D \text{ s.t. } D(x) \neq \bot}\alpha_{x, 0, e, D} \compress_{\DatabaseRegisterAt{0}{x}} \ket{x, e}_{\MessageRegister{0}\EnvironmentRegister{0}} \ket{\bot}_{\DatabaseRegisterAt{0}{x}}\ket{D}_{\DatabaseRegisterAt{0}{\RandomOracleDomain/\{x\}}}}\\
&= \frac{1}{\sqrt{2^{\RandomOracleOutputLength}}} + \frac{1}{2^{\RandomOracleOutputLength}}\vecnorm{\sum_{x, e, y, D \text{ s.t. } D(x) \neq \bot}\alpha_{x, y, e, D}\ket{x, e}_{\MessageRegister{0}\EnvironmentRegister{0}}\ket{y}_{\DatabaseRegisterAt{0}{x}}\ket{D - x}_{\DatabaseRegisterAt{0}{\RandomOracleDomain/\{x\}}}}\\
& + \frac{1}{\sqrt{2^{\RandomOracleOutputLength}}} \vecnorm{\sum_{x, e, D \text{ s.t. } D(x) \neq \bot}\alpha_{x, 0, e, D}\ket{x, e}_{\MessageRegister{0}\EnvironmentRegister{0}} \ket{D - x}_{\DatabaseRegister{0}}}\\
&= \frac{1}{\sqrt{2^{\RandomOracleOutputLength}}} + \frac{1}{2^{\RandomOracleOutputLength}}\sqrt{\sum_{x, e, y, D \text{ s.t. } D(x) = \bot}\abs{\sum_{z}\alpha_{x, y, e, D + [x \mapsto z]}}^2}\\
&+ \vecnorm{\left(\sum_x \ketbra{x}{x}_{\MessageRegister{0}} \otimes \ketbra{\hat{0}^\RandomOracleOutputLength}{\hat{0}^\RandomOracleOutputLength}_{\DatabaseRegisterAt{0}{x}}\right) \sum_{x, e, D}\alpha_{x, 0, e, D}\ket{x, 0, e, D}_{\MessageRegister{0}\TargetRegister{0}\EnvironmentRegister{0}\DatabaseRegister{0}}}\\
&\leq \frac{1}{\sqrt{2^{\RandomOracleOutputLength}}} + \frac{1}{2^{\RandomOracleOutputLength}}\sqrt{\sum_{x, e, y, D \text{ s.t. } D(x) = \bot}2^{\RandomOracleOutputLength}\cdot \sum_{z}\abs{\alpha_{x, y, e, D + [x \mapsto z]}}^2}+ \vecnorm{(\id{\DatabaseRegister{0}} - \ProjectNoHatZero) \PureStateSampleOne_{\MessageRegister{0}\TargetRegister{0}\EnvironmentRegister{0}\DatabaseRegister{0}}}\\
&\leq \frac{2}{\sqrt{2^{\RandomOracleOutputLength}}} + \vecnorm{(\id{\DatabaseRegister{0}} - \ProjectNoHatZero) \PureStateSampleOne_{\MessageRegister{0}\TargetRegister{0}\EnvironmentRegister{0}\DatabaseRegister{0}}}\enspace.
\end{align*}

Therefore, we have that
\begin{align*}
&\vecnorm{\bra{0}_{\WorkingRegister{0}\TargetRegister{0}}\CNOT_{\MessageRegister{0}\Another{\MessageRegister{0}}}\OracleUnitary_{\Another{\MessageRegister{0}}\TargetRegister{0}\DatabaseRegister{0}}\invert_{\TargetRegister{0}\DatabaseRegister{0}\WorkingRegister{0}}\PureStateSampleOne_{\MessageRegister{0}\TargetRegister{0}\EnvironmentRegister{0}\DatabaseRegister{0}}\ket{0}_{\WorkingRegister{0}} - \bra{0}_{\TargetRegister{0}}\OracleUnitary_{\MessageRegister{0}\TargetRegister{0}\DatabaseRegister{0}}\PureStateSampleOne_{\MessageRegister{0}\TargetRegister{0}\EnvironmentRegister{0}\DatabaseRegister{0}}} \\
&\leq \frac{4}{\sqrt{2^{\RandomOracleOutputLength}}}  + \vecnorm{(\id{\DatabaseRegister{0}} - \ProjectNoHatZero) \PureStateSampleOne_{\MessageRegister{0}\TargetRegister{0}\EnvironmentRegister{0}\DatabaseRegister{0}}}\\
&+ \|\sum_{\substack{x, e, y, D \text{ s.t. } \accessVectorAt{D}{x} = y}}\alpha_{x, y, e, D}\compress_{\DatabaseRegisterAt{0}{x}}\ket{x, e, D}_{\MessageRegister{0}\EnvironmentRegister{0}\DatabaseRegister{0}} - \sum_{\substack{x, e, y, D \text{ s.t. } \accessVectorAt{D}{x} = y\\ \text{and } \forall x'' < x, \accessVectorAt{D}{x''} \neq y}}\alpha_{x, y, e, D}\compress_{\DatabaseRegisterAt{0}{x}}\ket{x, e, D}_{\MessageRegister{0}\EnvironmentRegister{0}\DatabaseRegister{0}}\|\\
&= \frac{4}{\sqrt{2^{\RandomOracleOutputLength}}} + \vecnorm{(\id{\DatabaseRegister{0}} - \ProjectNoHatZero) \PureStateSampleOne_{\MessageRegister{0}\TargetRegister{0}\EnvironmentRegister{0}\DatabaseRegister{0}}} + \|\sum_{\substack{x, e, y, D \text{ s.t. } \accessVectorAt{D}{x} = y\\ \text{and } \exists x'' < x, \accessVectorAt{D}{x''} = y}}\alpha_{x, y, e, D}\ket{x, e, D}_{\MessageRegister{0}\EnvironmentRegister{0}\DatabaseRegister{0}}\|\\
&= \frac{4}{\sqrt{2^{\RandomOracleOutputLength}}} + \vecnorm{(\id{\DatabaseRegister{0}} - \ProjectNoHatZero) \PureStateSampleOne_{\MessageRegister{0}\TargetRegister{0}\EnvironmentRegister{0}\DatabaseRegister{0}}} +  \sqrt{\sum_{\substack{x, e, y, D \text{ s.t. } \accessVectorAt{D}{x} = y\\ \text{and } \exists x'' < x, \accessVectorAt{D}{x''} = y}}\abs{\alpha_{x, y, e, D}}^2}\enspace.
\end{align*}

Recall that $\PureStateSampleOne_{\MessageRegister{0}\TargetRegister{0}\EnvironmentRegister{0}\DatabaseRegister{0}} \coloneq \sum_{x \in \RandomOracleDomain, y \in \Bits^{\RandomOracleOutputLength}, e, D}\alpha_{x, y, e, D}\ket{x, y, e, D}_{\MessageRegister{0}\TargetRegister{0}\EnvironmentRegister{0}\DatabaseRegister{0}}$. We can get that
\begin{align*}
\vecnorm{(\id{\DatabaseRegister{0}} - \ProjectNoCollision)\PureStateSampleOne_{\MessageRegister{0}\TargetRegister{0}\EnvironmentRegister{0}\DatabaseRegister{0}}}^2 = \sum_{x, y, e, D \notin \SetOfNoCollisionDatabase}\abs{\alpha_{x, y, e, D}}^2\enspace.
\end{align*}

Thus by the non-negativity of each $\abs{\alpha_{x, y, e, D}}$,
\begin{align*}
&\vecnorm{\bra{0}_{\WorkingRegister{0}\TargetRegister{0}}\CNOT_{\MessageRegister{0}\Another{\MessageRegister{0}}}\OracleUnitary_{\Another{\MessageRegister{0}}\TargetRegister{0}\DatabaseRegister{0}}\invert_{\TargetRegister{0}\DatabaseRegister{0}\WorkingRegister{0}}\PureStateSampleOne_{\MessageRegister{0}\TargetRegister{0}\EnvironmentRegister{0}\DatabaseRegister{0}}\ket{0}_{\WorkingRegister{0}} - \bra{0}_{\TargetRegister{0}}\OracleUnitary_{\MessageRegister{0}\TargetRegister{0}\DatabaseRegister{0}}\PureStateSampleOne_{\MessageRegister{0}\TargetRegister{0}\EnvironmentRegister{0}\DatabaseRegister{0}}} \\
&\leq \frac{4}{\sqrt{2^{\RandomOracleOutputLength}}} + \vecnorm{(\id{\DatabaseRegister{0}} - \ProjectNoCollision)\PureStateSampleOne_{\MessageRegister{0}\TargetRegister{0}\EnvironmentRegister{0}\DatabaseRegister{0}}} + \vecnorm{(\id{\DatabaseRegister{0}} - \ProjectNoHatZero) \PureStateSampleOne_{\MessageRegister{0}\TargetRegister{0}\EnvironmentRegister{0}\DatabaseRegister{0}}}\enspace.
\end{align*}

We now prove the second inequality. Define
\begin{align*}
\PureStateSampleTwoI{0}
\coloneq&
\ProjectNoCollision\OracleUnitary_{\MessageRegister{0}\TargetRegister{0}\DatabaseRegister{0}}\PureStateSampleTwo_{\MessageRegister{0}\EnvironmentRegister{0}\DatabaseRegister{0}}\ket{0}_{\WorkingRegister{0}\TargetRegister{0}}\enspace,\\
\PureStateSampleTwoI{1}
\coloneq&
\ProjectNoCollision\invert_{\TargetRegister{0}\DatabaseRegister{0}\WorkingRegister{0}}\OracleUnitary_{\Another{\MessageRegister{0}}\TargetRegister{0}\DatabaseRegister{0}}\CNOT_{\MessageRegister{0}\Another{\MessageRegister{0}}}\PureStateSampleTwo_{\MessageRegister{0}\EnvironmentRegister{0}\DatabaseRegister{0}}\ket{0}_{\WorkingRegister{0}\TargetRegister{0}}\\
=&\ProjectNoCollision\invert_{\TargetRegister{0}\DatabaseRegister{0}\WorkingRegister{0}}\CNOT_{\MessageRegister{0}\Another{\MessageRegister{0}}}\OracleUnitary_{\MessageRegister{0}\TargetRegister{0}\DatabaseRegister{0}}\PureStateSampleTwo_{\MessageRegister{0}\EnvironmentRegister{0}\DatabaseRegister{0}}\ket{0}_{\WorkingRegister{0}\TargetRegister{0}}\\
=&\invert_{\TargetRegister{0}\DatabaseRegister{0}\WorkingRegister{0}}\CNOT_{\MessageRegister{0}\Another{\MessageRegister{0}}}\PureStateSampleTwoI{0}\enspace,
\end{align*}
where we use that copying the value in $\MessageRegister{1}$ and querying $\Another{\MessageRegister{1}}$ is equivalent to querying $\MessageRegister{1}$ and then performing the copy because $\Another{\MessageRegister{1}}$ is initialized as all-zero states.

We show that to bound the left hand side of the second inequality, which equals $\matnorm{\PureStateSampleTwoI{0} - \PureStateSampleTwoI{1}}$, it suffices to bound the error $\vecnorm{\Pi_{\ne}\PureStateSampleTwoI{0}}$, where $\Pi_{\neq} \coloneq \sum_{x, y, D \text{ s.t. } \accessVectorAt{D}{x} \neq y}\ketbra{x, y, D}{x, y, D}_{\MessageRegister{0}\TargetRegister{0}\DatabaseRegister{0}}$ projects to the bad branches such that $\accessVectorAt{D}{x} \neq y$ after making the query $\OracleUnitary_{\MessageRegister{0}\TargetRegister{0}\DatabaseRegister{0}}$ where $\TargetRegister{0}$ is initialized as all-zero states. 

This is because, for $\PureStateSampleTwoI{0} \coloneq \sum_{x, y, e, D \in \SetOfNoCollisionDatabase}\alpha_{x, y, e, D}\ket{x, y, e, D}_{\MessageRegister{0}\TargetRegister{0}\EnvironmentRegister{0}\DatabaseRegister{0}}\ket{0}_{\WorkingRegister{0}}$,
\begin{align*}
	\PureStateSampleTwoI{1} = \invert_{\TargetRegister{0}\DatabaseRegister{0}\WorkingRegister{0}}\CNOT_{\MessageRegister{0}\Another{\MessageRegister{0}}}\PureStateSampleTwoI{0} = \sum_{x, y, e, D \in \SetOfNoCollisionDatabase}\alpha_{x, y, e, D}\ket{x, y, e, D}_{\MessageRegister{0}\TargetRegister{0}\EnvironmentRegister{0}\DatabaseRegister{0}}\ket{g(x, y, D)}_{\WorkingRegister{0}}\enspace,
\end{align*}
for a classical function $g$ such that $g(x, y, D) = 0$ if and only if $\accessVectorAt{D}{x} = y$. Thus
\begin{align*}
	\vecnorm{\PureStateSampleTwoI{1} - \PureStateSampleTwoI{0}} &= \vecnorm{\sum_{x, y, e, D \in \SetOfNoCollisionDatabase}\alpha_{x, y, e, D}\ket{x, y, e, D}_{\MessageRegister{0}\TargetRegister{0}\EnvironmentRegister{0}\DatabaseRegister{0}}(\ket{g(x, y, D)}_{\WorkingRegister{0}} - \ket{0}_{\WorkingRegister{0}})}\\
	&= \sqrt{2} \vecnorm{\sum_{x, y, e, D \in \SetOfNoCollisionDatabase \text{ s.t. } g(x, y, D) \neq 0}\alpha_{x, y, e, D}\ket{x, y, e, D}_{\MessageRegister{0}\TargetRegister{0}\EnvironmentRegister{0}\DatabaseRegister{0}}}\\
	&= \sqrt{2}\vecnorm{\Pi_{\ne}\PureStateSampleTwoI{0}}\enspace.
\end{align*}

Intuitively, the error $\vecnorm{\Pi_{\ne}\PureStateSampleTwoI{0}}$ should be small because $\TargetRegister{1}$ is initialized as 0, and after making a query, it should store $\accessVectorAt{D}{x}$. This can be shown by direct calculation. 

We decompose the input as $\PureStateSampleTwo=\ProjectNoHatZero\PureStateSampleTwo +(\id{\DatabaseRegister{0}}-\ProjectNoHatZero)\PureStateSampleTwo$. Then
\begin{align*}
&\vecnorm{\Pi_{\ne}\PureStateSampleTwoI{0}}\\
&= \vecnorm{\Pi_{\ne}\ProjectNoCollision\OracleUnitary_{\MessageRegister{0}\TargetRegister{0}\DatabaseRegister{0}}\PureStateSampleTwo_{\MessageRegister{0}\EnvironmentRegister{0}\DatabaseRegister{0}}\ket{0}_{\WorkingRegister{0}\TargetRegister{0}}}\\
&\leq \vecnorm{\Pi_{\ne}\OracleUnitary_{\MessageRegister{0}\TargetRegister{0}\DatabaseRegister{0}}\PureStateSampleTwo_{\MessageRegister{0}\EnvironmentRegister{0}\DatabaseRegister{0}}\ket{0}_{\TargetRegister{0}}}\\
&\leq \vecnorm{\Pi_{\ne}\OracleUnitary_{\MessageRegister{0}\TargetRegister{0}\DatabaseRegister{0}}\ProjectNoHatZero\PureStateSampleTwo_{\MessageRegister{0}\EnvironmentRegister{0}\DatabaseRegister{0}}\ket{0}_{\TargetRegister{0}}} + \vecnorm{\Pi_{\ne}\OracleUnitary_{\MessageRegister{0}\TargetRegister{0}\DatabaseRegister{0}}(\id{\DatabaseRegister{0}}-\ProjectNoHatZero)\PureStateSampleTwo_{\MessageRegister{0}\EnvironmentRegister{0}\DatabaseRegister{0}}\ket{0}_{\TargetRegister{0}}}\\
&\leq \vecnorm{\Pi_{\ne}\compress_{\DatabaseRegister{0}}\sum_{x}\ketbra{x}{x}_{\MessageRegister{0}}\otimes \CNOT_{\DatabaseRegisterAt{0}{x}\TargetRegister{0}}\compress_{\DatabaseRegister{0}}\ProjectNoHatZero\PureStateSampleTwo_{\MessageRegister{0}\EnvironmentRegister{0}\DatabaseRegister{0}}\ket{0}_{\TargetRegister{0}}} + \vecnorm{(\id{\DatabaseRegister{0}}-\ProjectNoHatZero)\PureStateSampleTwo_{\MessageRegister{0}\EnvironmentRegister{0}\DatabaseRegister{0}}}\enspace.
\end{align*}

We write the state $\compress_{\DatabaseRegister{0}}\ProjectNoHatZero\PureStateSampleTwo_{\MessageRegister{0}\EnvironmentRegister{0}\DatabaseRegister{0}} = \sum_{x, e, f}\alpha_{x, e, f}\ket{x, e, f}_{\MessageRegister{0}\EnvironmentRegister{0}\DatabaseRegister{0}}$. Here $f$ does not contain $\bot$ on any location since the state before applying $\compress_{\DatabaseRegister{0}}$ is supported over $\ProjectNoHatZero$. We bound the first term in the above inequality:
\begin{align*}
	&\vecnorm{\Pi_{\ne}\compress_{\DatabaseRegister{0}}\sum_{x}\ketbra{x}{x}_{\MessageRegister{0}}\otimes \CNOT_{\DatabaseRegisterAt{0}{x}\TargetRegister{0}}\compress_{\DatabaseRegister{0}}\ProjectNoHatZero\PureStateSampleTwo_{\MessageRegister{0}\EnvironmentRegister{0}\DatabaseRegister{0}}\ket{0}_{\TargetRegister{0}}}\\
	&=\vecnorm{\Pi_{\ne}\compress_{\DatabaseRegister{0}}\sum_{x, e, f}\alpha_{x, e, f}\ket{x, e, f}_{\MessageRegister{0}\EnvironmentRegister{0}\DatabaseRegister{0}}\ket{\accessVectorAt{f}{x}}_{\TargetRegister{0}}}\\
	&= \vecnorm{\Pi_{\ne}\sum_{x, e, f}\alpha_{x, e, f}\compress_{\DatabaseRegisterAt{0}{x}}\ket{x, e, f}_{\MessageRegister{0}\EnvironmentRegister{0}\DatabaseRegister{0}}\ket{\accessVectorAt{f}{x}}_{\TargetRegister{0}}}\\
	&= \vecnorm{\Pi_{\ne}\sum_{x, e, f}\alpha_{x, e, f}\ket{x, e, f - x}_{\MessageRegister{0}\EnvironmentRegister{0}\DatabaseRegisterAt{0}{\RandomOracleDomain/\{x\}}}(\ket{f(x)} + \frac{1}{\sqrt{2^{\RandomOracleOutputLength - 1}}}\ket{\theta})_{\DatabaseRegisterAt{0}{x}}\ket{\accessVectorAt{f}{x}}_{\TargetRegister{0}}}\\
	&= \frac{1}{\sqrt{2^{\RandomOracleOutputLength - 1}}}\vecnorm{\Pi_{\ne}\sum_{x, e, f}\alpha_{x, e, f}\ket{x, e, f - x}_{\MessageRegister{0}\EnvironmentRegister{0}\DatabaseRegisterAt{0}{\RandomOracleDomain/\{x\}}}\ket{\theta}_{\DatabaseRegisterAt{0}{x}}\ket{\accessVectorAt{f}{x}}_{\TargetRegister{0}}}\\
	&\leq \frac{1}{\sqrt{2^{\RandomOracleOutputLength - 1}}}\enspace.
\end{align*} 

The second inequality of this claim follows from the above equations.
\end{proof}

\iffull
\begin{claim}
\else
\begin{numberedclaim}
\fi
\label{claim:diff-between-the-two-operations}
Define $V_\Check \coloneq \ProjectNoCollision \OracleUnitary_{\MessageRegister{0}\StandardBasisHashValueRegister{0}\DatabaseRegister{0}}\ProjectNoCollision\HadamardGate^{\otimes \MessageLength}_{\MessageRegister{0}}\OracleUnitary_{\MessageRegister{0}\HadamardBasisHashValueRegister{0}\DatabaseRegister{0}}\ProjectNoCollision\HadamardGate^{\otimes \MessageLength}_{\MessageRegister{0}}$ to be the operation that $\NoCollisionAlgVariant{\Check}$ does, and define $\ket{\NullStateInExp_{\Check}} \coloneq \ket{0^{2\RandomOracleOutputLength}}_{\StandardBasisHashValueRegister{0}\HadamardBasisHashValueRegister{0}}$ to be the state that $\NoCollisionAlgVariant{\Check}$ projects onto to see if the opening is valid.

Define $V_{\ExtractMessage} \coloneq \SWAP_{\StandardBasisMessageRegister{0}\MessageRegister{0}} \CNOT_{\StandardBasisMessageRegister{0}\HadamardBasisMessageRegister{0}}\HadamardGate_{\HadamardBasisMessageRegister{0}}^{\otimes \MessageLength}\ProjectNoCollision\OracleUnitary_{\StandardBasisMessageRegister{0}\StandardBasisHashValueRegister{0}\DatabaseRegister{0}}\ProjectNoCollision\invert_{\StandardBasisHashValueRegister{0}\DatabaseRegister{0}\WorkingRegister{0}_1}\OracleUnitary_{\HadamardBasisMessageRegister{0}\HadamardBasisHashValueRegister{0}\DatabaseRegister{0}}\ProjectNoCollision\invert_{\HadamardBasisHashValueRegister{0}\DatabaseRegister{0}\WorkingRegister{0}_2}$ to be the operation that $\Extractor.\NoCollisionAlgVariant{\ExtractMessage}$ does ($\SWAP_{\StandardBasisMessageRegister{0}\MessageRegister{0}}$ is just for renaming the registers), and define $\ket{\NullStateInExp_{\AltCheck}} \coloneq \frac{1}{\sqrt{|\RandomOracleDomain|}}\sum_{x \in \RandomOracleDomain}\ket{0, 0, 0^\RandomOracleOutputLength, 0^\RandomOracleOutputLength, x, x}_{\StandardBasisSuccessRegister{0}\HadamardBasisSuccessRegister{0}\StandardBasisHashValueRegister{0}\HadamardBasisHashValueRegister{0}\StandardBasisMessageRegister{0}\HadamardBasisMessageRegister{0}}$ to be the state that $\Extractor.\NoCollisionAlgVariant{\ExtractMessage}$ projects onto to see if the opening is valid.

Then for every subnormalized state $\PureStateSampleOne_{\CommitmentRegister{0}\OpeningRegister{0}\DatabaseRegister{0}\EnvironmentRegister{0}}$,
\begin{align*}
&\matnorm{(\bra{\NullStateInExp_{\Check}}V_{\Check} - \bra{\NullStateInExp_{\AltCheck}}V_{\ExtractMessage}\ket{0}_{\StandardBasisWorkingRegister{0}\HadamardBasisWorkingRegister{0}})\PureStateSampleOne_{\CommitmentRegister{0}\OpeningRegister{0}\DatabaseRegister{0}\EnvironmentRegister{0}}}\\
&\leq \frac{8}{\sqrt{2^\RandomOracleOutputLength}} + \vecnorm{(\id{\DatabaseRegister{0}} - \ProjectNoHatZero) \PureStateSampleOneI{1}}	+ \vecnorm{(\id{\DatabaseRegister{0}} - \ProjectNoHatZero) \PureStateSampleOneI{2}}\enspace,
\end{align*}
where $\PureStateSampleOneI{1} \coloneq \ProjectNoCollision\HadamardGate^{\otimes \MessageLength}_{\MessageRegister{0}}\OracleUnitary_{\MessageRegister{0}\HadamardBasisHashValueRegister{0}\DatabaseRegister{0}}\ProjectNoCollision\HadamardGate^{\otimes \MessageLength}_{\MessageRegister{0}}\PureStateSampleOne_{\CommitmentRegister{0}\OpeningRegister{0}\DatabaseRegister{0}\EnvironmentRegister{0}}$ and $\PureStateSampleOneI{2} \coloneq \ProjectNoCollision\HadamardGate^{\otimes \MessageLength}_{\MessageRegister{0}}\PureStateSampleOne_{\CommitmentRegister{0}\OpeningRegister{0}\DatabaseRegister{0}\EnvironmentRegister{0}}$.

Furthermore, for every subnormalized state $\PureStateSampleTwo_{\MessageRegister{0}\EnvironmentRegister{0}\DatabaseRegister{0}}$,
\begin{align*}
&\matnorm{(V_{\Check}^\dagger\ket{\NullStateInExp_{\Check}}\ket{0}_{\StandardBasisWorkingRegister{0}\HadamardBasisWorkingRegister{0}} - V_{\ExtractMessage}^\dagger\ket{\NullStateInExp_{\AltCheck}})\PureStateSampleTwo_{\MessageRegister{0}\EnvironmentRegister{0}\DatabaseRegister{0}}}\\
&\leq \frac{4}{\sqrt{2^\RandomOracleOutputLength}} + \sqrt{2}\vecnorm{(\id{\DatabaseRegister{0}} - \ProjectNoHatZero) \PureStateSampleTwoI{1}}	+ \sqrt{2}\vecnorm{(\id{\DatabaseRegister{0}} - \ProjectNoHatZero) \PureStateSampleTwoI{2}}\enspace,
\end{align*}
where $\PureStateSampleTwoI{1} \coloneq \HadamardGate^{\otimes \MessageLength}_{\MessageRegister{0}}\ProjectNoCollision\OracleUnitary_{\MessageRegister{0}\StandardBasisHashValueRegister{0}\DatabaseRegister{0}}\ProjectNoCollision\PureStateSampleTwo_{\MessageRegister{0}\EnvironmentRegister{0}\DatabaseRegister{0}}\ket{\NullStateInExp_{\Check}}$ and $\PureStateSampleTwoI{2} \coloneq \ProjectNoCollision\PureStateSampleTwo_{\MessageRegister{0}\EnvironmentRegister{0}\DatabaseRegister{0}}$.
\iffull
\end{claim}
\else
\end{numberedclaim}
\fi

\begin{proof}
For the first inequality, invoking \Cref{claim:diff=collision} on the subnormalized state $\PureStateSampleOneI{1}$, we obtain that
\begin{align}
\label{eqn:diff-the-first-O}
&\vecnorm{\bra{0}_{\WorkingRegister{0}_1\StandardBasisHashValueRegister{0}}\CNOT_{\MessageRegister{0}\StandardBasisMessageRegister{0}}\OracleUnitary_{\StandardBasisMessageRegister{0}\StandardBasisHashValueRegister{0}\DatabaseRegister{0}}\invert_{\StandardBasisHashValueRegister{0}\DatabaseRegister{0}\WorkingRegister{0}_1}\PureStateSampleOneI{1}\ket{0}_{\WorkingRegister{0}_1} - \bra{0}_{\StandardBasisHashValueRegister{0}}\OracleUnitary_{\MessageRegister{0}\StandardBasisHashValueRegister{0}\DatabaseRegister{0}}\PureStateSampleOneI{1}} \nonumber\\
&\leq \frac{4}{\sqrt{2^{\RandomOracleOutputLength}}} + \vecnorm{(\id{\DatabaseRegister{0}} - \ProjectNoCollision)\PureStateSampleOneI{1}} + \vecnorm{(\id{\DatabaseRegister{0}} - \ProjectNoHatZero) \PureStateSampleOneI{1}}\nonumber\\
&= \frac{4}{\sqrt{2^{\RandomOracleOutputLength}}} + \vecnorm{(\id{\DatabaseRegister{0}} - \ProjectNoHatZero) \PureStateSampleOneI{1}}\enspace,
\end{align}
where $\WorkingRegister{1}_1 \coloneq (\StandardBasisMessageRegister{1}, \StandardBasisSuccessRegister{1})$, $ \MessageRegister{1} \coloneq \OpeningRegister{1}$, and $(\StandardBasisHashValueRegister{1}, \HadamardBasisHashValueRegister{1}) \coloneq \CommitmentRegister{1}$.

Notice that by the definition of $V_\Check$ and $\ket{\NullStateInExp_{\Check}}$, \[\bra{\NullStateInExp_{\Check}}V_{\Check}\PureStateSampleOne_{\CommitmentRegister{0}\OpeningRegister{0}\DatabaseRegister{0}\EnvironmentRegister{0}} = \bra{0}_{\HadamardBasisHashValueRegister{0}}\bra{0}_{\StandardBasisHashValueRegister{0}}\ProjectNoCollision\OracleUnitary_{\MessageRegister{0}\StandardBasisHashValueRegister{0}\DatabaseRegister{0}}\PureStateSampleOneI{1}\enspace.\]

By \Cref{eqn:diff-the-first-O} and the triangle inequality,
\begin{align}
\label{eqn:replace-the-first-O}
&\vecnorm{\bra{0}_{\WorkingRegister{0}_1\StandardBasisHashValueRegister{0}\HadamardBasisHashValueRegister{0}}\ProjectNoCollision\CNOT_{\MessageRegister{0}\StandardBasisMessageRegister{0}}\OracleUnitary_{\StandardBasisMessageRegister{0}\StandardBasisHashValueRegister{0}\DatabaseRegister{0}}\invert_{\StandardBasisHashValueRegister{0}\DatabaseRegister{0}\WorkingRegister{0}_1}\PureStateSampleOneI{1}\ket{0}_{\WorkingRegister{0}_1} - \bra{\NullStateInExp_{\Check}}V_{\Check}\PureStateSampleOne_{\CommitmentRegister{0}\OpeningRegister{0}\DatabaseRegister{0}\EnvironmentRegister{0}}}\nonumber\\
&\leq \frac{4}{\sqrt{2^{\RandomOracleOutputLength}}} + \vecnorm{(\id{\DatabaseRegister{0}} - \ProjectNoHatZero) \PureStateSampleOneI{1}}\enspace.
\end{align}

Again invoking \Cref{claim:diff=collision} on the subnormalized state $\PureStateSampleOneI{2}$, we obtain that
\begin{align}
\label{eqn:diff-the-second-O}
&\vecnorm{\bra{0}_{\WorkingRegister{0}_2\HadamardBasisHashValueRegister{0}}\CNOT_{\MessageRegister{0}\HadamardBasisMessageRegister{0}}\OracleUnitary_{\HadamardBasisMessageRegister{0}\HadamardBasisHashValueRegister{0}\DatabaseRegister{0}}\invert_{\HadamardBasisHashValueRegister{0}\DatabaseRegister{0}\WorkingRegister{0}_2}\PureStateSampleOneI{2}\ket{0}_{\WorkingRegister{0}_2} - \bra{0}_{\HadamardBasisHashValueRegister{0}}\OracleUnitary_{\MessageRegister{0}\HadamardBasisHashValueRegister{0}\DatabaseRegister{0}}\PureStateSampleOneI{2}} \nonumber\\
&\leq \frac{4}{\sqrt{2^{\RandomOracleOutputLength}}} + \vecnorm{(\id{\DatabaseRegister{0}} - \ProjectNoCollision)\PureStateSampleOneI{2}} + \vecnorm{(\id{\DatabaseRegister{0}} - \ProjectNoHatZero) \PureStateSampleOneI{2}}\nonumber\\
&\leq \frac{4}{\sqrt{2^{\RandomOracleOutputLength}}} + \vecnorm{(\id{\DatabaseRegister{0}} - \ProjectNoHatZero) \PureStateSampleOneI{2}}\enspace,
\end{align}
where $\WorkingRegister{1}_2 \coloneq (\HadamardBasisMessageRegister{1}, \HadamardBasisSuccessRegister{1})$.

Notice that $\PureStateSampleOneI{1} = \ProjectNoCollision\HadamardGate^{\otimes \MessageLength}_{\MessageRegister{0}}\OracleUnitary_{\MessageRegister{0}\HadamardBasisHashValueRegister{0}\DatabaseRegister{0}}\PureStateSampleOneI{2}$. By \Cref{eqn:diff-the-second-O} and the triangle inequality,
\begin{align}
\label{eqn:replace-the-second-O}
&\vecnorm{\bra{0}_{\WorkingRegister{0}_2\HadamardBasisHashValueRegister{0}}\ProjectNoCollision\HadamardGate^{\otimes \MessageLength}_{\MessageRegister{0}}\CNOT_{\MessageRegister{0}\HadamardBasisMessageRegister{0}}\OracleUnitary_{\HadamardBasisMessageRegister{0}\HadamardBasisHashValueRegister{0}\DatabaseRegister{0}}\invert_{\HadamardBasisHashValueRegister{0}\DatabaseRegister{0}\WorkingRegister{0}_2}\PureStateSampleOneI{2}\ket{0}_{\WorkingRegister{0}_2} - \bra{0}_{\HadamardBasisHashValueRegister{0}}\PureStateSampleOneI{1}}\nonumber\\
&\leq \frac{4}{\sqrt{2^{\RandomOracleOutputLength}}} + \vecnorm{(\id{\DatabaseRegister{0}} - \ProjectNoHatZero) \PureStateSampleOneI{2}}\enspace.
\end{align}

\Cref{eqn:replace-the-first-O,eqn:replace-the-second-O} imply that
\begin{align*}
	&\vecnorm{\bra{\NullStateInExp_{\Check}}V_{\Check}\PureStateSampleOne_{\CommitmentRegister{0}\OpeningRegister{0}\DatabaseRegister{0}\EnvironmentRegister{0}} - \bra{0}_{\WorkingRegister{0}_1\WorkingRegister{0}_2\StandardBasisHashValueRegister{0}\HadamardBasisHashValueRegister{0}}V_{\Check}'\PureStateSampleOne_{\CommitmentRegister{0}\OpeningRegister{0}\DatabaseRegister{0}\EnvironmentRegister{0}}\ket{0}_{\WorkingRegister{0}_1\WorkingRegister{0}_2}}\\
	&\leq \frac{8}{\sqrt{2^{\RandomOracleOutputLength}}} + \vecnorm{(\id{\DatabaseRegister{0}} - \ProjectNoHatZero) \PureStateSampleOneI{1}} + \vecnorm{(\id{\DatabaseRegister{0}} - \ProjectNoHatZero) \PureStateSampleOneI{2}}\enspace,
\end{align*}
where $V_{\Check}' \coloneq \ProjectNoCollision\CNOT_{\MessageRegister{0}\StandardBasisMessageRegister{0}}\OracleUnitary_{\StandardBasisMessageRegister{0}\StandardBasisHashValueRegister{0}\DatabaseRegister{0}}\invert_{\StandardBasisHashValueRegister{0}\DatabaseRegister{0}\WorkingRegister{0}_1}\ProjectNoCollision\HadamardGate^{\otimes \MessageLength}_{\MessageRegister{0}}\CNOT_{\MessageRegister{0}\HadamardBasisMessageRegister{0}}\OracleUnitary_{\HadamardBasisMessageRegister{0}\HadamardBasisHashValueRegister{0}\DatabaseRegister{0}}\invert_{\HadamardBasisHashValueRegister{0}\DatabaseRegister{0}\WorkingRegister{0}_2}\ProjectNoCollision\HadamardGate^{\otimes \MessageLength}_{\MessageRegister{0}}$.

It remains to show that
$\bra{0}_{\WorkingRegister{0}_1\WorkingRegister{0}_2\StandardBasisHashValueRegister{0}\HadamardBasisHashValueRegister{0}}V_{\Check}'\ket{0}_{\WorkingRegister{0}_1\WorkingRegister{0}_2} = \bra{\NullStateInExp_{\AltCheck}}V_{\ExtractMessage}\ket{0}_{\WorkingRegister{0}_1\WorkingRegister{0}_2}$.

We use \Cref{claim:check_0=check_EPR} to prove the above equation. Let $\WorkingRegister{1} \coloneq (\WorkingRegister{1}_1, \WorkingRegister{1}_2)$.
\begin{align*}
&\bra{0}_{\WorkingRegister{0}_1\WorkingRegister{0}_2\StandardBasisHashValueRegister{0}\HadamardBasisHashValueRegister{0}}V_{\Check}'\ket{0}_{\WorkingRegister{0}}\\
&= \bra{0}_{\WorkingRegister{0}_1\WorkingRegister{0}_2\StandardBasisHashValueRegister{0}\HadamardBasisHashValueRegister{0}}\ProjectNoCollision\CNOT_{\MessageRegister{0}\StandardBasisMessageRegister{0}}\OracleUnitary_{\StandardBasisMessageRegister{0}\StandardBasisHashValueRegister{0}\DatabaseRegister{0}}\invert_{\StandardBasisHashValueRegister{0}\DatabaseRegister{0}\WorkingRegister{0}_1}\ProjectNoCollision\HadamardGate^{\otimes \MessageLength}_{\MessageRegister{0}}\CNOT_{\MessageRegister{0}\HadamardBasisMessageRegister{0}}\OracleUnitary_{\HadamardBasisMessageRegister{0}\HadamardBasisHashValueRegister{0}\DatabaseRegister{0}}\invert_{\HadamardBasisHashValueRegister{0}\DatabaseRegister{0}\WorkingRegister{0}_2}\ProjectNoCollision\HadamardGate^{\otimes \MessageLength}_{\MessageRegister{0}}\ket{0}_{\WorkingRegister{0}}\\
&= \bra{0}_{\WorkingRegister{0}_1\WorkingRegister{0}_2\StandardBasisHashValueRegister{0}\HadamardBasisHashValueRegister{0}}\ProjectNoCollision\CNOT_{\MessageRegister{0}\StandardBasisMessageRegister{0}}\HadamardGate^{\otimes \MessageLength}_{\MessageRegister{0}}\CNOT_{\MessageRegister{0}\HadamardBasisMessageRegister{0}}\HadamardGate^{\otimes \MessageLength}_{\MessageRegister{0}}\OracleUnitary_{\StandardBasisMessageRegister{0}\StandardBasisHashValueRegister{0}\DatabaseRegister{0}}\invert_{\StandardBasisHashValueRegister{0}\DatabaseRegister{0}\WorkingRegister{0}_1}\ProjectNoCollision\OracleUnitary_{\HadamardBasisMessageRegister{0}\HadamardBasisHashValueRegister{0}\DatabaseRegister{0}}\invert_{\HadamardBasisHashValueRegister{0}\DatabaseRegister{0}\WorkingRegister{0}_2}\ProjectNoCollision\ket{0}_{\WorkingRegister{0}}\\
&= \bra{\CheckStateEPR}_{\StandardBasisMessageRegister{0}\HadamardBasisMessageRegister{0}}\bra{0}_{\StandardBasisSuccessRegister{0}\HadamardBasisSuccessRegister{0}\StandardBasisHashValueRegister{0}\HadamardBasisHashValueRegister{0}}\SWAP_{\StandardBasisMessageRegister{0}\MessageRegister{0}}\CNOT_{\StandardBasisMessageRegister{0}\HadamardBasisMessageRegister{0}}\HadamardGate_{\HadamardBasisMessageRegister{0}}^{\otimes \MessageLength}\ProjectNoCollision\OracleUnitary_{\StandardBasisMessageRegister{0}\StandardBasisHashValueRegister{0}\DatabaseRegister{0}}\ProjectNoCollision\invert_{\StandardBasisHashValueRegister{0}\DatabaseRegister{0}\WorkingRegister{0}_1}\OracleUnitary_{\HadamardBasisMessageRegister{0}\HadamardBasisHashValueRegister{0}\DatabaseRegister{0}}\ProjectNoCollision\invert_{\HadamardBasisHashValueRegister{0}\DatabaseRegister{0}\WorkingRegister{0}_2}\ket{0}_{\WorkingRegister{0}}\\
&= \bra{\NullStateInExp_{\AltCheck}}V_{\ExtractMessage}\ket{0}_{\WorkingRegister{0}}\enspace, \tag{Definitions of $\ket{\NullStateInExp_{\AltCheck}}$ and $V_{\ExtractMessage}$}
\end{align*}
where we use the commutativity of unitaries on different registers, and the commutativity of $\ProjectNoCollision$ and $\invert$.

The second inequality follows similarly by invoking \Cref{claim:diff=collision} on the states $\PureStateSampleTwoI{1}$ and $\PureStateSampleTwoI{2}$. 
\end{proof}

Now we are ready to prove \Cref{lemma:indistinguishablility-if-no-collisions}.
\begin{proof}[Proof of \Cref{lemma:indistinguishablility-if-no-collisions}]
By the triangle inequality,
\begin{align}\label{eqn:sqrt-prob-state-diff}
&\abs{\sqrt{\prob{\NoCollisionVariant{\SimulateWorld}(\Adversary, \Distinguisher)}} - \sqrt{\prob{\NoCollisionVariant{\OfflineExtractWorld}(\Adversary, \Distinguisher)}}}^2\nonumber\\
&\leq \matnorm{(\bra{\NullStateInExp_{\Check}}V_{\Check} - \bra{\NullStateInExp_{\AltCheck}}V_{\ExtractMessage}\ket{0}_{\StandardBasisWorkingRegister{0}\HadamardBasisWorkingRegister{0}})\PureStateSampleOne_{\CommitmentRegister{0}\OpeningRegister{0}\DatabaseRegister{0}\EnvironmentRegister{0}}}^2\enspace,
\end{align}
where $\PureStateSampleOne_{\CommitmentRegister{0}\OpeningRegister{0}\DatabaseRegister{0}\EnvironmentRegister{0}}$ is the joint state after the adversary provides the opening.

By \Cref{claim:diff-between-the-two-operations}, 
\begin{align*}
&\matnorm{(\bra{\NullStateInExp_{\Check}}V_{\Check} - \bra{\NullStateInExp_{\AltCheck}}V_{\ExtractMessage}\ket{0}_{\StandardBasisWorkingRegister{0}\HadamardBasisWorkingRegister{0}})\PureStateSampleOne_{\CommitmentRegister{0}\OpeningRegister{0}\DatabaseRegister{0}\EnvironmentRegister{0}}}\\
&\leq \frac{8}{\sqrt{2^\RandomOracleOutputLength}} + \vecnorm{(\id{\DatabaseRegister{0}} - \ProjectNoHatZero) \PureStateSampleOneI{1}}	+ \vecnorm{(\id{\DatabaseRegister{0}} - \ProjectNoHatZero) \PureStateSampleOneI{2}}\enspace,
\end{align*}
where $\PureStateSampleOneI{1} \coloneq \ProjectNoCollision\HadamardGate^{\otimes \MessageLength}_{\MessageRegister{0}}\OracleUnitary_{\MessageRegister{0}\HadamardBasisHashValueRegister{0}\DatabaseRegister{0}}\ProjectNoCollision\HadamardGate^{\otimes \MessageLength}_{\MessageRegister{0}}\PureStateSampleOne_{\CommitmentRegister{0}\OpeningRegister{0}\DatabaseRegister{0}\EnvironmentRegister{0}}$ and $\PureStateSampleOneI{2} \coloneq \ProjectNoCollision\HadamardGate^{\otimes \MessageLength}_{\MessageRegister{0}}\PureStateSampleOne_{\CommitmentRegister{0}\OpeningRegister{0}\DatabaseRegister{0}\EnvironmentRegister{0}}$.

We bound the term $\vecnorm{(\id{\DatabaseRegister{0}} - \ProjectNoHatZero) \PureStateSampleOneI{1}}$. Notice that $\PureStateSampleOne_{\CommitmentRegister{0}\OpeningRegister{0}\DatabaseRegister{0}\EnvironmentRegister{0}}$  has database size at most $\NumberOfQueries$,
\begin{align*}
	&\vecnorm{(\id{\DatabaseRegister{0}} - \ProjectNoHatZero) \PureStateSampleOneI{1}}\\
	&=\vecnorm{(\id{\DatabaseRegister{0}} - \ProjectNoHatZero) \ProjectSizeDatabase{\NumberOfQueries + 1}\PureStateSampleOneI{1}}\\
	&\leq 2\matnorm{\ProjectSizeDatabase{\NumberOfQueries + 1}\Commutator{\NoCollision}{\ProjectNoHatZero}\ProjectSizeDatabase{\NumberOfQueries + 1}} + \vecnorm{(\id{\DatabaseRegister{0}} - \ProjectNoHatZero) \PureStateSampleOne_{\CommitmentRegister{0}\OpeningRegister{0}\DatabaseRegister{0}\EnvironmentRegister{0}}}\\
	&\leq 2(\NumberOfQueries + 1) \cdot 2^{-(\RandomOracleOutputLength - 3)/2} + \vecnorm{(\id{\DatabaseRegister{0}} - \ProjectNoHatZero) \PureStateSampleOne_{\CommitmentRegister{0}\OpeningRegister{0}\DatabaseRegister{0}\EnvironmentRegister{0}}}\enspace.
\end{align*}

Similarly,
\begin{align*}
	&\vecnorm{(\id{\DatabaseRegister{0}} - \ProjectNoHatZero) \PureStateSampleOneI{2}}\\
	&=\vecnorm{(\id{\DatabaseRegister{0}} - \ProjectNoHatZero) \ProjectSizeDatabase{\NumberOfQueries + 1}\PureStateSampleOneI{2}}\\
	&\leq \matnorm{\ProjectSizeDatabase{\NumberOfQueries}\Commutator{\NoCollision}{\ProjectNoHatZero}\ProjectSizeDatabase{\NumberOfQueries}} + \vecnorm{(\id{\DatabaseRegister{0}} - \ProjectNoHatZero) \PureStateSampleOne_{\CommitmentRegister{0}\OpeningRegister{0}\DatabaseRegister{0}\EnvironmentRegister{0}}}\\
	&\leq \NumberOfQueries \cdot 2^{-(\RandomOracleOutputLength - 3)/2} + \vecnorm{(\id{\DatabaseRegister{0}} - \ProjectNoHatZero) \PureStateSampleOne_{\CommitmentRegister{0}\OpeningRegister{0}\DatabaseRegister{0}\EnvironmentRegister{0}}}\enspace.
\end{align*}

Furthermore, as the error for $(\id{\DatabaseRegister{0}} - \ProjectNoHatZero)$ only accumulates when the adversary queries $\invert$,
\begin{align*}
	&\vecnorm{(\id{\DatabaseRegister{0}} - \ProjectNoHatZero) \PureStateSampleOne_{\CommitmentRegister{0}\OpeningRegister{0}\DatabaseRegister{0}\EnvironmentRegister{0}}}\\
	&\leq \sum_{i \in [\NumberOfQueries]}\matnorm{\ProjectSizeDatabase{\NumberOfQueries}\Commutator{\invert_{\TargetRegister{0}\DatabaseRegister{0}\WorkingRegister{0}}}{\ProjectNoHatZero}\ProjectSizeDatabase{\NumberOfQueries}}\sqrt{\QuantumQueryMassFunc{\Algorithm, \invert, i}}\\
	&\leq 2^{-\RandomOracleOutputLength/2 + 4}\sqrt{\NumberOfQueries} \sqrt{\NumberOfQueries\sum_{i \in [\NumberOfQueries]}\QuantumQueryMassFunc{\Algorithm, \invert, i}}\\
	&\leq 2^{-\RandomOracleOutputLength/2 + 4} \cdot \NumberOfQueries \cdot \sqrt{\QuantumTotalQueryMass_2}\enspace.
\end{align*}

Combining the above inequalities, we obtain
\begin{align}
&\matnorm{(\bra{\NullStateInExp_{\Check}}V_{\Check} - \bra{\NullStateInExp_{\AltCheck}}V_{\ExtractMessage}\ket{0}_{\StandardBasisWorkingRegister{0}\HadamardBasisWorkingRegister{0}})\PureStateSampleOne_{\CommitmentRegister{0}\OpeningRegister{0}\DatabaseRegister{0}\EnvironmentRegister{0}}}\nonumber\\
&\leq (12\NumberOfQueries + 16 + 32\NumberOfQueries\sqrt{\QuantumTotalQueryMass_2}) \cdot 2^{-\RandomOracleOutputLength/2}\enspace.
\end{align}

Then the claim follows from \Cref{eqn:sqrt-prob-state-diff}.
\end{proof}

\subsection{Proof of \Cref{claim:almost-commutativity-of-ProjectNoHatZero}}\label{subsec:Proof_of_almost-commutativity-of-ProjectNoHatZero}

We first prove a bound for a more general operator $A$.

\begin{claim}\label{claim:NoHatZero_Commute_Step_1}
For an operator $A$, if $A = A^\dagger$, and for every $x \in \RandomOracleDomain$, $A$ commutes with $\ketbra{\bot}{\bot}_{\DatabaseRegisterAt{0}{x}}$, then the following inequality holds:
\begin{align*}
	&\matnorm{\ProjectSizeDatabase{\NumberOfQueries}\Commutator{\ProjectNoHatZero}{A}\ProjectSizeDatabase{\NumberOfQueries}}	\leq 2\sqrt{\NumberOfQueries}\max_{x \in \RandomOracleDomain}\matnorm{\ProjectSizeDatabase{\NumberOfQueries}\left(\id{} - \ketbra{\hat{0}^\RandomOracleOutputLength}{\hat{0}^\RandomOracleOutputLength}\right)_{\DatabaseRegisterAt{0}{x}}A\ketbra{\hat{0}^\RandomOracleOutputLength}{\hat{0}^\RandomOracleOutputLength}_{\DatabaseRegisterAt{0}{x}}\ProjectSizeDatabase{\NumberOfQueries}}\enspace.
\end{align*}
\end{claim}

\begin{proof}
To bound $\matnorm{\ProjectSizeDatabase{\NumberOfQueries}\Commutator{\ProjectNoHatZero}{A}\ProjectSizeDatabase{\NumberOfQueries}}$, it suffices to bound $\matnorm{\ProjectSizeDatabase{\NumberOfQueries}\ProjectNoHatZero A\left(\id{} - \ProjectNoHatZero\right)\ProjectSizeDatabase{\NumberOfQueries}}$ because by the triangle inequality,
\begin{align*}
	&\matnorm{\ProjectSizeDatabase{\NumberOfQueries}\Commutator{\ProjectNoHatZero}{A}\ProjectSizeDatabase{\NumberOfQueries}}\\
	\leq &\matnorm{\ProjectSizeDatabase{\NumberOfQueries}\ProjectNoHatZero A\left(\id{} - \ProjectNoHatZero\right)\ProjectSizeDatabase{\NumberOfQueries}} + \matnorm{\ProjectSizeDatabase{\NumberOfQueries}\left( -\id{} + \ProjectNoHatZero\right) A \ProjectNoHatZero\ProjectSizeDatabase{\NumberOfQueries}}\\
	= & \matnorm{\ProjectSizeDatabase{\NumberOfQueries}\ProjectNoHatZero A\left(\id{} - \ProjectNoHatZero\right)\ProjectSizeDatabase{\NumberOfQueries}} + \matnorm{\ProjectSizeDatabase{\NumberOfQueries}\ProjectNoHatZero A^\dagger \left(\id{} - \ProjectNoHatZero\right)\ProjectSizeDatabase{\NumberOfQueries}}\\
	= & 2\matnorm{\ProjectSizeDatabase{\NumberOfQueries}\ProjectNoHatZero A\left(\id{} - \ProjectNoHatZero\right)\ProjectSizeDatabase{\NumberOfQueries}}\enspace.
\end{align*}

We define a family of projectors $\ProjectWhereNonEmpty{z} \coloneq \bigotimes_{x \in \RandomOracleDomain}\left(\frac{\id{}}{2} + (-1)^{\accessVectorAt{z}{x} + 1}\left(\frac{\id{}}{2} - \ketbra{\bot}{\bot}\right)\right)_{\DatabaseRegisterAt{0}{x}}$ for each vector $z \in \Bits^\RandomOracleDomain$. That is, $\ProjectWhereNonEmpty{z}$ projects onto $\ketbra{\bot}{\bot}$ on the registers $\DatabaseRegisterAt{1}{x}$ when $\accessVectorAt{z}{x} = 0$, and projects onto $\id{} - \ketbra{\bot}{\bot}$ otherwise.

Then $\{\ProjectWhereNonEmpty{z}\}_z$ is in fact a measurement, where the outcome $z$ indicates	 the locations where the database is non-empty. We observe that $\ProjectNoHatZero$ commutes with $\ProjectWhereNonEmpty{z}$, and in addition, since $A$ commutes with $\ketbra{\bot}{\bot}_{\DatabaseRegisterAt{0}{x}}$ for every $x \in \RandomOracleDomain$, $A$ commutes with $\ProjectWhereNonEmpty{z}$, for each $z \in \Bits^\RandomOracleDomain$. As a result,
\begin{align}\label{eqn:break_down_matnorm_into_max_over_delta_z}
	&\matnorm{\ProjectSizeDatabase{\NumberOfQueries}\ProjectNoHatZero A\left(\id{} - \ProjectNoHatZero\right)\ProjectSizeDatabase{\NumberOfQueries}}\nonumber\\
	= & \matnorm{\sum_{z, z' \in \Bits^\RandomOracleDomain}\ProjectWhereNonEmpty{z}\ProjectSizeDatabase{\NumberOfQueries}\ProjectNoHatZero A\left(\id{} - \ProjectNoHatZero\right)\ProjectSizeDatabase{\NumberOfQueries}\ProjectWhereNonEmpty{z'}}\nonumber\\
	= & \matnorm{\sum_{z, z' \in \Bits^\RandomOracleDomain}\ProjectWhereNonEmpty{z}\ProjectWhereNonEmpty{z'}\ProjectSizeDatabase{\NumberOfQueries}\ProjectNoHatZero A\left(\id{} - \ProjectNoHatZero\right)\ProjectSizeDatabase{\NumberOfQueries}}\nonumber\\
	= & \matnorm{\sum_{z\in \Bits^\RandomOracleDomain}\ProjectWhereNonEmpty{z}\ProjectWhereNonEmpty{z}\ProjectSizeDatabase{\NumberOfQueries}\ProjectNoHatZero A\left(\id{} - \ProjectNoHatZero\right)\ProjectSizeDatabase{\NumberOfQueries}}\nonumber\\
	= & \matnorm{\sum_{z\in \Bits^\RandomOracleDomain}\ProjectWhereNonEmpty{z}\ProjectSizeDatabase{\NumberOfQueries}\ProjectNoHatZero A\left(\id{} - \ProjectNoHatZero\right)\ProjectSizeDatabase{\NumberOfQueries}\ProjectWhereNonEmpty{z}}\nonumber\\
	\le & \max_{z \in \Bits^\RandomOracleDomain} \matnorm{\ProjectWhereNonEmpty{z}\ProjectSizeDatabase{\NumberOfQueries}\ProjectNoHatZero A\left(\id{} - \ProjectNoHatZero\right)\ProjectSizeDatabase{\NumberOfQueries}\ProjectWhereNonEmpty{z}}\nonumber\\
	= & \max_{z \in \Bits^\RandomOracleDomain \text{ s.t. } \HammingWeight{z} \leq \NumberOfQueries} \matnorm{\ProjectWhereNonEmpty{z}\ProjectSizeDatabase{\NumberOfQueries}\ProjectNoHatZero A\left(\id{} - \ProjectNoHatZero\right)\ProjectSizeDatabase{\NumberOfQueries}\ProjectWhereNonEmpty{z}}
\end{align}
where we use \Cref{eqn:NormOfSumOfOrthogonalOperator}, and the fact that for $z$ with Hamming weight more than $\NumberOfQueries$, $\ProjectWhereNonEmpty{z}\ProjectSizeDatabase{\NumberOfQueries} = 0$.

Now let's bound $\matnorm{\ProjectWhereNonEmpty{z}\ProjectSizeDatabase{\NumberOfQueries}\ProjectNoHatZero A\left(\id{} - \ProjectNoHatZero\right)\ProjectSizeDatabase{\NumberOfQueries}\ProjectWhereNonEmpty{z}}$.

For $z$ with Hamming weight $\HammingWeight{z} \leq \NumberOfQueries$, and every normalized state $\PureStateSampleOne$, we can write $\left(\id{} - \ProjectNoHatZero\right)\ProjectSizeDatabase{\NumberOfQueries}\ProjectWhereNonEmpty{z}\PureStateSampleOne$ in the following form
\[\left(\id{} - \ProjectNoHatZero\right)\ProjectSizeDatabase{\NumberOfQueries}\ProjectWhereNonEmpty{z}\PureStateSampleOne = \sum_{i = 1}^{\HammingWeight{z}}\alpha_i \ket{\hat{0}^\RandomOracleOutputLength}_{\DatabaseRegisterAt{0}{x_i}}\PureStateSampleOneI{i}_{\setcomplement{\DatabaseRegisterAt{0}{x_i}}}\]
where $x_i$ is the $i^{\text{th}}$ element in $\RandomOracleDomain$ such that $\accessVectorAt{z}{x} = 1$, $\setcomplement{\DatabaseRegisterAt{1}{x_i}}$ is all the registers excluding $\DatabaseRegisterAt{1}{x_i}$, and $\ket{\hat{0}^\RandomOracleOutputLength}_{\DatabaseRegisterAt{0}{x_i}}\PureStateSampleOneI{i}_{\setcomplement{\DatabaseRegisterAt{0}{x_i}}}$ are orthonormal states.

Plug $\left(\id{} - \ProjectNoHatZero\right)\ProjectSizeDatabase{\NumberOfQueries}\ProjectWhereNonEmpty{z}\PureStateSampleOne$ into the following, and we can get that
\begin{align*}
	&\vecnorm{\ProjectWhereNonEmpty{z}\ProjectSizeDatabase{\NumberOfQueries}\ProjectNoHatZero A\left(\id{} - \ProjectNoHatZero\right)\ProjectSizeDatabase{\NumberOfQueries}\ProjectWhereNonEmpty{z}\PureStateSampleOne}\\
	=&\vecnorm{\ProjectWhereNonEmpty{z}\ProjectSizeDatabase{\NumberOfQueries}\ProjectNoHatZero A\ProjectSizeDatabase{\NumberOfQueries}\left(\id{} - \ProjectNoHatZero\right)\ProjectSizeDatabase{\NumberOfQueries}\ProjectWhereNonEmpty{z}\PureStateSampleOne}\\
	\leq & \sum_{i = 1}^{\HammingWeight{z}}\abs{\alpha_i}\vecnorm{\ProjectSizeDatabase{\NumberOfQueries}\ProjectNoHatZero A\ProjectSizeDatabase{\NumberOfQueries}\ket{\hat{0}^\RandomOracleOutputLength}_{\DatabaseRegisterAt{0}{x_i}}\PureStateSampleOneI{i}_{\setcomplement{\DatabaseRegisterAt{0}{x_i}}}}\\
	= & \sum_{i = 1}^{\HammingWeight{z}}\abs{\alpha_i}\vecnorm{\ProjectNoHatZero \ProjectSizeDatabase{\NumberOfQueries}A\ProjectSizeDatabase{\NumberOfQueries}\ketbra{\hat{0}^\RandomOracleOutputLength}{\hat{0}^\RandomOracleOutputLength}_{\DatabaseRegisterAt{0}{x_i}}\PureStateSampleOneI{i}_{\setcomplement{\DatabaseRegisterAt{0}{x_i}}}\ket{\hat{0}^\RandomOracleOutputLength}_{\DatabaseRegisterAt{0}{x_i}}}\\
	\le & \sum_{i = 1}^{\HammingWeight{z}}\abs{\alpha_i}\vecnorm{\left(\id{} - \ketbra{\hat{0}^\RandomOracleOutputLength}{\hat{0}^\RandomOracleOutputLength}\right)_{\DatabaseRegisterAt{0}{x_i}} \ProjectSizeDatabase{\NumberOfQueries}A\ProjectSizeDatabase{\NumberOfQueries}\ketbra{\hat{0}^\RandomOracleOutputLength}{\hat{0}^\RandomOracleOutputLength}_{\DatabaseRegisterAt{0}{x_i}}\PureStateSampleOneI{i}_{\setcomplement{\DatabaseRegisterAt{0}{x_i}}}\ket{\hat{0}^\RandomOracleOutputLength}_{\DatabaseRegisterAt{0}{x_i}}}\\
	\le & \left(\sum_{i = 1}^{\HammingWeight{z}}\abs{\alpha_i}\right) \max_{x \in \RandomOracleDomain}\matnorm{\ProjectSizeDatabase{\NumberOfQueries}\left(\id{} - \ketbra{\hat{0}^\RandomOracleOutputLength}{\hat{0}^\RandomOracleOutputLength}\right)_{\DatabaseRegisterAt{0}{x}}A\ketbra{\hat{0}^\RandomOracleOutputLength}{\hat{0}^\RandomOracleOutputLength}_{\DatabaseRegisterAt{0}{x}}\ProjectSizeDatabase{\NumberOfQueries}}\\
	\le & \sqrt{\NumberOfQueries} \max_{x \in \RandomOracleDomain}\matnorm{\ProjectSizeDatabase{\NumberOfQueries}\left(\id{} - \ketbra{\hat{0}^\RandomOracleOutputLength}{\hat{0}^\RandomOracleOutputLength}\right)_{\DatabaseRegisterAt{0}{x}}A\ketbra{\hat{0}^\RandomOracleOutputLength}{\hat{0}^\RandomOracleOutputLength}_{\DatabaseRegisterAt{0}{x}}\ProjectSizeDatabase{\NumberOfQueries}}\enspace,
\end{align*}
where the last inequality is due to  Cauchy-Schwarz, $\HammingWeight{z} \leq \NumberOfQueries$ and $\vecnorm{\left(\id{} - \ProjectNoHatZero\right)\ProjectSizeDatabase{\NumberOfQueries}\ProjectWhereNonEmpty{z}\PureStateSampleOne} \leq 1$.

Combining \Cref{eqn:break_down_matnorm_into_max_over_delta_z} and the above inequality, we get \Cref{claim:NoHatZero_Commute_Step_1}.
\end{proof}

Both $\ProjectNoCollision$ and $\invert$ commute with $\ketbra{\bot}{\bot}_{\DatabaseRegisterAt{0}{x}}$ for every $x \in \RandomOracleDomain$, and thus \Cref{claim:almost-commutativity-of-ProjectNoHatZero} follows from \Cref{claim:NoHatZero_Commute_Step_1} and the following two claims.

\begin{claim}\label{claim:Hat0CommuteWithNoCollision}
For every $x \in \RandomOracleDomain$ and every query bound $\NumberOfQueries$,
	\begin{align*}
	\matnorm{\ProjectSizeDatabase{\NumberOfQueries}\left(\id{} - \ketbra{\hat{0}^\RandomOracleOutputLength}{\hat{0}^\RandomOracleOutputLength}\right)_{\DatabaseRegisterAt{0}{x}}\ProjectNoCollision\ketbra{\hat{0}^\RandomOracleOutputLength}{\hat{0}^\RandomOracleOutputLength}_{\DatabaseRegisterAt{0}{x}}\ProjectSizeDatabase{\NumberOfQueries}} \leq 2^{-(\RandomOracleOutputLength - 1)/2}\sqrt{\NumberOfQueries}\enspace.	
	\end{align*}
\end{claim}

\begin{proof}
	For every normalized state $\PureStateSampleOne$, we can write $\ketbra{\hat{0}^\RandomOracleOutputLength}{\hat{0}^\RandomOracleOutputLength}_{\DatabaseRegisterAt{0}{x}}\ProjectSizeDatabase{\NumberOfQueries}\PureStateSampleOne$ in the following form 
	\[\ketbra{\hat{0}^\RandomOracleOutputLength}{\hat{0}^\RandomOracleOutputLength}_{\DatabaseRegisterAt{0}{x}}\ProjectSizeDatabase{\NumberOfQueries}\PureStateSampleOne = \frac{1}{\sqrt{2^\RandomOracleOutputLength}}\sum_{y \in \Bits^\RandomOracleOutputLength}\ket{y}_{\DatabaseRegisterAt{0}{x}}\sum_{\accessVectorAt{D}{\RandomOracleDomain/\{x\}} \text{ has size at most $\NumberOfQueries - 1$}}\alpha_{\accessVectorAt{D}{\RandomOracleDomain/\{x\}}}\ket{\accessVectorAt{D}{\RandomOracleDomain/\{x\}}}_{\DatabaseRegisterAt{0}{\RandomOracleDomain/\{x\}}}\enspace.\]
	
	Notice that
	\begin{align*}
		&\vecnorm{\ProjectSizeDatabase{\NumberOfQueries}\left(\id{} - \ketbra{\hat{0}^\RandomOracleOutputLength}{\hat{0}^\RandomOracleOutputLength}\right)_{\DatabaseRegisterAt{0}{x}}\ProjectNoCollision\ketbra{\hat{0}^\RandomOracleOutputLength}{\hat{0}^\RandomOracleOutputLength}_{\DatabaseRegisterAt{0}{x}}\ProjectSizeDatabase{\NumberOfQueries}\PureStateSampleOne}\\
		\le &\vecnorm{\left(\id{} - \ketbra{\hat{0}^\RandomOracleOutputLength}{\hat{0}^\RandomOracleOutputLength}\right)_{\DatabaseRegisterAt{0}{x}}\ProjectNoCollision\ketbra{\hat{0}^\RandomOracleOutputLength}{\hat{0}^\RandomOracleOutputLength}_{\DatabaseRegisterAt{0}{x}}\ProjectSizeDatabase{\NumberOfQueries}\PureStateSampleOne}\\
		= &\sqrt{\vecnorm{\ProjectNoCollision\ketbra{\hat{0}^\RandomOracleOutputLength}{\hat{0}^\RandomOracleOutputLength}_{\DatabaseRegisterAt{0}{x}}\ProjectSizeDatabase{\NumberOfQueries}\PureStateSampleOne}^2 - \vecnorm{\bra{\hat{0}^\RandomOracleOutputLength}_{\DatabaseRegisterAt{0}{x}}\ProjectNoCollision\ketbra{\hat{0}^\RandomOracleOutputLength}{\hat{0}^\RandomOracleOutputLength}_{\DatabaseRegisterAt{0}{x}}\ProjectSizeDatabase{\NumberOfQueries}\PureStateSampleOne}^2}
	\end{align*}
	
	To bound $\vecnorm{\ProjectNoCollision\ketbra{\hat{0}^\RandomOracleOutputLength}{\hat{0}^\RandomOracleOutputLength}_{\DatabaseRegisterAt{0}{x}}\ProjectSizeDatabase{\NumberOfQueries}\PureStateSampleOne}^2$, we plug in $\ketbra{\hat{0}^\RandomOracleOutputLength}{\hat{0}^\RandomOracleOutputLength}_{\DatabaseRegisterAt{0}{x}}\ProjectSizeDatabase{\NumberOfQueries}\PureStateSampleOne$ to get that
	\begin{align*}
		&\vecnorm{\ProjectNoCollision\ketbra{\hat{0}^\RandomOracleOutputLength}{\hat{0}^\RandomOracleOutputLength}_{\DatabaseRegisterAt{0}{x}}\ProjectSizeDatabase{\NumberOfQueries}\PureStateSampleOne}^2\\
		=& \vecnorm{\ProjectNoCollision\frac{1}{\sqrt{2^\RandomOracleOutputLength}}\sum_{y \in \Bits^\RandomOracleOutputLength}\ket{y}_{\DatabaseRegisterAt{0}{x}}\sum_{\accessVectorAt{D}{\RandomOracleDomain/\{x\}} \text{ has size at most $\NumberOfQueries - 1$}}\alpha_{\accessVectorAt{D}{\RandomOracleDomain/\{x\}}}\ket{\accessVectorAt{D}{\RandomOracleDomain/\{x\}}}_{\DatabaseRegisterAt{0}{\RandomOracleDomain/\{x\}}}}^2\\
		=& \vecnorm{\frac{1}{\sqrt{2^\RandomOracleOutputLength}}\sum_{\substack{\accessVectorAt{D}{\RandomOracleDomain/\{x\}} \text{ has size at most $\NumberOfQueries - 1$}\\ \text{and it doesn't have collisions}}}\alpha_{\accessVectorAt{D}{\RandomOracleDomain/\{x\}}}\sum_{y \in \Bits^\RandomOracleOutputLength \text{ and } y \notin \accessVectorAt{D}{\RandomOracleDomain/\{x\}}}\ket{y}_{\DatabaseRegisterAt{0}{x}}\ket{\accessVectorAt{D}{\RandomOracleDomain/\{x\}}}_{\DatabaseRegisterAt{0}{\RandomOracleDomain/\{x\}}}}^2\\
		=& \sum_{\substack{\accessVectorAt{D}{\RandomOracleDomain/\{x\}} \text{ has size at most $\NumberOfQueries - 1$}\\ \text{and it doesn't have collisions}\\y \in \Bits^\RandomOracleOutputLength \text{ and } y \notin \accessVectorAt{D}{\RandomOracleDomain/\{x\}}}}\frac{1}{2^\RandomOracleOutputLength}\abs{\alpha_{\accessVectorAt{D}{\RandomOracleDomain/\{x\}}}}^2\\
		\leq & \sum_{\substack{\accessVectorAt{D}{\RandomOracleDomain/\{x\}} \text{ has size at most $\NumberOfQueries - 1$}\\ \text{and it doesn't have collisions}}}\abs{\alpha_{\accessVectorAt{D}{\RandomOracleDomain/\{x\}}}}^2
	\end{align*}

	Similarly,
	\begin{align*}
		&\vecnorm{\bra{\hat{0}^\RandomOracleOutputLength}_{\DatabaseRegisterAt{0}{x}}\ProjectNoCollision\ketbra{\hat{0}^\RandomOracleOutputLength}{\hat{0}^\RandomOracleOutputLength}_{\DatabaseRegisterAt{0}{x}}\ProjectSizeDatabase{\NumberOfQueries}\PureStateSampleOne}^2\\
		=& \vecnorm{\bra{\hat{0}^\RandomOracleOutputLength}_{\DatabaseRegisterAt{0}{x}}\frac{1}{\sqrt{2^\RandomOracleOutputLength}}\sum_{\substack{\accessVectorAt{D}{\RandomOracleDomain/\{x\}} \text{ has size at most $\NumberOfQueries - 1$}\\ \text{and it doesn't have collisions}}}\alpha_{\accessVectorAt{D}{\RandomOracleDomain/\{x\}}}\sum_{y \in \Bits^\RandomOracleOutputLength \text{ and } y \notin \accessVectorAt{D}{\RandomOracleDomain/\{x\}}}\ket{y}_{\DatabaseRegisterAt{0}{x}}\ket{\accessVectorAt{D}{\RandomOracleDomain/\{x\}}}_{\DatabaseRegisterAt{0}{\RandomOracleDomain/\{x\}}}}^2\\
		=& \sum_{\substack{\accessVectorAt{D}{\RandomOracleDomain/\{x\}} \text{ has size at most $\NumberOfQueries - 1$}\\ \text{and it doesn't have collisions}}}\abs{\alpha_{\accessVectorAt{D}{\RandomOracleDomain/\{x\}}}\sum_{y \in \Bits^\RandomOracleOutputLength \text{ and } y \notin \accessVectorAt{D}{\RandomOracleDomain/\{x\}}}\frac{1}{2^\RandomOracleOutputLength}}^2\\
		\ge & \sum_{\substack{\accessVectorAt{D}{\RandomOracleDomain/\{x\}} \text{ has size at most $\NumberOfQueries - 1$}\\ \text{and it doesn't have collisions}}}\abs{\alpha_{\accessVectorAt{D}{\RandomOracleDomain/\{x\}}}}^2 \left(1 - \frac{\NumberOfQueries}{2^\RandomOracleOutputLength}\right)^2\enspace.
	\end{align*}
	
	As a result, for any normalized state $\PureStateSampleOne$,
	\begin{align*}
		&\vecnorm{\ProjectSizeDatabase{\NumberOfQueries}\left(\id{} - \ketbra{\hat{0}^\RandomOracleOutputLength}{\hat{0}^\RandomOracleOutputLength}\right)_{\DatabaseRegisterAt{0}{x}}\ProjectNoCollision\ketbra{\hat{0}^\RandomOracleOutputLength}{\hat{0}^\RandomOracleOutputLength}_{\DatabaseRegisterAt{0}{x}}\ProjectSizeDatabase{\NumberOfQueries}\PureStateSampleOne}\\
		\le &\sqrt{\vecnorm{\ProjectNoCollision\ketbra{\hat{0}^\RandomOracleOutputLength}{\hat{0}^\RandomOracleOutputLength}_{\DatabaseRegisterAt{0}{x}}\ProjectSizeDatabase{\NumberOfQueries}\PureStateSampleOne}^2 - \vecnorm{\bra{\hat{0}^\RandomOracleOutputLength}_{\DatabaseRegisterAt{0}{x}}\ProjectNoCollision\ketbra{\hat{0}^\RandomOracleOutputLength}{\hat{0}^\RandomOracleOutputLength}_{\DatabaseRegisterAt{0}{x}}\ProjectSizeDatabase{\NumberOfQueries}\PureStateSampleOne}^2}\\
		\le & \sqrt{\sum_{\substack{\accessVectorAt{D}{\RandomOracleDomain/\{x\}} \text{ has size at most $\NumberOfQueries - 1$}\\ \text{and it doesn't have collisions}}}\abs{\alpha_{\accessVectorAt{D}{\RandomOracleDomain/\{x\}}}}^2} \cdot \sqrt{1 - \left(1 - \frac{\NumberOfQueries}{2^\RandomOracleOutputLength}\right)^2}\\
		\le & 2^{-(\RandomOracleOutputLength - 1)/2}\sqrt{\NumberOfQueries}\enspace,
	\end{align*}
	which by definition, implies
	\begin{align*}
	\matnorm{\ProjectSizeDatabase{\NumberOfQueries}\left(\id{} - \ketbra{\hat{0}^\RandomOracleOutputLength}{\hat{0}^\RandomOracleOutputLength}\right)_{\DatabaseRegisterAt{0}{x}}\ProjectNoCollision\ketbra{\hat{0}^\RandomOracleOutputLength}{\hat{0}^\RandomOracleOutputLength}_{\DatabaseRegisterAt{0}{x}}\ProjectSizeDatabase{\NumberOfQueries}} \leq 2^{-(\RandomOracleOutputLength - 1)/2}\sqrt{\NumberOfQueries}\enspace.	
	\end{align*}
\end{proof}

\begin{claim}\label{claim:Hat0CommuteWithinv}
For every $x \in \RandomOracleDomain$ and every query bound $\NumberOfQueries$,
	\begin{align*}
	\matnorm{\ProjectSizeDatabase{\NumberOfQueries}\left(\id{} - \ketbra{\hat{0}^\RandomOracleOutputLength}{\hat{0}^\RandomOracleOutputLength}\right)_{\DatabaseRegisterAt{0}{x}}\invert_{\TargetRegister{0}\DatabaseRegister{0}\WorkingRegister{0}}\ketbra{\hat{0}^\RandomOracleOutputLength}{\hat{0}^\RandomOracleOutputLength}_{\DatabaseRegisterAt{0}{x}}\ProjectSizeDatabase{\NumberOfQueries}} \leq 2^{-\RandomOracleOutputLength/2 + 3}\enspace.	
	\end{align*}
\end{claim}

\begin{proof}
	Recall that 
	\[\invert_{\TargetRegister{0}\DatabaseRegister{0}\WorkingRegister{0}} = \sum_{\target \in \Bits^\RandomOracleOutputLength} \ketbra{\target}{\target}_{\TargetRegister{0}} \otimes \PurifiedInvY^{(\target)}_{\DatabaseRegister{0}\WorkingRegister{0}}\]
	is a controlled unitary. By \Cref{eqn:NormOfControlledOperator},
	\[\matnorm{\Commutator{\invert_{\TargetRegister{0}\DatabaseRegister{0}\WorkingRegister{0}}}{\ketbra{\hat{0}^\RandomOracleOutputLength}{\hat{0}^\RandomOracleOutputLength}_{\DatabaseRegisterAt{0}{x}}}} \leq \max_{y \in \Bits^\RandomOracleOutputLength}\matnorm{\Commutator{\PurifiedInvY^{(\target)}_{\DatabaseRegister{0}\WorkingRegister{0}}}{\ketbra{\hat{0}^\RandomOracleOutputLength}{\hat{0}^\RandomOracleOutputLength}_{\DatabaseRegisterAt{0}{x}}}}\enspace.\]
	
	For every $y \in \Bits^\RandomOracleOutputLength$, since $\ketbra{\hat{0}^\RandomOracleOutputLength}{\hat{0}^\RandomOracleOutputLength}_{\DatabaseRegisterAt{0}{x}}$ only acts on $\DatabaseRegisterAt{1}{x}$ within the database register $\DatabaseRegister{1}$, by \Cref{lem:commutivity_of_purified_measurement},
	\[\matnorm{\Commutator{\PurifiedInvY^{(\target)}_{\DatabaseRegister{0}\WorkingRegister{0}}}{\ketbra{\hat{0}^\RandomOracleOutputLength}{\hat{0}^\RandomOracleOutputLength}_{\DatabaseRegisterAt{0}{x}}}} \leq 4 \matnorm{\Commutator{\ketbra{y}{y}_{\DatabaseRegisterAt{0}{x}}}{\ketbra{\hat{0}^\RandomOracleOutputLength}{\hat{0}^\RandomOracleOutputLength}_{\DatabaseRegisterAt{0}{x}}}} \leq 2^{-\RandomOracleOutputLength/2 + 3}\enspace.\]
	
	Therefore, $\invert_{\TargetRegister{0}\DatabaseRegister{0}\WorkingRegister{0}}$ and $\ketbra{\hat{0}^\RandomOracleOutputLength}{\hat{0}^\RandomOracleOutputLength}_{\DatabaseRegisterAt{0}{x}}$ almost commute with each other, \[\matnorm{\Commutator{\invert_{\TargetRegister{0}\DatabaseRegister{0}\WorkingRegister{0}}}{\ketbra{\hat{0}^\RandomOracleOutputLength}{\hat{0}^\RandomOracleOutputLength}_{\DatabaseRegisterAt{0}{x}}}} \leq 2^{-\RandomOracleOutputLength/2 + 3}\enspace.\]
	
	Now we are ready to bound $\matnorm{\ProjectSizeDatabase{\NumberOfQueries}\left(\id{} - \ketbra{\hat{0}^\RandomOracleOutputLength}{\hat{0}^\RandomOracleOutputLength}\right)_{\DatabaseRegisterAt{0}{x}}\invert_{\TargetRegister{0}\DatabaseRegister{0}\WorkingRegister{0}}\ketbra{\hat{0}^\RandomOracleOutputLength}{\hat{0}^\RandomOracleOutputLength}_{\DatabaseRegisterAt{0}{x}}\ProjectSizeDatabase{\NumberOfQueries}}$.
	\begin{align*}
		&\matnorm{\ProjectSizeDatabase{\NumberOfQueries}\left(\id{} - \ketbra{\hat{0}^\RandomOracleOutputLength}{\hat{0}^\RandomOracleOutputLength}\right)_{\DatabaseRegisterAt{0}{x}}\invert_{\TargetRegister{0}\DatabaseRegister{0}\WorkingRegister{0}}\ketbra{\hat{0}^\RandomOracleOutputLength}{\hat{0}^\RandomOracleOutputLength}_{\DatabaseRegisterAt{0}{x}}\ProjectSizeDatabase{\NumberOfQueries}}\\
		\leq & \matnorm{\left(\id{} - \ketbra{\hat{0}^\RandomOracleOutputLength}{\hat{0}^\RandomOracleOutputLength}\right)_{\DatabaseRegisterAt{0}{x}}\invert_{\TargetRegister{0}\DatabaseRegister{0}\WorkingRegister{0}}\ketbra{\hat{0}^\RandomOracleOutputLength}{\hat{0}^\RandomOracleOutputLength}_{\DatabaseRegisterAt{0}{x}}}\\
		\leq & \matnorm{\left(\id{} - \ketbra{\hat{0}^\RandomOracleOutputLength}{\hat{0}^\RandomOracleOutputLength}\right)_{\DatabaseRegisterAt{0}{x}}\ketbra{\hat{0}^\RandomOracleOutputLength}{\hat{0}^\RandomOracleOutputLength}_{\DatabaseRegisterAt{0}{x}}\invert_{\TargetRegister{0}\DatabaseRegister{0}\WorkingRegister{0}}} + \matnorm{\left(\id{} - \ketbra{\hat{0}^\RandomOracleOutputLength}{\hat{0}^\RandomOracleOutputLength}\right)_{\DatabaseRegisterAt{0}{x}}\Commutator{\invert_{\TargetRegister{0}\DatabaseRegister{0}\WorkingRegister{0}}}{\ketbra{\hat{0}^\RandomOracleOutputLength}{\hat{0}^\RandomOracleOutputLength}_{\DatabaseRegisterAt{0}{x}}}}\\
		\leq & \matnorm{\Commutator{\invert_{\TargetRegister{0}\DatabaseRegister{0}\WorkingRegister{0}}}{\ketbra{\hat{0}^\RandomOracleOutputLength}{\hat{0}^\RandomOracleOutputLength}_{\DatabaseRegisterAt{0}{x}}}}\\
		\leq & 2^{-\RandomOracleOutputLength/2 + 3}\enspace.
	\end{align*}

\end{proof}

\doclearpage
\section{Defining extractable quantum state vector commitments}
\label{sec:def-QVC}

We define the syntax of quantum state vector commitments, and then provide a formal definition for extractable quantum state vector commitments.

\subsection{Syntax for quantum state vector commitments}
\label{subsec:def-QVC-syntax}

We consider non-interactive quantum state vector commitments in the quantum random oracle setting.

\begin{definition}[Quantum state vector commitments in the QROM]
In the quantum random oracle model, a non-interactive \emph{quantum state vector commitment} (QVC) for block size $\QVCBlockSize$ and message length $\QVCMessageLength$ is a tuple of polynomial-time oracle-aided quantum algorithms $\QVC = \QVCTuple$ where
\begin{itemize}[noitemsep]
\item $\QVCCommit^{\RandomOracle{0}}(1^\Security, \MessageRegister{1}) \to ({\CommitmentRegister{1}, \QVCAuxiliaryRegister{1}})$: $\QVCCommit^{\RandomOracle{0}}$ has oracle access to $\RandomOracle{1}$ and is given as inputs the security parameter $1^\Security$ and a $\QVCMessageLength\cdot\QVCBlockSize$-qubit quantum state on the register $\QVCMessageRegister$, which might entangle with the register $\EnvironmentRegister{1}$. It outputs a state on two registers $\CommitmentRegister{1}$ and $\QVCAuxiliaryRegister{1}$. The register $\QVCCommitmentRegister$ will be sent to the receiver in the commitment phase and the register $\QVCAuxiliaryRegister{1}$ contains necessary information for the local opening in the open phase.
\item $\QVCOpen^{\RandomOracle{0}}(1^\Security, \QVCQuerySet, \QVCAuxiliaryRegister{1}) \to (\OpeningRegister{1}, \QVCAuxiliaryNotUsedRegister)$: To open the commitment on positions $\QVCQuerySet \subseteq \{1, 2, \ldots, \QVCMessageLength\}$, $\QVCOpen^{\RandomOracle{0}}$, with oracle access to $\RandomOracle{1}$, takes the security parameter $1^\Security$, the query set $\QVCQuerySet$, and a state on $\QVCAuxiliaryRegister{1}$ as inputs, and outputs a state on registers $(\OpeningRegister{1}, \QVCAuxiliaryNotUsedRegister)$. The state in register $\OpeningRegister{1}$ is interpreted as the opening and will be sent to the receiver for local decommitment, and $\QVCAuxiliaryNotUsedRegister$ contains the unused part in the register $\QVCAuxiliaryRegister{1}$.
\item $\QVCQuery^{\RandomOracle{0}}(1^\Security, \QVCQuerySet, \CommitmentRegister{1}, \OpeningRegister{1}, \AnswerRegister{1}) \to (\QVCValidityBit, \CommitmentRegister{1}, \OpeningRegister{1}, \AnswerRegister{1})$: To make a query to the message, $\QVCQuery^{\RandomOracle{0}}$, with oracle access to $\RandomOracle{1}$, takes as inputs the security parameter $1^{\Security}$, the query set $\QVCQuerySet \subseteq \{1, 2, \cdots, \QVCMessageLength\}$, and a state on registers $(\CommitmentRegister{1}, \OpeningRegister{1}, \AnswerRegister{1})$. It first checks the validity of the opening and outputs a bit $\QVCValidityBit$ to indicate whether the decommitment is valid. If it is valid, it recovers the quantum message on positions $\QVCQuerySet$, and swaps it with the register $\AnswerRegister{1}$, then it does the reverse to get a new state on registers $(\QVCCommitmentRegister, \QVCOpeningRegister)$, and outputs a state on registers $(\QVCCommitmentRegister, \QVCOpeningRegister, \AnswerRegister{1})$ together with the validity bit $\QVCValidityBit$.
\item $\QVCUpdate^{\RandomOracle{0}}(1^\Security, \QVCQuerySet, \QVCOpeningRegister, \QVCAuxiliaryNotUsedRegister) \to \QVCAuxiliaryRegister{1}$: $\QVCUpdate^{\RandomOracle{0}}$, with oracle access to $\RandomOracle{1}$, takes as inputs the security parameter $1^{\Security}$, a query set $\QVCQuerySet$, and a state on registers $(\QVCOpeningRegister, \QVCAuxiliaryNotUsedRegister)$, and reorganizes the states according to $\QVCQuerySet$ to erase the information about $\QVCQuerySet$. It outputs a state on the register $\QVCAuxiliaryRegister{1}$.
\item $\QVCRecover^{\RandomOracle{0}}(1^\Security,\QVCCommitmentRegister,\QVCAuxiliaryRegister{1}) \to (\QVCValidityBit, \QVCMessageRegister)$: To recover the messages, $\QVCRecover^{\RandomOracle{0}}$, with oracle access to $\RandomOracle{1}$, takes as inputs the security parameter $1^{\Security}$ and a state on two registers $(\QVCCommitmentRegister,\QVCAuxiliaryRegister{1})$. It outputs a validity bit $\QVCValidityBit$ and a state on the register $\QVCMessageRegister$, which is supposed to be the updated message after the queries.
\end{itemize}

We slightly overload the notations $\QVCOpen$, $\QVCQuery$, and $\QVCUpdate$ to also denote the corresponding coherent implementations when the query set is provided in the register $\QVCQuerySetRegister{1}$. In particular, we write $(\QVCQuerySetRegister{1}, \OpeningRegister{1}, \QVCAuxiliaryNotUsedRegister) \gets \QVCOpen^{\RandomOracle{0}}(1^\Security, \QVCQuerySetRegister{1}, \QVCAuxiliaryRegister{1})$, $(\QVCValidityBit, \QVCQuerySetRegister{1}, \CommitmentRegister{1}, \OpeningRegister{1}, \AnswerRegister{1}) \gets \QVCQuery^{\RandomOracle{0}}(1^\Security, \QVCQuerySetRegister{1}, \CommitmentRegister{1}, \OpeningRegister{1}, \AnswerRegister{1})$, and $(\QVCQuerySetRegister{1}, \QVCAuxiliaryRegister{1}) \gets \QVCUpdate^{\RandomOracle{0}}(1^\Security, \QVCQuerySetRegister{1}, \QVCOpeningRegister, \QVCAuxiliaryNotUsedRegister)$ for the coherent implementations. When it is clear from the context, we sometimes omit the security parameter $\Security$.

The vector commitment $\QVC$ should satisfy the following correctness and efficiency requirements.

\parhead{Perfect correctness} An honest sender should be able to produce a valid decommitment for every part of the quantum message and recover the updated message after the queries. Namely, for every integer $\Security$, $\QVCBlockSize$, function $\RandomOracle{1}$ with output length $\RandomOracleOutputLength$, unbounded quantum adversary $\Adversary$, and unbounded quantum distinguisher $\Distinguisher$,
\begin{align*}
&\prob{
\Distinguisher^{\RandomOracle{0}}(\ValidityBit, \QVCMessageRegister, \EnvironmentRegister{1}, \QVCQuerySetRegister{1}, \AnswerRegister{1})
\;\middle\vert\;
\begin{array}{l}
(\QVCMessageRegister, \EnvironmentRegister{1}, \QVCQuerySetRegister{1}, \AnswerRegister{1}) \gets \Adversary^{\RandomOracle{0}}\\
({\CommitmentRegister{1}, \QVCAuxiliaryRegister{1}}) \gets \QVCCommit^{\RandomOracle{0}}(1^\Security, \MessageRegister{1})\\
(\QVCQuerySetRegister{1}, \OpeningRegister{1}, \QVCAuxiliaryNotUsedRegister) \gets \QVCOpen^{\RandomOracle{0}}(1^\Security, \QVCQuerySetRegister{1}, \QVCAuxiliaryRegister{1})\\
(\QVCValidityBit_1, \QVCQuerySetRegister{1}, \CommitmentRegister{1}, \OpeningRegister{1}, \AnswerRegister{1}) \gets \QVCQuery^{\RandomOracle{0}}(1^\Security, \QVCQuerySetRegister{1}, \CommitmentRegister{1}, \OpeningRegister{1}, \AnswerRegister{1})\\
(\QVCQuerySetRegister{1}, \QVCAuxiliaryRegister{1}) \gets \QVCUpdate^{\RandomOracle{0}}(1^\Security, \QVCQuerySetRegister{1}, \QVCOpeningRegister, \QVCAuxiliaryNotUsedRegister)\\
(\QVCValidityBit_2, \QVCMessageRegister) \gets \QVCRecover^{\RandomOracle{0}}(1^\Security,\QVCCommitmentRegister,\QVCAuxiliaryRegister{1})\\
\QVCValidityBit \coloneq \QVCValidityBit_1 \land \QVCValidityBit_2
\end{array}
}\\
&=
\prob{
\Distinguisher^{\RandomOracle{0}}(1, \QVCMessageRegister, \EnvironmentRegister{1}, \QVCQuerySetRegister{1}, \AnswerRegister{1})
\;\middle\vert\;
\begin{array}{l}
(\QVCMessageRegister, \EnvironmentRegister{1}, \QVCQuerySetRegister{1}, \AnswerRegister{1}) \gets \Adversary^{\RandomOracle{0}}\\
(\QVCQuerySetRegister{1}, \QVCMessageRegister, \AnswerRegister{1}) \gets \QueryUnitary(\QVCQuerySetRegister{1}, \QVCMessageRegister, \AnswerRegister{1})
\end{array}
}
\enspace,
\end{align*}
where the random oracle output length $\RandomOracleOutputLength$ is a function of the security parameter $\Security$ specified by the scheme, and we slightly abuse notation by overloading $\QueryUnitary$ to also denote the unitary that, controlled on $\QVCQuerySetRegister{1}$, swaps $\AnswerRegister{1}$ with the corresponding subsets of the message register $\QVCMessageRegister$ (i.e., a parallel application of the unitary $\QueryUnitary$ from \Cref{sec:QIOP-detail}).

\parhead{Local opening} For a vector commitment scheme to be non-trivial, we require that the sizes of $\QVCCommitmentRegister$ and $\QVCOpeningRegister$ grow much slower than the message length. To be more specific, the size of $\QVCCommitmentRegister$ is $\poly(\Security)$, and the size of $\QVCOpeningRegister$ is $\poly(\abs{\QVCQuerySet}, \log \QVCMessageLength, \Security)$.
\end{definition}

\subsection{The extractability definition}

Now we're ready to define the extractability for quantum state vector commitments. The syntax of a quantum state vector commitment is different from the syntax for the basic quantum state commitment.

\begin{definition}[Quantum message extractor]
\label{def:quantum_message_extractor_VC}
A \emph{quantum message extractor} $\Extractor$ for a quantum state vector commitment scheme with query access to $\Simulator$ and $\ExtractOracle$ (both have the state register $\StateRegister{1}$) has the following syntax:
\begin{enumerate}[noitemsep]
\item $\Extractor.\ExtractMessage^{\Simulator(\StateRegister{1}), \ExtractOracle(\StateRegister{1})}(\CommitmentRegister{1}) \to (\MessageRegister{1}, \AltOpeningRegister{1})$: $\Extractor.\ExtractMessage$ takes as inputs a commitment in register $\CommitmentRegister{1}$, makes queries to $\Simulator$ and $\ExtractOracle$, and outputs a quantum state on two registers $\MessageRegister{1}$ and $\AltOpeningRegister{1}$ (which will be used for $\AltCheck$).
\item $\Extractor.\AltCheck(\QVCQuerySet,\OpeningRegister{1}, \AltOpeningRegister{1}) \to (\ValidityBit, \OpeningRegister{1}, \AltOpeningRegister{1})$: $\Extractor.\AltCheck$ takes as inputs the query set $\QVCQuerySet$, the opening in register $\OpeningRegister{1}$, and the auxiliary information in register $\AltOpeningRegister{1}$, and outputs $\ValidityBit = 0$ or $1$, indicating whether the sender gives a valid opening, along with a quantum state on registers $(\OpeningRegister{1}, \AltOpeningRegister{1})$.
\item $\Extractor.\AltCommit^{\Simulator(\StateRegister{1}), \ExtractOracle(\StateRegister{1})}(\AltOpeningRegister{1}, \QVCMessageRegister) \to \QVCCommitmentRegister$: $\Extractor.\AltCommit$ takes as inputs a quantum state on the register $\AltOpeningRegister{1}$, and the updated extracted quantum message on the register $\QVCMessageRegister$, makes queries to $\Simulator$ and $\ExtractOracle$, and outputs a quantum state on the register $\QVCCommitmentRegister$.
\end{enumerate}

We slightly overload the notation $\Extractor.\AltCheck$ to also denote the corresponding coherent implementations when the query set is provided in the register $\QVCQuerySetRegister{1}$. In particular, we write $(\ValidityBit, \QVCQuerySetRegister{1}, \OpeningRegister{1}, \AltOpeningRegister{1}) \gets \Extractor.\AltCheck(\QVCQuerySetRegister{1},\OpeningRegister{1}, \AltOpeningRegister{1})$ for the coherent implementation.
\end{definition}

\begin{definition}[Extractability]
\label{def:qvc_extractability}
A quantum state vector commitment scheme $\QVC = \QVCTuple$ is \emph{$(\ExtractionErrorI{1}, \ExtractionErrorI{2}, \ExtractionErrorI{3})$-extractable} if there exist a $\ExtractionErrorI{1}$-quantum simulator $\Simulator$ with state register $\StateRegister{1}$ for the random oracle, a stateful oracle $\ExtractOracle$ with the same state register $\StateRegister{1}$, and a polynomial-time quantum message extractor $\Extractor$ with query access to $\Simulator$ and $\ExtractOracle$ as in \Cref{def:quantum_message_extractor_VC} such that the following holds:

\begin{enumerate}
\item ($\ExtractOracle$ does not disturb the simulation.) For every integer $\RandomOracleOutputLength$, $\matnorm{\Commutator{\ExtractOracle}{\Simulator}}^2 \leq \ExtractionErrorI{2}(\RandomOracleOutputLength)$, where $\RandomOracleOutputLength$ is the output length of the random oracle.
\item ($\Extractor$ gives the only state that the adversary $\Adversary$ can open to.) For every integer $\RandomOracleOutputLength$, $\NumberOfQueries$, real numbers $\QuantumTotalQueryMass_1, \QuantumTotalQueryMass_2 \in [0, \NumberOfQueries]$ such that $\QuantumTotalQueryMass_1 + \QuantumTotalQueryMass_2 \leq \NumberOfQueries$, $\NumberOfQueries$-query $\NumberOfCommitmentsPhases$-phase quantum adversary $\Adversary$ such that $\QuantumTotalQueryMassFunc{\Adversary, \Simulator} \leq \QuantumTotalQueryMass_1$ and $\QuantumTotalQueryMassFunc{\Adversary, \ExtractOracle} \leq \QuantumTotalQueryMass_2$, and unbounded quantum distinguisher $\Distinguisher$,
\[\abs{\sqrt{\prob{\QVCSimulateWorld(\Adversary, \Distinguisher)}} - \sqrt{\prob{\QVCExtractWorld(\Adversary, \Distinguisher)}}}^2 \leq \ExtractionErrorI{3}(\NumberOfQueries, \NumberOfCommitmentsPhases, \QuantumTotalQueryMass_1, \QuantumTotalQueryMass_2, \RandomOracleOutputLength, \QVCMessageLength)\enspace,\]
where $\RandomOracleOutputLength$ is the output length of the random oracle, $\QVCMessageLength$ is the message length of the scheme, and the games $\QVCSimulateWorld(\Adversary, \Distinguisher)$ and $\QVCExtractWorld(\Adversary, \Distinguisher)$ are defined below:
\begin{itemize}
\item
\begin{itemize}[noitemsep]
\item[]$\QVCSimulateWorld(\Adversary, \Distinguisher)$:
\begin{enumerate}[nolistsep]
\item The game initializes the registers: $\StateRegister{1} \gets \ket{\bot}$, $\QVCQuerySetRegister{1} \gets \ket{\bar{0}}$, $\AnswerRegister{1} \gets \ket{\bar{0}}$, $ \OpeningRegister{1} \gets \ket{\bar{0}}$.
\item The adversary generates a commitment: $(\CommitmentRegister{1}, \EnvironmentRegister{1}) \gets \Adversary^{\Simulator(\StateRegister{1}), \ExtractOracle(\StateRegister{1})}$.
\item The adversary does the following:

For $i \in [\NumberOfCommitmentsPhases]$:
\begin{enumerate}[noitemsep]
\item The adversary generates a query location and the corresponding opening, together with an answer register for the superposition query: $(\QVCQuerySetRegister{1}, \OpeningRegister{1}, \AnswerRegister{1}, \EnvironmentRegister{1}) \gets \Adversary^{\Simulator(\StateRegister{1}), \ExtractOracle(\StateRegister{1})}(\QVCQuerySetRegister{1}, \OpeningRegister{1}, \AnswerRegister{1}, \EnvironmentRegister{1})$.
\item The game implements the query: $(\QVCValidityBit_i, \QVCQuerySetRegister{1}, \CommitmentRegister{1}, \OpeningRegister{1}, \AnswerRegister{1}) \gets \QVCQuery^{\Simulator(\StateRegister{1})}(1^\Security, \QVCQuerySetRegister{1}, \CommitmentRegister{1}, \OpeningRegister{1}, \AnswerRegister{1})$.
\end{enumerate}
\item The game generates the output: If $\land_{i \in [\NumberOfCommitmentsPhases]}\ValidityBit_i = 0$, output 0; otherwise, compute $\ValidityBit' \gets \Distinguisher(\QVCQuerySetRegister{1}, \CommitmentRegister{1}, \OpeningRegister{1}, \AnswerRegister{1}, \EnvironmentRegister{1}, \StateRegister{1})$ and output $\ValidityBit'$.
\end{enumerate}
\end{itemize}
\item
\begin{itemize}[noitemsep]
\item[]$\QVCExtractWorld(\Adversary, \Distinguisher)$:
\begin{enumerate}[nolistsep]
\item The game initializes the registers: $\StateRegister{1} \gets \ket{\bot}$, $\QVCQuerySetRegister{1} \gets \ket{\bar{0}}$, $\AnswerRegister{1} \gets \ket{\bar{0}}$, $ \OpeningRegister{1} \gets \ket{\bar{0}}$.
\item The adversary generates a commitment: $(\CommitmentRegister{1}, \EnvironmentRegister{1}) \gets \Adversary^{\Simulator(\StateRegister{1}), \ExtractOracle(\StateRegister{1})}$.
\item The game uses the extractor to extract the underlying message: $(\MessageRegister{1}, \AltOpeningRegister{1}) \gets \Extractor.\ExtractMessage^{\Simulator(\StateRegister{1}), \ExtractOracle(\StateRegister{1})}(\CommitmentRegister{1})$.
\item The adversary does the following:

For $i \in [\NumberOfCommitmentsPhases]$:
\begin{enumerate}[noitemsep]
\item The adversary generates a query location and the corresponding opening, together with an answer register for the superposition query: $(\QVCQuerySetRegister{1}, \OpeningRegister{1}, \AnswerRegister{1}, \EnvironmentRegister{1}) \gets \Adversary^{\Simulator(\StateRegister{1}), \ExtractOracle(\StateRegister{1})}(\QVCQuerySetRegister{1}, \OpeningRegister{1}, \AnswerRegister{1}, \EnvironmentRegister{1})$.
\item The game uses the alternative check to check if the opening is valid: $(\ValidityBit_i, \QVCQuerySetRegister{1}, \OpeningRegister{1}, \AltOpeningRegister{1}) \gets \Extractor.\AltCheck(\QVCQuerySetRegister{1},\OpeningRegister{1}, \AltOpeningRegister{1})$.
\item The game implements the query: $(\QVCQuerySetRegister{1}, \QVCMessageRegister, \AnswerRegister{1}) \gets \QueryUnitary(\QVCQuerySetRegister{1}, \QVCMessageRegister, \AnswerRegister{1})$.
\end{enumerate}
\item The game uses alternative commit to get the commitment: $\QVCCommitmentRegister \gets \Extractor.\AltCommit^{\Simulator(\StateRegister{1}), \ExtractOracle(\StateRegister{1})}(\AltOpeningRegister{1}, \QVCMessageRegister)$.
\item The game generates the output: If $\land_{i \in [\NumberOfCommitmentsPhases]}\ValidityBit_i = 0$, output 0; otherwise, compute $\ValidityBit' \gets \Distinguisher(\QVCQuerySetRegister{1}, \CommitmentRegister{1}, \OpeningRegister{1}, \AnswerRegister{1}, \EnvironmentRegister{1}, \StateRegister{1})$ and output $\ValidityBit'$.
\end{enumerate}
\end{itemize}
\end{itemize}
\end{enumerate}
\end{definition}

\doclearpage
\section{Construction of extractable quantum state vector commitments}

\subsection{Labeling the Merkle tree}

We make some conventions on how to label each vertex in the Merkle tree before giving the construction for the quantum state vector commitment.

A vertex $v$ in a Merkle tree of depth $d$ is labeled as a string $\indexForMTNode \in \Bits^{\leq d}$. The length of $\indexForMTNode$ indicates the depth of the vertex $v$ from the root, and the $i$-th bit of $\indexForMTNode$ indicates whether we need to take the left edge or the right edge on the path from the root to $v$ (0 for left and 1 for right). With the above convention, we have the following facts:
\begin{itemize}[nolistsep]
\item The root of the Merkle tree is labeled as the empty string $\emptystring$;
\item For a string $\indexForMTNode \in \Bits^{< d}$, the left child of a vertex $v$ with label $\indexForMTNode$ is labeled as $\indexForMTNode \parallel 0$;
\item For a string $\indexForMTNode \in \Bits^{< d}$, the right child of a vertex $v$ with label $\indexForMTNode$ is labeled as $\indexForMTNode \parallel 1$;
\item For a non-empty string $\indexForMTNode \in \Bits^{\leq d}$, let $\sibling{\indexForMTNode}$ be the label of the sibling of the vertex with label $\indexForMTNode$. Then $\sibling{\cdot}$ satisfies that $\sibling{\indexForMTNode}$ is the string of length $\abs{\indexForMTNode}$ and $\sibling{\indexForMTNode}$ equals $\indexForMTNode$ on every bit except the last one.
\item For a non-empty string $\indexForMTNode \in \Bits^{\leq d}$, let $\parent{\indexForMTNode}$ be the label of the parent of the vertex with label $\indexForMTNode$. Then $\parent{\indexForMTNode}$ is the prefix of $\indexForMTNode$ of length $\abs{\indexForMTNode} - 1$.
\item For a string $\indexForMTNode \in \Bits^{\leq d}$, let $\QSTCPath{\indexForMTNode}$ be the set of labels of the vertices on the path from the root to the vertex $v$ with label $\indexForMTNode$ (including the root and the vertex $v$). Then $\QSTCPath{\indexForMTNode}$ contains the $\abs{\indexForMTNode} + 1$ prefixes of $\indexForMTNode$ (including the empty string $\emptystring$ and the string $\indexForMTNode$ itself).
\end{itemize}

We overload the notation to use the string $\indexForMTNode \in \Bits^{\leq d}$ to mean the vertex with label $\indexForMTNode$.

Furthermore, we overload the functions $\sibling{\cdot}$ and $\QSTCPath{\cdot}$ to take sets as input. For every set $S \subseteq \Bits^{\leq d}$, we denote the set of siblings of the vertices inside $S$ as $\sibling{S} \coloneq \{\sibling{\indexForMTNode} : \indexForMTNode \in S \text{ and } \indexForMTNode \neq \emptystring\}$, and the set of vertices from the root to vertices inside $S$ as $\QSTCPath{S} \coloneq \bigcup_{\indexForMTNode \in S}\QSTCPath{\indexForMTNode}$.

For each set $S \subseteq \Bits^{d}$, we define $\QSTCAuthPath{S} \coloneq \sibling{\QSTCPath{S}} \setminus \QSTCPath{S}$.

\subsection{An extractable quantum state vector commitment scheme}

We apply the Merkle tree to the succinct quantum state commitment construction in \Cref{construction:basic_commitment} to obtain a quantum state vector commitment scheme.

\begin{construction}
\label{construction:quantum-state-vector-commitment}
We assume the number of blocks $\QVCMessageLength$ is a power of 2 (otherwise, we just add padding of 0 at the end). Denote $\QVCMessageLength = 2^{\QVCMessageDepth}$ for an integer $\QVCMessageDepth$.

Let $\RandomOracle{1}$ be sampled uniformly from all the functions with range $\Bits^{\RandomOracleOutputLength}$, where we set $\RandomOracleOutputLength$ to be the security parameter $\Security$, and let $\CM = (\Commit^{\RandomOracle{0}}, \Check^{\RandomOracle{0}})$ be the succinct quantum state commitment in \Cref{construction:basic_commitment}. Let $\CommitCircuit_{\RandomOracleOutputLength, \MessageLength}$ be the circuit for $\Commit^{\RandomOracle{0}}$ where the quantum message has $\MessageLength \coloneq 4\RandomOracleOutputLength$ qubits (in other words, it is the circuit in \Cref{construction:basic_commitment_circuit} for $\MessageLength = 4\RandomOracleOutputLength$).

We construct $\QSTC = \QVCTuple$ for block size $\QVCBlockSize \coloneq 2\RandomOracleOutputLength$ as follows.
\begin{itemize}
\item[] $\QSTC.\Commit^{\RandomOracle{0}}(\MessageRegister{1})$:
\begin{enumerate}[nolistsep]
\item Divide the $(\QVCMessageLength\cdot\QVCBlockSize)$-qubit register $\MessageRegister{1}$ into $\QVCMessageLength$ registers $(\accessVectorAt{\MessageRegister{1}}{\indexForMTNode})_{\indexForMTNode \in \Bits^{\QVCMessageDepth}}$ each of size $\QVCBlockSize$.
\item For each $\indexForMTNode \in \Bits^{\QVCMessageDepth - 1}$, define $\accessVectorAt{\MessageRegister{1}}{\indexForMTNode} \coloneq (\accessVectorAt{\MessageRegister{1}}{\indexForMTNode, 0}, \accessVectorAt{\MessageRegister{1}}{\indexForMTNode, 1})$.
\item Initialize $\QVCMessageLength - 1$ registers $(\accessVectorAt{\AncillasRegister{1}}{\indexForMTNode})_{\indexForMTNode \in \Bits^{< \QVCMessageDepth}}$, each of size $2\RandomOracleOutputLength$, as all-zero states.
\item For $j = d - 1, d - 2, \ldots, 0$:
\begin{enumerate}[nolistsep,noitemsep]
\item For each $\indexForMTNode \in \Bits^{j}$, apply $\CommitCircuit_{\RandomOracleOutputLength, \MessageLength}$, where the quantum message is on the registers $\accessVectorAt{\MessageRegister{1}}{\indexForMTNode}$ (the quantum message has $2\QVCBlockSize = \MessageLength$ qubits), the ancilla is on the register $\accessVectorAt{\AncillasRegister{1}}{\indexForMTNode}$, and the oracle access is implemented by $\RO$, to obtain $(\accessVectorAt{\CommitmentRegister{1}}{\indexForMTNode}, \accessVectorAt{\OpeningRegister{1}}{\indexForMTNode})$.
\item If $j \geq 1$, for each $\indexForMTNode \in \Bits^{j - 1}$, set $\accessVectorAt{\MessageRegister{1}}{\indexForMTNode} \coloneq (\accessVectorAt{\CommitmentRegister{1}}{\indexForMTNode, 0}, \accessVectorAt{\CommitmentRegister{1}}{\indexForMTNode, 1})$.
\end{enumerate}
\item Set $\CommitmentRegister{1} \coloneq \accessVectorAt{\CommitmentRegister{1}}{\emptystring}$ and $\QVCAuxiliaryRegister{1} \coloneq \bigotimes_{\indexForMTNode \in \Bits^{< \QVCMessageDepth}}\accessVectorAt{\OpeningRegister{1}}{\indexForMTNode}$.
\item Output $\CommitmentRegister{1}$ and $\QVCAuxiliaryRegister{1}$.
\end{enumerate}
\item[] $\QSTC.\Open^{\RandomOracle{0}}(\QVCQuerySet, \QVCAuxiliaryRegister{1})$:
\begin{enumerate}[nolistsep]
\item Parse $\QVCAuxiliaryRegister{1}$ as $\bigotimes_{\indexForMTNode \in \Bits^{< \QVCMessageDepth}}\accessVectorAt{\OpeningRegister{1}}{\indexForMTNode}$.
\item Set $\QVCOpeningRegister \coloneq \bigotimes_{\indexForMTNode \in \QSTCPath{\QVCQuerySet}\setminus \QVCQuerySet}\accessVectorAt{\OpeningRegister{1}}{\indexForMTNode}$ and $\QVCAuxiliaryNotUsedRegister \coloneq \bigotimes_{\indexForMTNode \in \Bits^{< \QVCMessageDepth}\setminus\QSTCPath{\QVCQuerySet}}\accessVectorAt{\OpeningRegister{1}}{\indexForMTNode}$.
\item Output $\QVCOpeningRegister$ and $\QVCAuxiliaryNotUsedRegister$.
\end{enumerate}
\item[] $\QSTC.\Query^{\RandomOracle{0}}(\QVCQuerySet, \CommitmentRegister{1}, \OpeningRegister{1}, \AnswerRegister{1})$:
\begin{enumerate}[nolistsep]
\item Set $\accessVectorAt{\CommitmentRegister{1}}{\emptystring} \coloneq \CommitmentRegister{1}$, parse $\QVCOpeningRegister$ as $\bigotimes_{\indexForMTNode \in \QSTCPath{\QVCQuerySet}\setminus\QVCQuerySet}\accessVectorAt{\OpeningRegister{1}}{\indexForMTNode}$, and parse $\AnswerRegister{1}$ as $\bigotimes_{\indexForMTNode \in \QVCQuerySet}\accessVectorAt{\AnswerRegister{1}}{\indexForMTNode}$.
\item For $j = 0, 1, \ldots, \QVCMessageDepth - 1$:
\begin{enumerate}[nolistsep]
\item For each $\indexForMTNode \in \QSTCPath{\QVCQuerySet} \cap \Bits^j$, apply the inverse of $\CommitCircuit_{\RandomOracleOutputLength, \QVCMessageLength}$ to $(\accessVectorAt{\CommitmentRegister{1}}{\indexForMTNode}, \accessVectorAt{\OpeningRegister{1}}{\indexForMTNode})$, where the oracle access is implemented by $\RO$, to obtain the committed quantum message in $\accessVectorAt{\MessageRegister{1}}{\indexForMTNode}$ and the ancilla in $\accessVectorAt{\AncillasRegister{1}}{\indexForMTNode}$.
\item If $j \neq \QVCMessageDepth - 1$, for each $\indexForMTNode \in \QSTCPath{\QVCQuerySet} \cap \Bits^j$, parse $\accessVectorAt{\MessageRegister{1}}{\indexForMTNode}$ as $(\accessVectorAt{\CommitmentRegister{1}}{\indexForMTNode, 0}, \accessVectorAt{\CommitmentRegister{1}}{\indexForMTNode, 1})$.
\item If $j = \QVCMessageDepth - 1$, for each $\indexForMTNode \in \QSTCPath{\QVCQuerySet} \cap \Bits^j$, parse $\accessVectorAt{\MessageRegister{1}}{\indexForMTNode}$ as $(\accessVectorAt{\MessageRegister{1}}{\indexForMTNode, 0}, \accessVectorAt{\MessageRegister{1}}{\indexForMTNode, 1})$.
\end{enumerate}
\item For $\indexForMTNode \in \QVCQuerySet$, apply $\SWAP$ operator: $(\accessVectorAt{\MessageRegister{1}}{\indexForMTNode}, \accessVectorAt{\AnswerRegister{1}}{\indexForMTNode}) \gets \SWAP(\accessVectorAt{\MessageRegister{1}}{\indexForMTNode}, \accessVectorAt{\AnswerRegister{1}}{\indexForMTNode})$.
\item Measure $\accessVectorAt{\AncillasRegister{1}}{\indexForMTNode}$ for each $\indexForMTNode \in \QSTCPath{\QVCQuerySet}\setminus\QVCQuerySet$ in the computational basis, and set $\ValidityBit = 1$ if all the measurement outcomes are 0, and otherwise set $\ValidityBit = 0$.
\item For each $\indexForMTNode \in \QSTCPath{\QVCQuerySet} \cap \Bits^{\QVCMessageDepth - 1}$, define $\accessVectorAt{\MessageRegister{1}}{\indexForMTNode} \coloneq (\accessVectorAt{\MessageRegister{1}}{\indexForMTNode, 0}, \accessVectorAt{\MessageRegister{1}}{\indexForMTNode, 1})$.
\item For $j = \QVCMessageDepth - 1, \ldots, 1, 0$:
\begin{enumerate}[nolistsep]
\item For each $\indexForMTNode \in \QSTCPath{\QVCQuerySet} \cap \Bits^j$, apply $\CommitCircuit_{\RandomOracleOutputLength, \MessageLength}$, where the quantum message is on the registers $\accessVectorAt{\MessageRegister{1}}{\indexForMTNode}$, the ancilla is on the register $\accessVectorAt{\AncillasRegister{1}}{\indexForMTNode}$, and the oracle access is implemented by $\RO$, to obtain $(\accessVectorAt{\CommitmentRegister{1}}{\indexForMTNode}, \accessVectorAt{\OpeningRegister{1}}{\indexForMTNode})$.
\item If $j \neq 0$, for each $\indexForMTNode \in \QSTCPath{\QVCQuerySet} \cap \Bits^{j - 1}$, set $\accessVectorAt{\MessageRegister{1}}{\indexForMTNode} \coloneq (\accessVectorAt{\CommitmentRegister{1}}{\indexForMTNode, 0}, \accessVectorAt{\CommitmentRegister{1}}{\indexForMTNode, 1})$.
\end{enumerate}
\item Set $\CommitmentRegister{1} \coloneq \accessVectorAt{\CommitmentRegister{1}}{\emptystring}$, $\OpeningRegister{1} \coloneq \bigotimes_{\indexForMTNode \in \QSTCPath{\QVCQuerySet}\setminus\QVCQuerySet}\accessVectorAt{\OpeningRegister{1}}{\indexForMTNode}$, and $\AnswerRegister{1} \coloneq \bigotimes_{\indexForMTNode \in \QVCQuerySet}\accessVectorAt{\AnswerRegister{1}}{\indexForMTNode}$.
\item Output $(\QVCValidityBit,\CommitmentRegister{1}, \OpeningRegister{1}, \AnswerRegister{1})$.
\end{enumerate}
\item[] $\QSTC.\Update^{\RandomOracle{0}}(\QVCQuerySet, \QVCOpeningRegister, \QVCAuxiliaryNotUsedRegister)$:
\begin{enumerate}[nolistsep]
\item Parse $\QVCOpeningRegister$ as $\bigotimes_{\indexForMTNode \in \QSTCPath{\QVCQuerySet}\setminus\QVCQuerySet}\accessVectorAt{\OpeningRegister{1}}{\indexForMTNode}$ and parse $\QVCAuxiliaryNotUsedRegister$ as $\bigotimes_{\indexForMTNode \in \Bits^{< \QVCMessageDepth}\setminus\QSTCPath{\QVCQuerySet}}\accessVectorAt{\OpeningRegister{1}}{\indexForMTNode}$.
\item Set $\QVCAuxiliaryRegister{1} \coloneq \bigotimes_{\indexForMTNode \in \Bits^{< \QVCMessageDepth}}\accessVectorAt{\OpeningRegister{1}}{\indexForMTNode}$.
\item Output $\QVCAuxiliaryRegister{1}$.
\end{enumerate}
\item[] $\QSTC.\Recover^{\RandomOracle{0}}(\QVCCommitmentRegister,\QVCAuxiliaryRegister{1})$:
\begin{enumerate}[nolistsep]
\item Set $\accessVectorAt{\CommitmentRegister{1}}{\emptystring} \coloneq \CommitmentRegister{1}$, and parse $\QVCAuxiliaryRegister{1}$ as $\bigotimes_{\indexForMTNode \in \Bits^{< \QVCMessageDepth}}\accessVectorAt{\OpeningRegister{1}}{\indexForMTNode}$.
\item For $j = 0, 1, \ldots, d - 1$:
\begin{enumerate}[nolistsep,noitemsep]
\item For each $\indexForMTNode \in \Bits^{j}$, apply the inverse of $\CommitCircuit_{\RandomOracleOutputLength, \QVCMessageLength}$ to $(\accessVectorAt{\CommitmentRegister{1}}{\indexForMTNode}, \accessVectorAt{\OpeningRegister{1}}{\indexForMTNode})$, where the oracle access is implemented by $\RO$, to obtain the committed quantum message in $\accessVectorAt{\MessageRegister{1}}{\indexForMTNode}$ and the ancilla in $\accessVectorAt{\AncillasRegister{1}}{\indexForMTNode}$.
\item If $j \neq \QVCMessageDepth - 1$, for each $\indexForMTNode \in \Bits^j$, parse $\accessVectorAt{\MessageRegister{1}}{\indexForMTNode}$ as $(\accessVectorAt{\CommitmentRegister{1}}{\indexForMTNode, 0}, \accessVectorAt{\CommitmentRegister{1}}{\indexForMTNode, 1})$.
\item If $j = \QVCMessageDepth - 1$, for each $\indexForMTNode \in \Bits^j$, parse $\accessVectorAt{\MessageRegister{1}}{\indexForMTNode}$ as $(\accessVectorAt{\MessageRegister{1}}{\indexForMTNode, 0}, \accessVectorAt{\MessageRegister{1}}{\indexForMTNode, 1})$.
\end{enumerate}
\item Measure $\accessVectorAt{\AncillasRegister{1}}{\indexForMTNode}$ for each $\indexForMTNode \in \Bits^{< \QVCMessageDepth}$ in the computational basis, and set $\ValidityBit = 1$ if all the measurement outcomes are 0, and otherwise set $\ValidityBit = 0$.
\item Output $(\QVCValidityBit, \QVCMessageRegister)$.
\end{enumerate}
\end{itemize}

Note that in our construction, $\QSTC.\Recover$ just does the inverse of $\QSTC.\Commit$.
\end{construction}

By construction, $\QSTC$ has the following efficiency:
\begin{itemize}[nolistsep]
\item The size of $\QVCCommitmentRegister$ is always $O(\RandomOracleOutputLength) = O(\Security)$ regardless of the message length.
\item The size of $\QVCOpeningRegister$ for a set $\QVCQuerySet$ is $O(\RandomOracleOutputLength\abs{\QSTCPath{\QVCQuerySet}\setminus\QVCQuerySet}) = O(\RandomOracleOutputLength\abs{\QVCQuerySet} \QVCMessageDepth) = O(\Security\abs{\QVCQuerySet} \log \QVCMessageLength)$.
\end{itemize}
Therefore, $\QSTC$ satisfies the local opening requirement of $\QVC$.

Moreover, the perfect completeness of $\QSTC$ follows from the perfect completeness of $\Commit$.

\subsection{Extractor for the vector commitment scheme}

We present the extractor for the quantum state vector commitment scheme $\QSTC$ in \Cref{construction:quantum-state-vector-commitment} before showing $\QSTC$ is extractable.

We begin with the oracle access of the extractor for $\QSTC$, which is exactly the same as the oracle access of the extractor for the basic commitment in \Cref{construction:basic_commitment_extractor}. Specifically, we use the unitary \[\OracleUnitary_{\QueryRegister{0}\AnswerRegister{0}\DatabaseRegister{0}} = \sum_{x \in \RandomOracleDomain}\ketbra{x}{x}_{\QueryRegister{0}}\otimes \compress_{\DatabaseRegisterAt{0}{x}}\CNOT_{\DatabaseRegisterAt{0}{x}\AnswerRegister{0}}\compress_{\DatabaseRegisterAt{0}{x}}\]
to simulate the random oracle, and the unitary $\invert_{\TargetRegister{0}\DatabaseRegister{0}\WorkingRegister{0}}$ as defined in \Cref{eqn:inv-def} to help the extractor, where $\WorkingRegister{1}$ can be parsed as $\FlagRegister{1}\Another{\MessageRegister{1}}$.

\begin{construction}[The extractor $\Extractor_{\QSTC}$]
\label{construction:vector_commitment_extractor}
Let $\Extractor_{\CM}$ be the extractor for the scheme $\CM$ in \Cref{construction:basic_commitment_extractor}. We consider the following $\Extractor_{\QSTC}$ for $\QSTC = \QVCTuple$ in \Cref{construction:quantum-state-vector-commitment}. $\Extractor_{\QSTC}$ has query access to $\OracleUnitary(\DatabaseRegister{1})$ and $\invert(\DatabaseRegister{1})$.

\begin{itemize}[noitemsep]
\item[] $\Extractor_{\QSTC}$ has the following interfaces:
\begin{itemize}[nolistsep]
\item $\Extractor_{\QSTC}.\ExtractMessage(\CommitmentRegister{1})$:
\begin{enumerate}[nolistsep]
\item Set $\accessVectorAt{\CommitmentRegister{1}}{\emptystring} \coloneq \CommitmentRegister{1}$.
\item For $j = 0, 1, \cdots, \QVCMessageDepth - 1$:
\begin{enumerate}[nolistsep]
\item For each $\indexForMTNode \in \Bits^{j}$, run $\Extractor_{\CM}.\ExtractMessage(\accessVectorAt{\CommitmentRegister{1}}{\indexForMTNode})$ with the working register $(\accessVectorAt{\StandardBasisWorkingRegister{1}}{\indexForMTNode}, \accessVectorAt{\HadamardBasisWorkingRegister{1}}{\indexForMTNode})$ initialized as all zero states by forwarding its queries to $\OracleUnitary(\DatabaseRegister{1})$ and $\invert(\DatabaseRegister{1})$ to obtain $\accessVectorAt{\MessageRegister{1}}{\indexForMTNode}$ and $\accessVectorAt{\AltOpeningRegister{1}}{\indexForMTNode}$.
\item If $j \neq \QVCMessageDepth - 1$, for each $\indexForMTNode \in \Bits^{j}$, parse $\accessVectorAt{\MessageRegister{1}}{\indexForMTNode}$ as $(\accessVectorAt{\CommitmentRegister{1}}{\indexForMTNode, 0},\accessVectorAt{\CommitmentRegister{1}}{\indexForMTNode, 1})$.
\item If $j = \QVCMessageDepth - 1$, for each $\indexForMTNode \in \Bits^{j}$, parse $\accessVectorAt{\MessageRegister{1}}{\indexForMTNode}$ as $(\accessVectorAt{\MessageRegister{1}}{\indexForMTNode, 0},\accessVectorAt{\MessageRegister{1}}{\indexForMTNode, 1})$.
\end{enumerate}
\item Set $\MessageRegister{1} \coloneq \bigotimes_{\indexForMTNode \in \Bits^{\QVCMessageDepth}}\accessVectorAt{\MessageRegister{1}}{\indexForMTNode}$ and $\AltOpeningRegister{1} \coloneq \bigotimes_{\indexForMTNode \in \Bits^{< \QVCMessageDepth}}\accessVectorAt{\AltOpeningRegister{1}}{\indexForMTNode}$.
\item Output $(\MessageRegister{1}, \AltOpeningRegister{1})$.
\end{enumerate}
\item $\Extractor_{\QSTC}.\AltCheck(\QVCQuerySet, \OpeningRegister{1}, \AltOpeningRegister{1})$:
\begin{enumerate}[nolistsep]
\item Parse $\QVCOpeningRegister$ as $\bigotimes_{\indexForMTNode \in \QSTCPath{\QVCQuerySet}\setminus \QVCQuerySet}\accessVectorAt{\OpeningRegister{1}}{\indexForMTNode}$ and $\AltOpeningRegister{1}$ as $\bigotimes_{\indexForMTNode \in \Bits^{< \QVCMessageDepth}}\accessVectorAt{\AltOpeningRegister{1}}{\indexForMTNode}$.
\item For each $\indexForMTNode \in \QSTCPath{\QVCQuerySet}\setminus\QVCQuerySet$, make the projective measurement as the algorithm $\Extractor_{\CM}.\AltCheck$ on registers $(\accessVectorAt{\OpeningRegister{1}}{\indexForMTNode}, \accessVectorAt{\AltOpeningRegister{1}}{\indexForMTNode})$ to get a result $\ValidityBitI{\indexForMTNode}$ (while the registers $(\accessVectorAt{\OpeningRegister{1}}{\indexForMTNode}, \accessVectorAt{\AltOpeningRegister{1}}{\indexForMTNode})$ still hold the post-measurement state).
\item Set $\ValidityBit \coloneq \bigwedge_{\indexForMTNode \in \QSTCPath{\QVCQuerySet}\setminus\QVCQuerySet}\ValidityBitI{\indexForMTNode}$, $\QVCOpeningRegister \coloneq \bigotimes_{\indexForMTNode \in \QSTCPath{\QVCQuerySet}\setminus \QVCQuerySet}\accessVectorAt{\OpeningRegister{1}}{\indexForMTNode}$, and $\AltOpeningRegister{1} \coloneq \bigotimes_{\indexForMTNode \in \Bits^{< \QVCMessageDepth}}\accessVectorAt{\AltOpeningRegister{1}}{\indexForMTNode}$.
\item Output $(\ValidityBit, \OpeningRegister{1}, \AltOpeningRegister{1})$.
\end{enumerate}
\item $\Extractor_{\QSTC}.\AltCommit(\AltOpeningRegister{1}, \MessageRegister{1})$:
\begin{enumerate}[nolistsep]
\item Parse $\AltOpeningRegister{1}$ as $\bigotimes_{\indexForMTNode \in \Bits^{< \QVCMessageDepth}}\accessVectorAt{\AltOpeningRegister{1}}{\indexForMTNode}$.
\item For each $\indexForMTNode \in \Bits^{\QVCMessageDepth - 1}$, define $\accessVectorAt{\MessageRegister{1}}{\indexForMTNode} \coloneq (\accessVectorAt{\MessageRegister{1}}{\indexForMTNode, 0}, \accessVectorAt{\MessageRegister{1}}{\indexForMTNode, 1})$.
\item For $j = \QVCMessageDepth - 1, \cdots, 1, 0$:
\begin{enumerate}[nolistsep]
\item For each $\indexForMTNode \in \Bits^{j}$ in the reverse of lexicographic order, apply the reverse of $\Extractor_{\CM}.\ExtractMessage$ by forwarding its queries to $\OracleUnitary(\DatabaseRegister{1})$ and $\invert(\DatabaseRegister{1})$ on $(\accessVectorAt{\MessageRegister{1}}{\indexForMTNode},\accessVectorAt{\AltOpeningRegister{1}}{\indexForMTNode})$ to obtain a state on the register $\accessVectorAt{\CommitmentRegister{1}}{\indexForMTNode}$ and the ancilla registers $(\accessVectorAt{\StandardBasisWorkingRegister{1}}{\indexForMTNode}, \accessVectorAt{\HadamardBasisWorkingRegister{1}}{\indexForMTNode})$.
\item If $j > 0$, for each $\indexForMTNode \in \Bits^{j - 1}$, set $\accessVectorAt{\MessageRegister{1}}{\indexForMTNode} \coloneq (\accessVectorAt{\CommitmentRegister{1}}{\indexForMTNode, 0},\accessVectorAt{\CommitmentRegister{1}}{\indexForMTNode, 1})$.
\end{enumerate}
\item Set $\CommitmentRegister{1} \coloneq \accessVectorAt{\CommitmentRegister{1}}{\emptystring}$.
\item Set $\WorkingRegister{1} \coloneq (\accessVectorAt{\StandardBasisWorkingRegister{1}}{\indexForMTNode}, \accessVectorAt{\HadamardBasisWorkingRegister{1}}{\indexForMTNode})_{\indexForMTNode \in \Bits^{<\QVCMessageDepth}}$.
\item Output $\CommitmentRegister{1}$.
\end{enumerate}
\end{itemize}
\end{itemize}

In the security analysis, we sometimes also let $\Extractor_{\QSTC}.\AltCommit$ output the register $\WorkingRegister{1}$, and we sometimes also let $\Extractor_{\QSTC}.\ExtractMessage$ take an additional register $\WorkingRegister{1}$ as an input register instead of initializing an all-zero register.

Note that in our construction, $\Extractor_{\QSTC}.\AltCommit$ just does the inverse of $\Extractor_{\QSTC}.\ExtractMessage$ except that $\Extractor_{\QSTC}.\AltCommit$ does not check whether the ancilla qubits return to the all-zero states.
\end{construction}

\subsection{Security analysis of the vector commitment scheme}

We show the scheme $\QSTC$ is an extractable quantum state vector commitment scheme.

\begin{theorem}[Extractability]
\label{thm:vc-extractability}
The scheme $\QSTC$ in \Cref{construction:quantum-state-vector-commitment} is $(\ExtractionErrorI{1}, \ExtractionErrorI{2}, \ExtractionErrorI{3})$-extractable with the quantum simulator $\OracleUnitary(\DatabaseRegister{1})$, the extractor oracle $\invert(\DatabaseRegister{1})$, and the extractor $\Extractor_{\QSTC}$ in \Cref{construction:vector_commitment_extractor} where
\begin{align*}
\ExtractionErrorI{1}&\coloneq 0\enspace,\\
\ExtractionErrorI{2}(\RandomOracleOutputLength) &\coloneq 2^{-\RandomOracleOutputLength + 7}\enspace,\\
\ExtractionErrorI{3}(\NumberOfQueries, \NumberOfCommitmentsPhases, \QuantumTotalQueryMass_1, \QuantumTotalQueryMass_2, \RandomOracleOutputLength, \QVCMessageLength) &\coloneq 2^{-\RandomOracleOutputLength + 22}\NumberOfCommitmentsPhases^2\QVCMessageLength^2(\NumberOfQueries + \NumberOfQueriesBound\NumberOfCommitmentsPhases)^2 \left(\QVCMessageLength^2+\QuantumTotalQueryMass_1+\QuantumTotalQueryMass_2 + \NumberOfCommitmentsPhases + 1
\right)\enspace.
\end{align*}
\end{theorem}

We first present a lemma that shows for the single opening case, the world $\QVCSimWorldSC$ that does the query using $\QSTC$, and the world $\QVCOfflineExtractWorldSC$ that first extracts the underlying message, implements the query, and then recover the commitment are indistinguishable, as long as there is no collision during the execution of the games. Then \Cref{thm:vc-extractability} follows from the standard hybrid arguments and the fact that a quantum algorithm cannot obtain a collision in the database with high probability.

\begin{lemma}
\label{lemma:indistinguishablility-if-no-collisions-vc}
Let $\QSTC$ be the quantum vector state commitment in \Cref{construction:quantum-state-vector-commitment}, and $\Extractor_{\QSTC}$ be the extractor in \Cref{construction:vector_commitment_extractor} with oracle access to $\OracleUnitary(\DatabaseRegister{1})$ and $\invert(\DatabaseRegister{1})$ that works on an internal database register $\DatabaseRegister{1}$.

For every integer $\RandomOracleOutputLength$, $\NumberOfQueries$, real numbers $\QuantumTotalQueryMass_1, \QuantumTotalQueryMass_2 \in [0, \NumberOfQueries]$ such that $\QuantumTotalQueryMass_1 + \QuantumTotalQueryMass_2 \leq \NumberOfQueries$, $\NumberOfQueries$-query quantum adversary $\Adversary$ such that $\QuantumTotalQueryMassFunc{\Adversary, \OracleUnitary} \leq \QuantumTotalQueryMass_1$ and $\QuantumTotalQueryMassFunc{\Adversary, \invert} \leq \QuantumTotalQueryMass_2$, and unbounded quantum distinguisher $\Distinguisher$,
\begin{align*}
&\abs{\sqrt{\prob{\NoCollisionVariant{\QVCSimWorldSC}(\Adversary, \Distinguisher)}} - \sqrt{\prob{\NoCollisionVariant{\QVCOfflineExtractWorldSC}(\Adversary, \Distinguisher)}}}^2\\
&\leq \ErrorBoundForSingleCommitOffline (\RandomOracleOutputLength, \QVCMessageLength, \NumberOfQueries, \QuantumTotalQueryMass_1, \QuantumTotalQueryMass_2)\enspace,
\end{align*}
where $\RandomOracleOutputLength$ is the output length of the random oracle and $\QVCMessageLength$ is the message length of the scheme, the error \[\ErrorBoundForSingleCommitOffline (\RandomOracleOutputLength, \QVCMessageLength, \NumberOfQueries, \QuantumTotalQueryMass_1, \QuantumTotalQueryMass_2) \coloneq 2^{-\RandomOracleOutputLength + 20}\QVCMessageLength^2(\NumberOfQueries + \NumberOfQueriesBound)^2 \left(\QVCMessageLength^2+\QuantumTotalQueryMass_2\right)\enspace,\] and the games $\NoCollisionVariant{\QVCSimWorldSC}(\Adversary, \Distinguisher)$ and $\NoCollisionVariant{\QVCOfflineExtractWorldSC}(\Adversary, \Distinguisher)$ are defined below:
\begin{itemize}
\item
\begin{itemize}[noitemsep]
\item[]$\NoCollisionVariant{\QVCSimWorldSC}(\Adversary, \Distinguisher)$:
\begin{enumerate}[nolistsep]
\item The game initializes the database register: $\DatabaseRegister{1} \gets \ket{\bot}$.
\item {The adversary generates the query set, the commitment, the opening, and the answer register that might be entangled with the environment register: $(\QVCQuerySetRegister{1}, \CommitmentRegister{1}, \OpeningRegister{1}, \AnswerRegister{1}, \EnvironmentRegister{1}) \gets \Adversary^{\OracleUnitary(\DatabaseRegister{1}), \invert(\DatabaseRegister{1})}$.}
\item {The game implements the query: $(\NoCollisionAlgVariant{\ValidityBit},\ValidityBit, \QVCQuerySetRegister{1}, \CommitmentRegister{1}, \OpeningRegister{1}, \AnswerRegister{1}) \gets \NoCollisionAlgVariant{\QSTC.\Query}^{\OracleUnitary(\DatabaseRegister{1})}(1^\Security, \QVCQuerySetRegister{1}, \CommitmentRegister{1}, \OpeningRegister{1}, \AnswerRegister{1})$.}
\item The game initializes the working register: set $\WorkingRegister{1} \coloneq (\accessVectorAt{\StandardBasisWorkingRegister{1}}{\indexForMTNode}, \accessVectorAt{\HadamardBasisWorkingRegister{1}}{\indexForMTNode})_{\indexForMTNode \in \Bits^{<\QVCMessageDepth}}$ and initialize it as all-zero states.
\item The game generates the output: If $\NoCollisionAlgVariant{\ValidityBit} \land \ValidityBit = 0$, output 0; otherwise, compute $\ValidityBit' \gets \Distinguisher(\QVCQuerySetRegister{1}, \CommitmentRegister{1}, \OpeningRegister{1}, \AnswerRegister{1}, \EnvironmentRegister{1}, \DatabaseRegister{1}, \WorkingRegister{1})$ and output $\ValidityBit'$.
\end{enumerate}
\end{itemize}
\item
\begin{itemize}[noitemsep]
\item[]$\NoCollisionVariant{\QVCOfflineExtractWorldSC}(\Adversary, \Distinguisher)$:
\begin{enumerate}[nolistsep]
\item The game initializes the database register: $\DatabaseRegister{1} \gets \ket{\bot}$.
\item {The adversary generates the query set, the commitment, the opening, and the answer register that might be entangled with the environment register: $(\QVCQuerySetRegister{1}, \CommitmentRegister{1}, \OpeningRegister{1}, \AnswerRegister{1}, \EnvironmentRegister{1}) \gets \Adversary^{\OracleUnitary(\DatabaseRegister{1}), \invert(\DatabaseRegister{1})}$.}
\item {The game uses the extractor to extract the underlying message: $(\NoCollisionAlgVariant{\ValidityBit}, \MessageRegister{1}, \AltOpeningRegister{1}) \gets \NoCollisionAlgVariant{\Extractor_{\QSTC}.\ExtractMessage}^{\OracleUnitary(\DatabaseRegister{1}), \invert(\DatabaseRegister{1})}(\CommitmentRegister{1})$.}
\item The game uses the alternative check to check if the opening is valid: $(\ValidityBit, \QVCQuerySetRegister{1}, \OpeningRegister{1}, \AltOpeningRegister{1}) \gets \Extractor_{\QSTC}.\AltCheck(\QVCQuerySetRegister{1},\OpeningRegister{1}, \AltOpeningRegister{1})$.
\item The game implements the query: $(\QVCQuerySetRegister{1}, \QVCMessageRegister, \AnswerRegister{1}) \gets \QueryUnitary(\QVCQuerySetRegister{1}, \QVCMessageRegister, \AnswerRegister{1})$.
\item The game uses the extractor to recover the commitment: $(\NoCollisionAlgVariant{\ValidityBit}', \CommitmentRegister{1}, \WorkingRegister{1}) \gets \NoCollisionAlgVariant{\Extractor_{\QSTC}.\AltCommit}^{\OracleUnitary(\DatabaseRegister{1}), \invert(\DatabaseRegister{1})}(\AltOpeningRegister{1}, \QVCMessageRegister)$.
\item The game generates the output: If $\NoCollisionAlgVariant{\ValidityBit}\land \NoCollisionAlgVariant{\ValidityBit}' \land \ValidityBit = 0$, output 0; otherwise, compute $\ValidityBit' \gets \Distinguisher(\QVCQuerySetRegister{1}, \CommitmentRegister{1}, \OpeningRegister{1}, \AnswerRegister{1}, \EnvironmentRegister{1}, \DatabaseRegister{1}, \WorkingRegister{1})$ and output $\ValidityBit'$.
\end{enumerate}
\end{itemize}
\end{itemize}
\end{lemma}

The proof of \Cref{lemma:indistinguishablility-if-no-collisions-vc} is deferred to \Cref{subsec:proof-of-indistinguishability-if-no-collisions-vs-single-query}. We first apply \Cref{lemma:indistinguishablility-if-no-collisions-vc} to show \Cref{thm:vc-extractability}.

\begin{proof}[Proof of \Cref{thm:vc-extractability}]
The extractor $\Extractor_\QSTC$ uses the same oracles $\invert(\DatabaseRegister{1})$ and $\OracleUnitary(\DatabaseRegister{1})$ as the extractor $\Extractor_{\CM}$ for the scheme $\CM$ in \Cref{construction:basic_commitment_extractor}. By \Cref{thm:extractability}, we have that
\begin{align*}
\ExtractionErrorI{1}&= 0\enspace,\\
\ExtractionErrorI{2}(\RandomOracleOutputLength) &= 2^{-\RandomOracleOutputLength + 7}\enspace,
\end{align*}
and two $\OracleUnitary$ oracles commute, and two $\invert$ oracles commute no matter which registers they act on.

To bound the term $\abs{\sqrt{\prob{\QVCSimulateWorld(\Adversary, \Distinguisher)}} - \sqrt{\prob{\QVCExtractWorld(\Adversary, \Distinguisher)}}}^2$, we first introduce a new game $\QVCOfflineExtractWorld$, which does the extraction after the adversary produces the openings. The differences between $\QVCOfflineExtractWorld$ and ${\QVCExtractWorld}$ are highlighted in blue.
\begin{itemize}[noitemsep]
\item[]$\QVCOfflineExtractWorld(\Adversary, \Distinguisher)$:
\begin{enumerate}[nolistsep]
\item The game initializes the registers: $\DatabaseRegister{1} \gets \ket{\bot}$, $\QVCQuerySetRegister{1} \gets \ket{\bar{0}}$, $\AnswerRegister{1} \gets \ket{\bar{0}}$, $ \OpeningRegister{1} \gets \ket{\bar{0}}$.
\item The adversary generates a commitment: $(\CommitmentRegister{1}, \EnvironmentRegister{1}) \gets \Adversary^{\OracleUnitary(\DatabaseRegister{1}), \invert(\DatabaseRegister{1})}$.
\item Set $\WorkingRegister{1} \coloneq (\accessVectorAt{\StandardBasisWorkingRegister{1}}{\indexForMTNode}, \accessVectorAt{\HadamardBasisWorkingRegister{1}}{\indexForMTNode})_{\indexForMTNode \in \Bits^{<\QVCMessageDepth}}$ and initialize it as all-zero states.
\item For $i \in [\NumberOfCommitmentsPhases]$:
\begin{enumerate}[noitemsep]
\item\label{step:QVCOffExtOpening} The adversary generates a query location and the corresponding opening, together with an answer register for the superposition query: $(\QVCQuerySetRegister{1}, \OpeningRegister{1}, \AnswerRegister{1}, \EnvironmentRegister{1}) \gets \Adversary^{\OracleUnitary(\DatabaseRegister{1}), \invert(\DatabaseRegister{1})}(\QVCQuerySetRegister{1}, \OpeningRegister{1}, \AnswerRegister{1}, \EnvironmentRegister{1})$.
\item\label{step:QVCOffExtExt} \textcolor{blue!70}{The game uses the extractor to extract the underlying message: $(\MessageRegister{1}, \AltOpeningRegister{1}) \gets \Extractor.\ExtractMessage^{\OracleUnitary(\DatabaseRegister{1}), \invert(\DatabaseRegister{1})}(\CommitmentRegister{1}, \WorkingRegister{1})$.}
\item The game uses the alternative check to check if the opening is valid: $(\ValidityBit_i, \QVCQuerySetRegister{1}, \OpeningRegister{1}, \AltOpeningRegister{1}) \gets \Extractor.\AltCheck(\QVCQuerySetRegister{1},\OpeningRegister{1}, \AltOpeningRegister{1})$.
\item The game implements the query: $(\QVCQuerySetRegister{1}, \QVCMessageRegister, \AnswerRegister{1}) \gets \QueryUnitary(\QVCQuerySetRegister{1}, \QVCMessageRegister, \AnswerRegister{1})$.
\item \textcolor{blue!70}{The game uses alternative commit to get the commitment: $(\QVCCommitmentRegister, \WorkingRegister{1}) \gets \Extractor.\AltCommit^{\OracleUnitary(\DatabaseRegister{1}), \invert(\DatabaseRegister{1})}(\AltOpeningRegister{1}, \QVCMessageRegister)$.}
\end{enumerate}
\item The game generates the output: If $\left(\land_{i \in [\NumberOfCommitmentsPhases]}\ValidityBit_i\right) = 0$, output 0; otherwise, compute $\ValidityBit' \gets \Distinguisher(\QVCQuerySetRegister{1}, \CommitmentRegister{1}, \OpeningRegister{1}, \AnswerRegister{1}, \EnvironmentRegister{1}, \DatabaseRegister{1})$ and output $\ValidityBit'$.
\end{enumerate}
\end{itemize}

As the unitary part of $\Extractor.\AltCommit$ is just the reverse of the unitary part of $\Extractor.\ExtractMessage$, the only difference between $\QVCOfflineExtractWorld$ and ${\QVCExtractWorld}$ is whether we do \Cref{step:QVCOffExtOpening} first or \Cref{step:QVCOffExtExt} first. As a result, for every $\NumberOfQueries$-query quantum adversary $\Adversary$ and unbounded quantum distinguisher $\Distinguisher$,
\begin{align}
\label{eqn:diff-Ext-OfflineExt-for-QVC}
&\abs{\sqrt{\prob{\QVCExtractWorld(\Adversary, \Distinguisher)}} - \sqrt{\prob{\QVCOfflineExtractWorld(\Adversary, \Distinguisher)}}}\nonumber\\
&\leq 4 \NumberOfQueries \QVCMessageLength \cdot \matnorm{\Commutator{\OracleUnitary}{\invert}} \leq 32\sqrt{2}\NumberOfQueries \QVCMessageLength \cdot 2^{-\RandomOracleOutputLength/2}\enspace.
\end{align}

It remains to bound $\abs{\sqrt{\prob{{\QVCOfflineExtractWorld(\Adversary, \Distinguisher)}}} - \sqrt{\prob{{\QVCSimulateWorld(\Adversary, \Distinguisher)}}}}$. We compare $\QVCSimulateWorld$ and $\QVCOfflineExtractWorld$ using the following hybrids. For $j\in\{0,\ldots,\NumberOfCommitmentsPhases\}$, let $\HybridI{j}$ use $\QSTC.\Query$ in the first $j$ phases and the offline extraction procedure in the remaining phases. Thus, $\HybridI{\NumberOfCommitmentsPhases}$ is $\QVCSimulateWorld$ and $\HybridI{0}$ is $\QVCOfflineExtractWorld$.

Fix $j\in[\NumberOfCommitmentsPhases]$. The games $\HybridI{j}$ and $\HybridI{j-1}$ differ only in phase $j$. Their common prefix, up to the submission of the opening for the $j$-th phase, is an ordinary $\QVCSimulateWorld$ prefix. We regard this prefix as the adversary of \Cref{lemma:indistinguishablility-if-no-collisions-vc}. Since the preceding $j-1$ phases make at most $4\QVCMessageLength(j-1)$ queries to $\OracleUnitary$ and none to $\invert$, this adversary has at most $\NumberOfQueries+4\QVCMessageLength(j-1)$ queries, with masses at most $\QuantumTotalQueryMass_1+4\QVCMessageLength(j-1)$ and $\QuantumTotalQueryMass_2$, respectively, where $\QuantumTotalQueryMass_1$ and $\QuantumTotalQueryMass_2$ refer to the adversary's query masses in $\QVCSimulateWorld$.

Let $\HybridIJ{j}{0}$ and $\HybridIJ{j}{1}$ be obtained from $\HybridI{j}$ and $\HybridI{j-1}$, respectively, by inserting no-collision checks only during the $j$-th phase. By \Cref{lemma:collision-free}, 
\begin{align*}
&\abs{\sqrt{\prob{\HybridI{j}(\Adversary,\Distinguisher)}}-\sqrt{\prob{\HybridIJ{j}{0}(\Adversary,\Distinguisher)}}}\\
&\leq 2\sqrt{6}\left(\NumberOfQueries+4\QVCMessageLength(j-1)+\NumberOfQueriesBound\right)\sqrt{\QuantumTotalQueryMass_1+4\QVCMessageLength j}\cdot 2^{-\RandomOracleOutputLength/2}\enspace,
\end{align*}
and 
\begin{align*}
&\abs{\sqrt{\prob{\HybridIJ{j}{1}(\Adversary,\Distinguisher)}}-\sqrt{\prob{\HybridI{j-1}(\Adversary,\Distinguisher)}}}\\
&\leq2\sqrt{6}\left(\NumberOfQueries+4\QVCMessageLength(j-1)+\NumberOfQueriesBound\right)\sqrt{\QuantumTotalQueryMass_1+4\QVCMessageLength j}\cdot 2^{-\RandomOracleOutputLength/2}\enspace.
\end{align*}

Moreover, by \Cref{lemma:indistinguishablility-if-no-collisions-vc},
\begin{align*}
&\abs{\sqrt{\prob{\HybridIJ{j}{0}(\Adversary,\Distinguisher)}}-\sqrt{\prob{\HybridIJ{j}{1}(\Adversary,\Distinguisher)}}}^2\\
&\leq\ErrorBoundForSingleCommitOffline\left(\RandomOracleOutputLength,\QVCMessageLength,\NumberOfQueries+4\QVCMessageLength(j-1),\QuantumTotalQueryMass_1+4\QVCMessageLength(j-1),\QuantumTotalQueryMass_2\right)\enspace.
\end{align*}

Therefore, by the triangle inequality, we obtain
\begin{align*}
&\abs{\sqrt{\prob{\HybridI{j}(\Adversary,\Distinguisher)}}-\sqrt{\prob{\HybridI{j-1}(\Adversary,\Distinguisher)}}}^2\\
&\leq \left(4\sqrt{6}\left(\NumberOfQueries+4\QVCMessageLength(j-1)+\NumberOfQueriesBound\right)\sqrt{\QuantumTotalQueryMass_1+4\QVCMessageLength j}\cdot 2^{-\RandomOracleOutputLength/2}\right.\\
&\left.\quad+\left(\ErrorBoundForSingleCommitOffline\left(\RandomOracleOutputLength,\QVCMessageLength,\NumberOfQueries+4\QVCMessageLength(j-1),\QuantumTotalQueryMass_1+4\QVCMessageLength(j-1),\QuantumTotalQueryMass_2\right)\right)^{1/2}\right)^2\\
&\leq \left(4\sqrt{6}\left(\NumberOfQueries+4\QVCMessageLength(\NumberOfCommitmentsPhases -1)+\NumberOfQueriesBound\right)\sqrt{\QuantumTotalQueryMass_1+4\QVCMessageLength \NumberOfCommitmentsPhases}\cdot 2^{-\RandomOracleOutputLength/2}\right.\\
&\left.\quad+\left(\ErrorBoundForSingleCommitOffline\left(\RandomOracleOutputLength,\QVCMessageLength,\NumberOfQueries+4\QVCMessageLength(\NumberOfCommitmentsPhases-1),\QuantumTotalQueryMass_1+4\QVCMessageLength(\NumberOfCommitmentsPhases-1),\QuantumTotalQueryMass_2\right)\right)^{1/2}\right)^2\\
&\leq 2\left(96\left(\NumberOfQueries+4\QVCMessageLength(\NumberOfCommitmentsPhases -1)+\NumberOfQueriesBound\right)^2(\QuantumTotalQueryMass_1+4\QVCMessageLength \NumberOfCommitmentsPhases)\cdot 2^{-\RandomOracleOutputLength} +   2^{-\RandomOracleOutputLength + 20}\QVCMessageLength^2(\NumberOfQueries + \NumberOfQueriesBound\NumberOfCommitmentsPhases)^2 \left(\QVCMessageLength^2+\QuantumTotalQueryMass_2\right)\right)\\
&\leq 2^{-\RandomOracleOutputLength + 21}\QVCMessageLength^2(\NumberOfQueries + \NumberOfQueriesBound\NumberOfCommitmentsPhases)^2 \left(\QVCMessageLength^2+\QuantumTotalQueryMass_1+\QuantumTotalQueryMass_2 + \NumberOfCommitmentsPhases\right)\enspace,
\end{align*}
where the third inequality follows from the Cauchy--Schwarz inequality.

Summing over $j\in[\NumberOfCommitmentsPhases]$, we obtain
\begin{align*}
&\abs{\sqrt{\prob{\QVCSimulateWorld(\Adversary,\Distinguisher)}}-\sqrt{\prob{\QVCOfflineExtractWorld(\Adversary,\Distinguisher)}}}^2\\
&\leq 2^{-\RandomOracleOutputLength + 21}\NumberOfCommitmentsPhases^2\QVCMessageLength^2(\NumberOfQueries + \NumberOfQueriesBound\NumberOfCommitmentsPhases)^2 \left(\QVCMessageLength^2+\QuantumTotalQueryMass_1+\QuantumTotalQueryMass_2 + \NumberOfCommitmentsPhases
\right)\enspace.
\end{align*}

Therefore,
\begin{align*}
&\abs{\sqrt{\prob{\QVCSimulateWorld(\Adversary,\Distinguisher)}}-\sqrt{\prob{\QVCExtractWorld(\Adversary,\Distinguisher)}}}^2\\
&\leq 2\abs{\sqrt{\prob{\QVCExtractWorld(\Adversary, \Distinguisher)}} - \sqrt{\prob{\QVCOfflineExtractWorld(\Adversary, \Distinguisher)}}}^2\\
&\quad+ 2\abs{\sqrt{\prob{\QVCSimulateWorld(\Adversary,\Distinguisher)}}-\sqrt{\prob{\QVCOfflineExtractWorld(\Adversary,\Distinguisher)}}}^2\\
&\leq 2^{-\RandomOracleOutputLength + 22}\NumberOfCommitmentsPhases^2\QVCMessageLength^2(\NumberOfQueries + \NumberOfQueriesBound\NumberOfCommitmentsPhases)^2 \left(\QVCMessageLength^2+\QuantumTotalQueryMass_1+\QuantumTotalQueryMass_2 + \NumberOfCommitmentsPhases + 1\right)\\
&\leq \ExtractionErrorI{3}(\NumberOfQueries, \NumberOfCommitmentsPhases, \QuantumTotalQueryMass_1, \QuantumTotalQueryMass_2, \RandomOracleOutputLength, \QVCMessageLength)\enspace.
\end{align*}

\end{proof}
\subsection{Proof of \Cref{lemma:indistinguishablility-if-no-collisions-vc}}
\label{subsec:proof-of-indistinguishability-if-no-collisions-vs-single-query}

\Cref{lemma:indistinguishablility-if-no-collisions-vc} mainly follows from a standard hybrid argument over \Cref{lemma:indistinguishablility-if-no-collisions}.

\begin{proof}[Proof of \Cref{lemma:indistinguishablility-if-no-collisions-vc}]
As in the proof of \Cref{lemma:indistinguishablility-if-no-collisions}, by the triangle inequality, to prove \Cref{lemma:indistinguishablility-if-no-collisions-vc}, it is sufficient to upper bound the $\ell_2$ norm of the difference between the corresponding subnormalized states at the point immediately before applying $\Distinguisher$, denoted by $\PureStateSampleOneI{\NoCollisionVariant{\QVCSimWorldSC}}$ and $\PureStateSampleOneI{\NoCollisionVariant{\QVCOfflineExtractWorldSC}}$, respectively. Specifically, it suffices to show that \[\vecnorm{\PureStateSampleOneI{\NoCollisionVariant{\QVCSimWorldSC}} - \PureStateSampleOneI{\NoCollisionVariant{\QVCOfflineExtractWorldSC}}}^2 \leq \ErrorBoundForSingleCommitOffline (\RandomOracleOutputLength, \QVCMessageLength, \NumberOfQueries, \QuantumTotalQueryMass_1, \QuantumTotalQueryMass_2)\enspace.\]

We observe that, in both games $\NoCollisionVariant{\QVCSimWorldSC}(\Adversary, \Distinguisher)$ and $\NoCollisionVariant{\QVCOfflineExtractWorldSC}(\Adversary, \Distinguisher)$, all operations performed prior to invoking the distinguisher $\Distinguisher$ are controlled by the query set $\QVCQuerySet$ stored in the register $\QVCQuerySetRegister{1}$. Thus it suffices to establish the above inequality for every fixed query set $\QVCQuerySet \subseteq \{1, 2, \ldots, \QVCMessageLength\}$, rather than for a query set chosen by the adversary in superposition, as 
\begin{align*}
	&\vecnorm{\PureStateSampleOneI{\NoCollisionVariant{\QVCSimWorldSC}} - \PureStateSampleOneI{\NoCollisionVariant{\QVCOfflineExtractWorldSC}}}^2\\
	&= \sum_{\QVCQuerySet}\vecnorm{\PureStateSampleOneI{\NoCollisionVariant{\QVCSimWorldSC} \text{ when $\QVCQuerySet$ is chosen}} - \PureStateSampleOneI{\NoCollisionVariant{\QVCOfflineExtractWorldSC} \text{ when $\QVCQuerySet$ is chosen}}}^2 \prob{\QVCQuerySet \text{ is chosen}}\enspace.
\end{align*}

To this end, we fix a set $\QVCQuerySet \subseteq \{1, 2, \cdots, \QVCMessageLength\}$, and use the standard hybrid argument together with \Cref{lemma:indistinguishablility-if-no-collisions} to replace each call to the algorithm $\Extractor_{\CM}.\ExtractMessage$ with $\CM.\Check$, and each call to the inverse of $\Extractor_{\CM}.\ExtractMessage$ with $\CM.\Commit$. We order the elements of $\QSTCPath{\QVCQuerySet}\setminus \QVCQuerySet$ first by increasing length and then lexicographically, and write $\QSTCPath{\QVCQuerySet}\setminus \QVCQuerySet = \{\indexForMTNode_1, \indexForMTNode_2, \ldots, \indexForMTNode_{\QVCPathSize}\}$. Since $\QSTCPath{\QVCQuerySet}\setminus \QVCQuerySet \subseteq \Bits^{<\QVCMessageDepth}$, the number of elements $\QVCPathSize < 2^{\QVCMessageDepth} = \QVCMessageLength$.

We define the following intermediate hybrid games:
\begin{itemize}
\item [] $\HybridI{0}(\Adversary, \Distinguisher)$: it does the same as $\NoCollisionVariant{\QVCOfflineExtractWorldSC}$ except that instead of letting the adversary produce $\QVCQuerySetRegister{1}$, we fix the query set to be $\QVCQuerySet$.
\item [] $\HybridIJ{1}{i}(\Adversary, \Distinguisher)$:
\begin{enumerate}[nolistsep]
\item The game initializes the database register: $\DatabaseRegister{1} \gets \ket{\bot}$.
\item {The adversary generates the query set, the commitment, the opening, and the answer register that might be entangled with the environment register: $(\QVCQuerySetRegister{1}, \CommitmentRegister{1}, \OpeningRegister{1}, \AnswerRegister{1}, \EnvironmentRegister{1}) \gets \Adversary^{\OracleUnitary(\DatabaseRegister{1}), \invert(\DatabaseRegister{1})}$.}
\item Parse $\OpeningRegister{1}$ as $\bigotimes_{\indexForMTNode \in \QSTCPath{\QVCQuerySet}\setminus \QVCQuerySet}\accessVectorAt{\OpeningRegister{1}}{\indexForMTNode}$ and $\AnswerRegister{1}$ as $\bigotimes_{\indexForMTNode \in \QVCQuerySet}\accessVectorAt{\AnswerRegister{1}}{\indexForMTNode}$.
\item The game uses the extractor to extract the underlying message:
\begin{enumerate}[nolistsep]
\item Set $\accessVectorAt{\CommitmentRegister{1}}{\emptystring} \coloneq \CommitmentRegister{1}$.
\item For $\indexForMTNode \in \Bits^{< \QVCMessageDepth}$:
\begin{enumerate}[nolistsep]
\item If $\indexForMTNode \in \{\indexForMTNode_1, \indexForMTNode_2, \ldots, \indexForMTNode_{i}\}$,
\begin{enumerate}[nolistsep]
\item Apply $\NoCollisionAlgVariant{\CM.\Check}$ to $(\accessVectorAt{\CommitmentRegister{1}}{\indexForMTNode}, \accessVectorAt{\OpeningRegister{1}}{\indexForMTNode})$, where the oracle access is implemented by $\OracleUnitary(\DatabaseRegister{1})$, to obtain the committed quantum message in $\accessVectorAt{\MessageRegister{1}}{\indexForMTNode}$ and two validity bits $\accessVectorAt{\ValidityBit}{\indexForMTNode}$ and $\accessVectorAt{\NoCollisionAlgVariant{\ValidityBit}}{\indexForMTNode}$.
\item Output 0 if $\accessVectorAt{\ValidityBit}{\indexForMTNode} \land \accessVectorAt{\NoCollisionAlgVariant{\ValidityBit}}{\indexForMTNode} = 0$.
\item Prepare a state $\ket{\NullStateInExp_{\AltCheck}}$ on the registers $(\accessVectorAt{\OpeningRegister{1}}{\indexForMTNode}, \accessVectorAt{\AltOpeningRegister{1}}{\indexForMTNode})$.
\end{enumerate}
\item If $\indexForMTNode \notin \{\indexForMTNode_1, \indexForMTNode_2, \ldots, \indexForMTNode_{i}\}$,
\begin{enumerate}[nolistsep]
\item Run $\NoCollisionAlgVariant{\Extractor_{\CM}.\ExtractMessage}(\accessVectorAt{\CommitmentRegister{1}}{\indexForMTNode})$ by forwarding its queries to $\OracleUnitary(\DatabaseRegister{1})$ and $\invert(\DatabaseRegister{1})$ to obtain $(\accessVectorAt{\MessageRegister{1}}{\indexForMTNode}, \accessVectorAt{\AltOpeningRegister{1}}{\indexForMTNode})$, and a validity bit $\accessVectorAt{\NoCollisionAlgVariant{\ValidityBit}}{\indexForMTNode}$, and output 0 if the bit $\accessVectorAt{\NoCollisionAlgVariant{\ValidityBit}}{\indexForMTNode} = 0$.
\item If $\indexForMTNode \in \QSTCPath{\QVCQuerySet}\setminus \QVCQuerySet$, run $\Extractor_{\CM}.\AltCheck(\accessVectorAt{\OpeningRegister{1}}{\indexForMTNode}, \accessVectorAt{\AltOpeningRegister{1}}{\indexForMTNode})$ to get a validity bit $\accessVectorAt{\ValidityBit}{\indexForMTNode}$ while maintaining the post-measurement state on registers $(\accessVectorAt{\OpeningRegister{1}}{\indexForMTNode}, \accessVectorAt{\AltOpeningRegister{1}}{\indexForMTNode})$, and output 0 if the bit $\accessVectorAt{\ValidityBit}{\indexForMTNode} = 0$.
\end{enumerate}
\item If $\abs{\indexForMTNode} \neq \QVCMessageDepth - 1$, parse $\accessVectorAt{\MessageRegister{1}}{\indexForMTNode}$ as $(\accessVectorAt{\CommitmentRegister{1}}{\indexForMTNode, 0},\accessVectorAt{\CommitmentRegister{1}}{\indexForMTNode, 1})$.
\item If $\abs{\indexForMTNode} = \QVCMessageDepth - 1$, parse $\accessVectorAt{\MessageRegister{1}}{\indexForMTNode}$ as $(\accessVectorAt{\MessageRegister{1}}{\indexForMTNode, 0},\accessVectorAt{\MessageRegister{1}}{\indexForMTNode, 1})$.
\end{enumerate}
\item Set $\MessageRegister{1} \coloneq \bigotimes_{\indexForMTNode \in \Bits^{\QVCMessageDepth}}\accessVectorAt{\MessageRegister{1}}{\indexForMTNode}$.
\end{enumerate}
\item The game implements the query: for $\indexForMTNode \in \QVCQuerySet$, $(\accessVectorAt{\QVCMessageRegister}{\indexForMTNode}, \accessVectorAt{\AnswerRegister{1}}{\indexForMTNode}) \gets \SWAP(\accessVectorAt{\QVCMessageRegister}{\indexForMTNode}, \accessVectorAt{\AnswerRegister{1}}{\indexForMTNode})$.
\item The game uses the extractor to recover the commitment: $(\NoCollisionAlgVariant{\ValidityBit}', \CommitmentRegister{1}, \WorkingRegister{1}) \gets \NoCollisionAlgVariant{\Extractor_{\QSTC}.\AltCommit}^{\OracleUnitary(\DatabaseRegister{1}), \invert(\DatabaseRegister{1})}(\AltOpeningRegister{1}, \QVCMessageRegister)$.
\item Output 0 if $\NoCollisionAlgVariant{\ValidityBit}' = 0$.
\item Set $\OpeningRegister{1} \coloneq \bigotimes_{\indexForMTNode \in \QSTCPath{\QVCQuerySet}\setminus\QVCQuerySet}\accessVectorAt{\OpeningRegister{1}}{\indexForMTNode}$.
\item The game generates the output: compute $\ValidityBit' \gets \Distinguisher(\QVCQuerySetRegister{1}, \CommitmentRegister{1}, \OpeningRegister{1}, \AnswerRegister{1}, \EnvironmentRegister{1}, \DatabaseRegister{1}, \WorkingRegister{1})$ and output $\ValidityBit'$.
\end{enumerate}
\item [] $\HybridIJ{2}{i}(\Adversary, \Distinguisher)$:
\begin{enumerate}[nolistsep]
\item The game initializes the database register: $\DatabaseRegister{1} \gets \ket{\bot}$.
\item {The adversary generates the query set, the commitment, the opening, and the answer register that might be entangled with the environment register: $(\QVCQuerySetRegister{1}, \CommitmentRegister{1}, \OpeningRegister{1}, \AnswerRegister{1}, \EnvironmentRegister{1}) \gets \Adversary^{\OracleUnitary(\DatabaseRegister{1}), \invert(\DatabaseRegister{1})}$.}
\item Parse $\OpeningRegister{1}$ as $\bigotimes_{\indexForMTNode \in \QSTCPath{\QVCQuerySet}\setminus \QVCQuerySet}\accessVectorAt{\OpeningRegister{1}}{\indexForMTNode}$ and $\AnswerRegister{1}$ as $\bigotimes_{\indexForMTNode \in \QVCQuerySet}\accessVectorAt{\AnswerRegister{1}}{\indexForMTNode}$.
\item The game uses the extractor to extract the underlying message:
\begin{enumerate}[nolistsep]
\item Set $\accessVectorAt{\CommitmentRegister{1}}{\emptystring} \coloneq \CommitmentRegister{1}$.
\item For $\indexForMTNode \in \Bits^{< \QVCMessageDepth}$:
\begin{enumerate}[nolistsep]
\item If $\indexForMTNode \in \QSTCPath{\QVCQuerySet}\setminus \QVCQuerySet$,
\begin{enumerate}[nolistsep]
\item Apply $\NoCollisionAlgVariant{\CM.\Check}$ to $(\accessVectorAt{\CommitmentRegister{1}}{\indexForMTNode}, \accessVectorAt{\OpeningRegister{1}}{\indexForMTNode})$, where the oracle access is implemented by $\OracleUnitary(\DatabaseRegister{1})$, to obtain the committed quantum message in $\accessVectorAt{\MessageRegister{1}}{\indexForMTNode}$ and two validity bits $\accessVectorAt{\ValidityBit}{\indexForMTNode}$ and $\accessVectorAt{\NoCollisionAlgVariant{\ValidityBit}}{\indexForMTNode}$.
\item Output 0 if $\accessVectorAt{\ValidityBit}{\indexForMTNode} \land \accessVectorAt{\NoCollisionAlgVariant{\ValidityBit}}{\indexForMTNode} = 0$.
\end{enumerate}
\item If $\indexForMTNode \notin \QSTCPath{\QVCQuerySet}\setminus \QVCQuerySet$,
\begin{enumerate}[nolistsep]
\item Run $\NoCollisionAlgVariant{\Extractor_{\CM}.\ExtractMessage}(\accessVectorAt{\CommitmentRegister{1}}{\indexForMTNode})$ by forwarding its queries to $\OracleUnitary(\DatabaseRegister{1})$ and $\invert(\DatabaseRegister{1})$ to obtain $(\accessVectorAt{\MessageRegister{1}}{\indexForMTNode}, \accessVectorAt{\AltOpeningRegister{1}}{\indexForMTNode})$, and a validity bit $\accessVectorAt{\NoCollisionAlgVariant{\ValidityBit}}{\indexForMTNode}$.
\item Output 0 if $\accessVectorAt{\NoCollisionAlgVariant{\ValidityBit}}{\indexForMTNode} = 0$.
\end{enumerate}
\item If $\abs{\indexForMTNode} \neq \QVCMessageDepth - 1$, parse $\accessVectorAt{\MessageRegister{1}}{\indexForMTNode}$ as $(\accessVectorAt{\CommitmentRegister{1}}{\indexForMTNode, 0},\accessVectorAt{\CommitmentRegister{1}}{\indexForMTNode, 1})$.
\item If $\abs{\indexForMTNode} = \QVCMessageDepth - 1$, parse $\accessVectorAt{\MessageRegister{1}}{\indexForMTNode}$ as $(\accessVectorAt{\MessageRegister{1}}{\indexForMTNode, 0},\accessVectorAt{\MessageRegister{1}}{\indexForMTNode, 1})$.
\end{enumerate}
\end{enumerate}
\item The game implements the query: for $\indexForMTNode \in \QVCQuerySet$, $(\accessVectorAt{\QVCMessageRegister}{\indexForMTNode}, \accessVectorAt{\AnswerRegister{1}}{\indexForMTNode}) \gets \SWAP(\accessVectorAt{\QVCMessageRegister}{\indexForMTNode}, \accessVectorAt{\AnswerRegister{1}}{\indexForMTNode})$.
\item \textcolor{blue!70}{The game uses the extractor and $\CM.\Commit$ to obtain the commitment:
\begin{enumerate}[nolistsep]
\item For each $\indexForMTNode \in \Bits^{\QVCMessageDepth - 1}$, define $\accessVectorAt{\MessageRegister{1}}{\indexForMTNode} \coloneq (\accessVectorAt{\MessageRegister{1}}{\indexForMTNode, 0}, \accessVectorAt{\MessageRegister{1}}{\indexForMTNode, 1})$.
\item For each $\indexForMTNode \in \Bits^{<\QVCMessageDepth}$ in the reverse order:
\begin{enumerate}[nolistsep]
\item If $\indexForMTNode \in \{\indexForMTNode_1, \indexForMTNode_2, \ldots, \indexForMTNode_{i}\}$,
\begin{enumerate}[nolistsep]
\item Apply $\NoCollisionAlgVariant{\CM.\Commit}$ to $\accessVectorAt{\MessageRegister{1}}{\indexForMTNode}$, where the oracle access is implemented by $\OracleUnitary(\DatabaseRegister{1})$, to obtain the commitment and the opening $(\accessVectorAt{\CommitmentRegister{1}}{\indexForMTNode}, \accessVectorAt{\OpeningRegister{1}}{\indexForMTNode})$ and a validity bit $\accessVectorAt{\NoCollisionAlgVariant{\ValidityBit}}{\indexForMTNode}$.
\item Output 0 if $\accessVectorAt{\NoCollisionAlgVariant{\ValidityBit}}{\indexForMTNode} = 0$.
\item Initialize the register $(\accessVectorAt{\StandardBasisWorkingRegister{1}}{\indexForMTNode}, \accessVectorAt{\HadamardBasisWorkingRegister{1}}{\indexForMTNode})$ as all-zero states.
\end{enumerate}
\item If $\indexForMTNode \notin \{\indexForMTNode_1, \indexForMTNode_2, \ldots, \indexForMTNode_{i}\}$,
\begin{enumerate}[nolistsep]
\item If $\indexForMTNode \in \{\indexForMTNode_{i + 1}, \indexForMTNode_{i + 2}, \ldots, \indexForMTNode_{\QVCPathSize}\}$, prepare a state $\ket{\NullStateInExp_{\AltCheck}}$ on the registers $(\accessVectorAt{\OpeningRegister{1}}{\indexForMTNode}, \accessVectorAt{\AltOpeningRegister{1}}{\indexForMTNode})$.
\item Denote the reverse of $\Extractor_{\CM}.\ExtractMessage$ as an algorithm $\Extractor_{\CM}.\AltCommit$.
\item Run $\NoCollisionAlgVariant{\Extractor_{\CM}.\AltCommit}(\accessVectorAt{\AltOpeningRegister{1}}{\indexForMTNode}, \accessVectorAt{\MessageRegister{1}}{\indexForMTNode})$ by forwarding its queries to $\OracleUnitary(\DatabaseRegister{1})$ and $\invert(\DatabaseRegister{1})$ to obtain $\accessVectorAt{\CommitmentRegister{1}}{\indexForMTNode}$, the ancilla registers $(\accessVectorAt{\StandardBasisWorkingRegister{1}}{\indexForMTNode}, \accessVectorAt{\HadamardBasisWorkingRegister{1}}{\indexForMTNode})$, and a validity bit $\accessVectorAt{\NoCollisionAlgVariant{\ValidityBit}}{\indexForMTNode}$.
\item Output 0 if the bit $\accessVectorAt{\NoCollisionAlgVariant{\ValidityBit}}{\indexForMTNode} = 0$.
\end{enumerate}
\item If $\indexForMTNode \neq \emptystring$ and the last bit of $\indexForMTNode$ is 0, set $\accessVectorAt{\MessageRegister{1}}{\parent{\indexForMTNode}}$ as $(\accessVectorAt{\CommitmentRegister{1}}{\indexForMTNode},\accessVectorAt{\CommitmentRegister{1}}{\sibling{\indexForMTNode}})$.
\end{enumerate}
\item Set $\CommitmentRegister{1} \coloneq \accessVectorAt{\CommitmentRegister{1}}{\emptystring}$.
\end{enumerate}}
\item Set $\OpeningRegister{1} \coloneq \bigotimes_{\indexForMTNode \in \QSTCPath{\QVCQuerySet}\setminus\QVCQuerySet}\accessVectorAt{\OpeningRegister{1}}{\indexForMTNode}$ and $\WorkingRegister{1} \coloneq (\accessVectorAt{\StandardBasisWorkingRegister{1}}{\indexForMTNode}, \accessVectorAt{\HadamardBasisWorkingRegister{1}}{\indexForMTNode})_{\indexForMTNode \in \Bits^{<\QVCMessageDepth}}$.
\item The game generates the output: compute $\ValidityBit' \gets \Distinguisher(\QVCQuerySetRegister{1}, \CommitmentRegister{1}, \OpeningRegister{1}, \AnswerRegister{1}, \EnvironmentRegister{1}, \DatabaseRegister{1}, \WorkingRegister{1})$ and output $\ValidityBit'$.
\end{enumerate}
\item [] $\HybridIJ{3}{i}(\Adversary, \Distinguisher)$:
\begin{enumerate}[nolistsep]
\item The game initializes the database register: $\DatabaseRegister{1} \gets \ket{\bot}$.
\item {The adversary generates the query set, the commitment, the opening, and the answer register that might be entangled with the environment register: $(\QVCQuerySetRegister{1}, \CommitmentRegister{1}, \OpeningRegister{1}, \AnswerRegister{1}, \EnvironmentRegister{1}) \gets \Adversary^{\OracleUnitary(\DatabaseRegister{1}), \invert(\DatabaseRegister{1})}$.}
\item Parse $\OpeningRegister{1}$ as $\bigotimes_{\indexForMTNode \in \QSTCPath{\QVCQuerySet}\setminus \QVCQuerySet}\accessVectorAt{\OpeningRegister{1}}{\indexForMTNode}$ and $\AnswerRegister{1}$ as $\bigotimes_{\indexForMTNode \in \QVCQuerySet}\accessVectorAt{\AnswerRegister{1}}{\indexForMTNode}$.
\item \textcolor{blue!70}{Initialize a counter $\Counter \gets 2^{\QVCMessageDepth} - 1 - i$.}
\item\label{step:check-collision1} \textcolor{blue!70}{Check whether the database register $\DatabaseRegister{1}$ has a collision by measuring it with $\{\ProjectNoCollision, \id{\DatabaseRegister{0}} - \ProjectNoCollision\}$. Output 0 if there is a collision.}
\item The game uses the extractor to extract the underlying message:
\begin{enumerate}[nolistsep]
\item Set $\accessVectorAt{\CommitmentRegister{1}}{\emptystring} \coloneq \CommitmentRegister{1}$.
\item For $\indexForMTNode \in \Bits^{< \QVCMessageDepth}$:
\begin{enumerate}[nolistsep]
\item If $\indexForMTNode \in \QSTCPath{\QVCQuerySet}\setminus \QVCQuerySet$,
\begin{enumerate}[nolistsep]
\item Apply $\NoCollisionAlgVariant{\CM.\Check}$ to $(\accessVectorAt{\CommitmentRegister{1}}{\indexForMTNode}, \accessVectorAt{\OpeningRegister{1}}{\indexForMTNode})$, where the oracle access is implemented by $\OracleUnitary(\DatabaseRegister{1})$, to obtain the committed quantum message in $\accessVectorAt{\MessageRegister{1}}{\indexForMTNode}$ and two validity bits $\accessVectorAt{\ValidityBit}{\indexForMTNode}$ and $\accessVectorAt{\NoCollisionAlgVariant{\ValidityBit}}{\indexForMTNode}$.
\item Output 0 if $\accessVectorAt{\ValidityBit}{\indexForMTNode} \land \accessVectorAt{\NoCollisionAlgVariant{\ValidityBit}}{\indexForMTNode} = 0$.
\end{enumerate}
\item If $\indexForMTNode \notin \QSTCPath{\QVCQuerySet}\setminus \QVCQuerySet$ and \textcolor{blue!70}{$\Counter >0$},
\begin{enumerate}[nolistsep]
\item Run $\NoCollisionAlgVariant{\Extractor_{\CM}.\ExtractMessage}(\accessVectorAt{\CommitmentRegister{1}}{\indexForMTNode})$ by forwarding its queries to $\OracleUnitary(\DatabaseRegister{1})$ and $\invert(\DatabaseRegister{1})$ to obtain $(\accessVectorAt{\MessageRegister{1}}{\indexForMTNode}, \accessVectorAt{\AltOpeningRegister{1}}{\indexForMTNode})$, and a validity bit $\accessVectorAt{\NoCollisionAlgVariant{\ValidityBit}}{\indexForMTNode}$.
\item Output 0 if $\accessVectorAt{\NoCollisionAlgVariant{\ValidityBit}}{\indexForMTNode} = 0$.
\end{enumerate}
\item \textcolor{blue!70}{If $\indexForMTNode \in \QSTCPath{\QVCQuerySet} \setminus \QVCQuerySet$ or $\Counter >0$,
\begin{enumerate}[nolistsep]
\item If $\abs{\indexForMTNode} < \QVCMessageDepth - 1$, parse $\accessVectorAt{\MessageRegister{1}}{\indexForMTNode}$ as $(\accessVectorAt{\CommitmentRegister{1}}{\indexForMTNode, 0},\accessVectorAt{\CommitmentRegister{1}}{\indexForMTNode, 1})$.
\item If $\abs{\indexForMTNode} = \QVCMessageDepth - 1$, parse $\accessVectorAt{\MessageRegister{1}}{\indexForMTNode}$ as $(\accessVectorAt{\MessageRegister{1}}{\indexForMTNode, 0},\accessVectorAt{\MessageRegister{1}}{\indexForMTNode, 1})$.
\end{enumerate}}
\item \textcolor{blue!70}{Decrease the counter by 1: $\Counter \gets \Counter - 1$.}
\end{enumerate}
\end{enumerate}
\item The game implements the query: for $\indexForMTNode \in \QVCQuerySet$, $(\accessVectorAt{\QVCMessageRegister}{\indexForMTNode}, \accessVectorAt{\AnswerRegister{1}}{\indexForMTNode}) \gets \SWAP(\accessVectorAt{\QVCMessageRegister}{\indexForMTNode}, \accessVectorAt{\AnswerRegister{1}}{\indexForMTNode})$.
\item The game uses the extractor and $\CM.\Commit$ to obtain the commitment:
\begin{enumerate}[nolistsep]
\item For each $\indexForMTNode \in \Bits^{<\QVCMessageDepth}$ in the reverse order:
\begin{enumerate}[nolistsep]
\item \textcolor{blue!70}{Increase the counter by 1: $\Counter \gets \Counter + 1$.}
\item \textcolor{blue!70}{If $\indexForMTNode \in \QSTCPath{\QVCQuerySet}\setminus \QVCQuerySet$},
\begin{enumerate}[nolistsep]
\item \textcolor{blue!70}{If $\abs{\indexForMTNode} = \QVCMessageDepth - 1$, define $\accessVectorAt{\MessageRegister{1}}{\indexForMTNode} \coloneq (\accessVectorAt{\MessageRegister{1}}{\indexForMTNode, 0}, \accessVectorAt{\MessageRegister{1}}{\indexForMTNode, 1})$.}
\item \textcolor{blue!70}{If $\abs{\indexForMTNode} < \QVCMessageDepth - 1$, define $\accessVectorAt{\MessageRegister{1}}{\indexForMTNode} \coloneq (\accessVectorAt{\CommitmentRegister{1}}{\indexForMTNode, 0}, \accessVectorAt{\CommitmentRegister{1}}{\indexForMTNode, 1})$.}
\item Apply $\NoCollisionAlgVariant{\CM.\Commit}$ to $\accessVectorAt{\MessageRegister{1}}{\indexForMTNode}$, where the oracle access is implemented by $\OracleUnitary(\DatabaseRegister{1})$, to obtain the commitment and the opening $(\accessVectorAt{\CommitmentRegister{1}}{\indexForMTNode}, \accessVectorAt{\OpeningRegister{1}}{\indexForMTNode})$ and a validity bit $\accessVectorAt{\NoCollisionAlgVariant{\ValidityBit}}{\indexForMTNode}$.
\item Output 0 if $\accessVectorAt{\NoCollisionAlgVariant{\ValidityBit}}{\indexForMTNode} = 0$.
\item Initialize the register $(\accessVectorAt{\StandardBasisWorkingRegister{1}}{\indexForMTNode}, \accessVectorAt{\HadamardBasisWorkingRegister{1}}{\indexForMTNode})$ as all-zero states.
\end{enumerate}
\item \textcolor{blue!70}{If $\indexForMTNode \notin \QSTCPath{\QVCQuerySet}\setminus \QVCQuerySet$ and $\Counter > 0$},
\begin{enumerate}[nolistsep]
\item \textcolor{blue!70}{If $\abs{\indexForMTNode} = \QVCMessageDepth - 1$, define $\accessVectorAt{\MessageRegister{1}}{\indexForMTNode} \coloneq (\accessVectorAt{\MessageRegister{1}}{\indexForMTNode, 0}, \accessVectorAt{\MessageRegister{1}}{\indexForMTNode, 1})$.}
\item \textcolor{blue!70}{If $\abs{\indexForMTNode} < \QVCMessageDepth - 1$, define $\accessVectorAt{\MessageRegister{1}}{\indexForMTNode} \coloneq (\accessVectorAt{\CommitmentRegister{1}}{\indexForMTNode, 0}, \accessVectorAt{\CommitmentRegister{1}}{\indexForMTNode, 1})$.}
\item Denote the reverse of $\Extractor_{\CM}.\ExtractMessage$ as an algorithm $\Extractor_{\CM}.\AltCommit$.
\item Run $\NoCollisionAlgVariant{\Extractor_{\CM}.\AltCommit}(\accessVectorAt{\AltOpeningRegister{1}}{\indexForMTNode}, \accessVectorAt{\MessageRegister{1}}{\indexForMTNode})$ by forwarding its queries to $\OracleUnitary(\DatabaseRegister{1})$ and $\invert(\DatabaseRegister{1})$ to obtain $\accessVectorAt{\CommitmentRegister{1}}{\indexForMTNode}$, the ancilla registers $(\accessVectorAt{\StandardBasisWorkingRegister{1}}{\indexForMTNode}, \accessVectorAt{\HadamardBasisWorkingRegister{1}}{\indexForMTNode})$, and a validity bit $\accessVectorAt{\NoCollisionAlgVariant{\ValidityBit}}{\indexForMTNode}$.
\item Output 0 if the bit $\accessVectorAt{\NoCollisionAlgVariant{\ValidityBit}}{\indexForMTNode} = 0$.
\end{enumerate}
\item \textcolor{blue!70}{If $\indexForMTNode \notin \QSTCPath{\QVCQuerySet}\setminus \QVCQuerySet$ and $\Counter \leq 0$},
\begin{enumerate}[nolistsep]
\item \textcolor{blue!70}{Initialize the register $(\accessVectorAt{\StandardBasisWorkingRegister{1}}{\indexForMTNode}, \accessVectorAt{\HadamardBasisWorkingRegister{1}}{\indexForMTNode})$ as all-zero states.}
\end{enumerate}
\end{enumerate}
\item Set $\CommitmentRegister{1} \coloneq \accessVectorAt{\CommitmentRegister{1}}{\emptystring}$.
\end{enumerate}
\item\label{step:check-collision2} \textcolor{blue!70}{Check whether the database register $\DatabaseRegister{1}$ has a collision by measuring it with $\{\ProjectNoCollision, \id{\DatabaseRegister{0}} - \ProjectNoCollision\}$. Output 0 if there is a collision.}
\item Set $\OpeningRegister{1} \coloneq \bigotimes_{\indexForMTNode \in \QSTCPath{\QVCQuerySet}\setminus\QVCQuerySet}\accessVectorAt{\OpeningRegister{1}}{\indexForMTNode}$ and $\WorkingRegister{1} \coloneq (\accessVectorAt{\StandardBasisWorkingRegister{1}}{\indexForMTNode}, \accessVectorAt{\HadamardBasisWorkingRegister{1}}{\indexForMTNode})_{\indexForMTNode \in \Bits^{<\QVCMessageDepth}}$.
\item The game generates the output: compute $\ValidityBit' \gets \Distinguisher(\QVCQuerySetRegister{1}, \CommitmentRegister{1}, \OpeningRegister{1}, \AnswerRegister{1}, \EnvironmentRegister{1}, \DatabaseRegister{1}, \WorkingRegister{1})$ and output $\ValidityBit'$.
\end{enumerate}
\item [] $\HybridI{4}(\Adversary, \Distinguisher)$: it does the same as $\NoCollisionVariant{\QVCSimWorldSC}$ except that instead of letting the adversary produce $\QVCQuerySetRegister{1}$, we fix the query set to be $\QVCQuerySet$.
\end{itemize}

Let $\PureStateSampleOneI{\HybridI{i}}$ and $\PureStateSampleOneI{\HybridIJ{i}{j}}$ denote the corresponding subnormalized states at the point immediately before applying $\Distinguisher$ in games $\HybridI{i}$ and $\HybridIJ{i}{j}$, respectively, for each $i$, $j$. We show through several lemmas that the states $\PureStateSampleOneI{\NoCollisionVariant{\QVCSimWorldSC}}$ and $\PureStateSampleOneI{\NoCollisionVariant{\QVCOfflineExtractWorldSC}}$ are close in $\ell_2$ norm.

\iffull
\begin{claim}
\else
\begin{numberedclaim}
\fi
\label{claim:Hybrid-state-1}
$\PureStateSampleOneI{\NoCollisionVariant{\QVCOfflineExtractWorldSC}} = \PureStateSampleOneI{\HybridIJ{1}{0}}$, $\PureStateSampleOneI{\HybridIJ{1}{\QVCPathSize}} = \PureStateSampleOneI{\HybridIJ{2}{0}}$, $\PureStateSampleOneI{\HybridIJ{2}{\QVCPathSize}} = \PureStateSampleOneI{\HybridIJ{3}{0}}$, and $\PureStateSampleOneI{\HybridIJ{3}{2^{\QVCMessageDepth} - 1}} = \PureStateSampleOneI{\NoCollisionVariant{\QVCSimWorldSC}}$.
\iffull
\end{claim}
\else
\end{numberedclaim}
\fi

\begin{proof}
Since $\Extractor_{\CM}.\AltCheck$ does not act on the database register $\DatabaseRegister{1}$, by the construction of $\Extractor_{\QSTC}.\AltCheck$, the game $\HybridIJ{1}{0}(\Adversary, \Distinguisher)$ is equivalent to the game $\HybridI{0}(\Adversary, \Distinguisher)$, which is equivalent to $\NoCollisionVariant{\QVCOfflineExtractWorldSC}$ when the query set is fixed to $\QVCQuerySet$, and thus
\[\PureStateSampleOneI{\NoCollisionVariant{\QVCOfflineExtractWorldSC}} = \PureStateSampleOneI{\HybridI{0}} = \PureStateSampleOneI{\HybridIJ{1}{0}}\enspace.\]

The game $\HybridIJ{2}{0}$ is equivalent to the game $\HybridIJ{1}{\QVCPathSize}$ by construction, since in the game $\HybridIJ{1}{\QVCPathSize}$, checking whether $\indexForMTNode \in \{\indexForMTNode_1, \indexForMTNode_2, \ldots, \indexForMTNode_{i}\}$ is just the same as checking whether $\indexForMTNode \in \QSTCPath{\QVCQuerySet} \setminus \QVCQuerySet$, and in the game $\HybridIJ{2}{0}$, we always take the branch $\indexForMTNode \notin \{\indexForMTNode_1, \indexForMTNode_2, \ldots, \indexForMTNode_{i}\}$. Similarly, the game $\HybridIJ{3}{0}$ is equivalent to the game $\HybridIJ{2}{\QVCPathSize}$ by construction, and thus
\[\PureStateSampleOneI{\HybridIJ{1}{\QVCPathSize}} = \PureStateSampleOneI{\HybridIJ{2}{0}}\enspace,\]
and
\[\PureStateSampleOneI{\HybridIJ{2}{\QVCPathSize}} = \PureStateSampleOneI{\HybridIJ{3}{0}}\enspace.\]

By the construction of $\HybridIJ{3}{2^{\QVCMessageDepth} - 1}$, in the procedure, the counter $\Counter$ is always non-positive since the initialization. Moreover, \Cref{step:check-collision1,step:check-collision2} can be absorbed in $\NoCollisionAlgVariant{\QSTC.\Query}^{\OracleUnitary(\DatabaseRegister{1})}(1^\Security, \QVCQuerySetRegister{1}, \CommitmentRegister{1}, \OpeningRegister{1}, \AnswerRegister{1})$ as the no collision variant of an algorithm would always check whether the database contains a collision before and after applying a unitary. Therefore, the game $\HybridIJ{3}{2^{\QVCMessageDepth} - 1}(\Adversary, \Distinguisher)$ is equivalent to the game $\HybridI{4}(\Adversary, \Distinguisher)$, which is equivalent to the game $\NoCollisionVariant{\QVCSimWorldSC}$ when the query set is fixed to $\QVCQuerySet$, and thus
\[\PureStateSampleOneI{\HybridIJ{3}{2^{\QVCMessageDepth} - 1}} = \PureStateSampleOneI{\HybridI{4}} = \PureStateSampleOneI{\NoCollisionVariant{\QVCSimWorldSC}}\enspace.\]
\end{proof}

\iffull
\begin{claim}
\else
\begin{numberedclaim}
\fi
\label{claim:Hybrid-state-2}
For $i \in [\QVCPathSize]$, 
\begin{align*}
	&\vecnorm{\PureStateSampleOneI{\HybridIJ{1}{i - 1}} - \PureStateSampleOneI{\HybridIJ{1}{i}}}\\
	&\leq (16\QVCMessageLength + 5)(\NumberOfQueries + \NumberOfQueriesBound) \cdot 2^{-(\RandomOracleOutputLength - 3)/2} + 2^{-\RandomOracleOutputLength/2 + 5}(\NumberOfQueries + \NumberOfQueriesBound)\sqrt{\QuantumTotalQueryMass_2 + 4\QVCMessageLength} + 8 \cdot 2^{-\RandomOracleOutputLength/2}\enspace.
\end{align*}
\iffull
\end{claim}
\else
\end{numberedclaim}
\fi

\begin{proof}
The only difference between the games $\HybridIJ{1}{i - 1}$ and $\HybridIJ{1}{i}$ is in how the location $\indexForMTNode_i$ is handled. In $\HybridIJ{1}{i}$, the location $\indexForMTNode_i$ is opened using $\NoCollisionAlgVariant{\CM.\Check}$, after which the registers $(\accessVectorAt{\OpeningRegister{1}}{\indexForMTNode_i}, \accessVectorAt{\AltOpeningRegister{1}}{\indexForMTNode_i})$ are initialized to the state $\ket{\NullStateInExp_{\AltCheck}}$. In $\HybridIJ{1}{i - 1}$, the location $\indexForMTNode_i$ is instead extracted using $\NoCollisionAlgVariant{\Extractor_{\CM}.\ExtractMessage}$ and is later checked using $\Extractor_{\CM}.\AltCheck$.

By \Cref{claim:diff-between-the-two-operations}, replacing $\NoCollisionAlgVariant{\CM.\Check}$ with $\NoCollisionAlgVariant{\Extractor_{\CM}.\ExtractMessage}$ followed by $\Extractor_{\CM}.\AltCheck$ incurs a loss of at most 
\begin{align*}
\frac{8}{\sqrt{2^\RandomOracleOutputLength}} + \vecnorm{(\id{\DatabaseRegister{0}} - \ProjectNoHatZero) \PureStateSampleOneI{1}}	+ \vecnorm{(\id{\DatabaseRegister{0}} - \ProjectNoHatZero) \PureStateSampleOneI{2}}	
\end{align*}
where $\PureStateSampleOneI{1} = \ProjectNoCollision\HadamardGate^{\otimes \MessageLength}_{\MessageRegister{0}}\OracleUnitary_{\MessageRegister{0}\HadamardBasisHashValueRegister{0}\DatabaseRegister{0}}\ProjectNoCollision\HadamardGate^{\otimes \MessageLength}_{\MessageRegister{0}}\PureStateSampleOne_{\CommitmentRegister{0}\OpeningRegister{0}\DatabaseRegister{0}\EnvironmentRegister{0}}$, $\PureStateSampleOneI{2} = \ProjectNoCollision\HadamardGate^{\otimes \MessageLength}_{\MessageRegister{0}}\PureStateSampleOne_{\CommitmentRegister{0}\OpeningRegister{0}\DatabaseRegister{0}\EnvironmentRegister{0}}$, and the state $\PureStateSampleOne_{\CommitmentRegister{0}\OpeningRegister{0}\DatabaseRegister{0}\EnvironmentRegister{0}}$ is the joint state when the game handles the location $\indexForMTNode_i$.

Since $\Adversary$ makes at most $\NumberOfQueries$ queries and the game makes at most $\NumberOfQueriesBound$ queries to the two oracles, the states $\PureStateSampleOneI{1}$ and $\PureStateSampleOneI{2}$ are invariant under $\ProjectSizeDatabase{\NumberOfQueries + \NumberOfQueriesBound}$.

As a result, by a similar reasoning as \Cref{lemma:indistinguishablility-if-no-collisions}, since to get $\PureStateSampleOne_{\CommitmentRegister{0}\OpeningRegister{0}\DatabaseRegister{0}\EnvironmentRegister{0}}$, we make at most $\NumberOfQueriesBound + 1$ calls to $\ProjectNoCollision$, and make at most $\NumberOfQueries + \NumberOfQueriesBound$ queries to $\invert$ with total query mass at most $\QuantumTotalQueryMass_2 + 4\QVCMessageLength$,
\begin{align*}
	&\vecnorm{(\id{\DatabaseRegister{0}} - \ProjectNoHatZero) \PureStateSampleOneI{1}}	+ \vecnorm{(\id{\DatabaseRegister{0}} - \ProjectNoHatZero) \PureStateSampleOneI{2}}\\
	&=\vecnorm{(\id{\DatabaseRegister{0}} - \ProjectNoHatZero) \ProjectSizeDatabase{\NumberOfQueries + \NumberOfQueriesBound}\PureStateSampleOneI{1}}	+ \vecnorm{(\id{\DatabaseRegister{0}} - \ProjectNoHatZero) \ProjectSizeDatabase{\NumberOfQueries + \NumberOfQueriesBound}\PureStateSampleOneI{2}}\\
	&\leq 3\matnorm{\ProjectSizeDatabase{\NumberOfQueries + \NumberOfQueriesBound}\Commutator{\NoCollision}{\ProjectNoHatZero}\ProjectSizeDatabase{\NumberOfQueries + \NumberOfQueriesBound}} + 2\vecnorm{(\id{\DatabaseRegister{0}} - \ProjectNoHatZero) \PureStateSampleOne_{\CommitmentRegister{0}\OpeningRegister{0}\DatabaseRegister{0}\EnvironmentRegister{0}}}\\
	&\leq 3(\NumberOfQueries + \NumberOfQueriesBound) \cdot 2^{-(\RandomOracleOutputLength - 3)/2} + 2\matnorm{\ProjectSizeDatabase{\NumberOfQueries + \NumberOfQueriesBound}\Commutator{\invert_{\TargetRegister{0}\DatabaseRegister{0}\WorkingRegister{0}}}{\ProjectNoHatZero}\ProjectSizeDatabase{\NumberOfQueries + \NumberOfQueriesBound}}\sqrt{(\NumberOfQueries + \NumberOfQueriesBound)(\QuantumTotalQueryMass_2 + 4\QVCMessageLength)}\\
	&+ 2(\NumberOfQueriesBound + 1)\matnorm{\ProjectSizeDatabase{\NumberOfQueries + \NumberOfQueriesBound}\Commutator{\NoCollision}{\ProjectNoHatZero}\ProjectSizeDatabase{\NumberOfQueries + \NumberOfQueriesBound}}\\
	&\leq (16\QVCMessageLength + 5)(\NumberOfQueries + \NumberOfQueriesBound) \cdot 2^{-(\RandomOracleOutputLength - 3)/2} + 2^{-\RandomOracleOutputLength/2 + 5}(\NumberOfQueries + \NumberOfQueriesBound)\sqrt{\QuantumTotalQueryMass_2 + 4\QVCMessageLength}\enspace,
\end{align*}
where we notice that only queries to $\ProjectNoCollision$ or $\invert$ would contribute error to $\vecnorm{(\id{\DatabaseRegister{0}} - \ProjectNoHatZero) \PureStateSampleOne_{\CommitmentRegister{0}\OpeningRegister{0}\DatabaseRegister{0}\EnvironmentRegister{0}}}$.

Initializing the registers $(\accessVectorAt{\OpeningRegister{1}}{\indexForMTNode_i}, \accessVectorAt{\AltOpeningRegister{1}}{\indexForMTNode_i})$ to $\ket{\NullStateInExp_{\AltCheck}}$ in the game $\HybridIJ{1}{i}$ incurs no additional loss, since conditioned on the validity bit being $1$, these registers are also in the state $\ket{\NullStateInExp_{\AltCheck}}$ in the game $\HybridIJ{1}{i - 1}$.

Therefore,
\begin{align*}
	&\vecnorm{\PureStateSampleOneI{\HybridIJ{1}{i - 1}} - \PureStateSampleOneI{\HybridIJ{1}{i}}}\\
	&\leq (16\QVCMessageLength + 5)(\NumberOfQueries + \NumberOfQueriesBound) \cdot 2^{-(\RandomOracleOutputLength - 3)/2} + 2^{-\RandomOracleOutputLength/2 + 5}(\NumberOfQueries + \NumberOfQueriesBound)\sqrt{\QuantumTotalQueryMass_2 + 4\QVCMessageLength} + 8 \cdot 2^{-\RandomOracleOutputLength/2}\enspace.
\end{align*}
\end{proof}

\iffull
\begin{claim}
\else
\begin{numberedclaim}
\fi
\label{claim:Hybrid-state-3}
For $i \in [\QVCPathSize]$, 
\begin{align*}
	&\vecnorm{\PureStateSampleOneI{\HybridIJ{2}{i - 1}} - \PureStateSampleOneI{\HybridIJ{2}{i}}}\\
	&\leq (24\QVCMessageLength + 8)(\NumberOfQueries + \NumberOfQueriesBound) \cdot 2^{-(\RandomOracleOutputLength - 3)/2} + 2^{-(\RandomOracleOutputLength - 11)/2}(\NumberOfQueries + \NumberOfQueriesBound)\sqrt{\QuantumTotalQueryMass_2 + 4\QVCMessageLength} + 4 \cdot 2^{-\RandomOracleOutputLength/2}\enspace.
\end{align*}
\iffull
\end{claim}
\else
\end{numberedclaim}
\fi

\begin{proof}
The only difference between the games $\HybridIJ{2}{i - 1}$ and $\HybridIJ{2}{i}$ is in how the location $\indexForMTNode_i$ is handled after the game implements the query. In $\HybridIJ{2}{i }$, the commitment and the opening of the location $\indexForMTNode_i$ are obtained by initializing the ancilla qubits as $\ket{\NullStateInExp_{\Check}}$ and applying the commitment scheme; the working registers $(\accessVectorAt{\StandardBasisWorkingRegister{1}}{\indexForMTNode_i}, \accessVectorAt{\HadamardBasisWorkingRegister{1}}{\indexForMTNode_i})$ are then initialized to all-zero states. In $\HybridIJ{2}{i - 1}$, the commitment and the opening of the location $\indexForMTNode_i$ are instead obtained by initializing $(\accessVectorAt{\OpeningRegister{1}}{\indexForMTNode_i}, \accessVectorAt{\AltOpeningRegister{1}}{\indexForMTNode_i})$ as the state $\ket{\NullStateInExp_{\AltCheck}}$, then applying $\NoCollisionAlgVariant{\Extractor_{\CM}.\AltCommit}$ while keeping the working registers $(\accessVectorAt{\StandardBasisWorkingRegister{1}}{\indexForMTNode_i}, \accessVectorAt{\HadamardBasisWorkingRegister{1}}{\indexForMTNode_i})$.
This difference is exactly captured by 
\[\matnorm{(V_{\Check}^\dagger\ket{\NullStateInExp_{\Check}}\ket{0}_{\StandardBasisWorkingRegister{0}\HadamardBasisWorkingRegister{0}} - V_{\ExtractMessage}^\dagger\ket{\NullStateInExp_{\AltCheck}})\PureStateSampleTwo_{\MessageRegister{0}\EnvironmentRegister{0}\DatabaseRegister{0}}}\enspace,\] 
where $\PureStateSampleTwo_{\MessageRegister{0}\EnvironmentRegister{0}\DatabaseRegister{0}}$ is the joint state when the game handles the location $\indexForMTNode_i$ after the query is implemented, $V_{\Check}^\dagger$ is exactly what $\NoCollisionAlgVariant{\CM.\Commit}$ does, and $V_{\ExtractMessage}^\dagger$ is exactly what $\NoCollisionAlgVariant{\Extractor_{\CM}.\AltCommit}$ does.

By \Cref{claim:diff-between-the-two-operations},
\begin{align*}
&\matnorm{(V_{\Check}^\dagger\ket{\NullStateInExp_{\Check}}\ket{0}_{\StandardBasisWorkingRegister{0}\HadamardBasisWorkingRegister{0}} - V_{\ExtractMessage}^\dagger\ket{\NullStateInExp_{\AltCheck}})\PureStateSampleTwo_{\MessageRegister{0}\EnvironmentRegister{0}\DatabaseRegister{0}}}\\
&\leq \frac{4}{\sqrt{2^\RandomOracleOutputLength}} + \sqrt{2}\vecnorm{(\id{\DatabaseRegister{0}} - \ProjectNoHatZero) \PureStateSampleTwoI{1}}	+ \sqrt{2}\vecnorm{(\id{\DatabaseRegister{0}} - \ProjectNoHatZero) \PureStateSampleTwoI{2}}\enspace,
\end{align*}
where $\PureStateSampleTwoI{1} \coloneq \HadamardGate^{\otimes \MessageLength}_{\MessageRegister{0}}\ProjectNoCollision\OracleUnitary_{\MessageRegister{0}\StandardBasisHashValueRegister{0}\DatabaseRegister{0}}\ProjectNoCollision\PureStateSampleTwo_{\MessageRegister{0}\EnvironmentRegister{0}\DatabaseRegister{0}}\ket{\NullStateInExp_{\Check}}$ and $\PureStateSampleTwoI{2} \coloneq \ProjectNoCollision\PureStateSampleTwo_{\MessageRegister{0}\EnvironmentRegister{0}\DatabaseRegister{0}}$.

Therefore, by the same reasoning as \Cref{claim:Hybrid-state-2},
\begin{align*}
	&\vecnorm{\PureStateSampleOneI{\HybridIJ{2}{i - 1}} - \PureStateSampleOneI{\HybridIJ{2}{i}}}\\
	&\leq (24\QVCMessageLength + 8)(\NumberOfQueries + \NumberOfQueriesBound) \cdot 2^{-(\RandomOracleOutputLength - 3)/2} + 2^{-(\RandomOracleOutputLength - 11)/2}(\NumberOfQueries + \NumberOfQueriesBound)\sqrt{\QuantumTotalQueryMass_2 + 4\QVCMessageLength} + 4 \cdot 2^{-\RandomOracleOutputLength/2}\enspace.
\end{align*}
\end{proof}

\iffull
\begin{claim}
\else
\begin{numberedclaim}
\fi
\label{claim:Hybrid-state-4}
For $i \in [2^\QVCMessageDepth - 1]$, $\vecnorm{\PureStateSampleOneI{\HybridIJ{3}{i - 1}} - \PureStateSampleOneI{\HybridIJ{3}{i}}} \leq 2^{-\RandomOracleOutputLength/2}\left((16\QVCMessageLength + 8)\sqrt{6(\NumberOfQueries + 8 \QVCMessageLength)} + 64\sqrt{2}\QVCMessageLength\right).$
\iffull
\end{claim}
\else
\end{numberedclaim}
\fi

\begin{proof}
The difference between the games $\HybridIJ{3}{i - 1}$ and $\HybridIJ{3}{i}$ is that $\HybridIJ{3}{i - 1}$ does an extra $\NoCollisionAlgVariant{\Extractor_{\CM}.\ExtractMessage}(\accessVectorAt{\CommitmentRegister{1}}{\indexForMTNode})$ before the game implements the query, and $\HybridIJ{3}{i - 1}$ does an extra $\NoCollisionAlgVariant{\Extractor_{\CM}.\AltCommit}(\accessVectorAt{\AltOpeningRegister{1}}{\indexForMTNode}, \accessVectorAt{\MessageRegister{1}}{\indexForMTNode})$ after the game implements the query, while $\HybridIJ{3}{i}$ initializes $(\accessVectorAt{\StandardBasisWorkingRegister{1}}{\indexForMTNode}, \accessVectorAt{\HadamardBasisWorkingRegister{1}}{\indexForMTNode})$ as all-zeros in the computational basis.

Notice that $\Extractor_{\CM}.\AltCommit$ does exactly the inverse of $\Extractor_{\CM}.\ExtractMessage$. The claim follows from the following two facts:
\begin{itemize}
\item Replace this extra $\NoCollisionAlgVariant{\Extractor_{\CM}.\ExtractMessage}(\accessVectorAt{\CommitmentRegister{1}}{\indexForMTNode})$ with ${\Extractor_{\CM}.\ExtractMessage}(\accessVectorAt{\CommitmentRegister{1}}{\indexForMTNode})$, and $\NoCollisionAlgVariant{\Extractor_{\CM}.\AltCommit}(\accessVectorAt{\AltOpeningRegister{1}}{\indexForMTNode}, \accessVectorAt{\MessageRegister{1}}{\indexForMTNode})$ with ${\Extractor_{\CM}.\AltCommit}(\accessVectorAt{\AltOpeningRegister{1}}{\indexForMTNode}, \accessVectorAt{\MessageRegister{1}}{\indexForMTNode})$ in game $\HybridIJ{3}{i - 1}$ to form a game $\HybridIJPrime{3}{i - 1}$. As the database size is always bounded by $\NumberOfQueries + \NumberOfQueriesBound$ and in each game, we always check whether the database has a collision at the beginning and at the end, this step would incur an error
\[\vecnorm{\PureStateSampleOneI{\HybridIJ{3}{i - 1}} - \PureStateSampleOneI{\HybridIJPrime{3}{i - 1}}} \leq 4 \matnorm{(\id{\DatabaseRegister{1}} - \ProjectNoCollision)\left(\ProjectSizeDatabase{\NumberOfQueries + \NumberOfQueriesBound}\OracleUnitary\ProjectSizeDatabase{\NumberOfQueries + \NumberOfQueriesBound}\right)\ProjectNoCollision} \leq 4 \cdot \sqrt{6(\NumberOfQueries + \NumberOfQueriesBound)} \cdot 2^{-\RandomOracleOutputLength/2}\enspace.\]
\item Moving the operation ${\Extractor_{\CM}.\AltCommit}(\accessVectorAt{\AltOpeningRegister{1}}{\indexForMTNode}, \accessVectorAt{\MessageRegister{1}}{\indexForMTNode})$ in game $\HybridIJPrime{3}{i - 1}$ to just after the operation ${\Extractor_{\CM}.\ExtractMessage}(\accessVectorAt{\CommitmentRegister{1}}{\indexForMTNode})$ results in exactly the same game as $\HybridIJ{3}{i}$, where $(\accessVectorAt{\StandardBasisWorkingRegister{1}}{\indexForMTNode}, \accessVectorAt{\HadamardBasisWorkingRegister{1}}{\indexForMTNode})$ is initialized as all-zero states. This step would incur an error
\begin{align*}
&\vecnorm{\PureStateSampleOneI{\HybridIJ{3}{i}} - \PureStateSampleOneI{\HybridIJPrime{3}{i - 1}}}\\
&\leq 8\QVCPathSize \cdot \matnorm{\Commutator{\invert}{O}} + (8\QVCPathSize + 2) \cdot \matnorm{\Commutator{\ProjectNoCollision}{\ProjectSizeDatabase{\NumberOfQueries + \NumberOfQueriesBound}\OracleUnitary\ProjectSizeDatabase{\NumberOfQueries + \NumberOfQueriesBound}}}\\
&\leq 8\QVCPathSize \cdot 2^{-(\RandomOracleOutputLength - 7)/2} + (16\QVCPathSize + 4) \cdot \sqrt{6(\NumberOfQueries + \NumberOfQueriesBound)} \cdot 2^{-\RandomOracleOutputLength/2}\enspace.
\end{align*}
\end{itemize}

Notice that $\QVCPathSize \leq \QVCMessageLength$. A triangle inequality concludes the proof.
\end{proof}

Combining \Cref{claim:Hybrid-state-1,claim:Hybrid-state-2,claim:Hybrid-state-3,claim:Hybrid-state-4}, we obtain
\begin{align*}
&\vecnorm{\PureStateSampleOneI{\NoCollisionVariant{\QVCSimWorldSC}} - \PureStateSampleOneI{\NoCollisionVariant{\QVCOfflineExtractWorldSC}}}^2\\
&\leq \left(\QVCMessageLength(40\QVCMessageLength + 13)(\NumberOfQueries + \NumberOfQueriesBound) \cdot 2^{-(\RandomOracleOutputLength - 3)/2} + 2^{-\RandomOracleOutputLength/2 + 7}\QVCMessageLength(\NumberOfQueries + \NumberOfQueriesBound)\sqrt{\QuantumTotalQueryMass_2 + 4\QVCMessageLength} + 12\QVCMessageLength \cdot 2^{-\RandomOracleOutputLength/2}\right.\\
& \left.+ 2^{-\RandomOracleOutputLength/2} \QVCMessageLength(16\QVCMessageLength + 8)\sqrt{6(\NumberOfQueries + 8 \QVCMessageLength)} + 2^{-(\RandomOracleOutputLength - 13)/2}\QVCMessageLength^2\right)^2\\
&\leq 2^{-\RandomOracleOutputLength + 2}\left(8\QVCMessageLength^2(40\QVCMessageLength + 13)^2(\NumberOfQueries + \NumberOfQueriesBound)^2 + 2^{14}\QVCMessageLength^2(\NumberOfQueries + \NumberOfQueriesBound)^2(\QuantumTotalQueryMass_2 + 4\QVCMessageLength) + 6\QVCMessageLength^2(16\QVCMessageLength + 8)^2(\NumberOfQueries + 8 \QVCMessageLength) + 2^{15}\QVCMessageLength^4\right)\\
&\leq 2^{-\RandomOracleOutputLength + 20}\QVCMessageLength^2(\NumberOfQueries + \NumberOfQueriesBound)^2 \left(\QVCMessageLength^2+\QuantumTotalQueryMass_2\right)\\
&\leq \ErrorBoundForSingleCommitOffline (\RandomOracleOutputLength, \QVCMessageLength, \NumberOfQueries, \QuantumTotalQueryMass_1, \QuantumTotalQueryMass_2)\enspace,
\end{align*}
where we use the triangle inequality in the second line, and the Cauchy--Schwarz inequality in the third line.

\end{proof}

\doclearpage
\section{Quantum interactive arguments based on QIOPs}

We show how to obtain a quantum-communication succinct interactive argument from public-query QIOPs.

\begin{theorem}
\label{thm:quantum-IBCS-soundness}
Let $\QIOP$ be a quantum interactive oracle proof for relation $\Relation$ in the form of \Cref{sec:QIOP-def} with public-query soundness $\QIOPPublicQuerySoundness$. Let $\QVC$ be an $(\ExtractionErrorI{1}, \ExtractionErrorI{2}, \ExtractionErrorI{3})$-extractable quantum state vector commitment scheme (\Cref{def:qvc_extractability}). Let $\RandomOracleOutputLength$ be the output length of the random oracle. Then the protocol $(\ArgProver, \ArgVerifier) = \IBCS[\QIOP, \QVC]$ in \Cref{construction:IBCS} is a quantum-communication interactive argument for relation $\Relation$ with soundness
\[\ExtractionErrorI{1}(\NumberOfQueries + \IBCSVerifierQueryComplexity, \RandomOracleOutputLength) + 2\QIOPPublicQuerySoundness(\InstanceSize) + 2\QIOPRoundComplexity \cdot \sum_{i \in [\QIOPRoundComplexity]}\ExtractionErrorI{3}(\NumberOfQueries, \QIOPQueryDepth, \NumberOfQueries, 0, \RandomOracleOutputLength, \QIOPProofSize{i})\enspace,\]
where $\IBCSVerifierQueryComplexity$ is the number of queries that the argument verifier $\ArgVerifier$ makes to the random oracle $\RandomOracle{1}$. Furthermore, when $\QVC$ is instantiated with the construction $\QSTC$ in \Cref{construction:quantum-state-vector-commitment}, $\IBCSVerifierQueryComplexity= O(\QIOPQueryComplexity \log \QIOPProofMaxSize)$.
\end{theorem}

\subsection{Our transformation}

We describe below our construction of the quantum-communication succinct interactive argument in the random oracle model, which we denote $(\ArgProver, \ArgVerifier) \coloneq \IBCS[\QIOP, \QVC]$. This is a quantum analogue of the IBCS transformation for IOPs \cite{CDGS23}.

\begin{construction}
\label{construction:IBCS}
Let $(\QIOPProver, \QIOPVerifier) \coloneq \QIOP$. The argument prover $\ArgProver$ receives an instance $\Instance$ and a quantum witness $\QuantumWitness$, and the argument verifier receives the same instance $\Instance$. $\ArgProver$ and $\ArgVerifier$ both have access to a random function $\RandomOracle{1}$ sampled from $\UniformFrom{\RandomOracleOutputLength}$. We do domain separation on $\RandomOracle{1}$ to obtain $\QIOPRoundComplexity$ random functions $\{\RandomOracle{1}_i\}_{i \in [\QIOPRoundComplexity]}$, where $\RandomOracle{1}_i(x) = \RandomOracle{1}(i, x)$ for $x \in \Bits^*$ and $i \in [\QIOPRoundComplexity]$. Then $\ArgProver$ and $\ArgVerifier$ interact as follows.
\begin{enumerate}[noitemsep]
\item $\ArgProver$'s initialization: Initialize $\ProverPrivateRegister{1}$ to the quantum witness $\QuantumWitness$, padded with all-zero states, let $\VerifierMessageRegister{0}$ be an empty register, and set $\QIOPNotReturnIdxSet{1} \coloneq \{1\}, \QIOPReturnIdxSet{1} \coloneq \emptyset$.
\item $\ArgVerifier$'s initialization: Initialize $\VerifierPrivateRegister{1}$ to all-zero states and set $\QIOPNotReturnIdxSet{1} \coloneq \{1\}, \QIOPReturnIdxSet{1} \coloneq \emptyset$.
\item For $i \in [\QIOPRoundComplexity]$:
\begin{enumerate}[nolistsep]
\item $\ArgProver$'s $i$-th commitment.
\begin{enumerate}[nolistsep]
\item Compute the returned QIOP proofs: for $j \in \QIOPReturnIdxSet{i}$, $\ProverMessageRegister{j} \gets \QVCRecover^{\RandomOracle{0}_j}(1^\Security,\QVCCommitmentRegister_j,\QVCAuxiliaryRegister{1}_j)$.
\item Compute the $i$-th QIOP state: $(\ProverMessageRegister{i}, \ProverPrivateRegister{i + 1}) \gets \QIOPProver(\Instance, \VerifierMessageRegister{i - 1}, \ProverPrivateRegister{i}, (\ProverMessageRegister{j})_{j \in \QIOPReturnIdxSet{i}})$.
\item Parse $\ProverMessageRegister{i}$ as $\QIOPProofSize{i}$ subregisters: $(\accessVectorAt{\ProverMessageRegister{i}}{j})_{j \in [\QIOPProofSize{i}]} \coloneq \ProverMessageRegister{i}$.
\item Compute a QVC commitment to the QIOP state where the message length is set to $\QIOPProofSize{i}$: \\
$({\CommitmentRegister{1}_i, \QVCAuxiliaryRegister{1}}_i) \gets \QVCCommit^{\RandomOracle{0}_i}(\ProverMessageRegister{i})$.
\item Send the QVC commitment $\CommitmentRegister{1}_i$ to $\ArgVerifier$.
\end{enumerate}
\item $\ArgProver$ answers the queries of $\QIOPVerifier$ that the argument verifier $\ArgVerifier$ simulates.
\begin{enumerate}[nolistsep]
\item $\ArgVerifier$ initializes $(\QueryLocationRegister_{\iota}, \AnswerRegister{1}_{\iota})_{\iota \in [\QIOPQueryWidthI{i}]}$ as the all-zero states and sets $\VerifierPrivateRegister{i}^{(0)} \coloneq \VerifierPrivateRegister{i}$.
\item For $q \in [\QIOPQueryDepthI{i}]$:
\begin{enumerate}[nolistsep]
\item $\ArgVerifier$ computes the query locations: $((\QueryLocationRegister_{\iota}, \AnswerRegister{1}_{\iota})_{\iota \in [\QIOPQueryWidthI{i}]}, \VerifierPrivateRegister{i}^{(q)}) \gets \QIOPVerifier_i^{(q - 1)}(\Instance, (\QueryLocationRegister_{\iota}, \AnswerRegister{1}_{\iota})_{\iota \in [\QIOPQueryWidthI{i}]}, \VerifierPrivateRegister{i}^{(q - 1)})$.
\item $\ArgVerifier$ sends $(\QueryLocationRegister_{\iota})_{\iota \in [\QIOPQueryWidthI{i}]}$ to $\ArgProver$.
\item $\ArgProver$ computes an opening: \\
$((\QueryLocationRegister_{\iota})_{\iota \in [\QIOPQueryWidthI{i}]}, (\OpeningRegister{1}_{i'}, \QVCAuxiliaryNotUsedRegister_{i'})_{i' \in \QIOPNotReturnIdxSet{i}}) \gets \QVCOpen^{(\RandomOracle{0}_{i'})_{i' \in \QIOPNotReturnIdxSet{i}}}(1^\Security, (\QueryLocationRegister_{\iota})_{\iota \in [\QIOPQueryWidthI{i}]}, (\QVCAuxiliaryRegister{1}_{i'})_{i' \in \QIOPNotReturnIdxSet{i}})$.
\item $\ArgProver$ sends $((\QueryLocationRegister_{\iota})_{\iota \in [\QIOPQueryWidthI{i}]}, (\OpeningRegister{1}_{i'})_{i' \in \QIOPNotReturnIdxSet{i}})$ to $\ArgVerifier$.
\item $\ArgVerifier$ implements the query for $\QIOPVerifier$:\\
$(\QVCValidityBit_{i, q}, (\QueryLocationRegister_{\iota}, \AnswerRegister{1}_{\iota})_{\iota \in [\QIOPQueryWidthI{i}]}, (\CommitmentRegister{1}_{i'}, \OpeningRegister{1}_{i'})_{i' \in \QIOPNotReturnIdxSet{i}}) \gets \QVCQuery^{(\RandomOracle{0}_{i'})_{i' \in \QIOPNotReturnIdxSet{i}}}(1^\Security, (\QueryLocationRegister_{\iota}, \AnswerRegister{1}_{\iota})_{\iota \in [\QIOPQueryWidthI{i}]}, (\CommitmentRegister{1}_{i'}, \OpeningRegister{1}_{i'})_{i' \in \QIOPNotReturnIdxSet{i}})$.
\item $\ArgVerifier$ sends $((\QueryLocationRegister_{\iota})_{\iota \in [\QIOPQueryWidthI{i}]}, (\OpeningRegister{1}_{i'})_{i' \in \QIOPNotReturnIdxSet{i}})$ to $\ArgProver$.
\item $\ArgProver$ erases the information about the query locations:\\
$((\QueryLocationRegister_{\iota})_{\iota \in [\QIOPQueryWidthI{i}]}, (\QVCAuxiliaryRegister{1}_{i'})_{i' \in \QIOPNotReturnIdxSet{i}}) \gets \QVCUpdate^{(\RandomOracle{0}_{i'})_{i' \in \QIOPNotReturnIdxSet{i}}}(1^\Security, (\QueryLocationRegister_{\iota})_{\iota \in [\QIOPQueryWidthI{i}]}, (\OpeningRegister{1}_{i'}, \QVCAuxiliaryNotUsedRegister_{i'})_{i' \in \QIOPNotReturnIdxSet{i}})$.
\item $\ArgProver$ sends $(\QueryLocationRegister_{\iota})_{\iota \in [\QIOPQueryWidthI{i}]}$ to $\ArgVerifier$.
\end{enumerate}
\end{enumerate}
\item If $i \neq \QIOPRoundComplexity$, $\ArgVerifier$ generates the $i$-th message and returns some of the previous proofs.
\begin{enumerate}[nolistsep]
\item $\ArgVerifier$ computes the next message and the indices of proofs to be returned:\\
$(\QIOPReturnIdxSet{i + 1}, \VerifierMessageRegister{i}, \VerifierPrivateRegister{i + 1}) \gets \QIOPVerifier_{i}^{(\QIOPQueryDepthI{i})}(\Instance, (\QueryLocationRegister_{\iota}, \AnswerRegister{1}_{\iota})_{\iota \in [\QIOPQueryWidthI{i}]}, \VerifierPrivateRegister{i}^{(\QIOPQueryDepthI{i})})$.
\item $\ArgVerifier$ sends $(\QIOPReturnIdxSet{i + 1}, \VerifierMessageRegister{i}, (\CommitmentRegister{1}_{i'})_{i' \in \QIOPReturnIdxSet{i + 1}})$ to $\ArgProver$.
\item Both $\ArgProver$ and $\ArgVerifier$ set $\QIOPNotReturnIdxSet{i + 1} \coloneq (\QIOPNotReturnIdxSet{i} \cup \{i + 1\}) \setminus \QIOPReturnIdxSet{i + 1}$.
\end{enumerate}
\item If $i = \QIOPRoundComplexity$, $\ArgVerifier$ decides whether to accept:
\begin{enumerate}[nolistsep]
\item $\ArgVerifier$ computes whether $\QIOPVerifier$ accepts: $\ValidityBit_{\scriptscriptstyle{\QIOP}} \gets \QIOPVerifier_{\QIOPRoundComplexity}^{(\QIOPQueryDepthI{\QIOPRoundComplexity})}(\Instance, (\QueryLocationRegister_{\iota}, \AnswerRegister{1}_{\iota})_{\iota \in [\QIOPQueryWidthI{\QIOPRoundComplexity}]}, \VerifierPrivateRegister{\QIOPRoundComplexity}^{(\QIOPQueryDepthI{\QIOPRoundComplexity})})$.
\item Compute $\ValidityBit_{\scriptscriptstyle{\QVC}} \coloneq \land_{i \in [\QIOPRoundComplexity]} \left(\land_{q \in [\QIOPQueryDepthI{i}]} \QVCValidityBit_{i, q}\right)$.
\item Output $\ValidityBit_{\scriptscriptstyle{\QVC}} \land \ValidityBit_{\scriptscriptstyle{\QIOP}}$.
\end{enumerate}
\end{enumerate}
\end{enumerate}
Here in the construction, we overload $\QVCOpen$, $\QVCQuery$, and $\QVCUpdate$ to handle multiple proofs coherently, using their respective random oracles, and combining validity bits by conjunction.
\end{construction}

When $\QVC$ satisfies perfect completeness, the protocol $(\ArgProver, \ArgVerifier) = \IBCS[\QIOP, \QVC]$ in \Cref{construction:IBCS} has the same completeness guarantee as the quantum interactive oracle proof $\QIOP$.

Moreover, when we instantiate the quantum state vector commitment scheme with $\QSTC$ from \Cref{construction:quantum-state-vector-commitment}, the efficiency measures of the interactive argument $(\ArgProver, \ArgVerifier) \coloneq \IBCS[\QIOP, \QSTC]$ satisfy the following:
\begin{itemize}[noitemsep]
\item The \emph{round complexity}
$\IBCSRoundComplexity = 4\sum_{i \in [\QIOPRoundComplexity]}\QIOPQueryDepthI{i} + 2\QIOPRoundComplexity - 1$.
\item The \emph{prover-to-verifier communication}
$\IBCSProverMsgSize = O(\RandomOracleOutputLength\QIOPRoundComplexity + \sum_{i \in [\QIOPRoundComplexity]}\QIOPQueryDepthI{i} \cdot \RandomOracleOutputLength \cdot \QIOPQueryWidthI{i} \cdot \log \QIOPProofMaxSize) = O(\Security\QIOPRoundComplexity + \Security \QIOPQueryComplexity \log \QIOPProofMaxSize)$.
\item The \emph{verifier-to-prover communication}
$\IBCSVerifierMsgSize = \QIOPVerifierMsgSize + \IBCSProverMsgSize = O(\QIOPVerifierMsgSize + \Security\QIOPRoundComplexity + \Security \QIOPQueryComplexity \log \QIOPProofMaxSize)$.
\item The \emph{verifier's query complexity} to the $i$-th random oracle
$\IBCSVerifierQueryComplexityI{i} = O(\QIOPQueryDepthI{i}\QIOPQueryWidthI{i} \log \QIOPProofSize{i})$ and the \emph{verifier's query complexity} $\IBCSVerifierQueryComplexity = O(\QIOPQueryComplexity \log \QIOPProofMaxSize)$.
\end{itemize}

Therefore, when we instantiate the quantum state vector commitment scheme with $\QSTC$, the resulting protocol $\IBCS[\QIOP, \QSTC]$ is a quantum-communication succinct interactive argument in the quantum random oracle model as long as the underlying $\QIOP$ is efficient.
\begin{corollary}
\label{cor:IBCS-MT-soundness}
Let $\QIOP$ be a quantum interactive oracle proof for relation $\Relation$ in the form of \Cref{sec:QIOP-def} with public-query soundness $\QIOPPublicQuerySoundness$. Let $\QSTC$ be the quantum state vector commitment scheme in \Cref{construction:quantum-state-vector-commitment}. Then the protocol $(\ArgProver, \ArgVerifier) = \IBCS[\QIOP, \QSTC]$ is a quantum-communication interactive argument for relation $\Relation$ with soundness
\[2\QIOPPublicQuerySoundness(\InstanceSize) + 2\QIOPRoundComplexity \cdot \sum_{i \in [\QIOPRoundComplexity]}\ExtractionErrorI{3}(\NumberOfQueries, \QIOPQueryDepth, \NumberOfQueries, 0, \Security, \QIOPProofSize{i})\enspace,\]
where $\ExtractionErrorI{3}(\NumberOfQueries, \NumberOfCommitmentsPhases, \QuantumTotalQueryMass_1, \QuantumTotalQueryMass_2, \RandomOracleOutputLength, \QVCMessageLength) \coloneq 2^{-\RandomOracleOutputLength + 22}\NumberOfCommitmentsPhases^2\QVCMessageLength^2(\NumberOfQueries + \NumberOfQueriesBound\NumberOfCommitmentsPhases)^2 \left(\QVCMessageLength^2+\QuantumTotalQueryMass_1+\QuantumTotalQueryMass_2 + \NumberOfCommitmentsPhases + 1\right)$.

Furthermore, if $\QIOP$ satisfies that $\QIOPPublicQuerySoundness(\InstanceSize) = \negl{\InstanceSize}$, $\sum_{i \in [\QIOPRoundComplexity]}\QIOPQueryDepthI{i}(\InstanceSize) = O(\poly(\log \InstanceSize))$, $\QIOPQueryComplexity(\InstanceSize) = O(\poly(\log \InstanceSize))$, $\QIOPProofSize{i}(\InstanceSize) = O(\poly(\InstanceSize))$, $\QIOPRoundComplexity(\InstanceSize) = O(\poly(\log\InstanceSize))$, and $\QIOPVerifierMsgSize(\InstanceSize) = O(\poly(\log\InstanceSize))$, then $(\ArgProver, \ArgVerifier) = \IBCS[\QIOP, \QSTC]$ is a quantum-communication succinct interactive argument for the same relation $\Relation$ with round complexity $O(\poly (\log \InstanceSize))$, total communication complexity $O(\poly (\log \InstanceSize, \Security))$, and soundness error $\negl{\InstanceSize} + 2^{-\Security} \NumberOfQueries^3 \poly(\InstanceSize)$ against $\NumberOfQueries$-query quantum adversaries. $(\ArgProver, \ArgVerifier)$ also has the same completeness guarantee as the quantum interactive oracle proof $\QIOP$.
\end{corollary}

\begin{proof}
\Cref{cor:IBCS-MT-soundness} follows from \Cref{thm:quantum-IBCS-soundness}, \Cref{thm:vc-extractability}, and the efficiency of the quantum-state vector commitment scheme $\QSTC$ from \Cref{construction:quantum-state-vector-commitment}.
\end{proof}

\subsection{The malicious QIOP prover}

We construct a malicious QIOP prover $\QIOPAdv$ to attack the public-query soundness of the underlying quantum interactive oracle proof system $\QIOP$ based on a malicious prover $\ArgAdv$ for the argument $(\ArgProver, \ArgVerifier) = \IBCS[\QIOP, \QVC]$. Here $\ArgAdv$ only has oracle access to $\QIOPRoundComplexity$ random functions $\{\RandomOracle{1}_i\}_{i \in [\QIOPRoundComplexity]}$.

\begin{construction}
\label{construction:QIOPAdv}
By \Cref{def:qvc_extractability}, $\QVC$ has a quantum extractor $\Extractor$ with query access to the quantum simulator $\Simulator$ and $\ExtractOracle$.
We construct $\QIOPAdv(\ArgAdv)$ as follows.
\begin{enumerate}[noitemsep]
\item[] $\QIOPAdv(\ArgAdv)$:
\begin{enumerate}[nolistsep]
\item Initialize the internal state register $\StateRegister{1}_i$ for the simulator $\Simulator$ for the random oracle in $\QVC$: for $i \in [\QIOPRoundComplexity]$, $\StateRegister{1}_i \gets \ket{\bot}$.
\item Set $\QIOPNotReturnIdxSet{1} \coloneq \{1\}, \QIOPReturnIdxSet{1} \coloneq \emptyset$.
\item Simulate $\ArgAdv$ by answering the random oracle queries with $\Simulator$ and the internal state registers $(\StateRegister{1}_i)_{i \in [\QIOPRoundComplexity]}$ until $\ArgAdv$ outputs an instance $\Instance$. Send the instance $\Instance$ to the QIOP verifier $\QIOPVerifier$.
\item For $i \in [\QIOPRoundComplexity]$:
\begin{enumerate}[nolistsep]
\item Simulate $\ArgAdv$ by answering the random oracle queries with $\Simulator$ and the internal state registers $(\StateRegister{1}_i)_{i \in [\QIOPRoundComplexity]}$ until $\ArgAdv$ outputs a commitment $\CommitmentRegister{1}_i$.
\item Run the extractor $\Extractor$ on $\CommitmentRegister{1}_i$ to get the $i$-th prover's message: $(\ProverMessageRegister{i}, \AltOpeningRegister{1}_i) \gets \Extractor.\ExtractMessage^{\Simulator(\StateRegister{1}_i), \ExtractOracle(\StateRegister{1}_i)}(\CommitmentRegister{1}_i)$.
\item Send the $i$-th prover's message $\ProverMessageRegister{i}$ to the trusted third party $\OracleParty$ that implements the queries to the prover's messages for $\QIOPVerifier$.
\item For $q \in [\QIOPQueryDepthI{i}]$:
\begin{enumerate}
\item On receiving the leaked query locations $(\QueryLocationRegister_{\iota})_{\iota \in [\QIOPQueryWidthI{i}]}$ from $\QIOPVerifier$, send $(\QueryLocationRegister_{\iota})_{\iota \in [\QIOPQueryWidthI{i}]}$ to $\ArgAdv$.
\item Simulate $\ArgAdv$ by answering the random oracle queries with $\Simulator$ and the internal state registers $(\StateRegister{1}_i)_{i \in [\QIOPRoundComplexity]}$ until $\ArgAdv$ outputs $((\QueryLocationRegister_{\iota})_{\iota \in [\QIOPQueryWidthI{i}]}, (\OpeningRegister{1}_{i'})_{i' \in \QIOPNotReturnIdxSet{i}})$.
\item Check the opening:

$(\ValidityBit_i, (\QueryLocationRegister_{\iota})_{\iota \in [\QIOPQueryWidthI{i}]}, (\OpeningRegister{1}_{i'}, \AltOpeningRegister{1}_{i'})_{i' \in \QIOPNotReturnIdxSet{i}}) \gets \Extractor.\AltCheck((\QueryLocationRegister_{\iota})_{\iota \in [\QIOPQueryWidthI{i}]}, (\OpeningRegister{1}_{i'}, \AltOpeningRegister{1}_{i'})_{i' \in \QIOPNotReturnIdxSet{i}})$.
\item Abort if $\ValidityBit_i = 0$.
\item Send $(\QueryLocationRegister_{\iota})_{\iota \in [\QIOPQueryWidthI{i}]}$ to the QIOP verifier $\QIOPVerifier$, who asks the trusted third party $\OracleParty$ to implement the query with $(\QueryLocationRegister_{\iota})_{\iota \in [\QIOPQueryWidthI{i}]}$.
\item On receiving $(\QueryLocationRegister_{\iota})_{\iota \in [\QIOPQueryWidthI{i}]}$ from the QIOP verifier $\QIOPVerifier$, send $((\QueryLocationRegister_{\iota})_{\iota \in [\QIOPQueryWidthI{i}]}, (\OpeningRegister{1}_{i'})_{i' \in \QIOPNotReturnIdxSet{i}})$ to $\ArgAdv$.
\item Simulate $\ArgAdv$ by answering the random oracle queries with $\Simulator$ and the internal state registers $(\StateRegister{1}_i)_{i \in [\QIOPRoundComplexity]}$ until $\ArgAdv$ outputs $(\QueryLocationRegister_{\iota})_{\iota \in [\QIOPQueryWidthI{i}]}$.
\item Send $(\QueryLocationRegister_{\iota})_{\iota \in [\QIOPQueryWidthI{i}]}$ to the QIOP verifier $\QIOPVerifier$.
\end{enumerate}
\item If $i \neq \QIOPRoundComplexity$:
\begin{enumerate}[nolistsep]
\item On receiving the returned index set and the $i$-th verifier's message $(\QIOPReturnIdxSet{i + 1}, \VerifierMessageRegister{i})$ from $\QIOPVerifier$, and the corresponding prover's messages $(\ProverMessageRegister{i'})_{i' \in \QIOPReturnIdxSet{i + 1}}$ from the trusted third party $\OracleParty$, compute the corresponding commitments in $\QIOPReturnIdxSet{i + 1}$ with the extractor $\Extractor$: for $i' \in \QIOPReturnIdxSet{i + 1}$,
$\CommitmentRegister{1}_{i'} \gets \Extractor.\AltCommit^{\Simulator(\StateRegister{1}_{i'}), \ExtractOracle(\StateRegister{1}_{i'})}(\AltOpeningRegister{1}_{i'}, \ProverMessageRegister{i'})$.
\item Send $(\QIOPReturnIdxSet{i + 1}, \VerifierMessageRegister{i}, (\CommitmentRegister{1}_{i'})_{i' \in \QIOPReturnIdxSet{i + 1}})$ to $\ArgAdv$.
\item Set $\QIOPNotReturnIdxSet{i + 1} \coloneq (\QIOPNotReturnIdxSet{i} \cup \{i + 1\}) \setminus \QIOPReturnIdxSet{i + 1}$.
\end{enumerate}
\end{enumerate}
\end{enumerate}
\end{enumerate}
\end{construction}

\subsection{Proof of \Cref{thm:quantum-IBCS-soundness}}
\label{subsec:proof-of-IBCS-soundness}

We first show that in the public-query soundness game, the malicious QIOP prover $\QIOPAdv(\ArgAdv)$ in \Cref{construction:QIOPAdv} has a similar winning probability as the malicious argument prover $\ArgAdv$.

\begin{lemma}
\label{lemma:diff-in-succ-probability}
For every integer $\InstanceSize$, $\NumberOfQueries$, $\RandomOracleOutputLength$, and a $\NumberOfQueries$-query argument adversary $\ArgAdv$,
\begin{align*}
&\prob{
\begin{array}{l}
\abs{\Instance} \leq \InstanceSize \\
\land\, \Instance \notin \GetLanguage{\Relation}\\
\land\, b = 1
\end{array}
\;\middle\vert\;
\begin{array}{l}
\RandomOracle{1} \gets \UniformFrom{\RandomOracleOutputLength}\\
\Instance \gets \ArgAdv^{\RandomOracle{0}}\\
b \gets \langle \ArgAdv^{\RandomOracle{0}}, \ArgVerifier^{\RandomOracle{0}}(\Instance) \rangle
\end{array}}\\
&\leq
2 \cdot \prob{
\begin{array}{l}
\abs{\Instance} \leq \InstanceSize \\
\land\, \Instance \notin \GetLanguage{\Relation}\\
\land\, b = 1
\end{array}
\;\middle\vert\;
\begin{array}{l}
\Instance \gets \QIOPAdv(\ArgAdv)\\
b \gets \langle \QIOPAdv(\ArgAdv), \QIOPVerifier(\Instance) \rangle_{\pq}
\end{array}}
+ \ExtractionErrorI{1}(\NumberOfQueries + \IBCSVerifierQueryComplexity, \RandomOracleOutputLength)
+ 2\QIOPRoundComplexity \cdot \sum_{i \in [\QIOPRoundComplexity]}\ExtractionErrorI{3}(\NumberOfQueries, \QIOPQueryDepth, \NumberOfQueries, 0, \RandomOracleOutputLength, \QIOPProofSize{i})\enspace.
\end{align*}
\end{lemma}

We first show how \Cref{lemma:diff-in-succ-probability} implies \Cref{thm:quantum-IBCS-soundness}.

\begin{proof}[Proof of \Cref{thm:quantum-IBCS-soundness}]
\Cref{thm:quantum-IBCS-soundness} follows from the definition of public-query soundness (\Cref{def:pq-qiop}). Specifically, if $\QIOP$ has public-query soundness $\QIOPPublicQuerySoundness$, then
\[\prob{
\begin{array}{l}
\abs{\Instance} \leq \InstanceSize \\
\land\, \Instance \notin \GetLanguage{\Relation}\\
\land\, b = 1
\end{array}
\;\middle\vert\;
\begin{array}{l}
\Instance \gets \QIOPAdv(\ArgAdv)\\
b \gets \langle \QIOPAdv(\ArgAdv), \QIOPVerifier(\Instance) \rangle_{\pq}
\end{array}} \leq \QIOPPublicQuerySoundness(\InstanceSize)\enspace.\]
\end{proof}

Next we prove \Cref{lemma:diff-in-succ-probability} via two claims.

The first claim replaces the random oracle with the simulator $\Simulator$.
\iffull
\begin{claim}
\else
\begin{numberedclaim}
\fi
\label{claim:hybrid-0}
For every integer $\InstanceSize$, $\NumberOfQueries$, $\RandomOracleOutputLength$, and a $\NumberOfQueries$-query argument adversary $\ArgAdv$,
\begin{align*}
&\prob{
\begin{array}{l}
\abs{\Instance} \leq \InstanceSize \\
\land\, \Instance \notin \GetLanguage{\Relation}\\
\land\, b = 1
\end{array}
\;\middle\vert\;
\begin{array}{l}
\RandomOracle{1} \gets \UniformFrom{\RandomOracleOutputLength}\\
\Instance \gets \ArgAdv^{\RandomOracle{0}}\\
b \gets \langle \ArgAdv^{\RandomOracle{0}}, \ArgVerifier^{\RandomOracle{0}}(\Instance) \rangle
\end{array}}\\
&\leq
\prob{
\begin{array}{l}
\abs{\Instance} \leq \InstanceSize \\
\land\, \Instance \notin \GetLanguage{\Relation}\\
\land\, b = 1
\end{array}
\;\middle\vert\;
\begin{array}{l}
\StateRegister{1} \gets \ket{\bot}\\
\Instance \gets \ArgAdv^{\Simulator(\StateRegister{1})}\\
b \gets \langle \ArgAdv^{\Simulator(\StateRegister{1})}, \ArgVerifier^{\Simulator(\StateRegister{1})}(\Instance) \rangle
\end{array}}
+ \ExtractionErrorI{1}(\NumberOfQueries + \IBCSVerifierQueryComplexity, \RandomOracleOutputLength)\enspace.
\end{align*}
\iffull
\end{claim}
\else
\end{numberedclaim}
\fi

\begin{proof}
This follows from the definition of quantum simulator (\Cref{def:quantum_simulator}), and that the interaction between $\ArgAdv$ and $\ArgVerifier$, and the final check on $\Instance$ can be combined into a single distinguisher with $\NumberOfQueries + \IBCSVerifierQueryComplexity$ queries to the random oracle.
\end{proof}

The second claim replaces all the queries implemented by $\QVC.\Query$ to the $i$-th prover's message during the interaction between the malicious argument prover $\ArgAdv$ and the argument verifier $\ArgVerifier$ with queries implemented by extractors one by one for $i \in [\QIOPRoundComplexity]$. We first define the hybrid games.
\begin{itemize}
\item [] $\HybridI{i^*}(\ArgAdv)$:
\begin{enumerate}[noitemsep]
\item Initialize the internal state register $\StateRegister{1}_i$: for $i \in [\QIOPRoundComplexity]$, $\StateRegister{1}_i \gets \ket{\bot}$.
\item Initialize the internal state of $\ArgVerifier$, the register $\VerifierPrivateRegister{1}$, to all-zero states, set $\VerifierMessageRegister{0}$ to be the empty register, and set $\QIOPNotReturnIdxSet{1} \coloneq \{1\}, \QIOPReturnIdxSet{1} \coloneq \emptyset$.
\item Simulate $\ArgAdv$ by answering the random oracle queries with $\Simulator$ and the internal state registers $(\StateRegister{1}_i)_{i \in [\QIOPRoundComplexity]}$ until $\ArgAdv$ outputs an instance $\Instance$.
\item Simulate the interaction of $\ArgAdv$ and $\ArgVerifier$ for the first $i^*$ commitments of $\ArgAdv$ with $\QVC$.

For $i \in [i^*]$:
\begin{enumerate}[nolistsep]
\item Simulate $\ArgAdv$ by answering the random oracle queries with $\Simulator$ on the input $(\QIOPReturnIdxSet{i}, \VerifierMessageRegister{i - 1}, (\CommitmentRegister{1}_{i'})_{i' \in \QIOPReturnIdxSet{i}})$ until $\ArgAdv$ outputs a commitment $\CommitmentRegister{1}_i$.
\item Initialize $\ArgVerifier$'s registers $(\QueryLocationRegister_{\iota}, \AnswerRegister{1}_{\iota})_{\iota \in [\QIOPQueryWidthI{i}]}$ as the all-zero states and set $\VerifierPrivateRegister{i}^{(0)} \coloneq \VerifierPrivateRegister{i}$.
\item For $q \in [\QIOPQueryDepthI{i}]$:
\begin{enumerate}
\item Simulate $\ArgVerifier$ to get the query locations: $((\QueryLocationRegister_{\iota}, \AnswerRegister{1}_{\iota})_{\iota \in [\QIOPQueryWidthI{i}]}, \VerifierPrivateRegister{i}^{(q)}) \gets \QIOPVerifier_i^{(q - 1)}(\Instance, (\QueryLocationRegister_{\iota}, \AnswerRegister{1}_{\iota})_{\iota \in [\QIOPQueryWidthI{i}]}, \VerifierPrivateRegister{i}^{(q - 1)})$.
\item Simulate $\ArgAdv$ by answering the random oracle queries with $\Simulator$ on the leaked query locations $(\QueryLocationRegister_{\iota})_{\iota \in [\QIOPQueryWidthI{i}]}$ until $\ArgAdv$ outputs $((\QueryLocationRegister_{\iota})_{\iota \in [\QIOPQueryWidthI{i}]}, (\OpeningRegister{1}_{i'})_{i' \in \QIOPNotReturnIdxSet{i}})$.
\item Implement the access to the underlying message for $\ArgVerifier$:
$(\QVCValidityBit_{i, q}, (\QueryLocationRegister_{\iota}, \AnswerRegister{1}_{\iota})_{\iota \in [\QIOPQueryWidthI{i}]}, (\CommitmentRegister{1}_{i'}, \OpeningRegister{1}_{i'})_{i' \in \QIOPNotReturnIdxSet{i}}) \gets \QVCQuery^{(\Simulator(\StateRegister{1}_{i'}))_{i' \in \QIOPNotReturnIdxSet{i}}}(1^\Security, (\QueryLocationRegister_{\iota}, \AnswerRegister{1}_{\iota})_{\iota \in [\QIOPQueryWidthI{i}]}, (\CommitmentRegister{1}_{i'}, \OpeningRegister{1}_{i'})_{i' \in \QIOPNotReturnIdxSet{i}})$.
\item Output 0 and abort if $\QVCValidityBit_{i, q} = 0$.
\item Simulate $\ArgAdv$ by answering the random oracle queries with $\Simulator$ on the registers $((\QueryLocationRegister_{\iota})_{\iota \in [\QIOPQueryWidthI{i}]}, (\OpeningRegister{1}_{i'})_{i' \in \QIOPNotReturnIdxSet{i}})$ until $\ArgAdv$ outputs the updated $(\QueryLocationRegister_{\iota})_{\iota \in [\QIOPQueryWidthI{i}]}$.
\end{enumerate}
\item If $i \neq \QIOPRoundComplexity$:
\begin{enumerate}[nolistsep]
\item Compute the $i$-th $\ArgVerifier$'s message and the index set for the proofs to be returned: $(\QIOPReturnIdxSet{i + 1}, \VerifierMessageRegister{i}, \VerifierPrivateRegister{i + 1}) \gets \QIOPVerifier_{i}^{(\QIOPQueryDepthI{i})}(\Instance, (\QueryLocationRegister_{\iota}, \AnswerRegister{1}_{\iota})_{\iota \in [\QIOPQueryWidthI{i}]}, \VerifierPrivateRegister{i}^{(\QIOPQueryDepthI{i})})$.
\item Set $\QIOPNotReturnIdxSet{i + 1} \coloneq (\QIOPNotReturnIdxSet{i} \cup \{i + 1\}) \setminus \QIOPReturnIdxSet{i + 1}$.
\end{enumerate}
\item If $i = \QIOPRoundComplexity$, simulate $\ArgVerifier$ and check whether $\ArgAdv$ wins:
\begin{enumerate}
\item Compute $\ValidityBit_{\scriptscriptstyle{\QIOP}} \gets \QIOPVerifier_{\QIOPRoundComplexity}^{(\QIOPQueryDepthI{\QIOPRoundComplexity})}(\Instance, (\QueryLocationRegister_{\iota}, \AnswerRegister{1}_{\iota})_{\iota \in [\QIOPQueryWidthI{\QIOPRoundComplexity}]}, \VerifierPrivateRegister{\QIOPRoundComplexity}^{(\QIOPQueryDepthI{\QIOPRoundComplexity})})$.
\item Output 1 if $\ValidityBit_{\scriptscriptstyle{\QIOP}} = 1$, $\abs{\Instance} \leq \InstanceSize$, and $\Instance \notin \GetLanguage{\Relation}$; output 0 otherwise.
\end{enumerate}
\end{enumerate}
\item Simulate the interaction of $\QIOPAdv(\ArgAdv)$ and $\QIOPVerifier$ for the remaining $\QIOPRoundComplexity - i^*$ commitments of $\ArgAdv$ with the extractor.

For $i \in \{i^* + 1, i^* + 2, \cdots, \QIOPRoundComplexity\}$:
\begin{enumerate}[nolistsep]
\item Simulate $\ArgAdv$ by answering the random oracle queries with $\Simulator$ on the input $(\QIOPReturnIdxSet{i}, \VerifierMessageRegister{i - 1}, (\CommitmentRegister{1}_{i'})_{i' \in \QIOPReturnIdxSet{i}})$ until $\ArgAdv$ outputs a commitment $\CommitmentRegister{1}_i$.
\item Run the extractor $\Extractor$ on $\CommitmentRegister{1}_i$ to get the $i$-th prover's message: $(\ProverMessageRegister{i}, \AltOpeningRegister{1}_i) \gets \Extractor.\ExtractMessage^{\Simulator(\StateRegister{1}_i), \ExtractOracle(\StateRegister{1}_i)}(\CommitmentRegister{1}_i)$.
\item Initialize $\ArgVerifier$'s registers $(\QueryLocationRegister_{\iota}, \AnswerRegister{1}_{\iota})_{\iota \in [\QIOPQueryWidthI{i}]}$ as the all-zero states and set $\VerifierPrivateRegister{i}^{(0)} \coloneq \VerifierPrivateRegister{i}$.
\item For $q \in [\QIOPQueryDepthI{i}]$:
\begin{enumerate}
\item Simulate $\ArgVerifier$ to get the query locations: $((\QueryLocationRegister_{\iota}, \AnswerRegister{1}_{\iota})_{\iota \in [\QIOPQueryWidthI{i}]}, \VerifierPrivateRegister{i}^{(q)}) \gets \QIOPVerifier_i^{(q - 1)}(\Instance, (\QueryLocationRegister_{\iota}, \AnswerRegister{1}_{\iota})_{\iota \in [\QIOPQueryWidthI{i}]}, \VerifierPrivateRegister{i}^{(q - 1)})$.
\item Simulate $\ArgAdv$ by answering the random oracle queries with $\Simulator$ on the leaked query locations $(\QueryLocationRegister_{\iota})_{\iota \in [\QIOPQueryWidthI{i}]}$ until $\ArgAdv$ outputs $((\QueryLocationRegister_{\iota})_{\iota \in [\QIOPQueryWidthI{i}]}, (\OpeningRegister{1}_{i'})_{i' \in \QIOPNotReturnIdxSet{i}})$.
\item Simulate the trusted third party $\OracleParty$ that implements the access to prover's message registers: for $i' \in \QIOPNotReturnIdxSet{i} \cap \{i^* + 1, \cdots, \QIOPRoundComplexity\}$, $((\QueryLocationRegister_{\iota}, \AnswerRegister{1}_{\iota})_{\iota \in [\QIOPQueryWidthI{i}]},\ProverMessageRegister{i'}) \gets \OracleParty((\QueryLocationRegister_{\iota}, \AnswerRegister{1}_{\iota})_{\iota \in [\QIOPQueryWidthI{i}]},\ProverMessageRegister{i'})$.
\item Check the opening: for $i' \in \QIOPNotReturnIdxSet{i} \cap \{i^* + 1, \cdots, \QIOPRoundComplexity\}$, $(\ValidityBit_{i, q, i'}, (\QueryLocationRegister_{\iota})_{\iota \in [\QIOPQueryWidthI{i}]}, (\OpeningRegister{1}_{i'}, \AltOpeningRegister{1}_{i'})) \gets \Extractor.\AltCheck((\QueryLocationRegister_{\iota})_{\iota \in [\QIOPQueryWidthI{i}]}, (\OpeningRegister{1}_{i'}, \AltOpeningRegister{1}_{i'}))$.
\item Output 0 and abort if $\land_{i' \in \QIOPNotReturnIdxSet{i} \cap \{i^* + 1, \cdots, \QIOPRoundComplexity\}}\ValidityBit_{i, q, i'} = 0$.
\item Implement the access to the first $i^*$ messages for $\ArgVerifier$:
$(\QVCValidityBit_{i, q}, (\QueryLocationRegister_{\iota}, \AnswerRegister{1}_{\iota})_{\iota \in [\QIOPQueryWidthI{i}]}, (\CommitmentRegister{1}_{i'}, \OpeningRegister{1}_{i'})_{i' \in \QIOPNotReturnIdxSet{i} \cap [i^*]}) \gets \QVCQuery^{(\Simulator(\StateRegister{1}_{i'}))_{i' \in \QIOPNotReturnIdxSet{i} \cap [i^*]}}(1^\Security, (\QueryLocationRegister_{\iota}, \AnswerRegister{1}_{\iota})_{\iota \in [\QIOPQueryWidthI{i}]}, (\CommitmentRegister{1}_{i'}, \OpeningRegister{1}_{i'})_{i' \in \QIOPNotReturnIdxSet{i} \cap [i^*]})$.
\item Output 0 and abort if $\QVCValidityBit_{i, q} = 0$.
\item Simulate $\ArgAdv$ by answering the random oracle queries with $\Simulator$ on the registers $((\QueryLocationRegister_{\iota})_{\iota \in [\QIOPQueryWidthI{i}]}, (\OpeningRegister{1}_{i'})_{i' \in \QIOPNotReturnIdxSet{i}})$ until $\ArgAdv$ outputs the updated $(\QueryLocationRegister_{\iota})_{\iota \in [\QIOPQueryWidthI{i}]}$.
\end{enumerate}
\item If $i \neq \QIOPRoundComplexity$:
\begin{enumerate}[nolistsep]
\item Compute the $i$-th $\ArgVerifier$'s message and the index set for the proofs to be returned:\\ $(\QIOPReturnIdxSet{i + 1}, \VerifierMessageRegister{i}, \VerifierPrivateRegister{i + 1}) \gets \QIOPVerifier_{i}^{(\QIOPQueryDepthI{i})}(\Instance, (\QueryLocationRegister_{\iota}, \AnswerRegister{1}_{\iota})_{\iota \in [\QIOPQueryWidthI{i}]}, \VerifierPrivateRegister{i}^{(\QIOPQueryDepthI{i})})$.
\item Compute the corresponding commitments in $\QIOPReturnIdxSet{i + 1}$ with the extractor $\Extractor$:\\ for $i' \in \QIOPReturnIdxSet{i + 1} \cap \{i^* + 1, \cdots, \QIOPRoundComplexity\}$, $\CommitmentRegister{1}_{i'} \gets \Extractor.\AltCommit^{\Simulator(\StateRegister{1}_{i'}), \ExtractOracle(\StateRegister{1}_{i'})}(\AltOpeningRegister{1}_{i'}, \ProverMessageRegister{i'})$.
\item Set $\QIOPNotReturnIdxSet{i + 1} \coloneq (\QIOPNotReturnIdxSet{i} \cup \{i + 1\}) \setminus \QIOPReturnIdxSet{i + 1}$.
\end{enumerate}
\item If $i = \QIOPRoundComplexity$, simulate $\ArgVerifier$ and check whether $\ArgAdv$ wins:
\begin{enumerate}
\item Compute $\ValidityBit_{\scriptscriptstyle{\QIOP}} \gets \QIOPVerifier_{\QIOPRoundComplexity}^{(\QIOPQueryDepthI{\QIOPRoundComplexity})}(\Instance, (\QueryLocationRegister_{\iota}, \AnswerRegister{1}_{\iota})_{\iota \in [\QIOPQueryWidthI{\QIOPRoundComplexity}]}, \VerifierPrivateRegister{\QIOPRoundComplexity}^{(\QIOPQueryDepthI{\QIOPRoundComplexity})})$.
\item Output 1 if $\ValidityBit_{\scriptscriptstyle{\QIOP}} = 1$, $\abs{\Instance} \leq \InstanceSize$, and $\Instance \notin \GetLanguage{\Relation}$; output 0 otherwise.
\end{enumerate}
\end{enumerate}
\end{enumerate}
\end{itemize}

Then by \Cref{construction:QIOPAdv} and the construction of hybrid games $\HybridI{i^*}(\ArgAdv)$,
\begin{align}
\label{eqn:hybrid-1}
\prob{\HybridI{\QIOPRoundComplexity}(\ArgAdv)} = \prob{
\begin{array}{l}
\abs{\Instance} \leq \InstanceSize \\
\land\, \Instance \notin \GetLanguage{\Relation}\\
\land\, b = 1
\end{array}
\;\middle\vert\;
\begin{array}{l}
\StateRegister{1} \gets \ket{\bot}\\
\Instance \gets \ArgAdv^{\Simulator(\StateRegister{1})}\\
b \gets \langle \ArgAdv^{\Simulator(\StateRegister{1})}, \ArgVerifier^{\Simulator(\StateRegister{1})}(\Instance) \rangle
\end{array}}\enspace,
\end{align}
and
\begin{align}
\label{eqn:hybrid-2}
\prob{\HybridI{0}(\ArgAdv)} = \prob{
\begin{array}{l}
\abs{\Instance} \leq \InstanceSize \\
\land\, \Instance \notin \GetLanguage{\Relation}\\
\land\, b = 1
\end{array}
\;\middle\vert\;
\begin{array}{l}
\Instance \gets \QIOPAdv(\ArgAdv)\\
b \gets \langle \QIOPAdv(\ArgAdv), \QIOPVerifier(\Instance) \rangle_{\pq}
\end{array}}\enspace.
\end{align}

Furthermore, the only difference between two consecutive hybrid games is whether the message underlying a commitment is accessed through the $\QVC$ interfaces or through the extractor, and thus can be bounded by the extractability of the quantum-state vector commitment scheme $\QVC$.
\iffull
\begin{claim}
\else
\begin{numberedclaim}
\fi
\label{claim:hybrid-3}
For every integer $\InstanceSize$, $\NumberOfQueries$, $\RandomOracleOutputLength$, $i^* \in [\QIOPRoundComplexity]$, and a $\NumberOfQueries$-query argument adversary $\ArgAdv$,
\[\abs{\sqrt{\prob{\HybridI{i^*}(\ArgAdv)}} - \sqrt{\prob{\HybridI{i^* - 1}(\ArgAdv)}}}^2 \leq \ExtractionErrorI{3}(\NumberOfQueries, \QIOPQueryDepth, \NumberOfQueries, 0, \RandomOracleOutputLength, \QIOPProofSize{i^*})\enspace.\]
\iffull
\end{claim}
\else
\end{numberedclaim}
\fi

\begin{proof}
The only difference between games $\HybridI{i^*}(\ArgAdv)$ and $\HybridI{i^* - 1}(\ArgAdv)$ is how they handle the $i^*$-th commitment. In the game $\HybridI{i^*}(\ArgAdv)$, all the queries of the underlying QIOP verifier $\QIOPVerifier$ to the $i^*$-th prover message are implemented by the interfaces of $\QVC$, while in the game $\HybridI{i^* - 1}(\ArgAdv)$, all the queries of the underlying QIOP verifier $\QIOPVerifier$ to the $i^*$-th prover message are implemented by the interfaces of the extractor.

Consider an adversary $\Adversary$ of $\QVC$ that does the following:
\begin{enumerate}[noitemsep]
\item Simulate $\HybridI{i^* - 1}(\ArgAdv)$ where the queries to $\Simulator(\StateRegister{1}_{i^*})$ are forwarded to the oracles (instead of implementing the access on its own) until the $i^*$-th commitment $\CommitmentRegister{1}_{i^*}$ is outputted.
\item Send $\CommitmentRegister{1}_{i^*}$ to the game as the commitment, and keep all other registers in the register $\EnvironmentRegister{1}$ as the internal state.
\item For $i \in \{i^*, \cdots, \QIOPRoundComplexity\}$:
\begin{enumerate}[nolistsep]
\item For $q \in [\QIOPQueryDepthI{i}]$:
\begin{enumerate}[nolistsep]
\item Continue to simulate $\HybridI{i^* - 1}(\ArgAdv)$ until queries to the $i^*$-th prover's message register are about to be implemented.
\item Compute the query set to the $i^*$-th prover's message and denote it as $\QVCQuerySetRegister{1}$, denote the corresponding answer register for these queries as $\AnswerRegister{1}$, and denote the opening register for these queries as $\OpeningRegister{1}_{i^*}$. Send $(\QVCQuerySetRegister{1},\OpeningRegister{1}_{i^*}, \AnswerRegister{1})$ to the game, and keep all other registers in the register $\EnvironmentRegister{1}$ as the internal state.
\end{enumerate}
\item Continue to simulate $\HybridI{i^* - 1}(\ArgAdv)$ until $\QIOPVerifier$ outputs $\QIOPReturnIdxSet{i + 1}$ or $\QIOPVerifier$ outputs the decision bit.
\item If $i^* \in \QIOPReturnIdxSet{i + 1}$ or $\QIOPVerifier$ outputs the decision bit, inform the game that all queries to the underlying message are done.
\end{enumerate}
\end{enumerate}

Consider a distinguisher $\Distinguisher$ that continues to simulate $\HybridI{i^* - 1}(\ArgAdv)$ from where the adversary $\Adversary$ stops.

Then the game $\HybridI{i^*}(\ArgAdv)$ is exactly $\QVCSimulateWorld(\Adversary, \Distinguisher)$, and the game $\HybridI{i^* - 1}(\ArgAdv)$ is exactly $\QVCExtractWorld(\Adversary, \Distinguisher)$. By the extractability of $\QVC$,
\begin{align*}
&\abs{\sqrt{\prob{\HybridI{i^*}(\ArgAdv)}} - \sqrt{\prob{\HybridI{i^* - 1}(\ArgAdv)}}}^2\\
&=\abs{\sqrt{\prob{\QVCExtractWorld(\Adversary, \Distinguisher)}} - \sqrt{\prob{\QVCSimulateWorld(\Adversary, \Distinguisher)}}}^2\\
&\leq \ExtractionErrorI{3}(\NumberOfQueries, \QIOPQueryDepth, \NumberOfQueries, 0, \RandomOracleOutputLength, \QIOPProofSize{i^*})\enspace,
\end{align*}
where we use the fact that the query depth of $\Adversary$ to the underlying message is at most $\sum_{i \geq i^*}\QIOPQueryDepthI{i} \leq \QIOPQueryDepth$, and moreover, $\Adversary$ has at most $\NumberOfQueries$ queries and at most query mass $\NumberOfQueries$ to $\Simulator(\StateRegister{1}_{i^*})$ since simulating $\QIOPVerifier$ does not require query access to $\Simulator(\StateRegister{1}_{i^*})$, and the queries of $\QVC$ and the extractor to $\Simulator(\StateRegister{1}_{i^*})$ and $\ExtractOracle(\StateRegister{1}_{i^*})$ are all performed by the game.
\end{proof}

We combine \Cref{claim:hybrid-0,claim:hybrid-3} to prove \Cref{lemma:diff-in-succ-probability}.
\begin{proof}[Proof of \Cref{lemma:diff-in-succ-probability}]
By \Cref{claim:hybrid-3,eqn:hybrid-1,eqn:hybrid-2} and the triangle inequality,
\begin{align*}
&\abs{
\sqrt{\prob{
\begin{array}{l}
\abs{\Instance} \leq \InstanceSize \\
\land\, \Instance \notin \GetLanguage{\Relation}\\
\land\, b = 1
\end{array}
\;\middle\vert\;
\begin{array}{l}
\StateRegister{1} \gets \ket{\bot}\\
\Instance \gets \ArgAdv^{\Simulator(\StateRegister{1})}\\
b \gets \langle \ArgAdv^{\Simulator(\StateRegister{1})}, \ArgVerifier^{\Simulator(\StateRegister{1})}(\Instance) \rangle
\end{array}}}
-
\sqrt{\prob{
\begin{array}{l}
\abs{\Instance} \leq \InstanceSize \\
\land\, \Instance \notin \GetLanguage{\Relation}\\
\land\, b = 1
\end{array}
\;\middle\vert\;
\begin{array}{l}
\Instance \gets \QIOPAdv(\ArgAdv)\\
b \gets \langle \QIOPAdv(\ArgAdv), \QIOPVerifier(\Instance) \rangle_{\pq}
\end{array}}
}
}^2\\
&\leq \left(\sum_{i \in [\QIOPRoundComplexity]}\sqrt{\ExtractionErrorI{3}(\NumberOfQueries, \QIOPQueryDepth, \NumberOfQueries, 0, \RandomOracleOutputLength, \QIOPProofSize{i})}\right)^2\\
&\leq \QIOPRoundComplexity \cdot \sum_{i \in [\QIOPRoundComplexity]}\ExtractionErrorI{3}(\NumberOfQueries, \QIOPQueryDepth, \NumberOfQueries, 0, \RandomOracleOutputLength, \QIOPProofSize{i})\enspace. \tag{By Cauchy--Schwarz inequality}
\end{align*}

Therefore,
\begin{align*}
&\prob{
\begin{array}{l}
\abs{\Instance} \leq \InstanceSize \\
\land\, \Instance \notin \GetLanguage{\Relation}\\
\land\, b = 1
\end{array}
\;\middle\vert\;
\begin{array}{l}
\RandomOracle{1} \gets \UniformFrom{\RandomOracleOutputLength}\\
\Instance \gets \ArgAdv^{\RandomOracle{0}}\\
b \gets \langle \ArgAdv^{\RandomOracle{0}}, \ArgVerifier^{\RandomOracle{0}}(\Instance) \rangle
\end{array}}\\
&\leq
\prob{
\begin{array}{l}
\abs{\Instance} \leq \InstanceSize \\
\land\, \Instance \notin \GetLanguage{\Relation}\\
\land\, b = 1
\end{array}
\;\middle\vert\;
\begin{array}{l}
\StateRegister{1} \gets \ket{\bot}\\
\Instance \gets \ArgAdv^{\Simulator(\StateRegister{1})}\\
b \gets \langle \ArgAdv^{\Simulator(\StateRegister{1})}, \ArgVerifier^{\Simulator(\StateRegister{1})}(\Instance) \rangle
\end{array}}
+ \ExtractionErrorI{1}(\NumberOfQueries + \IBCSVerifierQueryComplexity, \RandomOracleOutputLength) \tag{By \Cref{claim:hybrid-0}}\\
&\leq \left(\sqrt{\prob{
\begin{array}{l}
\abs{\Instance} \leq \InstanceSize \\
\land\, \Instance \notin \GetLanguage{\Relation}\\
\land\, b = 1
\end{array}
\;\middle\vert\;
\begin{array}{l}
\Instance \gets \QIOPAdv(\ArgAdv)\\
b \gets \langle \QIOPAdv(\ArgAdv), \QIOPVerifier(\Instance) \rangle_{\pq}
\end{array}}} + \sqrt{\QIOPRoundComplexity \cdot \sum_{i \in [\QIOPRoundComplexity]}\ExtractionErrorI{3}(\NumberOfQueries, \QIOPQueryDepth, \NumberOfQueries, 0, \RandomOracleOutputLength, \QIOPProofSize{i})}\right)^2 + \ExtractionErrorI{1}(\NumberOfQueries + \IBCSVerifierQueryComplexity, \RandomOracleOutputLength)\\
&\leq 2\prob{
\begin{array}{l}
\abs{\Instance} \leq \InstanceSize \\
\land\, \Instance \notin \GetLanguage{\Relation}\\
\land\, b = 1
\end{array}
\;\middle\vert\;
\begin{array}{l}
\Instance \gets \QIOPAdv(\ArgAdv)\\
b \gets \langle \QIOPAdv(\ArgAdv), \QIOPVerifier(\Instance) \rangle_{\pq}
\end{array}} + 2\QIOPRoundComplexity \cdot \sum_{i \in [\QIOPRoundComplexity]}\ExtractionErrorI{3}(\NumberOfQueries, \QIOPQueryDepth, \NumberOfQueries, 0, \RandomOracleOutputLength, \QIOPProofSize{i}) + \ExtractionErrorI{1}(\NumberOfQueries + \IBCSVerifierQueryComplexity, \RandomOracleOutputLength)\enspace. \tag{By Cauchy--Schwarz inequality}
\end{align*}
\end{proof}

\doclearpage

\appendix

\section{The QIOP that we use}\label{sec:QIOP-we-use}

We use the following QIOP in our transformation.

\begin{lemma}\label{lemma:QIOP-that-we-use}
There exists a QIOP for $\QMA$ with total proof length $\QIOPProofTotalSize = \NumberOfRepetition \cdot \poly(\InstanceSize)$, round complexity $\QIOPRoundComplexity = \NumberOfRepetition \cdot \poly(\log \InstanceSize)$, query complexity $\QIOPQueryComplexity = \NumberOfRepetition \cdot \poly(\log \InstanceSize)$, verifier-to-prover communication complexity $\QIOPVerifierMsgSize = \NumberOfRepetition \cdot \poly(\log \InstanceSize)$, completeness $\QIOPCompleteness = 1 - \NumberOfRepetition \cdot \negl{\InstanceSize}$, and public-query soundness $\QIOPPublicQuerySoundness = 3^{-\NumberOfRepetition}$ for instance size $\InstanceSize$ and positive integer $\NumberOfRepetition$.

In particular, taking $\NumberOfRepetition = \lceil (\log \InstanceSize)^2 \rceil$ yields a QIOP with polynomial total proof length, polylogarithmic round complexity, polylogarithmic query complexity, polylogarithmic verifier-to-prover communication, completeness negligibly close to $1$, and negligible public-query soundness.
\end{lemma}

\begin{proof}[Proof sketch]
We show how to construct a \emph{public-coin} QIOP for $\QMA$ with $\QIOPProofTotalSize = \poly(\InstanceSize)$, $\QIOPRoundComplexity = \poly(\log \InstanceSize)$, $\QIOPQueryComplexity = \poly(\log \InstanceSize)$, $\QIOPVerifierMsgSize = \poly(\log \InstanceSize)$,  $\QIOPCompleteness = 1 - \negl{\InstanceSize}$, and $\QIOPSoundness = \frac{1}{3}$ for instance size $\InstanceSize$. The lemma then follows by repeating this QIOP sequentially $\NumberOfRepetition$ times, accepting only if all executions accept, and applying \Cref{lemma:public-coin-public-query}.

Our construction is obtained by making three minor modifications to the protocol of \cite{SV26}. This protocol satisfies all our requirements except that it allows the verifier to return only part of the prover's first message register, has verifier-to-prover communication complexity $\QIOPVerifierMsgSize = \poly(\InstanceSize)$, and is written as a private-coin protocol.

\parhead{Splitting the prover's message}
The verifier in \cite{SV26} partitions the prover's first message register into blocks and returns one block in each round. This partition is fixed in advance and publicly known. Therefore, we can modify the protocol by having the prover send these blocks in separate rounds, with the verifier sending dummy messages between consecutive prover messages. Then each return in the modified protocol consists of an entire prover's message register from a single round.
	
\parhead{Reducing communication}
Returning prover's message registers do not contribute to $\QIOPVerifierMsgSize$. The only remaining long message is the final classical challenge specifying the Hamiltonian term used in the verification. As observed in Section~1.2.1 of \cite{SV26}, this communication can be reduced to $O(\log \InstanceSize)$ bits by using the derandomized amplification technique of \cite{BMVZ26}.
	
\parhead{Public-coin}
The measurement basis and Hamiltonian-term challenges are already public in \cite{SV26}. However, the choice of test branch and round, the sampled code stabilizers, and the PCPP verifier's randomness are private in their formulation.

We modify the protocol to make it public-coin as follows. Instead of choosing the test branch and round in advance, the verifier samples fresh public coins at each checkpoint to decide whether to continue or test and terminate, using conditional stopping probabilities that preserve the original distribution over tests. Each decision is made after all registers needed for the selected test have been received. The verifier then publicly samples the randomness needed for that test (e.g., the code stabilizers and PCPP randomness), performs the queries, and immediately accepts or rejects, without using any subsequent prover message. Thus, revealing the test randomness gives the prover neither an opportunity to alter the registers being tested nor an advantage in preparing subsequent proofs. The original soundness bound therefore continues to apply.
\end{proof}


\clearpage
\iffull
\section*{Acknowledgments}
\label{sec:acknowledgements}

The authors are supported in part by the Ethereum Foundation and the Global Chinese Community of Universal Digital Commons. The authors had access to OpenAI models through the ChatGPT for Academic Researchers program.

\fi

\section*{AI disclosure}

The authors used GPT-6 Astra to assist with identifying related work, improving the writing, and reviewing the proofs. The authors developed the proof ideas and wrote the final proofs, and take full responsibility for the paper.

\iffull
\printbibliography
\fi

\end{document}